\documentclass[a4paper,11pt]{article} 

\usepackage[T1]{fontenc}
\usepackage[top=1in,left=1in,right=1in,bottom=1in]{geometry}
\usepackage{ragged2e}
\usepackage{lmodern}
\usepackage{multicol}
\usepackage{amsmath}
\usepackage{amsthm}
\usepackage{amsfonts}
\usepackage{thmtools}
\usepackage{mathtools}
\usepackage{amssymb}
\usepackage{bm}
\usepackage{dsfont}
\usepackage{enumitem}
\usepackage{graphicx}
\usepackage[dvipsnames]{xcolor}
\usepackage[bottom]{footmisc}
\usepackage{hyperref}
\usepackage[nodisplayskipstretch]{setspace}
\usepackage{subcaption}
\usepackage{algorithm}
\usepackage[beginLComment=//~,endLComment=~]{algpseudocodex}
\usepackage{tabularx}
\usepackage{booktabs}
\usepackage{multirow}
\usepackage{makecell}
\usepackage{xurl}
\usepackage{wasysym}
\usepackage{caption}
\usepackage{orcidlink} 

\allowdisplaybreaks

\newcommand{\transp}{^\top}
\renewcommand{\P}{\mathbb{P}}
\newcommand{\E}{\mathbb{E}}
\newcommand{\Var}{\textnormal{Var}}

\newcommand{\Int}{\mathbb{Z}}
\newcommand{\Real}{\mathbb{R}}

\newcommand{\N}{\mathbb{N}}

\newcommand{\1}[1]{\mathds{1}\{#1\}}

\newcommand{\abs}[1]{\lvert#1\rvert}
\newcommand{\Abs}[1]{\left\lvert\,#1\,\right\rvert}

\newcommand{\tempo}[2]{{#1}\vert{#2}}
\let\otextsc\textsc
\renewcommand{\textsc}[1]{\textnormal{\otextsc{#1}}} %
\newcommand{\ul}[1]{\underline{#1}}

\newcommand{\vertiii}[1]{{%
    \left\vert\kern-0.25ex\left\vert\kern-0.25ex\left\vert #1 
    \right\vert\kern-0.25ex\right\vert\kern-0.25ex\right\vert
}}

\definecolor{MATblue}{HTML}{0072BD}
\definecolor{MATgreen}{HTML}{77AC30}
\definecolor{MATorange}{HTML}{D95319}
\definecolor{MATred}{HTML}{A2142F}
\definecolor{MATyellow}{HTML}{EDB120}
\definecolor{MATcyan}{HTML}{4DBEEE}
\definecolor{MATpurple}{HTML}{7E2F8E}

\hypersetup{
	colorlinks=true,
	linkcolor=MATred,
	citecolor=MATblue,
	urlcolor=MATblue
}

\newcommand{\raisedrule}[2][0em]{\leavevmode\leaders\hbox{\rule[#1]{1pt}{#2}}\hfill\kern0pt}

\theoremstyle{plain}
\declaretheorem[name=Definition,numberwithin=section]{definition}
\theoremstyle{plain}
\declaretheorem[name=Theorem,numberwithin=section]{theorem}
\declaretheorem[name=Corollary,numberwithin=section]{corollary}
\declaretheorem[name=Proposition,numberwithin=section]{proposition}
\declaretheorem[name=Lemma,numberwithin=section]{lemma}
\makeatletter
\renewenvironment{proof}[1][\proofname]{\par
  \pushQED{\qed}%
  \normalfont \topsep6\p@\@plus6\p@\relax
  \trivlist
  \item[\hskip\labelsep\bfseries
    #1\@addpunct{.}]\ignorespaces
}{%
  \popQED\endtrivlist\@endpefalse
}
\makeatother
\theoremstyle{definition}
\declaretheorem[name=Assumption]{assumption}

\declaretheorem[name=Remark,numberwithin=section]{remark}

\usepackage{natbib}
\definecolor{red}{rgb}{0, 0, 0}
\definecolor{brown}{rgb}{0, 0, 0}
\definecolor{green}{rgb}{0, 0, 0}

\begin{document}

\title{\vspace{-2em}\normalfont\linespread{1.2}\selectfont%
    From Many Models, One: Macroeconomic Forecasting with Reservoir Ensembles
}

\author{%
    Giovanni Ballarin\,\orcidlink{0000-0003-0312-1169}%
    \thanks{\linespread{1.2}\selectfont%
    Corresponding author -- University of St.~Gallen, Division of Mathematics and Statistics, Rosenbergstrasse 22, 9000 St.~Gallen, Switzerland. {\texttt{Giovanni.Ballarin@unisg.ch}}}
    \and
    Lyudmila Grigoryeva\,\orcidlink{0000-0002-4857-7779}%
    \thanks{\linespread{1.2}\selectfont%
    University of St.~Gallen, Division of Mathematics and Statistics, Rosenbergstrasse 22, 9000 St.~Gallen, Switzerland. {\texttt{Lyudmila.Grigoryeva@unisg.ch}}. Honorary Associate Professor. University of Warwick, Department of Statistics, Coventry, CV4 7AL, United Kingdom. {\texttt{Lyudmila.Grigoryeva@warwick.ac.uk}}}
    \and
    Yui Ching Li%
    \thanks{\linespread{1.2}\selectfont%
    Bank for International Settlements, Centralbahnplatz 2, 4002 Basel, Switzerland. %
    The views expressed in this article are solely those of the authors and do not necessarily reflect those of the Bank for International Settlements.}%
}
\date{\today}

\maketitle

\makeatletter
\let\ofootnote\footnote
\def\footnote{\@ifstar\footnote@star\footnote@nostar}
\def\footnote@star#1{{\let\thefootnote\relax\footnotetext{#1}}}
\def\footnote@nostar{\ofootnote}
\footnote*{%
	YCL and LG acknowledge the financial support of the FoKo of the University of St.~Gallen (Project Nr.~1022324, ``Machine Learning Techniques for Macroeconomic Forecasting and Pricing of Financial Derivatives''). GB is thankful to the Great Minds Postdoctoral Fellowship program of the University of St.~Gallen for financial support. 
	Python code accompanying this manuscript is available at the online repository \url{https://github.com/RCEconModelling/Ensemble-MFESN}.
}%
\makeatother

\begin{abstract}
 {\color{green}Model combination offers an alternative to selecting a single forecasting model, particularly when individual models are sensitive to random initialization and hyperparameter choices. We study online aggregation through Follow-the-Leader (FTL) and Hedge-type algorithms, with an application to ensembles of Multi-Frequency Echo State Networks (MFESNs) for macroeconomic forecasting. On the theoretical side, we derive Hoeffding- and Bernstein-type expected-regret bounds for FTL and decreasing Hedge under independent and $\varphi$-mixing expert losses. Under a positive gap between the best expert's expected loss and those of its competitors, the bounds are uniform in the number of forecasting rounds and exhibit logarithmic dependence on the number of experts. Monte Carlo experiments investigate the effects of temporal dependence and the separation between experts' expected losses. Empirically, we construct MFESN ensembles for forecasting U.S. quarterly GDP growth using mixed-frequency macroeconomic and financial indicators. The results show that adaptive aggregation, particularly FTL and AdaHedge, achieves substantial reductions in out-of-sample mean squared forecasting error relative to individual MFESN models and conventional forecast combinations.}
\end{abstract}


\maketitle
\newpage

{\color{brown}

\section{Introduction}
\label{section:intro}

Many mathematical, statistical, and econometric tools have been developed to solve the problem of predicting stochastic and deterministic processes. 
In empirical applications, however, it is often difficult to identify which model, or even which model class, is most appropriate for a given task.
Models that perform well on average may fail in specific settings. 
This situation is common in forecasting settings, for example, involving macroeconomic time series. 
Since the exact out-of-sample forecasting performance of a given forecasting approach is unknown before it is applied, it is natural, when model updating is infeasible or inconvenient, to revise the choice of model(s) as new information becomes available. 

The model combination approach formalizes the use of multiple prediction schemes. The forecaster relies on a finite set, referred to as an \textit{ensemble}, of distinct models or algorithms, termed \textit{experts}, to predict an outcome of interest. Experts in an ensemble are evaluated and reweighted over time based on their predictive losses. The objective is to construct data-driven weights that aggregate individual expert forecasts to minimize cumulative loss over time. Hence, ensemble methods can dynamically adapt predictions over the course of a forecasting task. The objective of a combination scheme is to achieve cumulative loss close to that of the best-performing expert in hindsight.

In this work, we bridge the literature on online ensemble learning (see, e.g., \citealp{shalev_online}) with econometric forecast combination \citep{timmermann2006forecast} in a general time series setup. 
In the i.i.d.\ setting, \cite{degennePerchet2016} established an $\mathcal{O}(\log K/\Delta)$ expected pseudo-regret bound for Follow-the-Leader (FTL) algorithm with heterogeneous sub-optimality gaps, while \cite{mourtadaOptimalityHedgeAlgorithm2019} obtained gap-dependent regret guarantees for decreasing Hedge. Building upon these results and the developments in \cite{rooijFollowLeaderIf2014}, we derive Hoeffding- and Bernstein-type expected regret bounds relative to the best expert in hindsight, including under temporal dependence. While the dependence on the number of experts is mild (logarithmic), the gap between the least expected loss and the runner-up plays a key role in our bounds. Using simulations, we show that our bounds, although conservative, capture the dependence of both Follow-the-Leader and decreasing Hedge on the separation between the best expert and the runner-up.
Second, on the empirical side, we apply a collection of online combination methods to a macroeconomic forecasting task: when compared to well-studied weighting schemes (average weighting, rolling MSE), ensemble combination methods achieve substantial reductions in mean-squared forecasting error. 
Following \cite{ballarin2022reservoir}, we construct our ensembles using \textit{Multi-Frequency Echo State Network} (MFESN) models, which have been shown to outperform state-of-the-art mixed-frequency forecasting methods such as MIxed DAta Sampling (MIDAS) \citep{Ghysels2007} and Dynamic Factor Models (DFMs) \citep{geweke1977dynamic,sargent1977business}.
The MFESN models are based on Echo State Networks (ESNs), a class of recurrent-type neural networks with randomly sampled state parameters \citep{jaeger2001,maass2,lukosevicius}.  
This model family has been actively studied in terms of approximation, generalization, and memory properties~\citep{RC6,RC7,RC8,RC10,RC12,RC23} and lends itself naturally to ensemble construction via random resampling. 
Using quarterly U.S. GDP growth and the set of $33$ monthly and daily financial and macroeconomic variables from \cite{ballarin2022reservoir}, we show that MFESN ensembles substantially reduce mean-squared forecasting errors. Allowing for heterogeneity in random parameter draws, the state leak rates, and the ridge regression penalty for training yields reductions of up to 54\% in mean squared forecasting error relative to the in-sample mean. These results set new benchmarks in the prediction of U.S. GDP growth with multi-frequency regressors, and are robust to loss thresholding, ensemble resampling, and reduction in the information set used for prediction.

Our work contributes to the rich and well-developed literature of forecast combination, which is a long-standing idea in econometrics and time series analysis \citep{Bates1969,Winkler1983,clemen1989combining,grangerInvitedReviewCombining1989},\footnote{Boosting, bagging, and random forest methods, which explicitly combine many learners into one predictor to improve performance, have also become increasingly popular \citep{Athey2019,atheyMachineLearningMethods2019}.} with well-documented empirical effectiveness in a variety of macroeconomic and financial forecasting contexts. Subsequent contributions have studied both linear and nonlinear combination schemes, as well as the role of time variation in forecast weights (see \citealp{timmermann2006forecast}, for a comprehensive survey). More recent work has emphasized adaptive and data-driven approaches to forecast combination in high-dimensional and potentially unstable environments, including settings with model misspecification and structural changes (e.g., \cite{hansenLeastsquaresForecastAveraging2008,chengForecastingFactoraugmentedRegression2015,kimMiningBigData2018,DenReijer2019,bolhuisDeusExMachina2020,fulekyMacroeconomicForecastingEra2020,masiniMachineLearningAdvances2023}); Bayesian~\citep{WANG20231518} and density combination schemes~\citep{wallisCombiningForecastsForty2011} have also been considered.

The literature on online combination schemes and ensembles is also rich. In the standard framework of \textit{prediction with expert advice} \citep{Littlestone1994,Freund1997,vovk1995game,Cesa-Bianchi2006}, or its variant known as the \textit{Hedge setting}, much of the work focuses on developing combination algorithms that enjoy guarantees on regret, defined as the difference between the cumulative loss of the forecaster and that of the best among $K$ experts. 
Among these methods, FTL  assigns uniform weights at each round to the model(s) with the smallest cumulative loss so far. Under jointly i.i.d. bounded loss vectors, FTL admits gap-independent expected-regret bounds of order $\mathcal{O}(\sqrt{T\log K})$ (with $T$ a time horizon), even when losses are dependent across experts (see \citealp{kotlowski2018}), while potentially performing poorly under worst-case data \citep{vanervenAdaptiveHedge2011}.
The {\it Hedge algorithm} \citep{Freund1997}, also called the exponentially weighted forecaster, provides a smoothed alternative to FTL and enjoys worst-case $\mathcal{O}(\sqrt{T \log K})$ regret bounds, including anytime guarantees with suitably chosen time-varying learning rates. In stochastic environments, allowing the learning rate of Hedge to be data-adaptive yields regret comparable to FTL, motivating schemes such as \textit{Adaptive Hedge} \citep{vanervenAdaptiveHedge2011,rooijFollowLeaderIf2014}.
Significant effort has been devoted to developing strategies that combine the best properties of these approaches; see, for example, \cite{koolenLearningLearningRate2014,wintenbergerOptimalLearningBernstein2017,mourtadaOptimalityHedgeAlgorithm2019,itoBestofBestWorlds2024,wintenbergerStochasticOnlineConvex2024}.  
}

The remainder of the paper is organized as follows. Section~\ref{section:preliminaries} introduces the ensemble forecasting framework and formalizes the online learning setting for combination strategies. Section~\ref{section:expert_ensembles} reviews popular combination strategies and presents theoretical bounds, deriving regret guarantees for Follow-the-Leader and decreasing Hedge under i.i.d.~and $\varphi$-mixing losses. It also compares online combination methods with simple averaging and summarizes our Monte Carlo simulation findings.   Section~\ref{section:ensemble_esn} describes the Multi-Frequency Echo State Network architecture and its extension to ensemble forecasting. Section~\ref{section:application} introduces four ensemble variants, evaluates their performance in forecasting U.S. GDP growth using mixed-frequency data setting of \cite{ballarin2022reservoir} and compares them with state-of-the-art methods. Section~\ref{section:conclusion} concludes.
All proofs, additional empirical results, and the Monte Carlo simulations are provided in the appendix.

\paragraph{Notation.}
For $n \in \N^+$, we define $[n] := \{1, \ldots, n\}$. Given a vector ${v} \in \mathbb{R}^n$ (or $\bm{v} \in \mathbb{R}^n$), we denote its entries by $v_i$, $i \in [n]$. For a vector ${v} \in \Real^n$, we write ${v} \geq 0$ whenever ${v}_i\geq 0$ for all $i\in [n]$. The symbol $\bm{1}_n \in \Real^n$ stands for the vector of ones. All random variables are defined on a fixed probability space $(\Omega, \mathcal{F}, \mathbb{P})$ and follow context-specific notation: quantities traditionally denoted by Latin or Greek lowercase letters (for example, $\ell$,~$\varepsilon$,~$\zeta$) remain lowercase even when random. For all other random variables, we use uppercase letters.

\section{Preliminaries}
\label{section:preliminaries}

In this section, we introduce the general econometric setting of interest, which involves constructing time series predictions by combining (or selecting) models from a set of experts.

\subsection{General Setup}
Let $\{Y_t\}_{t \in \Int}$, $Y_t \in \Real$, be the target time series, and $\{Z_t\}_{t \in \Int}$, $Z_t \in \Real^{d}$, a vector of regressors, which may include lagged values of $Y_t$. 
We consider the problem of predicting the target series $Y_{t+h}$ at the forecasting horizon $h \in \mathbb{N}^+$.
Let $(\Omega,\mathcal F,\mathbb P)$ be a fixed probability space on which all the random variables are defined. 
 {\color{green}Assuming that $Y_{t+h}$ is square-integrable, the $h$-step-ahead minimum mean squared error (MSE) predictor of $Y_{t+h}$ at time $t\in[T]$, with $h\in[H]$, is given by the conditional expectation function (CEF) %
\begin{equation}\label{eq:mse_pred_problem_general}
    g_{t,h}
    \in 
    \arg\min_{g \in \mathcal{G}}  \E\left[  ( Y_{t+h} - g(Z_t)) ^2 \right], \quad g_{t,h}(Z_t)=\E[Y_{t+h} \,\vert\, Z_t] \enspace a.s.,
\end{equation}
where $\mathcal{G}$ denotes the class of all Borel-measurable maps $g:\mathbb R^d\to\mathbb R$ such that $\mathbb{E}[g(Z_t)^2]<\infty$.}

The core challenge arising when solving~\eqref{eq:mse_pred_problem_general}, even in the one-step-ahead prediction setting ($h=1$), is that the conditional expectation function is generally unknown and must be estimated from data. This is typically addressed by structural assumptions on the joint data-generating process $(Y_t,Z_t)$. For example, standard choices include linear models~\citep{BrocDavisYellowBook}, semi- or nonparametric classes of sufficient regularity~\citep{stoneOptimalGlobalRates1982,Li2009,Tsybakov2009}, or a reproducing kernel Hilbert space setting~\citep{Steinwart2008,scholkopf2002learning,berlinet2011reproducing}.
In this work, we study a standard online aggregation of  multiple models (``experts''), each potentially excelling under different finite sample realizations, but always assuming that a best-performing model exists (see Section~\ref{section:expert_ensembles}).
The focus of the paper is on studying schemes to weight experts  to produce a combined forecast that over time aligns with the optimal expert.

\subsection{Prediction with Experts}

The concept of \textit{prediction with experts} goes back to the works of \cite{Littlestone1994,foster1991prediction,foster1993randomization,vovk1995game,cesa1997use}. 
A forecaster is interested in obtaining, at each time step $t \in [T]$, a prediction $\widehat{Y}_{t+h}$, for a horizon $h\in [H]$, based on $K \geq 1$ experts indexed by $k \in [K]$, each providing their own forecast $\widehat{Y}_{t+h}^{(k)}$. The forecaster constructs a convex combination of $K$ expert predictions using a weight vector $\bm{\omega}_{t,h} := (\omega^{(1)}_{t,h}, \ldots, \omega^{(K)}_{t,h})^\top \in \Real^K$ such that
\begin{equation*}
    \bm{\omega}_{t,h} \geq 0
    \quad\text{and}\quad
    \bm{\omega}_{t,h}^\top \bm{1}_K =  1 \enspace 
    \text{ for all }
    t \in [T] \text{ and for all }
    h \in [H],
\end{equation*}
and producing a forecast \begin{equation}
\label{predictor}
    \widehat{Y}_{t+h}:=\sum_{k=1}^K \omega_{t, h}^{(k)} \widehat{Y}_{t+h}^{(k)}.
\end{equation} We will also use $\Delta^{K-1}:=\{\bm{\omega} \in \mathbb{R}_{+}^K: \bm{\omega}^\top \bm{1}_K=1\}$ to indicate the real simplex.

The accuracy of each forecast is assessed using a loss function $\ell: \mathbb{R} \times \mathbb{R} \rightarrow \mathbb{R}_{+}$. We define, for each expert $k \in [K]$ and each $t \in [T]$, the associated loss at $h$-horizon forecasting task of predicting ${Y}_{t+h}$, $h \in [H]$, as 
\begin{equation*}
    \ell^{(k)}_{t,h} := \ell({Y}_{t+h}, \widehat{Y}_{t+h}^{(k)}),
\end{equation*} 
the vector containing the losses of all $K$ experts as $\bm{\ell}_{t,h} := (\ell^{(1)}_{t,h}, \ldots, \ell^{(K)}_{t,h})^\top$
and the forecaster's loss $\overline{\ell}_{t,h} := \bm{\omega}_{t,h}^\top \bm{\ell}_{t,h}$. We also define the cumulative loss of expert $k$ and of the forecaster {\color{green} over the first $t$ forecasting rounds as} %
\begin{equation*}
    L^{(k)}_{t,h} := \sum_{\tau=1}^{\color{green}{t}} \ell^{(k)}_{\tau,h}, \quad {\rm and}\quad \overline{L}_{t,h} := \sum_{\tau=1}^{\color{green}{t}} \overline{\ell}_{\tau,h},
\end{equation*}
respectively, with the convention that {\color{green} $L^{(k)}_{0,h}=0$ and $\overline{L}_{0,h}=0$,} and construct a vector $\bm{L}_{t,h} := (L^{(1)}_{t,h}, \ldots, L^{(K)}_{t,h})^\top$ of cumulative losses of all $K$ experts. {\color{green}  These cumulative losses are indexed by forecasting rounds
and need not be fully observed at calendar time $t$. When predicting $Y_{t+h}$ at time $t$, the cumulative losses available for weighting are $L^{(k)}_{\max(t-h,0),h}$. In particular, for $h=1$, the weights at round $t$ may depend on $L^{(k)}_{t-1,1}$.
}

The natural objective of online learning procedures is to minimize the {\color{green} cumulative} loss $\overline{L}_{t,h}$ incurred by the forecaster when combining $K$ experts' predictions.  Assuming $[K]$ does not change over time, one can introduce cumulative regret at $t \in [T]$ for horizon $h\in [H]$ as
\begin{equation*}
    \overline{R}_{t,h} := 
    \overline{L}_{t,h} - \min_{k \in [K]} L^{(k)}_{t,h} ,
\end{equation*}
which, for a given forecasting horizon $h$, is the difference between the forecaster's cumulative loss and that of the best expert in hindsight. {\color{green} Under Assumption~\ref{assumption:losses_iid_or_mixing} introduced below, the expected loss of each expert $k\in[K]$ is time-invariant and is denoted by $\mu_{k,h}=\E[\ell^{(k)}_{t,h}]$. Let $k^\star_h\in\arg\min_{k\in[K]}\mu_{k,h}$ denote an expert with minimal expected loss. Define $\ell^\star_{t,h}:= \ell^{(k^\star_h)}_{t,h}$. The instantaneous regret relative to this expert is $r^\star_{t,h}:=\overline{\ell}_{t,h}-\ell^\star_{t,h}$, while the corresponding regret of expert $k$ is $r^{(k)}_{t,h}:=\ell^{(k)}_{t,h}-\ell^\star_{t,h}$. By definition, $r^{(k^\star_h)}_{t,h}=0$
for all $t$.}
Therefore, the goal of an online combination scheme consists in minimizing cumulative regret $\overline{R}_{t,h}$ {\color{green} relative to the best expert in hindsight over the first $t$ forecasting rounds}. %
One may also consider an adversarial setting, where losses are defined via $\ell(y, \hat{y})$ and experts' $\widehat{Y}_{t+h}^{(k)}$ are fixed before the environment acts. Equivalently, the environment chooses $Y_{t+h}$, to induce
$
\boldsymbol{\ell}_{t, h}
$, such that the regret is maximized, leading to a worst-case scenario.

Throughout the paper, unless stated otherwise, we impose the following assumptions on expert and forecaster losses for all $h\in [H]$.

\begin{assumption}[Bounded losses]\label{assumption:losses_0-1}
    Losses $\ell^{(k)}_{t,h} \in [0,1]$ almost surely for all $k \in [K]$, $t \in \mathbb{N}$.%
\end{assumption}

\begin{assumption}[Temporal dependence]\label{assumption:losses_iid_or_mixing} 
    Losses $\{ \boldsymbol{\ell}_{t,h} \}_{t}$, $\boldsymbol{\ell}_{t,h}\in \mathbb{R}^K$, satisfy one of the following:
    \begin{description}
        \item[(i)] $\{ \ell^{(k)}_{t,h} \}_{t}$ is i.i.d. for every $k\in [K]$;
        \item[(ii)] $\{ \ell^{(k)}_{t,h} \}_{t}$, for every $k\in [K]$, is strictly stationary and $\varphi$-mixing with  coefficients $\{\varphi^{(k)}_n\}_{n\ge1}$ depending only on the lag $n$;%
        {\color{brown}
        \item[(iii)] The vector-valued process $\{ \boldsymbol{\ell}_{t,h} \}_{t}$ is i.i.d.;
        \item[(iv)] The vector-valued process $\{ \boldsymbol{\ell}_{t,h} \}_{t}$ is strictly stationary and $\varphi$-mixing with coefficients $\{\varphi_n\}_{n\ge1}$ depending only on the lag $n$.%
        }
    \end{description}
\end{assumption}
The definition of $\varphi$-mixing dependence in (ii) and {\color{green} (iv)} is in the sense of 
Definition~\ref{def:rio_uniform_mixing} in the Appendix, following \citealp{bradley2005basic}.
The dependence conditions are imposed on expert losses \emph{over time}, but no independence is assumed \emph{across} the $K$ losses.

\begin{remark}\label{remark:weak_stationarity} 
    Under Assumption~\ref{assumption:losses_0-1} and Assumption~\ref{assumption:losses_iid_or_mixing}, for all $k \in [K]$, $\{ \ell^{(k)}_{t,h} \}_t$ is a stationary process with time-invariant mean $\mu_{k,h} = \E[\ell^{(k)}_{t,h}]$ and variance $v_{k,h} = {\rm Var}[\ell^{(k)}_{t,h}]$.
 \end{remark}

\begin{assumption}[Uniform summability]\label{ass:uniform summability}
    Losses $\{ \boldsymbol{\ell}_{t,h} \}_{t}$, $\boldsymbol{\ell}_{t,h}\in \mathbb{R}^K$, satisfy one of the following:
    \begin{description}
        \item[(i)] There exists a finite constant $C_{1,\varphi}>0$ such that
        $
        \max _{k \in[K]} \sum_{n=1}^{\infty} \varphi_n^{(k)} < C_{1,\varphi},
        $ with $\{\varphi^{(k)}_n\}_{n\ge1}$ coefficients of $\{ \ell^{(k)}_{t,h} \}_{t}$, for every $k\in [K]$.
        \item[(ii)]  There exists a finite constant $\widetilde{C}_{1,\varphi}>0$ such that
        $
        \sum_{n=1}^{\infty} \varphi_n <\widetilde{C}_{1,\varphi} $ with $\{\varphi_n\}_{n\ge1}$ coefficients of vector-valued $\{ \boldsymbol{\ell}_{t,h} \}_{t}$.
    \end{description}
\end{assumption}

\begin{assumption}[Uniform root-summability]\label{ass:uniform summability roots}
    Losses $\{ \boldsymbol{\ell}_{t,h} \}_{t}$, $\boldsymbol{\ell}_{t,h}\in \mathbb{R}^K$, satisfy one of the following:
    \begin{description}
            \item[(i)] There exists a finite constant $C_{2,\varphi}>0$ such that
        $
        \max _{k \in[K]} \sum_{n=1}^{\infty} \sqrt{\varphi_n^{(k)}} < C_{2,\varphi},
        $ with $\{\varphi^{(k)}_n\}_{n\ge1}$ coefficients of $\{ \ell^{(k)}_{t,h} \}_{t}$, for every $k\in [K]$.
        \item[(ii)]  There exists a finite constant $\widetilde{C}_{2,\varphi}>0$ such that
        $
        \sum_{n=1}^{\infty} \sqrt{\varphi_n} <\widetilde{C}_{2,\varphi} $ with $\{\varphi_n\}_{n\ge1}$ coefficients of vector-valued $\{ \boldsymbol{\ell}_{t,h} \}_{t}$.
    \end{description}
\end{assumption}

Assumption~\ref{assumption:losses_0-1} is standard in online learning and, up to a fixed rescaling, requires losses to be bounded.
Assumption~\ref{assumption:losses_iid_or_mixing} imposes conditions on the temporal dependence of loss vectors, as the prediction algorithm evolves. 
While much of the literature focuses on the i.i.d. case, mixing conditions such as Assumption~\ref{assumption:losses_iid_or_mixing}(ii) are empirically more relevant in time-series applications. 
Notice also that for any $k\in [K]$, since $0 \leq \varphi_n^{(k)} \leq 1$, Assumption~\ref{ass:uniform summability roots}(i) implies Assumption~\ref{ass:uniform summability}(i) and Assumption~\ref{ass:uniform summability roots}(ii) implies Assumption~\ref{ass:uniform summability}(ii). Moreover, in each of these assumptions, part (ii) implies part (i). All assumptions are imposed on the \emph{loss processes} $\{\ell^{(k)}_{t,h}\}_t$, rather than on the joint data process $(Y_t,Z_t)$. %

In the remainder of the paper, we focus on a single forecasting horizon $h\in \mathbb{N}$ and hence drop the $h$ subscript where convenient. %
{\color{green} In particular, $\ell_t^{(k)}$ and $L_t^{(k)}=\sum_{s=1}^t\ell_s^{(k)}$ refer to the same forecasting-round indexing as $\ell_{t,h}^{(k)}$ and $L_{t,h}^{(k)}$, respectively. The regret bounds derived below assume that, at each round $t$, the losses from all preceding rounds are available when constructing the expert weights. This is the case for one-step-ahead forecasting ($h=1$), namely, when predicting $Y_{t+1}$, the outcomes up to $Y_t$ have already been observed, so the losses from rounds $1,\ldots,t-1$ are available. For $h>1$, the availability of past losses must account for the forecasting horizon, and the regret guarantees would require a separate analysis.}

\section{Expert Ensembles}
\label{section:expert_ensembles}

We now turn to an in-depth discussion of the Follow-the-Leader and Hedge online learning methods. 
Table~\ref{table:comb_schemes_summary} provides an overview of these schemes, as well as classical combination approaches based on averaging and rolling mean-squared-error weighting (see Appendix~\ref{appendix:Direct Combination Methods} for more details).
In Section~\ref{subsection:online_learning_guarantees}, we theoretically study the stochastic setting with and without temporal dependence, and derive regret guarantees for both FTL and decreasing Hedge.
\begin{table}[t!]
\hfuzz=1pt
\centering
\setlength\tabcolsep{0pt}
\setlength\extrarowheight{1pt}
\linespread{1.2}\selectfont\centering
\footnotesize
\begin{tabular*}{\textwidth}{@{\extracolsep{\fill}}*{5}{c}}
    \toprule
    \small Type &\small Scheme & \small Weighting, $\omega^{(k)}_{t}$ & \small Specification & \small Reference \\
    \toprule
    
    \multirow{6}{*}{\rotatebox{90}{Direct}}&Median &  $\frac{1}{\left|\mathcal{M}_t\right|} \1{k \in \mathcal{M}_t}$& {\makecell{ $\mathcal{M}_t$: middle-ranked  \\ expert(s) by forecast value }} &\multirow{6}{*}{\makecell{ \scriptsize\cite{Bates1969} \\ \scriptsize
    \cite{timmermann2006forecast} }}\\ 
    &\multicolumn{3}{l}{\raisedrule[0.2em]{0.2pt}} \\
     &
    Simple Averaging & $\frac{1}{K}$ & --- &  \\
    &\multicolumn{3}{l}{\raisedrule[0.2em]{0.2pt}} \\
    &Rolling MSE & $\frac{  \left({\mathrm{MSE}}^{(k)}_{t,r}+\varepsilon   \right)^{-1}} {\sum_{j=1}^K   \left({\mathrm{MSE}}^{(j)}_{t,r}+\varepsilon   \right)^{-1}} $ & \makecell{ fixed constant $\varepsilon>0$\\ rolling window length $r$} & \\[10pt]
    \midrule
    &Follow-the-Leader & $\frac{ 1 }{| \textsc{ftl}_t |} \1{k \in \textsc{ftl}_t}$ & $\textsc{ftl}_t = \arg\min_{j \in [K]} L^{(j)}_{t-1}$ & \scriptsize\cite{rooijFollowLeaderIf2014} \\[3pt]
    \midrule
    \multirow{9}{*}{\rotatebox{90}{Hedge}} &Constant & \multirow{9}{*}{ \makecell{$\frac{ \exp\left( -\eta_t L^{(k)}_{t-1} \right) }{ \sum_{j=1}^K \exp\left( -\eta_t L^{(j)}_{t-1} \right) }$\\[10pt] with learning rate $\eta_t$  } } & $\eta_t=\eta \propto \sqrt{\frac{\log K}{T}}$ & \multirow{6}{*}{ \begin{tabular}[t]{@{}c@{}} \scriptsize\cite{Littlestone1994}\\  \scriptsize\cite{Cesa-Bianchi2006}\\ \scriptsize\cite{chernovPredictionExpertAdvice2010} \end{tabular} } \\
    &\multicolumn{1}{l}{\raisedrule[0.2em]{0.2pt}} & & \multicolumn{1}{l}{\raisedrule[0.2em]{0.2pt}} & \\
    &Doubling & & \makecell{ $\eta_t \propto \sqrt{\frac{\log K}{2^{r-1}}}$\\[5pt]  doubling round $r=\left\lceil\log _2(t+1)\right\rceil$} & \\[4pt]
    &\multicolumn{1}{l}{\raisedrule[0.2em]{0.2pt}} & & \multicolumn{1}{l}{\raisedrule[0.2em]{0.2pt}} & \\
    &Decreasing & & $\eta_t \propto \sqrt{\frac{\log K}{t}}$ & \\
    &\multicolumn{1}{l}{\raisedrule[0.2em]{0.2pt}} & & \multicolumn{2}{l}{\raisedrule[0.2em]{0.2pt}} \\
    &Adaptive & & \makecell{ $\eta_{t} = \frac{\log K}{\nabla_{t-1}}$\\[5pt] cumulative mixability gap $\nabla_{t-1}$ } & \makecell{\scriptsize\cite{vanervenAdaptiveHedge2011}\\ \scriptsize\cite{rooijFollowLeaderIf2014} } \\
\bottomrule
\end{tabular*}
\caption{%
    Overview of online ensemble learning and combination schemes considered.
}
\label{table:comb_schemes_summary}
\end{table}

\subsection{Follow-the-Leader}
\label{subsection:Follow-the-Leader}
An intuitive first step to construct expert weights is the so-called {\it Follow-the-Leader} (FTL) strategy: at each round, uniform weights are assigned to the (subset of) model(s) with the smallest cumulative loss so far.
Formally, let $\textsc{ftl}_t = \arg\min_{k \in [K]} L^{(k)}_{t-1}$  be the subset of leader-experts at time $t$, that is, the experts with the least cumulative loss up to $t-1$. The Follow-the-Leader weights are defined as %
\begin{equation*}
    \omega^{(k)}_{\textsc{ftl},t} 
    =
    \frac{ 1 }{| \textsc{ftl}_t |} \1{k \in \textsc{ftl}_t} ,
\end{equation*}
with the convention that $\omega^{(k)}_{\textsc{ftl},1} = 1/K$.
FTL is known to perform well in several settings, particularly in the case of i.i.d. losses. However, it has the downside of being sensitive to adversarially generated losses, as it overreacts to short-term fluctuations. \cite{rooijFollowLeaderIf2014} derive an upper bound on the cumulative regret incurred by a forecaster using the FTL strategy due to the variability of the losses and the number of leader changes.

Let $s_t := \max_{k \in [K]} \ell^{(k)}_{t} - \min_{k \in [K]} \ell^{(k)}_{t}$ and $S_t := \max\{s_1, \ldots, s_t\}$ denote the instantaneous maximal loss differential and maximum loss differential over $t$ periods, respectively.  Define {\color{green} $c_t :=\1{\textsc{ftl}_t\neq\textsc{ftl}_{t+1}}$ and  $C_T := \sum_{t=1}^T c_t$, where $\textsc{ftl}_{t+1}$ denotes the set of experts with the smallest cumulative loss after observing the losses at time $t$. }

\begin{lemma}[\citealt{rooijFollowLeaderIf2014}, Lemma 10]
\label{lemma:rooij_ftl_lemma}
    The cumulative FTL regret satisfies
    \begin{equation*}
        \overline{R}_{\textsc{ftl},T}
        \leq
        S_T C_T .
    \end{equation*}
\end{lemma}
Lemma~\ref{lemma:rooij_ftl_lemma} shows that FTL is effective whenever leader changes are infrequent, and the loss range is small. Under i.i.d.~losses whose expectations are well separated, frequent leader changes are unlikely~\citep{rooijFollowLeaderIf2014}. In Section~\ref{subsection:online_learning_guarantees}, we provide a sharper result in a stochastic setting, while also allowing for dependence.

\subsection{Hedge and Adaptive Hedge}

{\color{green} While FTL assigns positive weight only to experts with the smallest cumulative loss, this approach can be generalized to weighting schemes that distribute positive weight across all experts.  A prominent example is the} {\it Hedge} family of algorithms \citep{Littlestone1994,Freund1997,Freund1999,cesa1997use}, also known as \textit{exponentially weighted average} forecasters.
In Hedge, experts are weighted according to
\begin{equation}\label{eq:hedge_base}
    \omega^{(k)}_{\textsc{hdg},t} 
    =
    \frac{ \exp( -\eta_t L^{(k)}_{t-1} ) }{ \sum_{j=1}^K \exp( -\eta_t L^{(j)}_{t-1} ) } ,
\end{equation}
where $\{\eta_t\}_{t=1}^T$ is a sequence of {\color{green} positive} learning rates \citep{Cesa-Bianchi2006}. We use the standard convention $\omega^{(k)}_{\textsc{hdg},1} = 1/K$. Variations of Hedge propose different sequences of learning rates aimed at minimizing the cumulative regret of the forecaster. 
A natural way to interpret \eqref{eq:hedge_base} is to note that it corresponds to the gradient of a scaled LogSumExp function of the negative cumulative losses
\begin{equation*}
\omega^{(k)}_{\textsc{hdg},t}
    =
    \frac{\partial \mathcal{L} (\bm{L}_{t-1}, \eta_t)}{\partial L^{(k)}_{t-1}}, \enspace k\in [K], \enspace \text{with}\enspace
    \mathcal{L} (\bm{L}_{t-1},\eta_t )
    :=
    -\frac{1}{\eta_t} \log\left( \sum_{k=1}^K \exp( -\eta_t L^{(k)}_{t-1} ) \right),
\end{equation*}
where $\eta_t>0$ acts as a smoothing parameter. 
Large  values of $\eta_t$ make the distribution over experts concentrate around the best expert(s) (recovering FTL), while small values yield nearly uniform weights. %
Equivalently, Hedge arises as the solution to an entropy-regularized linear optimization problem over the simplex, with {\color{green} negative} Shannon entropy acting as a regularizer (see \citealp{boyd2004convex}). From this perspective, Hedge can be viewed as a smoothed version of FTL, with $\eta_t$ controlling the per-round smoothing of the loss.

\paragraph{Constant Hedge.}
Setting $\eta_t:=\eta>0$ yields a constant-rate forecaster with a regret bound
\begin{equation}\label{eq:hedge_regret_bound}
    \overline{R}_T
      \le  
    \frac{\log K}{\eta}   +   \frac{\eta S_T^2 T}{8},
\end{equation}
where {\color{green} $S_T:=\max_{t\in[T]}s_t$ is the maximum loss range over the first $T$ rounds, with $s_t:=\max_{k\in[K]}\ell_t^{(k)} -\min_{k\in[K]}\ell_t^{(k)}$. The bound follows from Hoeffding's lemma and holds for arbitrary loss sequences}. Optimizing \eqref{eq:hedge_regret_bound} with respect to $\eta$ gives
\begin{equation*}
\eta^\star = \sqrt{\frac{8\log K}{S_T^2 T}},
\end{equation*}
{\color{green} for $S_T>0$, leading to $\overline{R}_T \le S_T \sqrt{T\log K /2}$.
However, this choice depends on the realized loss range $S_T$, which is not known in advance. If losses take values in $[0,1]$, then $S_T\le 1$ and choosing $\eta=\sqrt{8\log K/T}$ guarantees $\overline{R}_T\le\sqrt{T\log K/2}$}~\citep{Littlestone1994}.
A constant rate allows for the standard closed-form online update via \emph{instantaneous} losses:
\begin{equation*}
    \omega^{(k)}_{\textsc{hdg},t}
    =
    \begin{cases}
        1/K , & \text{if } t = 1 , \\[2pt]
        \dfrac{
            \omega^{(k)}_{\textsc{hdg},t-1} \exp( -\eta \ell^{(k)}_{t-1} )
        }{
            \sum_{j=1}^K \omega^{(j)}_{\textsc{hdg},t-1} \exp( -\eta \ell^{(j)}_{t-1} )
        } , & \text{otherwise}. 
    \end{cases}
\end{equation*} 
In practice, $\eta$ is a \textit{learning hyperparameter}: $\eta\to\infty$ recovers FTL; $\eta=0$ yields no learning. The downside of a constant $\eta$ is its lack of adaptivity over arbitrarily long learning periods.

\paragraph{Decreasing Hedge (DecHedge).}
{\it Decreasing Hedge} is obtained by setting $\eta_t = c_0 \sqrt{\log K/t}$ for a constant $c_0 > 0$ \citep{auer2016anAlgorithm}. {\color{green}  While the constant $c_0$ can be tuned for a specific learning problem, }\cite{mourtadaOptimalityHedgeAlgorithm2019} show that the choice $c_0 = 2$ {\color{green} achieves the minimax-optimal worst-case} $\mathcal{O}(\sqrt{T\log K })$ regret rate for arbitrary loss sequences in $[0,1]$~\citep{chernovPredictionExpertAdvice2010}.
However, in stochastic settings where losses are well-behaved, more refined guarantees are possible. In fact, the generic constant $c_0 = 2$ may be conservative when the environment is not adversarial. 

\paragraph{Adaptive Hedge (AdaHedge).}
{\it AdaHedge} dynamically tunes $\eta_t$ based on past performance, and neither horizon nor additional parameters are used \citep{vanervenAdaptiveHedge2011}.

Let $\bm{\omega}_t\in\Delta^{K-1}$ be the weight vector at time $t$, let $\bm{\ell}_t$ be the loss vector of $K$ experts, and let the forecaster's loss be $\overline{\ell}_t$. The \emph{mix loss} at time $t$ for a learning rate $\eta_t>0$ is
\begin{equation*}
    \overline{m}_t 
    := 
    -\frac{1}{\eta_t} \log  \left( \bm{\omega}_{t}\transp \exp({-\eta_t \bm{\ell}_{t}})   \right),
\qquad
\overline{M}_t := \sum_{s=1}^t \overline{m}_s,
\end{equation*}
(see, e.g., \citealp{Cesa-Bianchi2006}, Chapter 2). The difference $\delta_t := \overline{\ell}_t-\overline{m}_t\geq 0$ is called the \emph{mixability gap}, with cumulative gap $\nabla_t := \sum_{s=1}^t \delta_s$.
Consequently, the regret of the forecaster can be decomposed as
\begin{equation*}
    \overline{R}_{t} 
    = \overline{L}_{t} - \min_{k \in [K]} L^{(k)}_{t} 
    = \overline{M}_{t} - \min_{k \in [K]} L^{(k)}_{t} + \nabla_{t} , 
\end{equation*}
where \(\overline{M}_{t} - \min_{k \in [K]} L^{(k)}_{t}\) is the regret incurred under the mix loss. {\color{green} For a nonincreasing sequence of learning rates,
Lemma~2 in \cite{rooijFollowLeaderIf2014} implies
that}
\begin{equation*}
    \overline{M}_{t} - \min_{k \in [K]} L^{(k)}_{t}
    \leq
    \frac{\log K }{\eta_t}.
\end{equation*}
AdaHedge chooses the learning rate at time $t$ by balancing the complexity term ${\log K }/{\eta}$ with the observed cumulative gap $\nabla_{t-1}$,\footnote{At time $t$, while the forecaster makes their decision and the current loss is not realized yet, only the $t-1$ cumulative mixability gap is available.} leading to:
\begin{equation*}
    \eta_{\textsc{ah},t}
    :=
    \frac{\log K}{\nabla_{t-1}} ,
\end{equation*}
with the convention that $\eta_{\textsc{ah},t}=\infty$ if $\nabla_{t-1}=0$.
{\color{green} Since $\nabla_t$ is nondecreasing, the resulting learning rate sequence is nonincreasing, and \cite{rooijFollowLeaderIf2014} establish that $\overline{R}_{\textsc{ah},T}\leq 2\nabla_T$.}

\begin{lemma}[\citealt{rooijFollowLeaderIf2014}, Corollary 9]
    The AdaHedge regret satisfies
    \begin{equation*}
        \overline{R}_{\textsc{ah},T} 
        \leq
        \sqrt{ \sum_{t=1}^T s_t^2 \log K } + S_T\left(\frac{4}{3}\log K  + 2 \right) .
    \end{equation*}
\end{lemma}
If $s_t\le S$ for all $t$, this gives $\overline{R}_{\textsc{ah},T}\le S\sqrt{T \log K}+S  \left(\frac{4}{3}\ln K+2\right)$, matching the $\mathcal{O}(S\sqrt{T\ln K})$ worst-case rate.  
Additionally, the leading term depends on the {\color{green} observed loss ranges through $\sum_{t=1}^T s_t^2$ rather than directly on $T$, so the bound can be tighter when differences between experts' losses are small.}

Finally, for completeness, Appendix~\ref{appendix:Doubling-trick Hedge} also discusses the doubling-trick Hedge algorithm.

\subsection{Online Learning Guarantees under Dependence}
\label{subsection:online_learning_guarantees}

Following \cite{luoAchievingAllNo2015} and \cite{mourtadaOptimalityHedgeAlgorithm2019}, one can see that the complexity of the online forecasting problem for a forecaster following a certain scheme can be generally understood in terms of the gaps between the expected loss of the best expert and the expected losses of other $K-1$ experts. {\color{green} Throughout this subsection, we assume that $K\geq 2$}. Let the per-expert gap be defined as
\begin{equation*}
    \Delta_k := \mu_k -  \mu_{k^\star} \geq 0 , \enspace k\in[K],
\end{equation*}
where $\mu_k = \E[ \ell^{(k)}_t ]$ and $k^\star \in \arg\min_{k \in [K]} \mu_k$, and let 
\begin{equation*} 
    \Delta := \min_{k \neq k^\star} \Delta_k= \min_{k \neq k^\star} \mu_k -  \mu_{k^\star}
\end{equation*}
denote the so-called \textit{sub-optimality gap} between the expected loss of the best expert and the runner-up. 

{\color{green} We apply Hoeffding-type and Bernstein-type bounds to showcase the different kinds of guarantees that can be obtained, which depend on two constants that quantify the dependence of losses over time, and that are necessary ingredients in our proofs. Although our assumptions imply time-invariant loss variances, the Hoeffding-type bounds also apply to bounded, nonstationary loss processes with time-invariant expected losses and corresponding uniform mixing conditions. Similarly, the Bernstein-type bounds under temporal independence remain valid with uniform-in-time variance bounds. We do not establish this extension for the Bernstein-type bounds under $\varphi$-mixing.}

\subsubsection{Bounds for Follow-the-Leader}

Our first result is a set of regret guarantees for FTL online combination under both independence and mixing assumptions.
 
For the sake of space, we do not include here an additional bound obtained using a Bernstein inequality
which can be found in Appendix~\ref{appendix:proofs_ftl_regret}.

\begin{theorem}\label{theorem:ftl_regret_mixing}
    Let Assumptions~\ref{assumption:losses_0-1}-\ref{assumption:losses_iid_or_mixing} hold {\color{green} and that there exists a unique best expert $k^\star\in[K]$ with $\Delta>0$}.
    Define $v_{\max} := \max_{k\in[K]} v_k$, where $v_k:=\Var(\ell^{(k)}_t)$, $k\in [K]$. 
    Further assume that there are almost surely no ties {\color{green} for the FTL leader } at any time {\color{green} $t>1$}. 
    Then, for any $T\geq 1$, the expected regret of FTL satisfies
    \begin{equation*}
        \E\big[\overline{R}_{\textsc{ftl},T}\big]
        \leq
        \min \big\{ R^H_{\textsc{ftl}}, R^B_{\textsc{ftl}} \big\}
    \end{equation*}
    with the following cases:
    \begin{description}
        \item[(i)] If Assumption~\ref{assumption:losses_iid_or_mixing}(i) holds, then
        \begin{align*}
            R^H_{\textsc{ftl}} = 2 + \frac{2 \log K + 4}{\Delta^2} , 
            \quad
            R^B_{\textsc{ftl}} = 2 + \frac{(8v_{\max} + \tfrac{4}{3}\Delta)\left(\log(2K)+2\right)}{\Delta^2}.
        \end{align*}
        \item[(ii)] If Assumption~\ref{assumption:losses_iid_or_mixing}(ii) and Assumption~\ref{ass:uniform summability roots}(i) hold, then
        \begin{align*}
            R^H_{\textsc{ftl}} &= 3 + \frac{8 \overline{\theta}_{\textsc{ftl}}^H (\log K + 4)}{\Delta^2} , 
            \quad 
            \text{with} \enspace \overline{\theta}_{\textsc{ftl}}^H=1+4\max_{k\in[K]}\sum_{n=1}^{\infty}\varphi_n^{(k)}, \\
R^B_{\textsc{ftl}} &= 2+\frac{4(\overline{\rho}_{\textsc{ftl}}^B)^2(32v_{\max}+\Delta)(\log(2K)+2)}{\Delta^2},
            \quad 
            \text{with} \enspace \overline{\rho}_{\textsc{ftl}}^B=1+\max_{k\in[K]}\sum_{n=1}^{\infty}\sqrt{\varphi_n^{(k)}}.
        \end{align*}
    \end{description}
\end{theorem}

\begin{remark}
Under Assumption~\ref{assumption:losses_iid_or_mixing}(iii), the dependence on $\Delta$ in Theorem~\ref{theorem:ftl_regret_mixing}(i) can be improved from $\Delta^{-2}$ to $\Delta^{-1}$, as shown in Proposition~\ref{proposition:ftl_iid_sharp} in Appendix~\ref{appendix:proofs_ftl_regret}. Unlike Theorem~\ref{theorem:ftl_regret_mixing}, this sharper result does not require the absence of ties among cumulative expert losses. An alternative approach, which we do not pursue here, would be to impose suitable bounds on the cumulative covariance between the expert weights and contemporaneous excess losses, potentially allowing similar improvements under temporal dependence (see also  \cite{gasparinConformalOnlineModel2025}). The no-ties condition in Theorem~\ref{theorem:ftl_regret_mixing} holds, for example, if the joint distribution of the experts' cumulative losses is absolutely continuous at every round $t>1$.
\end{remark}

\subsubsection{Bounds for Hedge}
\label{subsec:bounds_for_hedge}

\cite{mourtadaOptimalityHedgeAlgorithm2019} showed that, in the stochastic setting with i.i.d.~losses, decreasing Hedge naturally adapts to the hardness of the combination problem, whereas constant and doubling trick Hedge do not. 

Let $\chi^{(k)}_t := (\ell^{(k)}_t - \ell^{(k^\star)}_t )$ for each $k \in [K]$, $k \not= k^\star$. The process $\{ \chi^{(k)}_t \}_{t \in \Int}$ is the excess loss of expert $k$ relative to the best expert $k^\star$.

\begin{theorem}\label{theorem:hedge_regret_mixing}
    Let {\color{green}$K\geq 3$} and suppose that Assumptions~\ref{assumption:losses_0-1}-\ref{assumption:losses_iid_or_mixing} hold. Let $\eta_t = 2\sqrt{\log K/t}$ for $t \in [T]$ be the learning rate for decreasing Hedge. Assume that there exists a unique best expert   $k^\star\in [K]$ {\color{green} with $\Delta>0$}. Define $\widetilde{v}_{\max}:=\max_{k\neq k^\star} \Var( \chi^{(k)}_t )$.
    Then, {\color{green} for every $T\geq 1$, }the expected regret of decreasing Hedge satisfies
    \begin{equation*}
        \E\big[\overline{R}_{\textsc{hdg},T}\big]
        \leq
        \min \big\{ R^H_{\textsc{hdg}}, R^B_{\textsc{hdg}} \big\}
    \end{equation*}
    with the following cases:
    \begin{description}
        \item[(i)] If {\color{brown} Assumption~\ref{assumption:losses_iid_or_mixing}(iii)} holds, then
        \begin{align*}      
R^H_{\textsc{hdg}}=
\min\left\{
\frac{4\Delta\log K+25}{\Delta^2},
\frac{5\log K+25}{\Delta}
\right\}.
        \end{align*}
        and
        \begin{align*}
            R^B_{\textsc{hdg}} = 1 + \sqrt{\log K} 
                + \frac{ {\color{green} \left(4\sqrt{\frac{2}{3}}+1\right)} \Delta \log K + 8 (\widetilde{v}_{\max} + \frac{1}{3} \Delta) + 16 }{\Delta^2} .
        \end{align*}
        \item[(ii)] If {\color{brown} Assumption~\ref{assumption:losses_iid_or_mixing}(iv)} and {\color{brown} Assumption~\ref{ass:uniform summability roots}(ii)} hold, then
        {\color{green}\begin{align*}
            R^H_{\textsc{hdg}}=2+\sqrt{\log K}+\frac{4\sqrt{2\overline{\theta}_{\textsc{hdg}}^H}\log K+4\overline{\theta}_{\textsc{hdg}}^H\left(\log(K-1)+\frac32\right)}{\Delta}+\frac{16(4\overline{\theta}_{\textsc{hdg}}^H+1)}{\Delta^2}
        ,
        \end{align*}
        where $\overline{\theta}_{\textsc{hdg}}^H=1+4\max_{k\neq k^\star}\sum_{n=1}^{\infty}\varphi_n^{\chi,k}<\infty$, and
        \begin{align*} 
            R^B_{\textsc{hdg}} = 
            & \, 1 + \sqrt{\log K} 
                + \frac{2\sqrt{33}\,\overline{\rho}_{\textsc{hdg}}^B\log K+4\overline{\theta}_{\textsc{hdg}}^H\left(\log(K-1)+\frac32\right)}{\Delta} \\
            & \qquad +\frac{4(\overline{\rho}_{\textsc{hdg}}^B)^2(32\widetilde{v}_{\max}+\Delta)+16}{\Delta^2}
        \end{align*}
        where $\overline{\rho}_{\textsc{hdg}}^B=1+\max_{k\neq k^\star}\sum_{n=1}^{\infty}\sqrt{\varphi_n^{\chi,k}}<\infty$.
        In both expressions $\varphi_n^{\chi,k}:=\varphi_n(\{\chi_t^{(k)}\}_{t\in\mathbb Z})$  denote the $\varphi$-mixing coefficients of the excess-loss process $\chi_t^{(k)}=\ell_t^{(k)}-\ell_t^{(k^\star)}$}.
    \end{description}
\end{theorem}

An additional but more complex bound can be constructed using a different Bernstein inequality, which, for clarity of exposition, we derive in Appendix~\ref{appendix:proofs_hedge_regret}.

\begin{remark}
The Hoeffding-type bound $(4\Delta\log K+25)/\Delta^2$ remains valid under the weaker Assumption~\ref{assumption:losses_iid_or_mixing}(i) (see  Proposition~\ref{proposition:hedge_hoeffding_regret_bound}(i) in Appendix~\ref{appendix:proofs_hedge_regret}). Under Assumption~\ref{assumption:losses_iid_or_mixing}(iii), the sharper bound $(5\log K+25)/\Delta$ follows from Theorem~2 of \cite{mourtadaOptimalityHedgeAlgorithm2019} and Lemma~\ref{lemma:mixing_DT}(i).
\end{remark}

{\color{red}

    The gap-dependent bounds of Theorems~\ref{theorem:ftl_regret_mixing} and \ref{theorem:hedge_regret_mixing} become uninformative a $\Delta \to 0$.  For decreasing Hedge, the gap-independent worst-case bound  $\mathcal{O}(\sqrt{T \log K})$ provides an alternative guarantee for a fixed forecasting horizon $T$.  No comparable worst-case guarantee holds for FTL without additional assumptions. However, under jointly i.i.d. loss vectors, FTL also admits a gap-independent expected-regret bound of order $\mathcal{O}(\sqrt{T\log K})$ (see \cite{degennePerchet2016,kotlowski2018} for related results in the stochastic setting). 
    In Appendix~\ref{appendix:gap_dependence_bounds}, we discuss these issues in greater detail, outlining a potential extension to settings with ``groups'' of experts, and then analyzing the impact of small sub-optimality gaps in our applications and simulations. 

}

{\color{green}

\subsection{Simulations}
\label{section:simulations}

To study the empirical performance of expert aggregation and assess how practically informative our regret guarantees are, we conduct two Monte Carlo simulation experiments, which, due to space constraints, we report in Appendix~\ref{appendix:simulations}. Here, we summarize our key findings.

First, we construct an ensemble forecasting scenario that is linear and Gaussian, where experts are constructed via random sampling of a subset of the regressors.\footnote{The design is inspired by the random projection forecasting approach of \cite{bootForecastingUsingRandom2019}.} We consider different dependence and expert specifications (using more or fewer regressors in the forecast), as well as ensemble and sample sizes. In this simplified linear scenario, online combination methods exhibit a substantially greater advantage over direct combination methods when the population oracle predictor is included in the ensemble, particularly when the other experts use small subsets of regressors. When the oracle is excluded, the performance differences between online and direct combination methods are generally smaller.

In a second set of simulations, we generate expert losses as bounded AR(1) processes and evaluate Follow-the-Leader and decreasing Hedge across different levels of persistence and sub-optimality gaps, which allows us to explicitly compute the regret bounds from Theorems~\ref{theorem:ftl_regret_mixing} and~\ref{theorem:hedge_regret_mixing}. To distinguish the second-best expert from the remaining experts, we also introduce an additional positive gap between their expected losses.
Overall, we observe that our bounds tend to be conservative. While they qualitatively capture the increase of expected regret as the sub-optimality gap decreases, the empirical mean regret curves obtained via Monte Carlo simulations are largely insensitive to the degree of temporal dependence.  

}

{\color{green}

\subsection{Online Combination versus Averaging}
\label{subsection:combination_compare}

Section~\ref{subsection:online_learning_guarantees} shows that, under the stated assumptions,  both Follow-the-Leader and decreasing Hedge algorithms admit expected regret upper bounds that are uniform in the terminal time $T$, for  any $K\ge 3 $ experts. These guarantees apply to arbitrary types of experts, which can be viewed as ``black box'' predictors. 

Simple averaging weights ($\omega_t^{(k)} = 1/K$) are often found to be the superior weighting scheme in applications, a phenomenon often referred to as the \textit{forecast combination puzzle}.
Classical econometric forecast combination theory often relies on assumptions concerning the unbiasedness of individual forecasts, either unconditionally or relative to a reference information set, together with restrictions  on the covariance matrix of forecast errors~\citep {timmermann2006forecast}. 
These assumptions may be reasonable in some settings, but  unbiasedness is generally not guaranteed for forecasts obtained  with nonparametric and machine learning methods. 
Some works have partially relaxed these conditions or have considered alternative settings. 
More recently, \cite{elliottCombiningForecastsWhy2026} showed that, under empirically motivated restrictions on the forecast error covariance matrix, the potential gains from optimal combination weights over simple averaging can be small even before accounting for estimation error.\footnote{Nonetheless, \cite{elliottCombiningForecastsWhy2026} still assume forecast unbiasedness in their analysis.} 
\cite{yangCOMBININGFORECASTINGPROCEDURES2004} obtained upper bounds on combination MSEs for weighting schemes similar to the Hedge-type methods we study, but he proceeded by leveraging particular distributional, moment, and uniform deviation assumptions on forecasts. 

The online combination framework we consider is different: we only impose assumptions on the experts' \textit{losses} used for weighting. Ensemble experts may target or be trained with different losses than the ones employed when combining them. 
The conditions on the loss process (Assumptions~\ref{assumption:losses_0-1}-\ref{ass:uniform summability roots}) that we leverage to develop our formal guarantees are not essential to obtain good empirical performance, as our simulations and empirical findings also reveal. 
Finally, the additional assumption that there exists a well-separated best expert (i.e., $\Delta > 0$) does not by itself introduce restrictions on the performance of the combined forecast. In fact, when using a convex loss, Jensen's inequality implies
$$ \ell(Y_{t+1},\widehat{Y}_{t+1}) \leq \sum_{k=1}^K\omega_t^{(k)}\ell(Y_{t+1},\widehat{Y}_{t+1}^{(k)}) = \overline{\ell}_t . $$
Hence, our regret bounds also provide upper bounds on the cumulative loss of the combined forecast, which may outperform the best individual expert.

}

\section{Ensemble Echo State Networks}
\label{section:ensemble_esn}

{\color{green}
In this section, we provide a concise overview of the Echo State Network (ESN) models that we employ in our empirical application. ESNs are nonlinear state-space models consisting of a nonlinear state equation with randomly sampled parameters that remain fixed during training (the {\it reservoir}), and a linear observation (readout) equation with coefficients estimated using ridge-regularized least squares. Due to the partial randomness inherent in ESN-type models, they are natural candidates for combination: we can construct many different ESNs via repeated sampling. In the following, we focus on presenting the generalization of these models to the multi-frequency time series setting proposed in \cite{ballarin2022reservoir}, which is the basis for our experts. More background information on ESNs can be found in Appendix~\ref{appendix:esn_in_brief}.
}%

\subsection{Multi-Frequency ESNs}
\label{subsec:Multi-Frequency ESNs}
\cite{ballarin2022reservoir} introduced a class of Echo State Network models designed for time series observed at multiple sampling frequencies. These models are collectively referred to as {\it Multi-Frequency ESNs} (MFESNs).\footnote{\color{green} For the readers unfamiliar with nonlinear state-space and reservoir computing models, we provide a self-contained discussion of Echo State Networks in Appendix~\ref{appendix:esn_in_brief}.}

The first category of MFESNs, known as {\it Single-reservoir MFESNs} (S-MFESNs), enables the  forecasting of lower-frequency targets by mapping high-frequency reservoir states to low-frequency outputs through a state alignment scheme. In an aligned S-MFESN, the forecast is generated based on the most recent state corresponding to the reference low-frequency time index $t$.
The second category, termed {\it Multi-reservoir MFESNs} (M-MFESNs), introduces multiple state equations, each associated with a distinct sampling frequency in the input data. For example, quarterly and monthly variables can be incorporated jointly as regressors. By employing multiple reservoirs evolving at their respective frequencies, these models flexibly capture the dynamics of different groups of input variables sampled at common frequencies.

We now briefly introduce the M-MFESN framework following \cite{ballarin2022reservoir}. For brevity, we do not discuss S-MFESN models explicitly, as they are empirically dominated by multi-reservoir MFESNs (see Section~\ref{section:application}). Suppose there are $Q$ groups of input time series with observations $\{ Z^{(q)}_{t,\tempo{s}{\kappa_q}} \}_{t,s}$,  where $Z^{(q)}_{t,\tempo{s}{\kappa_q}} \in \Real^{d_q}$, $\kappa_q\in \mathbb{N}^+$, $q \in [Q]$, $t \in [T]$, and $s \in \{0, \ldots, \kappa_q-1\}$, sampled at common frequencies $\{\kappa_1, \ldots, \kappa_Q\}$. Using the \textit{tempo} indexing of \cite{ballarin2022reservoir}, the high-frequency observations occur at fixed fractional sub-intervals relative to the reference low-frequency index $t$. For each frequency $\kappa_q $, the fractional index   $s=0$ indicates alignment with $t$, so that $(t,\tempo{0}{\kappa_q}) \equiv t$, and advancing $\kappa_q$ sub-intervals corresponds to one unit of low-frequency time, that is  $(t,\tempo{\kappa_q}{\kappa_q}) \equiv (t+1,\tempo{0}{\kappa_q}) \equiv t+1$.
For each input series group $q$, we define a frequency-specific $D_q$-dimensional reservoir state equation
\begin{equation*}
    X^{(q)}_{t,\tempo{s}{\kappa_q}} 
    = 
    \alpha_q X^{(q)}_{t,\tempo{s-1}{\kappa_q}} 
    + (1-\alpha_q) \sigma(\rho_q \overline{A}_q X^{(q)}_{t,\tempo{s-1}{\kappa_q}} + \gamma_q \overline{C}_q Z^{(q)}_{t,\tempo{s}{\kappa_q}} + \varsigma_q \overline{\zeta}_q) , \label{eq:m-mfesn_state}
\end{equation*}
where we use the convention $X_{t, \tempo{ -1} {\kappa_q}}^{(q)}:=X_{t-1,  \tempo{\kappa_q-1 }{\kappa_q}}^{(q)}$ and where $\overline{A}_q$, $\overline{C}_q$, and $\overline{\zeta}_q$ are group-specific normalized reservoir parameters, and ${\varphi}_q:= (\alpha_q, \rho_q, \gamma_q, \varsigma_q)$ collects  the corresponding reservoir hyperparameters, with $\alpha_q$ denoting the leak rate.
To construct the forecast at the reference time $t$, the states are aligned by stacking the frequency-specific states evaluated at $s=0$, that is,  $X_{t,Q} := (X^{(1)\top}_{t,\tempo{0}{\kappa_1}}, \ldots, X^{(Q)\top}_{t,\tempo{0}{\kappa_Q}} )\transp \in \Real^{D_Q}$, $D_Q := \sum_{q=1}^Q D_q$. The M-MFESN prediction equation is given by
\begin{align*}
    Y_{t+1} 
    &= b_Q + W_Q\transp X_{t,Q} + \epsilon_{t+1} 
    = b_Q + \sum_{q=1}^Q W_q\transp X^{(q)}_{t,\tempo{0}{\kappa_q}} + \epsilon_{t+1} .
\end{align*}
The readout $W_Q$ defines a linear map from the nonlinear, frequency-specific states to $Y_{t+1}$.

\subsection{Ensemble (Multi-Frequency) ESNs}
\label{subsection:esn_mfesns}

An important concern when working with models initialized or inherently constructed with random weights is the impact that such randomness may have on model performance.
A recent and growing body of work on the  
``lottery ticket'' hypothesis \citep{frankleLotteryTicketHypothesis2018,malachProvingLotteryTicket2020} argues that randomly initialized neural networks can contain subnetworks whose initial configurations are particularly favorable for training~\citep{maSanityChecksLottery2021,sreenivasanRareGemsFinding2022}. Related research has focused explicitly on designing weight distributions that are guaranteed to be a priori advantageous for initialization; see, e.g.,  \citealt{bolagerSamplingWeightsDeep2023}.
Echo State Networks, although proven to satisfy universal approximation and generalization guarantees~\citep{RC7,RC10,RC12}, also depend critically on the properties of randomly drawn reservoir parameters. 
Once used for forecasting, such systems may offer very different forecasting accuracy. It is hence suggestive to construct {\it ensembles of ESNs} (EN-ESN). In our multifrequency setting, we extend the ensemble approach to  MFESNs (EN-MFESN), considering both single- and multi-reservoir architectures. A multi-reservoir ensemble is indexed by two indices: $k$ for the expert within the ensemble, and $q$ for the frequency group within each MFESN expert.

Following the notation introduced in Section~\ref{subsec:Multi-Frequency ESNs}, consider an ensemble of $K \geq 1$ distinct multi-reservoir MFESN (M-MFESN) experts, indexed by $k\in[K]$, each comprising $Q$ frequency-specific reservoirs. The prediction of expert $k$ is given by  $\widehat{Y}^{(k)}_{t+1} = \widehat{b}^{(k)}_{Q, \lambda_k} + \widehat{W}^{(k)\top}_{Q,\lambda_k} X^{(k)}_{t,Q}$, where $\lambda_k$ denotes the ridge penalty of expert $k$, and the readout $\widehat{W}^{(k)\top}_{Q,\lambda_k}$ and $\widehat{b}^{(k)}_{Q, \lambda_k}$ are estimated as in \eqref{w_est}-\eqref{b_est}. The ensemble forecast is constructed according to \eqref{predictor} as
\begin{equation*}
    \widehat{Y}_{t+1}
    =
    \sum_{k=1}^K \omega^{(k)}_{t} \widehat{Y}^{(k)}_{t+1}
    =
    \widehat{b}_{K,t}
    +
    \sum_{k=1}^K \omega^{(k)}_{t} \widehat{W}^{(k)\top}_{Q,\lambda_k} X^{(k)}_{t,Q}
    , \enspace {\rm with} \enspace \widehat{b}_{K,t} := \sum_{k=1}^K \omega^{(k)}_{t} \widehat{b}^{(k)}_{Q,\lambda_k}.
\end{equation*}

The ensemble construction applies analogously to single-reservoir MFESNs. In Subsection~\ref{subsec:Models and Ensembles}, we introduce the four ensemble variants considered in our empirical application.

\section{GDP Forecasting with EN-MFESNs}
\label{section:application}

In this section, we evaluate and compare the one-step-ahead forecasting performance of our proposed EN-MFESN against individual MFESN models and state-of-the-art approaches. To ensure a fair comparison, we follow  \cite{ballarin2022reservoir} regarding the data, preprocessing procedures, and benchmark specifications. %

\subsection{Data}
We conduct our analysis for the original dataset used in \cite{ballarin2022reservoir}, for which we consider two predictor sets: \textit{small-MD}, comprising 9 predictors, and \textit{medium-MD}, consisting of 33 predictors. The data are collected at different sampling frequencies, while the target variable is U.S. quarterly GDP growth.
{\color{brown}
For the sake of space, in this section, we focus on the larger medium-MD data sample and a selection of core results; throughout our empirical analysis, we consider the one-step-ahead forecasting horizon only.\footnote{\color{brown}%
For further details on data sources, preprocessing, and availability, see Appendix~\ref{appendix:data} and Section~4.1 in \cite{ballarin2022reservoir}. To compare forecasts across multiple horizons, \cite{ballarin2022reservoir} uses an evaluation sample of $41$ quarters, whereas we use 48 quarters for one-step-ahead forecasting. 
}
}

\subsection{Models and Ensembles}
\label{subsec:Models and Ensembles}

\paragraph{Benchmarks.}
The first benchmark we consider is the in-sample mean of U.S. GDP growth computed over the estimation sample, which serves a natural baseline for comparing mean squared forecasting errors.
We also include an AR(1) model, standard in univariate time series analysis and macroeconomic forecasting \citep{Stock2002, baiForecastingEconomicTime2008}.
Our analysis demonstrates that AR(1) can outperform the in-sample mean benchmark and some ensemble specifications, making it a non-trivial reference for evaluating multi-frequency nonlinear state-space models.
Additionally, the two dynamic factor models (DFMs) from \cite{ballarin2022reservoir} are also included. DFMs are widely applied in macroeconomic analysis to recover low-dimensional common components from high-dimensional time series \citep{stock2016dynamic,Fuleky2020,barigozziDynamicFactorModels2024}.
\paragraph{MFESNs.}
Two classes of ESN-based forecasting models are considered. The first set comprises single-reservoir ESN models (S-MFESNs) with reservoir sizes  $D=30$ and $D=120$, respectively. The sparsity degree of state matrices is set to $10/D$. 
The second group of models includes multi-reservoir ESNs (M-MFESNs) with separate reservoirs for monthly and daily data input series, allowing us to capture the dynamics across both temporal resolutions.
For each model class, two MFESN variants are considered, denoted A and B, which differ in the reservoir size and sparsity for S-MFESNs and in the state hyperparameters for M-MFESNs.
Our implementation of MFESNs follows \cite{ballarin2022reservoir}  in the choice of architecture, reservoir sizes, activation functions, hyperparameters, and other relevant model-specific details. 
{\color{brown}
Time-series cross-validation (CV) is used to tune the ridge penalty separately for each MFESN expert, unless otherwise specified. CV  and ridge regression are performed using only the training sample, and all MFESN parameters are kept fixed during out-of-sample forecasting.
}%
A summary of our designs and specifications is provided in Appendix~\ref{appendix:models_summary}.%
{\color{green}
\paragraph{Ensembles.}
We consider MFESN ensembles obtained by resampling random reservoir parameters (RP), denoted EN-MFESN-RP, and ensembles that additionally vary the leak rate (EN-MFESN-$\alpha$RP). We also build ensembles that avoid cross-validation entirely: together with random coefficients, we additionally pick the ridge penalty of each MFESN expert from a pre-specified grid (EN-MFESN-$\lambda$RP), and, finally, we allow both $\alpha$ and $\lambda$ (EN-MFESN-$\alpha\lambda$RP) to vary.
For each class and type of multi-frequency ESN model (single- or multi-reservoir, type A or B), we generate $K$ variants of the ``baseline'' specification by independently resampling the random reservoir weight matrices. When varying the leak rate, we select $\alpha$ from a finite grid, $\alpha \in \{0.1, 0.3, 0.5, 0.7, 0.9\}$; for the ridge penalty, we pick $\lambda \in  \{0.01,0.1,1,10,100\}$.
To ensure comparability between ensemble specifications, a total ensemble size of $K=1000$ models is maintained throughout, allocating an equal number of models to each $\alpha$ or $\lambda$ value; in the $\alpha\lambda$RP setting, we assign $40$ models to each leak rate/ridge penalty pair, again to ensure equal allocation. 
}%
In the specific case of M-MFESN ensembles, we use the same leak rate for both monthly and daily frequency state equations.
We apply the direct combination strategies of {\color{green} median forecast (Median), simple averaging (SA),} and rolling MSE (RollMSE) (see Appendix~\ref{appendix:Direct Combination Methods}), as well as Constant Hedge (Hedge), Decreasing Hedge (DecHedge), Follow-the-Leader (FTL), and the Adaptive Hedge (AdaHedge) of \cite{rooijFollowLeaderIf2014}  introduced in Section~\ref{section:expert_ensembles}.
{\color{green} 
We use the standard squared loss both for the combinations and predictions evaluation. In Appendix~\ref{appendix:medium_md_additional}, we repeat our analysis with a quantile-clipped and scaled squared loss designed to lie in the unit interval (Assumption~\ref{assumption:losses_0-1}), and find that the forecasting performance of FTL and AdaHedge is relatively robust to this modification. Pseudocode for FTL and Hedge-type combination schemes is provided in Appendix~\ref{appendix:algorithms}.
}%

\begin{figure}[t!]
    \centering
    \begin{subfigure}[b]{0.45\textwidth}
        \includegraphics[width=\textwidth]{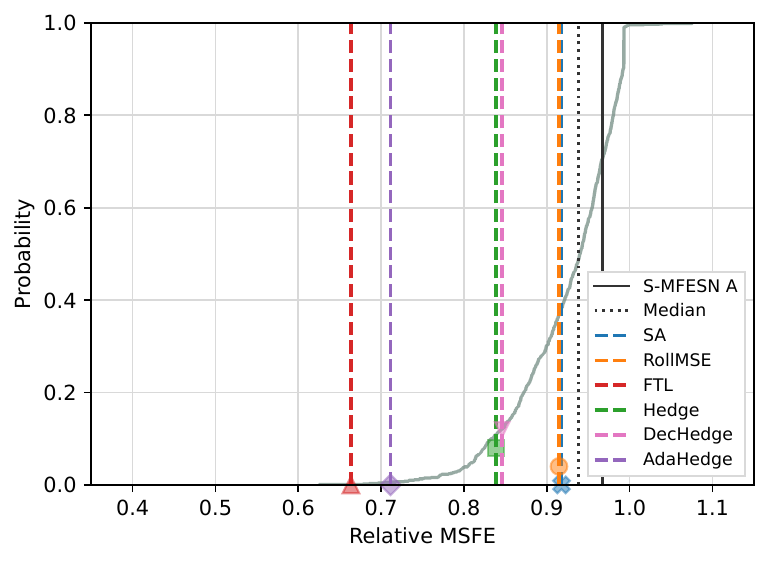}
        \caption{S-MFESN~A}
    \end{subfigure}
    \hfill
    \begin{subfigure}[b]{0.45\textwidth}
        \includegraphics[width=\textwidth]{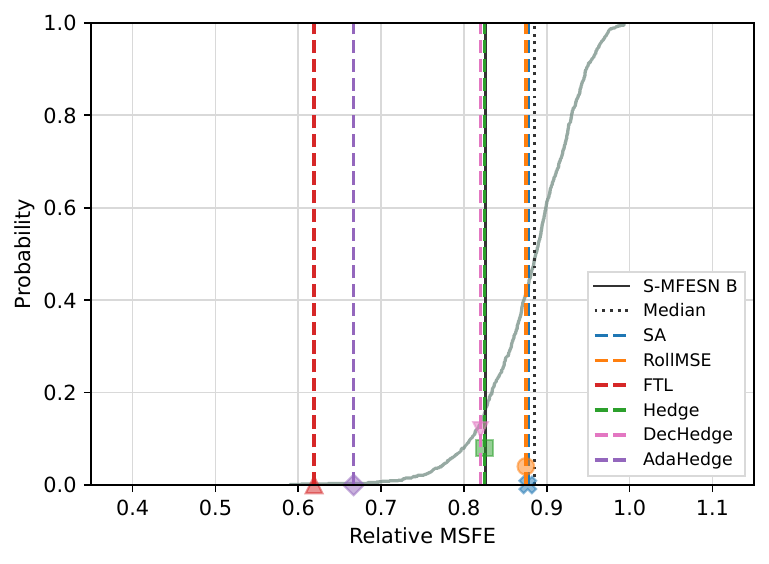}
        \caption{S-MFESN~B}
    \end{subfigure}
    
    \vspace{0.5cm} %
    
    \begin{subfigure}[b]{0.45\textwidth}
        \includegraphics[width=\textwidth]{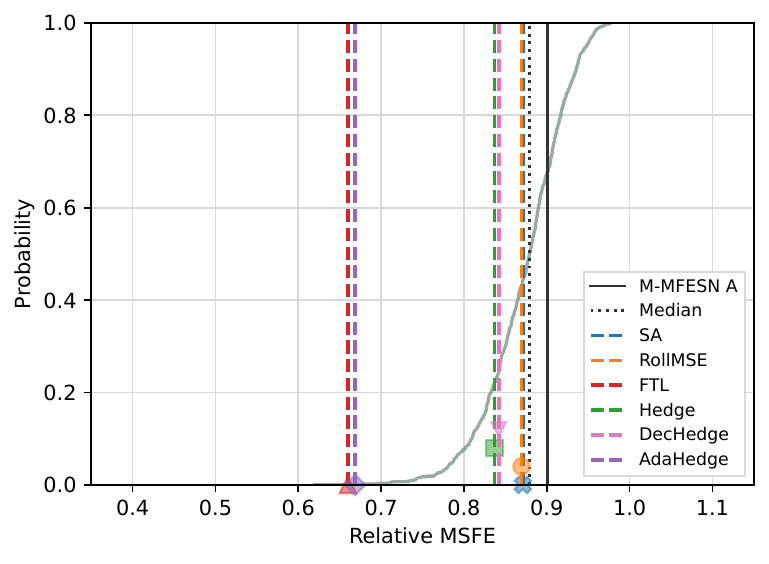}
        \caption{{\color{green} M-MFESN~A}}
    \end{subfigure}
    \hfill
    \begin{subfigure}[b]{0.45\textwidth}
        \includegraphics[width=\textwidth]{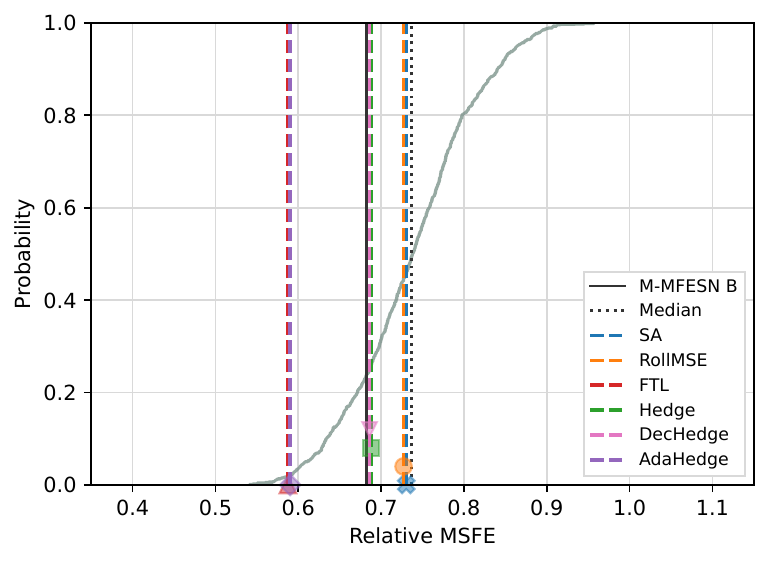}
        \caption{M-MFESN~B}
    \end{subfigure}
    \caption{EN-MFESN-RP ensembles of S-MFESN~A~(a), S-MFESN~B~(b), M-MFESN~A~(c), and M-MFESN~B~(d) models. Plots display MSFE relative to {\color{brown} the in-sample mean as a} baseline for ensemble aggregation methods (colored dashed lines), the Median combination performance (black dotted line), specific draws of S-MFESN~A/B and M-MFESN~A/B from \cite{ballarin2022reservoir} (black solid line), and the ECDF of the relative MSFE across individual ensemble models (gray curve). Colored markers are included to distinguish overlapping lines.}
    \label{fig:plot_ecdf_msfe_medium_seed}
\end{figure}

{\color{green}

\subsection{Ensemble Forecasting Performance}
\label{section:forecast_perf}

We first focus on analyzing two ensemble types, EN-MFESN-RP and EN-MFESN-$\alpha\lambda$RP, which represent two extremes in terms of expert variety. The former consists of experts sharing the same reservoir hyperparameters, while the latter additionally varies the leak rate and ridge penalty over predefined grids.
Our evaluation metric is the relative mean-squared forecasting error (MSFE) with respect to {\color{brown}the in-sample mean}.\footnote{Additional empirical results are reported in Appendix~\ref{appendix:additional_empirics}, including ensemble performance when grouping experts by hyperparameters, results for the small-MD dataset, robustness checks based on ensemble resampling, and the evolution of FTL and AdaHedge weights over the forecasting period.}

Figure~\ref{fig:plot_ecdf_msfe_medium_seed} presents the results for EN-MFESN-RP ensembles and reveals significant sensitivity to the random draw of reservoir state coefficients. Individual S-MFESN~A/B and M-MFESN~A/B model draws from the model specification of \citealt{ballarin2022reservoir} have heterogeneous performance: while B-type single- and multi-reservoir MFESNs perform favorably, A-type models often exhibit higher forecasting errors, in some cases exceeding the median relative MSFE with respect to {\color{brown}the in-sample mean}. 
Simple averaging and Hedge are similar in all Panels. The behavior of the latter depends on the choice of learning rate relative to the scale of cumulative loss differences between experts.
In contrast, both FTL and AdaHedge deliver consistent, sizable, and robust performance gains. We note that FTL achieves forecasting errors in or below the extreme left tail of the individual experts' MSFE distribution, indicating the benefits of adaptive expert selection.
Although FTL dominates AdaHedge throughout, the difference is higher for S-MFESN models than for M-MFESNs, consistent with their greater flexibility and the higher baseline accuracy of the latter.

When we enrich the pool of experts by drawing $\alpha$ and $\lambda$, we obtain further improvements in the performance of the best online combination methods across all MFESN model classes, as shown in Figure~\ref{figure:medium_alpha_lambda_common_grid}. Compared to EN-MFESN-RP, the best EN-MFESN-$\alpha\lambda$RP combinations achieve lower MSFE across all four model classes, with some experts from S- and M-MFESN~A designs reaching more than 50\% and MFESN-B up to 60\% reductions in relative error. The mean squared forecasting error ECDFs are quite different in shape compared to those in Figure~\ref{fig:plot_ecdf_msfe_medium_seed}, primarily because approximately $20\%$ of models are fitted with $\lambda = 100$, which collapses forecasts close to zero.
For M-MFESN~B, however, the direct combination methods yield worse results with the more diverse ensemble, whereas FTL and AdaHedge achieve substantial improvements (compare Panel~(d) in Figure~\ref{fig:plot_ecdf_msfe_medium_seed} and Figure~\ref{figure:medium_alpha_lambda_common_grid}).
This suggests that variety in forecasting methods leads to greater ensemble prediction accuracy when the combination strategy is able to eventually distinguish superior models from low performers.\footnote{Our result indeed echoes the recent finding of \cite{elliottCombiningForecastsWhy2026} that classical optimal combination weights can yield substantial gains primarily by discarding some forecasts and averaging the remaining ones.}

\begin{figure}[t!]
    \centering
    \begin{subfigure}[b]{0.45\textwidth}
        \includegraphics[width=\textwidth]{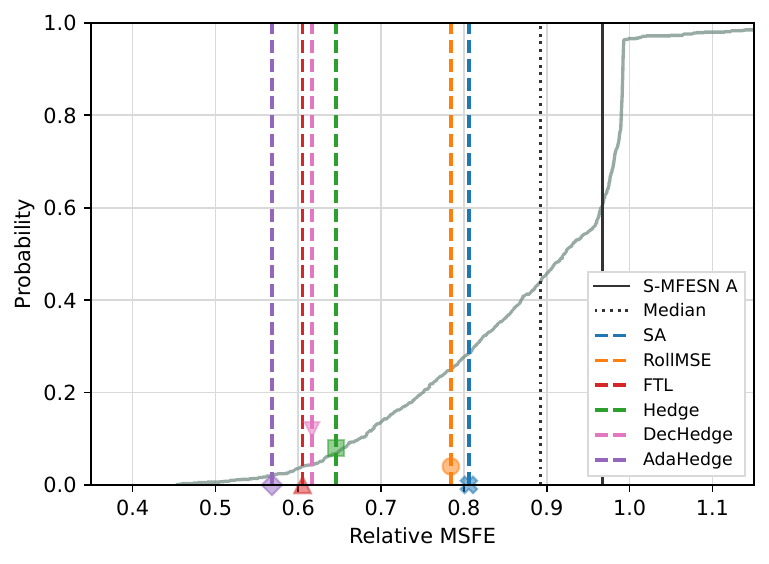}
        \caption{S-MFESN~A}
    \end{subfigure}
    \hfill
    \begin{subfigure}[b]{0.45\textwidth}
        \includegraphics[width=\textwidth]{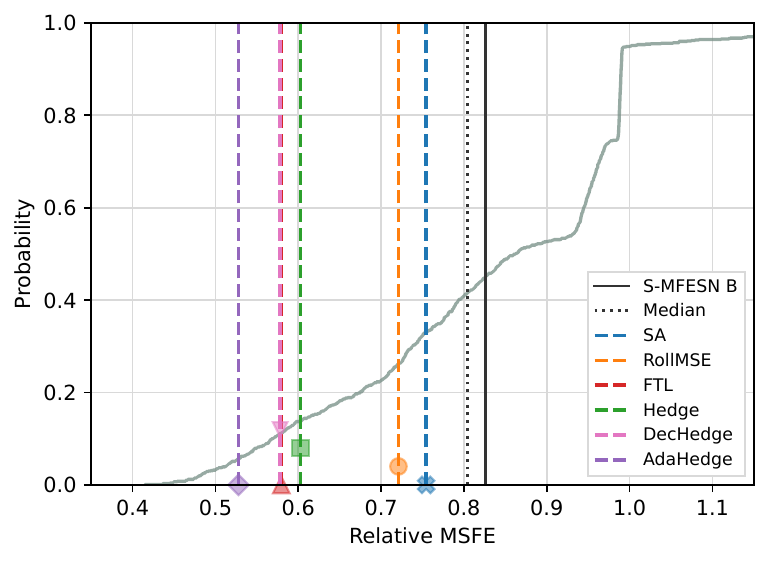}
        \caption{S-MFESN~B}
    \end{subfigure}

    \vspace{0.5cm}

    \begin{subfigure}[b]{0.45\textwidth}
        \includegraphics[width=\textwidth]{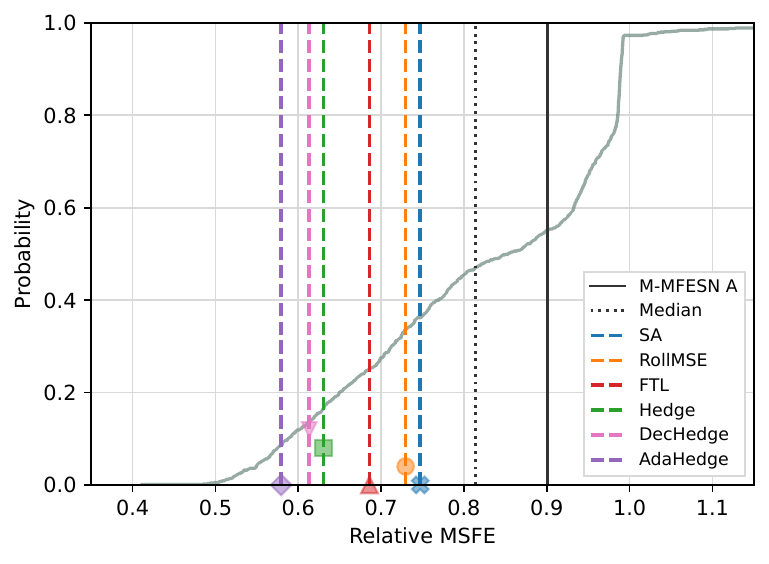}
        \caption{M-MFESN~A}
    \end{subfigure}
    \hfill
    \begin{subfigure}[b]{0.45\textwidth}
        \includegraphics[width=\textwidth]{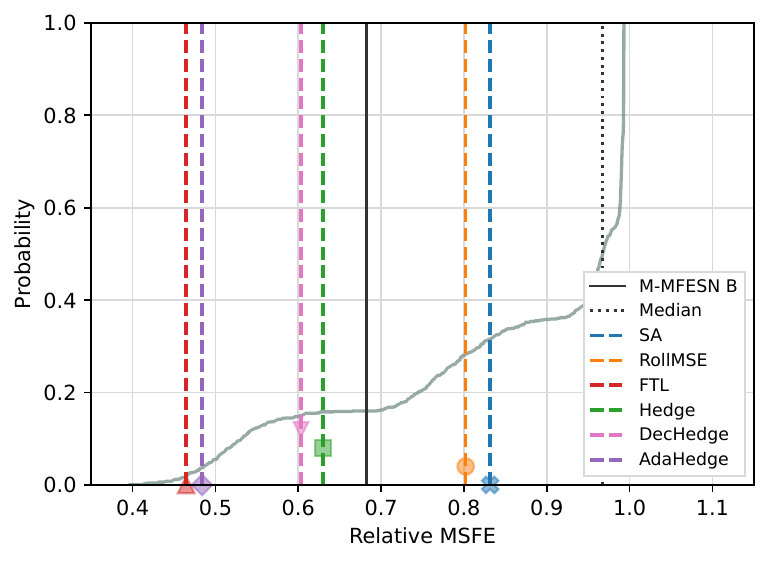}
        \caption{M-MFESN~B}
    \end{subfigure}

    \caption{EN-MFESN-$\alpha\lambda$RP ensembles of
    S-MFESN~A~(a), S-MFESN~B~(b), M-MFESN~A~(c), and
    M-MFESN~B~(d) models. Plots display MSFE relative to {\color{brown} the in-sample mean as a} baseline for ensemble aggregation methods (colored dashed lines), the Median combination performance (black dotted line), specific draws of S-MFESN~A/B and M-MFESN~A/B from \cite{ballarin2022reservoir} (black solid line), and the ECDF of the relative MSFE across individual ensemble models (gray curve). Colored markers are included to distinguish overlapping lines.}
    \label{figure:medium_alpha_lambda_common_grid}
\end{figure}

We summarize our results in Table~\ref{table:msfe_medium_results_all}, where benchmark forecasts are included for reference, as well as EN-MFESN-$\alpha$RP and EN-MFESN-$\lambda$RP ensemble designs (see Appendix~\ref{appendix:medium_md_additional} for additional plots). For benchmark methods, FTL, and Hedge-type combination schemes, we report relative MSFEs together with the associated percentage changes with respect to the in-sample mean.
Overall, FTL, decreasing Hedge, and AdaHedge are robust in achieving significant gains over the original MFESNs. More specifically, FTL is particularly effective for the M-MFESN~B models, while AdaHedge performs best in several other configurations, especially when the ridge penalty  $\lambda$ varies over a predefined grid.
In the EN-MFESN-RP setting, online ensemble combination yields meaningful forecasting error reductions: for the best model type, M-MFESN~B, FTL achieves reductions exceeding 41\% and {\color{green} 22}\% MSFE over the in-sample mean and AR(1), respectively. 
Even higher gains arise in the EN-MFESN-$\alpha$RP setting. For S-MFESN~A models, AdaHedge reduces MSFE by more than 42\%, while for M-MFESN~B, FTL achieves MSFE reductions exceeding 51\% and {\color{green} 36}\% relative to the in-sample mean and AR(1) benchmarks, respectively.

}

\begin{table}[p!]
\centering
\small
\setlength\tabcolsep{0pt}
\setlength\extrarowheight{1pt}
\linespread{1.01}\selectfont
\newcolumntype{C}{>{\centering\arraybackslash}X}
\color{green}
\begin{tabularx}{\textwidth}{@{} >{\centering\arraybackslash}p{2.5cm} @{\extracolsep{\fill}}*{8}{C} @{}}
    & & \multicolumn{7}{c}{Ensemble} \\
    \cline{3-9}
    \addlinespace[2pt]
    Model & Baseline & Median & SA & RollMSE & FTL & Hedge & DecHedge & AdaHedge \\
    \toprule
    Mean & 1.000 & -- & -- & -- & -- & -- & -- & -- \\
    AR(1) & 0.758 & -- & -- & -- & -- & -- & -- & -- \\
    DFM A & 0.841 & -- & -- & -- & -- & -- & -- & -- \\
    DFM B & 1.118 & -- & -- & -- & -- & -- & -- & -- \\
    \addlinespace[2pt]
    \multicolumn{9}{l}{\footnotesize EN-MFESN-RP (random parameters resampling)~\raisedrule[0.2em]{0.2pt}} \\[2pt]
    S-MFESN~A & 0.967 & 0.939 & 0.918 & 0.915 & {\bfseries 0.664} & 0.839 & 0.846 & \ul{0.711} \\[-4pt]
    & {\footnotesize \textit{-3.26\%}} & {\footnotesize \textit{-6.15\%}} & {\footnotesize \textit{-8.22\%}} & {\footnotesize \textit{-8.53\%}} & {\footnotesize\bfseries \textit{-33.64\%}} & {\footnotesize \textit{-16.12\%}} & {\footnotesize \textit{-15.38\%}} & {\footnotesize\ul{\textit{-28.86\%}}} \\
    S-MFESN~B & 0.826 & 0.885 & 0.877 & 0.875 & {\bfseries 0.619} & 0.825 & 0.820 & \ul{0.667} \\[-4pt]
    & {\footnotesize \textit{-17.39\%}} & {\footnotesize \textit{-11.46\%}} & {\footnotesize \textit{-12.29\%}} & {\footnotesize \textit{-12.53\%}} & {\footnotesize\bfseries \textit{-38.09\%}} & {\footnotesize \textit{-17.49\%}} & {\footnotesize \textit{-17.99\%}} & {\footnotesize\ul{\textit{-33.31\%}}} \\
    M-MFESN~A & 0.901 & 0.879 & 0.871 & 0.870 & {\bfseries 0.660} & 0.837 & 0.842 & \ul{0.669} \\[-4pt]
    & {\footnotesize \textit{-9.90\%}} & {\footnotesize \textit{-12.10\%}} & {\footnotesize \textit{-12.88\%}} & {\footnotesize \textit{-13.05\%}} & {\footnotesize\bfseries \textit{-34.00\%}} & {\footnotesize \textit{-16.34\%}} & {\footnotesize \textit{-15.80\%}} & {\footnotesize\ul{\textit{-33.13\%}}} \\
    M-MFESN~B & 0.682 & 0.737 & 0.730 & 0.727 & {\bfseries 0.587} & 0.688 & 0.686 & \ul{0.590} \\[-4pt]
    & {\footnotesize \textit{-31.79\%}} & {\footnotesize \textit{-26.29\%}} & {\footnotesize \textit{-26.96\%}} & {\footnotesize \textit{-27.29\%}} & {\footnotesize\bfseries \textit{-41.25\%}} & {\footnotesize \textit{-31.21\%}} & {\footnotesize \textit{-31.39\%}} & {\footnotesize\ul{\textit{-40.98\%}}} \\
    \addlinespace[2pt]
    \multicolumn{9}{l}{\footnotesize EN-MFESN-$\alpha$RP (random parameters resampling \& varying leak rates)~\raisedrule[0.2em]{0.2pt}} \\[2pt]
    S-MFESN~A & 0.967 & 0.896 & 0.836 & 0.820 & 0.670 & 0.667 & \ul{0.643} & {\bfseries 0.557} \\[-4pt]
    & {\footnotesize \textit{-3.26\%}} & {\footnotesize \textit{-10.44\%}} & {\footnotesize \textit{-16.44\%}} & {\footnotesize \textit{-17.97\%}} & {\footnotesize \textit{-32.98\%}} & {\footnotesize \textit{-33.27\%}} & {\footnotesize\ul{\textit{-35.66\%}}} & {\footnotesize\bfseries \textit{-44.25\%}} \\
    S-MFESN~B & 0.826 & 0.834 & 0.780 & 0.766 & {\bfseries 0.502} & 0.645 & 0.630 & \ul{0.518} \\[-4pt]
    & {\footnotesize \textit{-17.39\%}} & {\footnotesize \textit{-16.55\%}} & {\footnotesize \textit{-21.95\%}} & {\footnotesize \textit{-23.38\%}} & {\footnotesize\bfseries \textit{-49.80\%}} & {\footnotesize \textit{-35.47\%}} & {\footnotesize \textit{-37.01\%}} & {\footnotesize\ul{\textit{-48.20\%}}} \\
    M-MFESN~A & 0.901 & 0.871 & 0.829 & 0.818 & \ul{0.572} & 0.677 & 0.665 & {\bfseries 0.530} \\[-4pt]
    & {\footnotesize \textit{-9.90\%}} & {\footnotesize \textit{-12.95\%}} & {\footnotesize \textit{-17.13\%}} & {\footnotesize \textit{-18.16\%}} & {\footnotesize\ul{\textit{-42.78\%}}} & {\footnotesize \textit{-32.27\%}} & {\footnotesize \textit{-33.53\%}} & {\footnotesize\bfseries \textit{-47.04\%}} \\
    M-MFESN~B & 0.682 & 0.795 & 0.761 & 0.745 & {\bfseries 0.484} & 0.647 & 0.634 & \ul{0.527} \\[-4pt]
    & {\footnotesize \textit{-31.79\%}} & {\footnotesize \textit{-20.54\%}} & {\footnotesize \textit{-23.93\%}} & {\footnotesize \textit{-25.50\%}} & {\footnotesize\bfseries \textit{-51.58\%}} & {\footnotesize \textit{-35.27\%}} & {\footnotesize \textit{-36.55\%}} & {\footnotesize\ul{\textit{-47.29\%}}} \\
    \addlinespace[2pt]
    \multicolumn{9}{l}{\footnotesize EN-MFESN-$\lambda$RP (random parameters resampling \& varying ridge penalties)~\raisedrule[0.2em]{0.2pt}} \\[2pt]
    S-MFESN~A & 0.967 & 0.908 & 0.826 & 0.810 & 0.896 & 0.699 & \ul{0.675} & {\bfseries 0.626} \\[-4pt]
    & {\footnotesize \textit{-3.26\%}} & {\footnotesize \textit{-9.18\%}} & {\footnotesize \textit{-17.36\%}} & {\footnotesize \textit{-19.01\%}} & {\footnotesize \textit{-10.40\%}} & {\footnotesize \textit{-30.13\%}} & {\footnotesize\ul{\textit{-32.45\%}}} & {\footnotesize\bfseries \textit{-37.43\%}} \\
    S-MFESN~B & 0.826 & 0.846 & 0.807 & 0.780 & \ul{0.640} & 0.684 & 0.663 & {\bfseries 0.604} \\[-4pt]
    & {\footnotesize \textit{-17.39\%}} & {\footnotesize \textit{-15.36\%}} & {\footnotesize \textit{-19.29\%}} & {\footnotesize \textit{-21.99\%}} & {\footnotesize\ul{\textit{-35.98\%}}} & {\footnotesize \textit{-31.61\%}} & {\footnotesize \textit{-33.68\%}} & {\footnotesize\bfseries \textit{-39.60\%}} \\
    M-MFESN~A & 0.901 & 0.793 & 0.741 & 0.713 & 0.907 & 0.629 & \ul{0.616} & {\bfseries 0.588} \\[-4pt]
    & {\footnotesize \textit{-9.90\%}} & {\footnotesize \textit{-20.68\%}} & {\footnotesize \textit{-25.93\%}} & {\footnotesize \textit{-28.73\%}} & {\footnotesize \textit{-9.34\%}} & {\footnotesize \textit{-37.06\%}} & {\footnotesize\ul{\textit{-38.44\%}}} & {\footnotesize\bfseries \textit{-41.20\%}} \\
    M-MFESN~B & 0.682 & 0.957 & 0.785 & 0.749 & {\bfseries 0.470} & 0.588 & 0.558 & \ul{0.490} \\[-4pt]
    & {\footnotesize \textit{-31.79\%}} & {\footnotesize \textit{-4.34\%}} & {\footnotesize \textit{-21.48\%}} & {\footnotesize \textit{-25.12\%}} & {\footnotesize\bfseries \textit{-52.98\%}} & {\footnotesize \textit{-41.24\%}} & {\footnotesize \textit{-44.16\%}} & {\footnotesize\ul{\textit{-50.96\%}}} \\
    \addlinespace[2pt]
    \multicolumn{9}{l}{\footnotesize EN-MFESN-$\alpha\lambda$RP (random parameters resampling \& varying leak rates and ridge penalties)~\raisedrule[0.2em]{0.2pt}} \\[2pt]
    S-MFESN~A & 0.967 & 0.893 & 0.806 & 0.784 & \ul{0.605} & 0.646 & 0.617 & {\bfseries 0.568} \\[-4pt]
    & {\footnotesize \textit{-3.26\%}} & {\footnotesize \textit{-10.73\%}} & {\footnotesize \textit{-19.40\%}} & {\footnotesize \textit{-21.58\%}} & {\footnotesize\ul{\textit{-39.49\%}}} & {\footnotesize \textit{-35.44\%}} & {\footnotesize \textit{-38.35\%}} & {\footnotesize\bfseries \textit{-43.18\%}} \\
    S-MFESN~B & 0.826 & 0.804 & 0.754 & 0.721 & 0.579 & 0.603 & \ul{0.578} & {\bfseries 0.528} \\[-4pt]
    & {\footnotesize \textit{-17.39\%}} & {\footnotesize \textit{-19.59\%}} & {\footnotesize \textit{-24.56\%}} & {\footnotesize \textit{-27.90\%}} & {\footnotesize \textit{-42.05\%}} & {\footnotesize \textit{-39.75\%}} & {\footnotesize\ul{\textit{-42.19\%}}} & {\footnotesize\bfseries \textit{-47.22\%}} \\
    M-MFESN~A & 0.901 & 0.814 & 0.747 & 0.730 & 0.686 & 0.631 & \ul{0.613} & {\bfseries 0.579} \\[-4pt]
    & {\footnotesize \textit{-9.90\%}} & {\footnotesize \textit{-18.56\%}} & {\footnotesize \textit{-25.27\%}} & {\footnotesize \textit{-27.04\%}} & {\footnotesize \textit{-31.40\%}} & {\footnotesize \textit{-36.94\%}} & {\footnotesize\ul{\textit{-38.68\%}}} & {\footnotesize\bfseries \textit{-42.07\%}} \\
    M-MFESN~B & 0.682 & 0.967 & 0.832 & 0.802 & {\bfseries 0.465} & 0.630 & 0.603 & \ul{0.484} \\[-4pt]
    & {\footnotesize \textit{-31.79\%}} & {\footnotesize \textit{-3.26\%}} & {\footnotesize \textit{-16.83\%}} & {\footnotesize \textit{-19.79\%}} & {\footnotesize\bfseries \textit{-53.54\%}} & {\footnotesize \textit{-37.02\%}} & {\footnotesize \textit{-39.67\%}} & {\footnotesize\ul{\textit{-51.62\%}}} \\
    \bottomrule
\end{tabularx}
\color{black}
\caption{Relative MSFE of quarterly U.S. GDP growth predictions with respect to the in-sample mean. Baselines are benchmarks and models in \cite{ballarin2022reservoir}. Ensemble size is $K=1000$ for each MFESN specification. Performance changes in percentage with respect to {\color{green} the in-sample mean} are shown in italic. Best performing combinations are displayed in bold; second-best performing combinations are underlined.}
\label{table:msfe_medium_results_all}
\end{table}

\section{Conclusion}
\label{section:conclusion}

The ensemble combination approach is a well-studied framework that enables a forecaster to incorporate the predictions of a pool of models into a single, adaptive forecast. However, these techniques remain under-explored in the fields of applied statistics and econometrics.

Our work makes contributions in both theoretical and empirical directions. First, using Hoeffding- and Bernstein-type concentration inequalities, we derive bounds on the expected cumulative regret of Follow-the-Leader and decreasing Hedge under i.i.d.~and $\varphi$-mixing assumptions on the experts' losses. Under a positive sub-optimality gap, these bounds are uniform in the terminal time  $T$.
We then consider an application where model combination allows for robust and substantial performance improvements against state-of-the-art results: macroeconomic forecasting with mixed-frequency data using multi-frequency ESN models.
We find that the ensemble approach is highly effective at sharpening the performance of all MFESN models, with Follow-the-Leader and AdaHedge schemes being the overall best schemes; 
{\color{green}
these findings are confirmed when checking their robustness via ensemble resampling, loss thresholding, and restriction to a reduced set of regressors. 
}%
Our empirical analysis also shows that ensemble combination can alleviate the need for computationally effortful MFESN hyperparameter tuning by aggregating models with different leak rates and ridge penalties.

{\color{green}
A limitation of the expected regret guarantees we develop is that our simulations indicate they can be conservative. An interesting direction for future research on this topic is to consider alternative approaches for the analysis of online ensemble combination tailored specifically to the weakly-dependent stochastic setting.
}

{\color{green}
\paragraph{AI Disclosure Statement.}
OpenAI GPT-5.2, GPT-5.6 Luna, and Mistral Medium 3.5 were used during the revision to assist in reviewing theoretical results and in refining the manuscript. 
OpenAI Codex was used to assist in preparing code for some of the empirical experiments and simulations.
The entire manuscript has been written by the authors, who take full responsibility for the content.
}%

\bibliography{golib_rev1.bib}

\begin{thebibliography}{}

\bibitem[Alquier et~al., 2013]{alquierPredictionTimeSeries2013}
Alquier, P., Li, X., and Wintenberger, O. (2013).
\newblock {Prediction of time series by statistical learning: General losses
  and fast rates}.
\newblock {\em Dependence Modeling}, 1(2013):65--93.

\bibitem[Alquier and Wintenberger, 2012]{alquierFastRatesDependent2012}
Alquier, P. and Wintenberger, O. (2012).
\newblock Fast rates in learning with dependent observations.

\bibitem[Athey, 2019]{Athey2019}
Athey, S. (2019).
\newblock {The Impact of Machine Learning on Economics}.
\newblock In {\em The Economics of Artificial Intelligence: An Agenda}, pages
  507--547. University of Chicago Press.

\bibitem[Athey and Imbens, 2019]{atheyMachineLearningMethods2019}
Athey, S. and Imbens, G.~W. (2019).
\newblock {Machine learning methods that economists should know about}.
\newblock {\em Annual Review of Economics}, 11(1):685--725.

\bibitem[Auer and Chiang, 2016]{auer2016anAlgorithm}
Auer, P. and Chiang, C.-K. (2016).
\newblock {An algorithm with nearly optimal pseudo-regret for both stochastic
  and adversarial bandits}.
\newblock In {\em PMLR}, pages 116--120. PMLR.

\bibitem[Bai and Ng, 2008]{baiForecastingEconomicTime2008}
Bai, J. and Ng, S. (2008).
\newblock {Forecasting Economic Time Series Using Targeted Predictors}.
\newblock {\em Journal of Econometrics}, 146(2):304--317.

\bibitem[Ballarin et~al., 2024a]{ballarin2022reservoir}
Ballarin, G., Dellaportas, P., Grigoryeva, L., Hirt, M., van Huellen, S., and
  Ortega, J.-P. (2024a).
\newblock {Reservoir Computing for macroeconomic forecasting with mixed
  frequency data}.
\newblock {\em International Journal of Forecasting}, 40:1206--1237.

\bibitem[Ballarin et~al., 2024b]{RC23}
Ballarin, G., Grigoryeva, L., and Ortega, J.-P. (2024b).
\newblock {Memory of recurrent networks: Do we compute it right?}
\newblock {\em Journal of Machine Learning Research}, 25(243):1--38.

\bibitem[Barigozzi and Hallin, 2024]{barigozziDynamicFactorModels2024}
Barigozzi, M. and Hallin, M. (2024).
\newblock {Dynamic Factor Models: A Genealogy}.
\newblock In {Ngoc Thach}, N., Trung, N.~D., Ha, D.~T., and Kreinovich, V.,
  editors, {\em Partial Identification in Econometrics and Related Topics},
  pages 3--24. Springer Nature Switzerland, Cham.

\bibitem[Bates and Granger, 1969]{Bates1969}
Bates, J.~M. and Granger, C. (1969).
\newblock {The combination of forecasts}.
\newblock {\em Operations Research}, 20(4):451--468.

\bibitem[Bergmeir and Benitez, 2012]{bergmeir2012use}
Bergmeir, C. and Benitez, J.~M. (2012).
\newblock {On the use of cross-validation for time series predictor
  evaluation}.
\newblock {\em Information Sciences}, 191:192--213.

\bibitem[Bergmeir et~al., 2018]{bergmeir2018}
Bergmeir, C., Hyndman, R.~J., and Koo, B. (2018).
\newblock {A note on the validity of cross-validation for evaluating
  autoregressive time series prediction}.
\newblock {\em Computational Statistics \& Data Analysis}, 120:70--83.

\bibitem[Berlinet and Thomas-Agnan, 2011]{berlinet2011reproducing}
Berlinet, A. and Thomas-Agnan, C. (2011).
\newblock {\em {Reproducing Kernel Hilbert Spaces in Probability and
  Statistics}}.
\newblock Springer Science \& Business Media.

\bibitem[Bolager et~al., 2023]{bolagerSamplingWeightsDeep2023}
Bolager, E.~L., Burak, I., Datar, C., Sun, Q., and Dietrich, F. (2023).
\newblock {Sampling weights of deep neural networks}.
\newblock {\em Advances in Neural Information Processing Systems},
  36:63075--63116.

\bibitem[Bolhuis and Rayner, 2020]{bolhuisDeusExMachina2020}
Bolhuis, M.~A. and Rayner, B. (2020).
\newblock {Deus Ex Machina? A framework for macro forecasting with machine
  learning}.
\newblock {\em IMF Working Papers}, 2020(045).

\bibitem[Boot and Nibbering, 2019]{bootForecastingUsingRandom2019}
Boot, T. and Nibbering, D. (2019).
\newblock Forecasting using random subspace methods.
\newblock {\em Journal of Econometrics}, 209(2):391--406.

\bibitem[Boucheron et~al., 2013]{Boucheron2013}
Boucheron, S., Lugosi, G., and Massart, P. (2013).
\newblock {\em {Concentration Inequalities: A Nonasymptotic Theory of
  Independence}}.
\newblock Oxford University Press.

\bibitem[Boyd and Vandenberghe, 2004]{boyd2004convex}
Boyd, S.~P. and Vandenberghe, L. (2004).
\newblock {\em {Convex Optimization, Teil 1}}.
\newblock Cambridge University Press, Cambridge, England, UK.

\bibitem[Bradley, 2005]{bradley2005basic}
Bradley, R.~C. (2005).
\newblock Basic properties of strong mixing conditions. a survey and some open
  questions.
\newblock {\em Probability Surveys}, 2:107--144.

\bibitem[Brockwell and Davis, 2006]{BrocDavisYellowBook}
Brockwell, P.~J. and Davis, R.~A. (2006).
\newblock {\em {Time Series: Theory and Methods}}.
\newblock Springer-Verlag.

\bibitem[Casarin and Veggente, 2021]{casarinRandomProjectionMethods2021}
Casarin, R. and Veggente, V. (2021).
\newblock Random projection methods in economics and finance*.
\newblock In {\em The {{Essentials}} of {{Machine Learning}} in {{Finance}} and
  {{Accounting}}}. Routledge.

\bibitem[Cesa-Bianchi et~al., 1997]{cesa1997use}
Cesa-Bianchi, N., Freund, Y., Haussler, D., Helmbold, D.~P., Schapire, R.~E.,
  and Warmuth, M.~K. (1997).
\newblock {How to use expert advice}.
\newblock {\em Journal of the ACM (JACM)}, 44(3):427--485.

\bibitem[Cesa-Bianchi and Lugosi, 2006]{Cesa-Bianchi2006}
Cesa-Bianchi, N. and Lugosi, G. (2006).
\newblock {\em {Prediction, Learning, and Games}}.
\newblock Cambridge University Press.

\bibitem[Cheng and Hansen, 2015]{chengForecastingFactoraugmentedRegression2015}
Cheng, X. and Hansen, B.~E. (2015).
\newblock {Forecasting with factor-augmented regression: A frequentist model
  averaging approach}.
\newblock {\em Journal of Econometrics}, 186(2):280--293.

\bibitem[Chernov and Zhdanov, 2010]{chernovPredictionExpertAdvice2010}
Chernov, A. and Zhdanov, F. (2010).
\newblock {Prediction with Expert Advice under Discounted Loss}.
\newblock In Hutchison, D., Kanade, T., Kittler, J., Kleinberg, J.~M., Mattern,
  F., Mitchell, J.~C., Naor, M., Nierstrasz, O., {Pandu Rangan}, C., Steffen,
  B., Sudan, M., Terzopoulos, D., Tygar, D., Vardi, M.~Y., Weikum, G., Hutter,
  M., Stephan, F., Vovk, V., and Zeugmann, T., editors, {\em Algorithmic
  Learning Theory}, volume 6331, pages 255--269. Springer Berlin Heidelberg,
  Berlin, Heidelberg.

\bibitem[Christmann and Steinwart, 2008]{Steinwart2008}
Christmann, A. and Steinwart, I. (2008).
\newblock {\em {Support Vector Machines}}.
\newblock Springer.

\bibitem[Clemen, 1989]{clemen1989combining}
Clemen, R.~T. (1989).
\newblock {Combining forecasts: A review and annotated bibliography}.
\newblock {\em International journal of forecasting}, 5(4):559--583.

\bibitem[de~Rooij et~al., 2014]{rooijFollowLeaderIf2014}
de~Rooij, S., van Erven, T., Gr{\"{u}}nwald, P.~D., and Koolen, W.~M. (2014).
\newblock {Follow the Leader if you can, Hedge if you must}.
\newblock {\em Journal of Machine Learning Research}, 15(37):1281--1316.

\bibitem[Dedecker et~al., 2007]{Dedecker2007a}
Dedecker, J., Doukhan, P., Lang, G., Le{\'{o}}n, J.~R., Louhichi, S., and
  Prieur, C. (2007).
\newblock {\em {Weak Dependence: With Examples and Applications}}.
\newblock Springer Science+Business Media.

\bibitem[Degenne and Perchet, 2016]{degennePerchet2016}
Degenne, R. and Perchet, V. (2016).
\newblock Anytime optimal algorithms in stochastic multi-armed bandits.
\newblock In {\em Proceedings of the 33rd International Conference on Machine
  Learning}, volume~48 of {\em Proceedings of Machine Learning Research}, pages
  1587--1595. PMLR.

\bibitem[den Reijer and Johansson, 2019]{DenReijer2019}
den Reijer, A. and Johansson, A. (2019).
\newblock {Nowcasting Swedish GDP with a large and unbalanced data set}.
\newblock {\em Empirical Economics}, 57(4):1351--1373.

\bibitem[Dinh et~al., 2024]{dinhRandomSubspaceLocal2024}
Dinh, V.~H., Nibbering, D., and Wong, B. (2024).
\newblock Random {{Subspace Local Projections}}.
\newblock {\em The Review of Economics and Statistics}, pages 1--33.

\bibitem[Doukhan, 1995]{doukhan1995mixing}
Doukhan, P. (1995).
\newblock {\em Mixing: Properties and Examples}.
\newblock Springer-Verlag New York.

\bibitem[Doz and Fuleky, 2020]{Fuleky2020}
Doz, C. and Fuleky, P. (2020).
\newblock {\em {Dynamic Factor Models}}.
\newblock Springer International Publishing.

\bibitem[Elliott and Liao, 2026]{elliottCombiningForecastsWhy2026}
Elliott, G. and Liao, J. (2026).
\newblock Combining {{Forecasts}} -- {{On Why Averaging Beats Optimal Linear
  Weights}}.
\newblock {\em Journal of Business \& Economic Statistics}, 0(0):1--13.

\bibitem[Foster, 1991]{foster1991prediction}
Foster, D.~P. (1991).
\newblock {Prediction in the worst case}.
\newblock {\em The Annals of Statistics}, pages 1084--1090.

\bibitem[Foster and Vohra, 1993]{foster1993randomization}
Foster, D.~P. and Vohra, R.~V. (1993).
\newblock {A randomization rule for selecting forecasts}.
\newblock {\em Operations Research}, 41(4):704--709.

\bibitem[Frankle and Carbin, 2018]{frankleLotteryTicketHypothesis2018}
Frankle, J. and Carbin, M. (2018).
\newblock {The lottery ticket hypothesis: Finding sparse, trainable neural
  networks}.
\newblock In {\em International Conference on Learning Representations (ICLR)}.

\bibitem[Freund and Schapire, 1997]{Freund1997}
Freund, Y. and Schapire, R.~E. (1997).
\newblock {A decision-theoretic generalization of online learning and an
  application to boosting}.
\newblock {\em Journal of Computer and System Sciences}, 55(1):119--139.

\bibitem[Freund and Schapire, 1999]{Freund1999}
Freund, Y. and Schapire, R.~E. (1999).
\newblock {Adaptive game playing using multiplicative weights}.
\newblock {\em Games and Economic Behavior}, 29:79--103.

\bibitem[Fu et~al., 2023]{fuRobustRegressionGeneral2023}
Fu, S., Tian, Y., and Tang, L. (2023).
\newblock Robust regression under the general framework of bounded loss
  functions.
\newblock {\em European Journal of Operational Research}, 310(3):1325--1339.

\bibitem[Fuleky, 2020]{fulekyMacroeconomicForecastingEra2020}
Fuleky, P., editor (2020).
\newblock {\em {Macroeconomic forecasting in the era of big data: Theory and
  practice}}, volume~52 of {\em Advanced {{Studies}} in {{Theoretical}} and
  {{Applied Econometrics}}}.
\newblock Springer International Publishing, Cham.

\bibitem[Gasparin and Ramdas, 2025]{gasparinConformalOnlineModel2025}
Gasparin, M. and Ramdas, A. (2025).
\newblock {Conformal Online Model Aggregation}.

\bibitem[Geweke, 1977]{geweke1977dynamic}
Geweke, J. (1977).
\newblock {The dynamic factor analysis of economic time series}.
\newblock {\em Latent variables in socio-economic models}.

\bibitem[Ghysels et~al., 2007]{Ghysels2007}
Ghysels, E., Sinko, A., and Valkanov, R. (2007).
\newblock {MIDAS regressions: Further results and new directions}.
\newblock {\em Econometric Reviews}, 26(1):53--90.

\bibitem[Gonon et~al., 2020]{RC10}
Gonon, L., Grigoryeva, L., and Ortega, J.-P. (2020).
\newblock {Risk bounds for reservoir computing}.
\newblock {\em Journal of Machine Learning Research}, 21(240):1--61.

\bibitem[Gonon et~al., 2023]{RC12}
Gonon, L., Grigoryeva, L., and Ortega, J.-P. (2023).
\newblock {Approximation bounds for random neural networks and reservoir
  systems}.
\newblock {\em The Annals of Applied Probability}, 33(1):28--69.

\bibitem[Gonon and Ortega, 2020]{RC8}
Gonon, L. and Ortega, J.-P. (2020).
\newblock {Reservoir computing universality with stochastic inputs}.
\newblock {\em IEEE Transactions on Neural Networks and Learning Systems},
  31(1):100--112.

\bibitem[Granger, 1989]{grangerInvitedReviewCombining1989}
Granger, C. W.~J. (1989).
\newblock {Invited review combining forecasts -- twenty years later}.
\newblock {\em Journal of Forecasting}, 8(3):167--173.

\bibitem[Grigoryeva and Ortega, 2018a]{RC7}
Grigoryeva, L. and Ortega, J.-P. (2018a).
\newblock {Echo state networks are universal}.
\newblock {\em Neural Networks}, 108:495--508.

\bibitem[Grigoryeva and Ortega, 2018b]{RC6}
Grigoryeva, L. and Ortega, J.-P. (2018b).
\newblock {Universal discrete-time reservoir computers with stochastic inputs
  and linear readouts using non-homogeneous state-affine systems}.
\newblock {\em Journal of Machine Learning Research}, 19(24):1--40.

\bibitem[Hang, 2015]{hangStatisticalLearningKernelbased2015}
Hang, H. (2015).
\newblock {\em {Statistical Learning of Kernel-Based Methods for Non-i.i.d.
  Observations}}.
\newblock PhD thesis, University of Stuttgart.

\bibitem[Hang and Steinwart, 2017]{hangBernsteintypeInequalityMixing2017}
Hang, H. and Steinwart, I. (2017).
\newblock {A Bernstein-type inequality for some mixing processes and dynamical
  systems with an application to learning}.
\newblock {\em The Annals of Statistics}, 45(2):708--743.

\bibitem[Hansen, 2008]{hansenLeastsquaresForecastAveraging2008}
Hansen, B.~E. (2008).
\newblock {Least-squares forecast averaging}.
\newblock {\em Journal of Econometrics}, 146(2):342--350.

\bibitem[Hyndman and Athanasopoulos, 2018]{Hyndman2013}
Hyndman, R.~J. and Athanasopoulos, G. (2018).
\newblock {\em {Forecasting: Principles and Practice}}.
\newblock OTexts.

\bibitem[Ito et~al., 2024]{itoBestofBestWorlds2024}
Ito, S., Tsuchiya, T., and Honda, J. (2024).
\newblock {Adaptive learning rate for Follow-the-Regularized-Leader:
  Competitive analysis and best-of-both-worlds}.
\newblock {\em Proceedings of Machine Learning Research}, 247.

\bibitem[Jaeger, 2010]{jaeger2001}
Jaeger, H. (2010).
\newblock {The `echo state' approach to analysing and training recurrent neural
  networks with an erratum note}.
\newblock Technical report, German National Research Center for Information
  Technology.

\bibitem[Kallenberg, 2021]{kallenbergFoundationsModernProbability2021}
Kallenberg, O. (2021).
\newblock {\em Foundations of {{Modern Probability}}}, volume~99 of {\em
  Probability {{Theory}} and {{Stochastic Modelling}}}.
\newblock Springer International Publishing, Cham.

\bibitem[Khaleghi and Lugosi, 2023]{khaleghiInferringMixingProperties2023}
Khaleghi, A. and Lugosi, G. (2023).
\newblock Inferring the {{Mixing Properties}} of a {{Stationary Ergodic Process
  From}} a {{Single Sample-Path}}.
\newblock {\em IEEE Transactions on Information Theory}, 69(6):4014--4026.

\bibitem[Kim and Swanson, 2018]{kimMiningBigData2018}
Kim, H.~H. and Swanson, N.~R. (2018).
\newblock {Mining big data using parsimonious factor, machine learning,
  variable selection and shrinkage methods}.
\newblock {\em International Journal of Forecasting}, 34(2):339--354.

\bibitem[Koolen et~al., 2014]{koolenLearningLearningRate2014}
Koolen, W.~M., van Erven, T., and Gr{\"{u}}nwald, P. (2014).
\newblock {Learning the learning rate for prediction with expert advice}.
\newblock {\em Advances in Neural Information Processing Systems}, 27.

\bibitem[Kot{\l}owski, 2018]{kotlowski2018}
Kot{\l}owski, W. (2018).
\newblock On minimaxity of {Follow the Leader} strategy in the stochastic
  setting.
\newblock {\em Theoretical Computer Science}, 742:50--65.

\bibitem[Li and Racine, 2009]{Li2009}
Li, Q. and Racine, J. (2009).
\newblock {\em {Nonparametric Econometric Methods}}.
\newblock Number~25 in Advances in Econometrics. Emerald, Bingley, 1. ed
  edition.

\bibitem[Littlestone and Warmuth, 1994]{Littlestone1994}
Littlestone, N. and Warmuth, M.~K. (1994).
\newblock {The weighted majority algorithm}.
\newblock {\em Information and computation}, 108(2):212--261.

\bibitem[Luko{\v{s}}evi{\v{c}}ius and Jaeger, 2009]{lukosevicius}
Luko{\v{s}}evi{\v{c}}ius, M. and Jaeger, H. (2009).
\newblock {Reservoir computing approaches to recurrent neural network
  training}.
\newblock {\em Computer Science Review}, 3(3):127--149.

\bibitem[Luo and Schapire, 2015]{luoAchievingAllNo2015}
Luo, H. and Schapire, R.~E. (2015).
\newblock {Achieving all with no parameters: AdaNormalHedge}.
\newblock In {\em Proceedings of The 28th Conference on Learning Theory}, pages
  1286--1304. PMLR.

\bibitem[Ma et~al., 2021]{maSanityChecksLottery2021}
Ma, X., Yuan, G., Shen, X., Chen, T., Chen, X., Chen, X., Liu, N., Qin, M.,
  Liu, S., Wang, Z., and Wang, Y. (2021).
\newblock {Sanity checks for lottery tickets: Does your winning ticket really
  win the jackpot?}
\newblock In {\em Advances in Neural Information Processing Systems},
  volume~34, pages 12749--12760. Curran Associates, Inc.

\bibitem[Maass, 2011]{maass2}
Maass, W. (2011).
\newblock {Liquid state machines: Motivation, theory, and applications}.
\newblock In {Barry Cooper}, S.~S. and Sorbi, A., editors, {\em Computability
  In Context: Computation and Logic in the Real World}, chapter~8, pages
  275--296. World Scientific.

\bibitem[Malach et~al., 2020]{malachProvingLotteryTicket2020}
Malach, E., Yehudai, G., Shalev-Schwartz, S., and Shamir, O. (2020).
\newblock {Proving the lottery ticket hypothesis: Pruning is all you need}.
\newblock In {\em Proceedings of the 37th International Conference of Machine
  Learning}, pages 6682--6691. PMLR.

\bibitem[Masini et~al., 2023]{masiniMachineLearningAdvances2023}
Masini, R.~P., Medeiros, M.~C., and Mendes, E.~F. (2023).
\newblock {Machine learning advances for time series forecasting}.
\newblock {\em Journal of Economic Surveys}, 37(1):76--111.

\bibitem[McCracken and Ng, 2016]{McCracken2016}
McCracken, M.~W. and Ng, S. (2016).
\newblock {FRED-MD: A monthly database for macroeconomic research}.
\newblock {\em Journal of Business \& Economic Statistics}, 34(4):574--589.

\bibitem[McCracken and Ng, 2021]{McCracken2021}
McCracken, M.~W. and Ng, S. (2021).
\newblock {FRED-QD: A quarterly database for macroeconomic research}.
\newblock {\em {Federal Reserve Bank of St. Louis Review}}, First Quarter
  2021:1--44.

\bibitem[Mourtada and Ga{\"{i}}ffas,
  2019]{mourtadaOptimalityHedgeAlgorithm2019}
Mourtada, J. and Ga{\"{i}}ffas, S. (2019).
\newblock {On the optimality of the Hedge algorithm in the stochastic regime}.
\newblock {\em Journal of Machine Learning Research}, 20(83):1--28.

\bibitem[Rio, 2017]{rioAsymptoticTheoryWeakly2017}
Rio, E. (2017).
\newblock {\em {Asymptotic Theory of Weakly Dependent Random Processes}},
  volume~80 of {\em Probability {{Theory}} and {{Stochastic Modelling}}}.
\newblock Springer Berlin Heidelberg, Berlin, Heidelberg.

\bibitem[Samson, 2000]{samsonConcentrationMeasureInequalities2000}
Samson, P.-M. (2000).
\newblock {Concentration of measure inequalities for Markov chains and
  $\phi$-mixing processes}.
\newblock {\em The Annals of Probability}, 28(1):416--461.

\bibitem[Sargent and Sims, 1977]{sargent1977business}
Sargent, T.~J. and Sims, C.~A. (1977).
\newblock {Business cycle modeling without pretending to have too much a priori
  economic theory}.
\newblock {\em New methods in business cycle research}, 1:145--168.

\bibitem[Sch{\"{o}}lkopf et~al., 2002]{scholkopf2002learning}
Sch{\"{o}}lkopf, B., Smola, A.~J., Bach, F., and Others (2002).
\newblock {\em {Learning with kernels: support vector machines, regularization,
  optimization, and beyond}}.
\newblock MIT press.

\bibitem[Shalev-Shwartz, 2012]{shalev_online}
Shalev-Shwartz, S. (2012).
\newblock {Online learning and online convex optimization}.
\newblock {\em Found. Trends Mach. Learn.}, 4(2):107--194.

\bibitem[Sreenivasan et~al., 2022]{sreenivasanRareGemsFinding2022}
Sreenivasan, K., Sohn, J.-y., Yang, L., Grinde, M., Nagle, A., Wang, H., Xing,
  E., Lee, K., and Papailiopoulos, D. (2022).
\newblock {Rare gems: Finding lottery tickets at initialization}.
\newblock {\em Advances in Neural Information Processing Systems},
  35:14529--14540.

\bibitem[Stock and Watson, 1996]{Stock1996}
Stock, J.~H. and Watson, M.~W. (1996).
\newblock {Evidence on structural instability in macroeconomic time series
  relations}.
\newblock {\em Journal of Business \& Economic Statistics}, 14(1):11--30.

\bibitem[Stock and Watson, 2002]{Stock2002}
Stock, J.~H. and Watson, M.~W. (2002).
\newblock {Macroeconomic forecasting using diffusion indexes}.
\newblock {\em Journal of Business \& Economic Statistics}, 20(2):147--162.

\bibitem[Stock and Watson, 2016]{stock2016dynamic}
Stock, J.~H. and Watson, M.~W. (2016).
\newblock {Dynamic factor models, factor-augmented vector autoregressions, and
  structural vector autoregressions in macroeconomics}.
\newblock In {\em Handbook of macroeconomics}, volume~2, pages 415--525.
  Elsevier.

\bibitem[Stone, 1982]{stoneOptimalGlobalRates1982}
Stone, C.~J. (1982).
\newblock {Optimal global rates of convergence for nonparametric regression}.
\newblock {\em The Annals of Statistics}, 10(4).

\bibitem[Timmermann, 2006]{timmermann2006forecast}
Timmermann, A. (2006).
\newblock {Chapter 4 Forecast Combinations}.
\newblock In {\em Handbook of Economic Forecasting}, volume~1, pages 135--196.
  Elsevier, Waltham, MA, USA.

\bibitem[Tsybakov, 2009]{Tsybakov2009}
Tsybakov, A.~B. (2009).
\newblock {\em {Introduction to Nonparametric Estimation}}.
\newblock Springer Series in Statistics. Springer, New York ; London.

\bibitem[van Erven et~al., 2011]{vanervenAdaptiveHedge2011}
van Erven, T., Koolen, W.~M., de~Rooij, S., and Gr{\"{u}}nwald, P. (2011).
\newblock {Adaptive Hedge}.
\newblock In {\em Advances in Neural Information Processing Systems},
  volume~24. Curran Associates, Inc.

\bibitem[Vovk, 1995]{vovk1995game}
Vovk, V. (1995).
\newblock {A game of prediction with expert advice}.
\newblock In {\em Proceedings of the eighth annual conference on Computational
  learning theory}, pages 51--60.

\bibitem[Wallis, 2011]{wallisCombiningForecastsForty2011}
Wallis, K.~F. (2011).
\newblock {Combining forecasts -- forty years later}.
\newblock {\em Applied Financial Economics}, 21(1-2):33--41.

\bibitem[Wang et~al., 2023]{WANG20231518}
Wang, X., Hyndman, R.~J., Li, F., and Kang, Y. (2023).
\newblock {Forecast combinations: An over 50-year review}.
\newblock {\em International Journal of Forecasting}, 39(4):1518--1547.

\bibitem[Winkler and Makridakis, 1983]{Winkler1983}
Winkler, R.~L. and Makridakis, S. (1983).
\newblock {The Combination of Forecasts}.
\newblock {\em Journal of the Royal Statistical Society. Series A (General)},
  146(2):150--157.

\bibitem[Wintenberger, 2017]{wintenbergerOptimalLearningBernstein2017}
Wintenberger, O. (2017).
\newblock {Optimal learning with Bernstein online aggregation}.
\newblock {\em Machine Learning}, 106(1):119--141.

\bibitem[Wintenberger, 2024]{wintenbergerStochasticOnlineConvex2024}
Wintenberger, O. (2024).
\newblock {Stochastic online convex optimization. Application to probabilistic
  time series forecasting}.
\newblock {\em Electronic Journal of Statistics}, 18(1):429--464.

\bibitem[Yang, 2004]{yangCOMBININGFORECASTINGPROCEDURES2004}
Yang, Y. (2004).
\newblock {{Combining Forecasting Procedures}}: {{Some Theoretical Results}}.
\newblock {\em Econometric Theory}, 20(1):176--222.

\end{thebibliography}

\newpage
\appendix
\onehalfspacing
\counterwithin*{table}{section}
\renewcommand{\thetable}{\Alph{section}.\arabic{table}}
\counterwithin*{figure}{section}
\renewcommand{\thefigure}{\Alph{section}.\arabic{figure}}

\begin{center}
    {\LARGE Appendix}\\[1em]
    {\Large From Many Models, One: Macroeconomic Forecasting with Reservoir Ensembles}\\[1em]
    {\large%
        Giovanni Ballarin 
        \hspace{1em}
        Lyudmila Grigoryeva
        \hspace{1em}
        Yui Ching Li
    }\\[3em]
\end{center}

\section{Technical Results}

\subsection{Basic Results}

\begin{proposition}[\citealp{chernovPredictionExpertAdvice2010}]\label{proposition:regret_hedge}
    Let $\eta_1, \eta_2, \ldots$ be a decreasing sequence of positive learning rates, and assume that all expert losses lie in $[0,1]$. The Hedge algorithm satisfies the following regret bound:
    \begin{equation*}\label{eq:hedge_regret}
        \overline{R}_{\textsc{hdg},T}\leq \frac{1}{\eta_T}\log K +\frac{1}{8}\sum_{t=1}^{T}\eta_t.
    \end{equation*}
    The choice of $\eta_t=2\sqrt{\log K/t}$ yields a regret bound of $\sqrt{T \log K}$ for every $T\geq 1$. 
\end{proposition}
\begin{lemma}\label{lemma:elementary exp inequality} Let $a,b > 0$ be constants. Then  for any $x \geq 0$
 \begin{equation*}
 e^{-x/a}+e^{-x/b} \leq 2 e^{-{x}/{\max\{a,b\}}}.
 \end{equation*}
\end{lemma}

\begin{lemma}[\citealp{mourtadaOptimalityHedgeAlgorithm2019}, Lemma 13]\label{lemma:exp_series_limit}
    For every $\alpha > 0$,
    \begin{equation*}
        \sum_{t = 1}^{\infty} e^{-\alpha t}  \leq \frac{1}{\alpha} 
        \quad\text{and}\quad
        \sum_{t = 1}^{\infty} e^{-\alpha \sqrt{t}}  \leq \frac{2}{\alpha^2} .
    \end{equation*}
\end{lemma}

\begin{lemma}\label{lemma:logtail_t0}
    Let $b > 0$ and $\alpha > 0$ be constants. Then, there exists $\nu^* > 0$ such that, setting $t^* = \nu^* b$, it holds
    \begin{equation*}
        \frac{t^*}{(\log t^*)^\alpha} \geq b .
    \end{equation*}
\end{lemma}

\begin{proof}[Proof of Lemma~\ref{lemma:logtail_t0}]
    Consider the function $g(t)= t/(\log t)^\alpha$, $t>1$. For all $t \geq e^\alpha$, {\color{green} $g(t)$}  is strictly increasing and continuous.
    Now, letting $t = \nu b$ for some $\nu \geq b$, it clearly holds
    \begin{equation}
    \label{eq:tlogt}
        g(t)=\frac{ \nu b }{(\log \nu + \log b)^\alpha} 
        \geq
        \frac{ \nu b }{(2 \log \nu)^\alpha}
    \end{equation}
    since the denominator of the left-hand side is smaller than on the right-hand side by the lower bound assumed on $\nu$.
    Hence, we can define
    \begin{equation*}
        \nu^* 
        := 
        \inf \left\{ \nu  \geq \max\{ e^\alpha, b, {\color{green} e^{\alpha}/b}\} \,\big\vert\, \nu \geq (2 \log \nu)^\alpha \right\}.
    \end{equation*}
    By continuity of $g(t)$, $\nu^*$ is also in the set, that is,
    \begin{equation*}
        \frac{ \nu^* }{(2 \log \nu^*)^\alpha} \geq 1
    \end{equation*}
    and hence
    \begin{equation*}
        \frac{ \nu^* b }{(2 \log \nu^*)^\alpha} \geq b .
    \end{equation*}
    Setting $t^* = \nu^* b$, and using \eqref{eq:tlogt}, we obtain $g(t^*)
        \geq 
        b $ as required.
\end{proof}

\subsection{Concentration Inequalities under Independence and Weak Dependence}
\label{appendix:weak_depend}

First, we recall the formal definition of $\varphi$-mixing, following \citealp{bradley2005basic}.

\begin{definition}[$\varphi$-mixing coefficients, \citealp{bradley2005basic}, Definitions in (1.2) and (2.1)-(2.2)]
\label{def:rio_uniform_mixing}
    Let $\{U_t\}_{t\in\mathbb Z}$ be a sequence of random variables taking values in a measurable space. For $m\in\mathbb Z$, define
    $
    \mathcal{F}_m:=\sigma(U_s\mid s\leq m)$ and $ \mathcal{G}_m:=\sigma(U_s\mid s\geq m)$.
    For two sub-$\sigma$-algebras $\mathcal{A}$ and $\mathcal{B}$, define
    \begin{equation*}
        \varphi(\mathcal{A},\mathcal{B}):=\sup_{\substack{A\in\mathcal{A},\;B\in\mathcal{B}\\\P(A)>0}}\big| \P(B\mid A)-\P(B) \big|.
    \end{equation*}
    The $\varphi$-mixing coefficients of $\{U_t\}_{t\in\mathbb Z}$ are defined by
    \begin{equation*}
        \varphi_0:=1 \quad \text{and} \quad  \varphi_n:=\sup_{m\in\mathbb Z}\varphi(\mathcal{F}_m,\mathcal{G}_{m+n}),\quad n\geq1.
    \end{equation*}
    The sequence is said to be $\varphi$-mixing if $\varphi_n\to0$ as $n\to\infty$.
\end{definition}

\paragraph{Hoeffding-type bounds.}
{\color{green} 
We recall the classical Hoeffding inequality for independent bounded random variables and its extension under $\varphi$-mixing conditions.
\begin{lemma}[Hoeffding inequality for independent random variables, \citealp{Boucheron2013}, Theorem 2.8]
\label{lemma:iid_hoeffding}
Let $U_1,\ldots,U_n$ be independent, centered, real-valued random variables such that $a_i\leq U_i\leq b_i$ almost surely, for all $i\in[n]$. Then, for any $\xi>0$,
\begin{equation*}
\P\left(\sum_{i=1}^n U_i\geq\xi\right) \leq\exp\left(-\frac{2\xi^2}{\sum_{i=1}^n(b_i-a_i)^2}\right).
\end{equation*}
The same bound holds for $\P(\sum_{i=1}^n U_i\leq-\xi)$.
Consequently,
\begin{equation*}
\P\left(\left|\sum_{i=1}^n U_i\right|\geq\xi\right)\leq
2\exp\left(-\frac{2\xi^2}{\sum_{i=1}^n(b_i-a_i)^2}\right).
\end{equation*}
\end{lemma}
}

\begin{lemma}[Hoeffding inequality for $\varphi$-mixing sequences, \citealp{rioAsymptoticTheoryWeakly2017}, Corollary 2.1]\label{lemma:rio_hoeffding_uniform_mixing}
    Let $\{U_i\}_{i \geq 0}$ be a sequence of centered and real-valued bounded random variables such that $\abs{ U_i }^2 \leq M_i$. 
    Set $\theta_n := 1 + 4 \sum_{i=1}^{n-1} \varphi_i$. Then, for any $\xi>0$,
    \begin{equation*}
        \P\left(\left\lvert \sum_{i=1}^n U_i \right\rvert \geq \xi \right)
        \leq
        \sqrt{e} \exp\left( - \frac{\xi^2}{2 \theta_n \sum_{i=1}^n {\color{green} M_i}} \right) .
    \end{equation*}
\end{lemma}

\paragraph{Bernstein-type bounds.}

Bernstein's inequality provides a more general bound than Hoeffding's by taking into account the second moments of the random variables involved.
\begin{lemma}[Bernstein inequality  for independent random variables, \citealp{Boucheron2013}, Corollary 2.11]\label{lemma:iid_bernstein}
	Let $U_1,\ldots,U_n$ be independent centered real-valued bounded random variables such that $\abs{ U_i } \leq M$ for all $i$. Let $\Var(U_i)=v<\infty$ for all $i$. Then for any $\xi>0$,
    \begin{equation*}
        \P\left( \sum_{i=1}^n U_i  \geq \xi \right)
        \leq
        \exp\left( - \frac{\xi^2}{2 n v +\frac{2}{3} M \xi} \right) .
    \end{equation*}
\end{lemma}
\noindent {\color{green} The boundedness assumption ensures that the moment conditions of Theorem~2.10 in \cite{Boucheron2013} are satisfied with $v_0=nv$ and $c=M/3$.}

We recall two generalizations of classical Bernstein which apply to the $\varphi$-mixing setting.%

\begin{lemma}[Bernstein inequality for $\varphi$-mixing sequences, \cite{samsonConcentrationMeasureInequalities2000} and  \cite{alquierFastRatesDependent2012}, Lemma~3]\label{lemma:samson_phi_mixing}
    Let $\{U_i\}_{i \geq 0}$ be a real-valued strictly stationary, $\varphi$-mixing process such that $\abs{U_i}\leq M$ almost surely for all $i$, where $M>0$. Let $\varphi_i$ denote its $\varphi$-mixing coefficients and let $\Var(U_i)=v<\infty$ for all $i$.
    Define $\rho_n :=  1 + \sum_{i=1}^{n} \sqrt{\varphi_i} $.
    Then for any $\xi>0$,
    \begin{equation}
    \label{eq:right_bound}
        \P\left( \sum_{i=1}^n (U_i-\E[U_i])\geq\xi \right)
        \leq
        \exp\left( 
           -\frac{\xi^2}
    {32\rho_n nv+2M\rho_n^2\xi}\right)
    \end{equation}
    and
    \begin{equation}
    \label{eq:left_bound}
        \P\left( \sum_{i=1}^n (U_i-\E[U_i])\leq-\xi \right)
        \leq
        \exp\left( 
           -\frac{\xi^2}
    {32\rho_n nv+2M\rho_n^2\xi}\right).
    \end{equation}
\end{lemma}

\begin{proof}
The bounds follow from the exponential-moment inequality of \cite{samsonConcentrationMeasureInequalities2000} (see proof of Theorem~3, p.~460), as stated in Lemma~3 of \cite{alquierFastRatesDependent2012} (see also \cite{alquierPredictionTimeSeries2013}). Indeed, for $S_n:=\sum_{i=1}^n(U_i-\E[U_i])$, this inequality gives $\log\E[\exp(\lambda S_n)]\leq 8\rho_n nv\lambda^2$ for all $0\leq\lambda\leq 1/(M\rho_n^2)$. 

By the exponential Markov inequality, $\P(S_n\geq\xi)\leq \exp\left(-\lambda\xi+8\rho_n nv\lambda^2 \right)$. Choosing $\lambda=
{\xi}/({16\rho_n nv+M\rho_n^2\xi})$ ensures that $0<\lambda\leq1/(M\rho_n^2)$ and yields
    \begin{equation*}
        -\lambda\xi+8\rho_n nv\lambda^2\leq -\frac{\xi^2} {32\rho_n nv+2M\rho_n^2\xi},
    \end{equation*}
    which proves \eqref{eq:right_bound}. Applying the same argument to $\{-U_i\}_{i\geq0}$ gives \eqref{eq:left_bound}.
\end{proof}

\noindent
The next result provides a Bernstein-type inequality for geometrically $\varphi$-mixing sequences following Theorem 4.7 of \citealp{hangStatisticalLearningKernelbased2015} (see also Chapter 4 thereby for a detailed discussion).
The price to pay for this generalization is, however, an additional logarithmic factor in the denominator of the exponential which is tied to the dependence structure of the process. 

\begin{theorem}[Bernstein inequality for geometrically $\varphi$-mixing, \cite{hangBernsteintypeInequalityMixing2017} and \citealp{hangStatisticalLearningKernelbased2015}, Theorem~4.7]\label{theorem:huang_C_mixing}
Let $\{U_i\}_{i \geq 0}$ be a real-valued stationary geometrically $\varphi$-mixing process with coefficients satisfying $\varphi_n\leq c \exp{(-b n ^\gamma)}$, $n\geq 1$, for some $b>0$, $c> 0$, and $\gamma>0$.
Let $h:\mathbb{R}\to\mathbb R$ be bounded, measurable with $\|h\|_\infty\le M$, $\E[h(U_i)] = 0$, and
$\Var(h(U_i))\le \sigma^2$, where $M>0$. Then there exists
\begin{equation*}
    n_0   =   \max \big \{\min  \left\{m\ge 3 \:\vert\: m^2  \ge  {3232c} \textnormal{ and } {m}{(\log m)^{-2/\gamma}}  \ge  4  \right\}, e^{3/b}  \big\}
\end{equation*}
such that, for all $n\ge n_0$ and all $\varepsilon>0$,
\begin{equation*}
\P  \left(\sum_{i=1}^n h(U_i) \ge \varepsilon\right)
  \le  
2\exp  \left(
-\frac{\varepsilon^2}{ 8(\log n)^{2/\gamma}  \left({\color{green} n}\sigma^2+\varepsilon M/3  \right)}
\right).
\end{equation*}
\end{theorem}
The result follows from Theorem~4.7 of \cite{hangStatisticalLearningKernelbased2015} by taking the function seminorm to be identically zero and setting $A=M$ and $B=M$ in the notation of the original theorem.
\newpage
\section{Proofs}

This section contains the proofs for the theoretical results in the paper, as well as all intermediate technical results needed to derive them. 
First, we present the main concentration results for the experts' losses; then, we obtain regret bounds for FTL and Hedge methods.

\subsection{Concentration Results for FTL}
\label{appendix:proofs_concentration_ftl}
We exploit the theory from Appendix~\ref{appendix:weak_depend} to construct concentration-type bounds for the cumulative losses of experts under both i.i.d.~and dependent settings.

\subsubsection{Hoeffding-type}

\begin{lemma}[Hoeffding-type]
\label{lemma:ftl_concentr_lemma}
    Let Assumptions~\ref{assumption:losses_0-1}-\ref{assumption:losses_iid_or_mixing} hold and, without loss of generality, assume that the experts are ordered such that    %
    $
        \mu_1 < \mu_2 \leq \ldots \leq \mu_K 
    $, that is $\Delta = \mu_2-\mu_1>0$.
    \begin{description}
        \item[(i)] If Assumption~\ref{assumption:losses_iid_or_mixing}(i) holds, then
        \begin{equation*}
            \P\left( L^{(1)}_t > \min_{2 \leq k \leq K} L^{(k)}_t \right) 
            \leq
            \exp\left( - \frac{t \Delta^2}{2} \right) \left[ 2 + \sum_{k=3}^K \exp  \left( - 2 t (\mu_1 - \mu_k)(\mu_2 - \mu_k)   \right) \right] .
        \end{equation*}
        \item[(ii)] If Assumption~\ref{assumption:losses_iid_or_mixing}(ii) holds, then
        \begin{equation*}
            \P\left( L^{(1)}_t > \min_{2 \leq k \leq K} L^{(k)}_t \right) 
            \leq
            \sqrt{e}
            \sum_{k=1}^K 
            \exp\left( - \frac{t (\mu_1 - \mu_k + \Delta/2)^2}{ 2 \theta^{(k)}_t } \right),
        \end{equation*}
        where $\theta^{(k)}_t := 1 + 4\sum_{n=1}^{t-1}\varphi^{(k)}_n$.
    \end{description}
\end{lemma}

\begin{proof}[Proof of Lemma~\ref{lemma:ftl_concentr_lemma}]
    First, choose some $\delta$ such that $\mu_1 < \delta < \mu_2$. For notational convenience, define $L^{*\setminus 1}_{t} := \min_{2 \leq k \leq K} L^{(k)}_t$ and event $\mathcal{A} := \{ L^{*\setminus 1}_{t} \leq t\delta \}$, whose complement is denoted by $\mathcal{A}^\complement$.

    By the law of total probability,
    \begin{align*}
        \P( L^{(1)}_t > L^{*\setminus 1}_{t} )
        & = \P\left( \{ L^{(1)}_t > L^{*\setminus 1}_{t} \} \cap \mathcal{A} \right) + \P\left( \{ L^{(1)}_t > L^{*\setminus 1}_{t} \} \cap \mathcal{A}^\complement \right) \\
        & \leq \P( \mathcal{A} ) + \P\left( \{ L^{(1)}_t > L^{*\setminus 1}_{t} \} \cap \mathcal{A}^\complement \right) .
    \end{align*}
    To upper bound the first term, since $K$ is finite and fixed, one can directly apply the union bound and obtain
    \begin{align*}
        \P\left( \min_{2 \leq k \leq K} L^{(k)}_t \leq t\delta \right)
        & \leq \sum_{k=2}^K \P\left( L^{(k)}_t \leq t\delta \right) = \sum_{k=2}^K \P\left( L^{(k)}_t - t \mu_k \leq t(\delta - \mu_k) \right) \\
         & = \sum_{k=2}^K \P\left( \sum_{s=1}^t (\ell^{(k)}_s - \mu_k) \leq t(\delta - \mu_k) \right) .
    \end{align*}

    \medskip
    
    \noindent \textbf{Case (i).} We first consider the case where Assumption~\ref{assumption:losses_iid_or_mixing}(i) holds.
    By Hoeffding's inequality, since $\ell^{(k)}_t - \mu_k \in [-\mu_k, 1-\mu_k]$, for $k \neq 1$ it holds
    \begin{align*}
        \P\left( \sum_{s=1}^t (\ell^{(k)}_s - \mu_k) \leq t m_k \right) & \leq \exp  \left( - 2 t m_k^2   \right) ,
    \end{align*}
    where $m_k = \delta - \mu_k$. Hence, 
    \begin{equation*}
        \P( \mathcal{A} ) 
        \leq
        \sum_{k=2}^K \exp \left( - 2 t (\delta - \mu_k)^2 \right) .
    \end{equation*}
    Since
    \begin{equation*}
        \P\left( \{ L^{(1)}_t > L^{*\setminus 1}_{t} \} \cap \{ L^{*\setminus 1}_{t} > t\delta \} \right) 
        \leq 
        \P\left( L^{(1)}_t > t\delta \right) ,
    \end{equation*}
    by applying Hoeffding's inequality again we obtain
    \begin{equation*}
        \P\left( L^{(1)}_t > t\delta \right) 
        \leq 
        \exp \left( - 2 t (\delta - \mu_1)^2 \right) .
    \end{equation*}
    Set now $\delta = (\mu_1 + \mu_2)/2 = \mu_1 + \Delta/2$, so that
    \begin{align*}
        \P \left( L^{(1)}_t > L^{*\setminus 1}_t \right)
        & \leq
          2 \exp \left( - t \Delta^2/2 \right) + \sum_{k=3}^K \exp \left( - 2 t (\mu_1 - \mu_k + \Delta/2 )^2 \right) \\
        & = 
          \exp \left( - t \Delta^2/2 \right) \left[ 2 + \sum_{k=3}^K \exp \left( - 2 t (\mu_1 - \mu_k)(\mu_1 - \mu_k + \Delta) \right) \right] .
    \end{align*}
    The final bound follows by using $\Delta = \mu_2 - \mu_1$.

    \medskip
    
    \noindent \textbf{Case (ii).} When working under mixing conditions, we can leverage known generalizations of Hoeffding's inequality to dependent data. Under Assumption~\ref{assumption:losses_iid_or_mixing}(ii), the Hoeffding-type bound of Lemma~\ref{lemma:rio_hoeffding_uniform_mixing} (see also  \citealp{rioAsymptoticTheoryWeakly2017}) in our setting states that, for any $\xi > 0$,
    \begin{equation*}
        \P\left( \Abs{ L^{(k)}_t - {\color{green} t}\mu_k } \geq \xi \right) 
        \leq 
        \sqrt{e} \exp\left( -\frac{ \xi^2 }{ 2 \theta^{(k)}_t t } \right) . 
    \end{equation*}
    We follow the same arguments as in the above proof for case {(i)}, and find that
    \begin{align*}
        \P\left( \sum_{s=1}^t (\ell^{(k)}_s - \mu_k) \leq t m_k \right) 
        & \leq \sqrt{e} \exp\left( - \frac{ t m_k^2 }{ 2 \theta^{(k)}_t } \right) 
    \end{align*}
    and
    \begin{equation*}
        \P\left( L^{(1)}_t > t\delta \right) 
        \leq 
        \sqrt{e} \exp\left( - \frac{t (\delta - \mu_1)^2}{ 2 \theta^{(1)}_t } \right) .
    \end{equation*}
    Combining these bounds and choosing once more $\delta = \mu_1 + \Delta/2$ concludes the proof.
\end{proof}

\begin{corollary}\label{corollary:simple_ftl_conc_bound}
    Let the conditions of Lemma~\ref{lemma:ftl_concentr_lemma} hold. Then the following bounds hold:
    \begin{description}
        \item[(i)] Under the conditions of case {(i)} of Lemma~\ref{lemma:ftl_concentr_lemma}
        \begin{equation}\label{eq:simple_ftl_conc_bound}
        \P \left( L^{(1)}_t > \min_{2 \leq k \leq K} L^{(k)}_t \right) 
            \leq
            K \exp\left(- t \frac{\Delta^2}{2}\right).
        \end{equation}
        \item[(ii)] Under the conditions of case {(ii)} of Lemma~\ref{lemma:ftl_concentr_lemma}, assume, additionally, that Assumption~\ref{ass:uniform summability}(i) holds and define $\overline{\theta}_{\max} = \max_{k \in [K]} \sup_{t \geq 1} \theta^{(k)}_t$. Then 
        \begin{equation}\label{eq:simple_ftl_mixing_bound}
        \P\left( L^{(1)}_t > \min_{2 \leq k \leq K} L^{(k)}_t \right) 
        \leq
        K\sqrt{e}
        \exp\left( - \frac{t \Delta^2}{ 8 \overline{\theta}_{\max} } \right).
    \end{equation}
    \end{description}
\end{corollary}

\begin{proof}[Proof of Corollary~\ref{corollary:simple_ftl_conc_bound}]
    Part {(i)} is an immediate consequence of case {(i)} of Lemma~\ref{lemma:ftl_concentr_lemma} using that, since $\left(\mu_1-\mu_k\right)\left(\mu_2-\mu_k\right) \geq 0$ for all $k \geq 3$, then $\exp \left(-2 t\left(\mu_1-\mu_k\right)\left(\mu_2-\mu_k\right)\right) \leq 1$ for all $k \geq 3$, and hence $\left[ 2 + \sum_{k=3}^K \exp  \left( - 2 t (\mu_1 - \mu_k)(\mu_2 - \mu_k)   \right) \right]{\color{green} \le}K$.  Part {(ii)} is obtained directly from case {(ii)} of Lemma~\ref{lemma:ftl_concentr_lemma} by noticing that {\color{green} $\mu_1-\mu_k + \Delta/2 \leq -\Delta/2$} and $\theta_t^{(k)} \leq \overline{\theta}_{\max }$,  for all {\color{green} $k\ge 2$} and for all $t$, where by Assumption~\ref{ass:uniform summability}(i), $\overline{\theta}_{\max }\leq 1+4 C_{1,\varphi}<\infty$.
\end{proof}

\subsubsection{Bernstein-type}

\begin{lemma}[Bernstein-type]
\label{lem:ftl-bernstein}
    Let Assumptions~\ref{assumption:losses_0-1}-\ref{assumption:losses_iid_or_mixing} hold and let $v_k := \Var(\ell^{(k)}_t)$, $k\in [K]$. 
    \begin{description}
    \item[(i)] If Assumption~\ref{assumption:losses_iid_or_mixing}(i) holds, then for all $t\ge1$,
    \begin{align}\label{part i bound}
        \P\left( L^{(k^\star)}_t > \min_{k \neq k^\star} L^{(k)}_t \right) 
                    &\leq
        2\sum_{k\neq k^\star}
        \exp  \left(
        -\frac{t \Delta_k^2}{ 8 \max\{v_{k},v_{k^\star}\} + \tfrac{4}{3} \Delta_k}
        \right) .
    \end{align}
    \item[(ii)] If Assumption~\ref{assumption:losses_iid_or_mixing}(ii) holds, define ${\rho}^{(k)}_t := 1+\sum_{n=1}^{t}\sqrt{\varphi^{(k)}_n}$. Then for all $t\ge1$,
    \begin{align}\label{part ii bound}
        \P\left( L^{(k^\star)}_t > \min_{k \neq k^\star} L^{(k)}_t \right) 
                    &\leq
        2 \sum_{k\neq k^\star}
        \exp \left(
        -\frac{t \Delta_k^2}{ 
        {\color{green} 4\max\{(\rho_t^{(k)})^2,(\rho_t^{(k^\star)})^2\}(32\max\{v_{k^\star},v_k\}+\Delta_k)}
        }
        \right).
    \end{align}
    \item[(iii)] If Assumption~\ref{assumption:losses_iid_or_mixing}{(ii)} holds and for each $k\in [K]$, $\{\ell_t^{(k)}\}_t$ is  $\varphi$-mixing with coefficients satisfying 
    $\varphi^{(k)}_n\le c_k \exp{(-b_k n^{\gamma_k})}$ with $b_k>0$, $c_k\ge0$, and $\gamma_k>0$. Define
    $
    \gamma_{\min}:=\min_{k\in[K]}\gamma_k$,
    $b_{\min}:=\min_{k\in[K]} b_k$,
    $c_{\max}:=\max_{k\in[K]} c_k
    $,
    and
    \begin{equation}\label{t0}
        t_0 
        :=
        \max \big\{\min \left\{m \geq 3 \:\vert\: m^2 \geq 3232\, c_{\max} 
        \textnormal{ and } 
        m (\log m)^{-2 / \gamma_{\min}} \geq 4\right\}, e^{3 / b_{\min}} \big\}.
    \end{equation}
    Then for all $t\geq t_0$
    \begin{align}
        \P & \left( L^{(k^\star)}_t > \min_{k \neq k^\star} L^{(k)}_t \right) \nonumber \\
        &\leq 4 \sum_{k\neq k^\star} \exp \left(-\frac{t \Delta_k^2}{32( \max   \left\{ (\log t)^{2/{\gamma_{k^\star}}}v_{k^\star}, (\log t)^{2/{\gamma_k}} v_k\right\}
        +\frac{1}{6} (\log t)^{2/\min\{\gamma_{k^\star},\gamma_k\}} \Delta_k)}\right) . \label{eq:bernstein_part-ii-bound}
    \end{align}
    \end{description}
\end{lemma}

\begin{proof}[Proof of Lemma~\ref{lem:ftl-bernstein}]
    We start by defining an event of interest $\mathcal{A}:=\{ L_t^{(k^\star)} > L_t^{(k)}\}$, $k\in [K]$. Notice that, {\color{green} on the event $\mathcal{A}$,} $ L_t^{(k^\star)} - L_t^{(k)} \geq 0$ and hence for any $t>0$ it holds that 
    \begin{align}\label{step1}
        L_t^{(k^\star)} - L_t^{(k)} + t\Delta_k = (L_t^{(k^\star)} - t\mu_{k^\star}) - (L_t^{(k)} - t\mu_{k}) \geq t\Delta_k.
    \end{align} 
    Next, notice that, by Remark~\ref{remark:weak_stationarity}, one can define $S_t^{(k)}:=L_t^{(k)} - t\mu_{k} = \sum_{\tau=1}^t\left(\ell_\tau^{(k)}-\mu_{k}\right)$ for all $t\geq 1$ and $k\in [K]$, which allows us to write \eqref{step1} as
    \begin{align*}
        \mathcal{A} \, {\color{green} \subseteq} \,\{S_t^{(k^\star)} - S_t^{(k)} \geq t\Delta_k\} \subseteq \{ S_t^{(k^\star)} \geq \tfrac{1}{2} t\Delta_k \} \cup  \{ -S_t^{(k)} \geq \tfrac{1}{2} t\Delta_k \}
    \end{align*}
    and hence
    \begin{align}\label{eq:PA}
        \P(\mathcal{A}) \leq \P( S_t^{(k^\star)} \geq \tfrac{1}{2} t\Delta_k ) +  \P( S_t^{(k)} \leq -\tfrac{1}{2} t\Delta_k ).
    \end{align}
    
    \noindent \textbf{Case (i).}
    We first consider the case where Assumption~\ref{assumption:losses_iid_or_mixing}(i) holds. Together with Assumption~\ref{assumption:losses_0-1}, it implies that $\{S_t^{(k)}\}_t$ are sums of bounded zero-mean random variables with variances {\color{green} $\{v_k\}_{k\in[K]}$}. We use the one-sided Bernstein inequality (see Lemma~\ref{lemma:iid_bernstein}) and obtain 
    \begin{align*}
    	\P( S_t^{(k^\star)} \geq \tfrac{1}{2} t\Delta_k ) \leq \exp  \left(
        -\frac{t \Delta_k^2}{ 8 v_{k^\star} + \tfrac{4}{3} \Delta_k} 
        \right) \enspace \text{and}\enspace \P( S_t^{(k)} \leq -\tfrac{1}{2} t\Delta_k ) \leq \exp  \left(
        -\frac{t \Delta_k^2}{ 8 v_{k} + \tfrac{4}{3} \Delta_k} 
        \right).
    \end{align*}
    Substituting these bounds in \eqref{eq:PA} and using the elementary inequality in Lemma~\ref{lemma:elementary exp inequality}, one obtains that 
    \begin{align*}
        \P(\mathcal{A})  
        \leq 
        2 \exp \left(
        -\frac{t \Delta_k^2}{ 8 \max\{v_{k^\star},v_{k}\} + \tfrac{4}{3} \Delta_k}
        \right),
    \end{align*}
    which, by the union bound, yields the inequality in \eqref{part i bound}. 

    \medskip
    
    \noindent \textbf{Case (ii).} 
    We proceed analogously to the proof of case {(i)}, using the one-sided Bernstein inequalities in Lemma~\ref{lemma:samson_phi_mixing} for sums of $\varphi$-mixing random variables in \eqref{eq:PA}. By Assumption~\ref{assumption:losses_0-1}, we can take $M=1$. Choosing $\varepsilon= \Delta_k t/2$ in \eqref{eq:right_bound} and \eqref{eq:left_bound} yields the following bounds \begin{align*}
    	\P( S_t^{(k^\star)} \geq \tfrac{1}{2} t\Delta_k ) & \leq \exp \left(
        -\frac{t \Delta_k^2}{ {\color{green} 
        128\rho_t^{(k^\star)}v_{k^\star}+4(\rho_t^{(k^\star)})^2\Delta_k
        }} 
        \right), \\
        \P( S_t^{(k)} \leq -\tfrac{1}{2} t\Delta_k ) & \leq \exp \left(
        -\frac{t \Delta_k^2}{{\color{green} 128\rho_t^{(k)}v_{k}+4(\rho_t^{(k)})^2\Delta_k}} 
        \right).
    \end{align*}
    {\color{green} Since $\rho_t^{(k)}\geq1$ for all $k\in[K]$, these inequalities imply
    \begin{align*}
        \P\left(S_t^{(k^\star)}\geq\tfrac12 t\Delta_k\right) &\leq \exp\left(-\frac{t\Delta_k^2}{4(\rho_t^{(k^\star)})^2(32v_{k^\star}+\Delta_k)}\right),\\
        \P\left(S_t^{(k)}\leq-\tfrac12 t\Delta_k\right) &\leq \exp\left(-\frac{t\Delta_k^2}{4(\rho_t^{(k)})^2(32v_k+\Delta_k)}\right).
    \end{align*}}
    Using these inequalities in \eqref{eq:PA} together with the elementary result in Lemma~\ref{lemma:elementary exp inequality}, we obtain 
    \begin{align*}
        \P(\mathcal{A})  
        \leq 
        2 \exp \left(
        -\frac{t \Delta_k^2}{ {\color{green} 4\max\{(\rho_t^{(k)})^2,(\rho_t^{(k^\star)})^2\}(32\max\{v_{k^\star},v_k\}+\Delta_k)}}
        \right),
    \end{align*}
    which, by the union bound over $k\neq k^\star$, yields the inequality in \eqref{part ii bound} for all $t\geq 1$.

    \medskip
    
    \noindent \textbf{Case (iii).} 
    The proof is analogous to the proof of case {(i)}, but now using a one-sided Bernstein inequality for sums of geometrically $\varphi$-mixing random variables in \eqref{eq:PA}, see Theorem~\ref{theorem:huang_C_mixing} in \cite{hangStatisticalLearningKernelbased2015,hangBernsteintypeInequalityMixing2017}. Taking $\varepsilon=\frac{1}{2} \Delta_k t$ yields, for all $t\geq t_0$ with $t_0$ as in \eqref{t0}, obtained as the largest among all the experts, the following one-sided bounds:
    \begin{align*}
        \P \left(S_t^{(k^\star)}\geq \tfrac{t\Delta_k}{2}  \right)
         &\le
        2\exp  \left(
        -\frac{t \Delta_k^2}{ 32 (\log t)^{2/\gamma_{k^\star}} (v_{k^\star} + \tfrac{1}{6} \Delta_k)}\right),\\
        \P  \left(S_t^{(k)}\leq -\tfrac{t\Delta_k}{2}  \right)
         &\le
         2\exp  \left(
        -\frac{t \Delta_k^2}{ 32(\log t)^{2/\gamma_{k}} (v_k + \tfrac{1}{6}\Delta_k)}
        \right),
    \end{align*}
    where we used that $M=1$ by Assumption~\ref{assumption:losses_0-1}.
    Substituting these bounds in \eqref{eq:PA} and using the fact that $\exp(-\tfrac{x}{a})+\exp(-\tfrac{x}{b}) \leq 2\exp(-\tfrac{x}{\max\{a,b\}})$ for $x \geq 0$ and $a,b>0$, one obtains that 
    \begin{align*}
        \P(\mathcal{A})  \leq 4  \exp \left(-\frac{t \Delta_k^2}{32( \max   \left\{ (\log t)^{2/{\gamma_{k^\star}}}v_{k^\star}, (\log t)^{2/{\gamma_k}} v_k\right\}
        +\frac{1}{6} (\log t)^{2/\min\{\gamma_{k^\star},\gamma_k\}} \Delta_k)}\right),    
    \end{align*}
    which, by the union bound, yields \eqref{eq:bernstein_part-ii-bound}. 
\end{proof}
\begin{corollary}\label{corollary:simple_ftl_conc_bound_bernstein}
    Let the conditions of Lemma~\ref{lem:ftl-bernstein} hold and, without loss of generality, assume that
$\mu_1<\mu_2\leq\cdots\leq\mu_K$, so that
$k^\star=1$ and $\Delta=\mu_2-\mu_1>0$.
Let $v_{\max}:=\max_{k\in[K]} v_k$. Then the following bounds hold:
    \begin{description}
        \item[(i)] Under the conditions of case {(i)} of Lemma~\ref{lem:ftl-bernstein}
        \begin{equation*}\label{eq:simple_ftl_conc_bound_bernstein}
        \P \left( L^{(1)}_t > \min_{2 \leq k \leq K} L^{(k)}_t \right) 
            \leq 2 (K-1) 
        \exp  \left(
        -\frac{t \Delta^2}{8 v_{\max} +\tfrac{4}{3} \Delta}
        \right).
        \end{equation*}
        \item[(ii)] Under the conditions of case {(ii)} of Lemma~\ref{lem:ftl-bernstein}, assume, additionally, that Assumption~\ref{ass:uniform summability roots}(i) holds and define $\overline{\rho}_{\max} = \max_{k \in [K]} \sup_{t \geq 1} \rho^{(k)}_t$. Then 
        \begin{equation*}\label{eq:simple_ftl_mixing_bound_bernstein}
            \P\left( L^{(1)}_t > \min_{2 \leq k \leq K} L^{(k)}_t \right) 
            \leq {2} (K-1) 
            \exp  \left(
            -\frac{t \Delta^2}{4\overline{\rho}_{\max}^2(32v_{\max}+\Delta)}
            \right).
        \end{equation*}
        \item[(iii)] Under the conditions of case {(iii)} of Lemma~\ref{lem:ftl-bernstein},
        \begin{equation}\label{eq:simple_ftl_mixing_bound_bernstein_iii}
            \P\left( L^{(1)}_t > \min_{2 \leq k \leq K} L^{(k)}_t \right) 
            \leq 4 (K-1) 
            \exp  \left(
            -\frac{t (\log t)^{-2 / \gamma_{\min}} \Delta^2}{ 32(v_{\max} +\tfrac{1}{6} \Delta)}
            \right).
        \end{equation}
  
    \end{description}
\end{corollary}

\begin{proof}[Proof of Corollary~\ref{corollary:simple_ftl_conc_bound_bernstein}]
    Part {(i)} is an immediate consequence of case {(i)} of Lemma~\ref{lem:ftl-bernstein}. Using the elementary inequality in Lemma~\ref{lemma:elementary exp inequality}, and noticing that the exponent of the exponential is a decreasing function in $\Delta_k$ yields the required result. Part {(ii)} directly follows from case {\color{green}{(ii)} of Lemma~\ref{lem:ftl-bernstein}, using $\Delta_k\geq\Delta$, $(\rho_t^{(k)})^2\leq\overline{\rho}_{\max}^2$, and $v_k\leq v_{\max}$. Indeed, the function $x\mapsto x^2/(32v_{\max}+x)$ is increasing for $x>0$. Moreover, Assumption~\ref{ass:uniform summability roots}(i) ensures that $\overline{\rho}_{\max}^2\leq(1+C_{2,\varphi})^2<\infty$.} Additionally, by Assumption~\ref{ass:uniform summability roots}(i),  $\overline{\rho}_{\max }^2=\max _{k \in[K]} \sup _{t \geq 1}\left(\rho_t^{(k)}\right)^2 \leq\left(1+C_{2, \varphi}\right)^2<\infty$. Part {(iii)} directly follows from case {(iii)} of Lemma~\ref{lem:ftl-bernstein} by noticing that it holds that $\Delta_k\ge\Delta$ {\color{green} for all $k\neq k^\star$} and
    $$ \max\{(\log t)^{2/\gamma_{k^\star}}v_{k^\star},(\log t)^{2/\gamma_k}v_k\}
    \le (\log t)^{2/\gamma_{\min}} \max_{k\in[K]} v_k \quad {\rm for}\enspace {\rm all}\enspace k\in [K].$$
    {\color{green} Moreover, since $x \mapsto x^2/(v_{\max}+x/6)$ is increasing for $x>0$, we have $\frac{\Delta_k^2}{v_{\max}+\Delta_k/6}\geq\frac{\Delta^2}{v_{\max}+\Delta/6}$
 for all $k\neq k^\star$.}  
\end{proof}

\subsubsection{Combined Bound}

The following proposition provides upper bounds for the probability that the expert with the lowest expected loss commits a cumulative error that is larger than that of the runner-up expert. We combine both Hoeffding-type and Bernstein-type results to showcase the different kinds of guarantees that can be obtained.
Note that, although the presented bounds assume constant-in-time second moment of losses, they remain valid if one uses time-dependent variances or their uniform-in-time bound. 

\begin{proposition}\label{proposition:main_concentration}
	Let Assumptions~\ref{assumption:losses_0-1}-\ref{assumption:losses_iid_or_mixing} hold and, without loss of generality, assume that the experts are ordered such that    %
    $
        \mu_1 < \mu_2 \leq \ldots \leq \mu_K 
    $, that is $\Delta = \mu_2-\mu_1>0$. 
    Additionally, let $v_k:=\Var(\ell^{(k)}_t)$, $k\in [K]$. Then for all $t\geq t_0$, 
    \begin{equation*}
        \P\left( L^{(1)}_t > \min_{2 \leq k \leq K} L^{(k)}_t \right) 
        \leq
        \min\left\{ H(t), B(t) \right\}
    \end{equation*}
    with the following cases:
    \begin{description}
        \item[(i)] If the dependence conditions of Lemma~\ref{lem:ftl-bernstein}(i) hold, then
        \begin{align*}
            H(t)&=
            \exp\left( - \frac{t \Delta^2}{2} \right) \left[ 2 + \sum_{k=3}^K \exp \left( - 2 t (\mu_1 - \mu_k)(\mu_2 - \mu_k) \right) \right], \\B(t) &= 2\sum_{k= 2}^K
            \exp \left(
            -\frac{t \Delta_k^2}{ 8 \max\{v_{1},v_k\} + \tfrac{4}{3} \Delta_k}
            \right) ,
        \end{align*}
        where $t_0 = 1$.
        \item[(ii)] If the dependence conditions of Lemma~\ref{lem:ftl-bernstein}(ii) hold, then
        \begin{align*}
            H(t)&=
            \sqrt{e}
            \sum_{k=1}^K 
            \exp\left( - \frac{t (\mu_1 - \mu_k + \Delta/2)^2}{ 2 \theta^{(k)}_t } \right), \enspace \text{with} \enspace \theta^{(k)}_t = 1 + 4\sum_{n=1}^{t-1}\varphi^{(k)}_n, \\B(t) &= 
            2\sum_{k=2}^{K}\exp\left(-\frac{t\Delta_k^2}{4\max\{(\rho^{(1)}_t)^2,(\rho^{(k)}_t)^2\}(32\max\{v_1,v_k\}+\Delta_k)}\right),
        \end{align*}
        with $\rho^{(k)}_t=1+\sum_{n=1}^{t}\sqrt{\varphi_n^{(k)}}$.
        \item[(iii)] If the dependence conditions of Lemma~\ref{lem:ftl-bernstein}(iii) hold, then
        \begin{align*}
            H(t)&=
            \sqrt{e}
            \sum_{k=1}^K 
            \exp\left( - \frac{t (\mu_1 - \mu_k + \Delta/2)^2}{ 2 \theta^{(k)}_t } \right), \enspace \text{with} \enspace \theta^{(k)}_t := 1 + 4\sum_{n=1}^{t-1}\varphi^{(k)}_n, \\B(t) &= 4 \sum_{k = 2}^K \exp \left(-\frac{t \Delta_k^2}{32( \max   \left\{ (\log t)^{2/{\gamma_{1}}}v_{1}, (\log t)^{2/{\gamma_k}} v_k\right\}
            +\frac{1}{6} (\log t)^{2/\min\{\gamma_{1},\gamma_k\}} \Delta_k)}\right) , 
        \end{align*}
        where $t_0$ is given explicitly in \eqref{t0}, Lemma~\ref{lem:ftl-bernstein}, and $\gamma_k > 0$ for all $k \in [K]$.
    \end{description}
\end{proposition}

\begin{proof}[Proof of Proposition~\ref{proposition:main_concentration}]
    The assumptions allow us to apply both Lemma~\ref{lemma:ftl_concentr_lemma} and~\ref{lem:ftl-bernstein}.

    \noindent \textbf{Case (i).}
    Take $H(t)$ to be the Hoeffding-type bound from Lemma~\ref{lemma:ftl_concentr_lemma}(i) and $B(t)$ the Bernstein-type bound from Lemma~\ref{lem:ftl-bernstein}(i).

\noindent \textbf{Case (ii).}
    Under the $\varphi$-mixing assumptions, take $H(t)$ to be the Hoeffding-type bound from Lemma~\ref{lemma:ftl_concentr_lemma}(ii) and $B(t)$ the Bernstein-type bound from Lemma~\ref{lem:ftl-bernstein}(ii).
    
    \noindent \textbf{Case (iii).}
    Under geometrically $\varphi$-mixing assumptions, similarly take $H(t)$ to be the bound from Lemma~\ref{lemma:ftl_concentr_lemma}(ii) and $B(t)$ from Lemma~\ref{lem:ftl-bernstein}(iii). 

    In either case, when $t \geq t_0$, we can control the probability that $L^{(1)}_t$ exceeds the cumulative loss of at least one other expert by the minimum of the two bounds, since both apply.
\end{proof}

\subsection{Concentration Results for Hedge}
\label{appendix:proofs_concentration_hedge}

We exploit the theory from Appendix~\ref{appendix:weak_depend} to construct concentration-type bounds for the cumulative losses of experts under both i.i.d.~and dependent settings.

\subsubsection{Hoeffding-type}
\begin{lemma}[Hoeffding-type]
\label{lemma:hedge_concentr_lemma}
    Let Assumptions~\ref{assumption:losses_0-1}-\ref{assumption:losses_iid_or_mixing} hold. Let $\eta_t = 2\sqrt{\log K /t}$ for $t \in [T]$ be the learning rate for decreasing Hedge. Let  $k^\star\in [K]$ be the expert with the smallest expected loss among $K$ experts.  Then for $k\in [K]$ and any $t$
    \begin{description}
        \item[(i)] If {\color{brown} Assumption~\ref{assumption:losses_iid_or_mixing}(i)} holds, then
        \begin{align}\label{hedge_ass_1}
            \P\left( L_{t}^{(k)}-L_{t}^{(k^\star)}<\frac{\Delta_k t }{2} \right)
            &\leq
              \exp \left(-\frac{ \Delta_k ^2 t}{8 }\right) .
        \end{align}
        \item[(ii)] If {\color{brown} Assumptions~\ref{assumption:losses_iid_or_mixing}(iv) and \ref{ass:uniform summability}(ii)} hold, then
        \begin{align}\label{hedge_ass_2}
            \P\left( L_{t}^{(k)}-L_{t}^{(k^\star)}<\frac{\Delta_k t }{2} \right)
            &\leq
              \sqrt{e}\exp \left(-\frac{ \Delta_k^2 t}{{\color{green} 32\widetilde{\theta}_t^{(k)}}}\right)  ,
        \end{align}
        {\color{green} where $\widetilde{\theta}_t^{(k)}
        :=1+4\sum_{n=1}^{t-1}\varphi_n^{\chi,k}$, with $\varphi_n^{\chi,k}:=\varphi_n(\{\chi_t^{(k)}\}_{t\in\mathbb Z})$ denoting the $\varphi$-mixing coefficients of the excess-loss process $\chi_t^{(k)}=\ell_t^{(k)}-\ell_t^{(k^\star)}$.}

    \end{description}
\end{lemma}
\begin{proof}[Proof of Lemma~\ref{lemma:hedge_concentr_lemma}]

In order to derive the bound for $\P\left( L_{t}^{(k)}-L_{t}^{(k^\star)}<{\Delta_kt}/{2} \right)$, one constructs, for every $k\neq k^\star$, random variables $Z_t^{(k)}:=-\ell_t^{(k)} + \ell_t^{(k^\star)} + \Delta_k$,  $Z_t^{(k)}\in [-1+\Delta_k, 1+\Delta_k]\subset[-1,2]$, $\E[Z_t^{(k)}]=0$, and writes
\begin{align}
\label{eq:pz}
    \P\left( L_{t}^{(k)}-L_{t}^{(k^\star)}<\frac{\Delta_k t }{2}\right)&=\P\left( \sum_{s=1}^{t}Z_s^{(k)}>\frac{\Delta_k t}{2}\right).
\end{align}
We now apply the conditions in case (i) and case (ii) to conclude the proof.

\medskip

\noindent {\bf Case (i)}. Under {\color{brown} Assumption~\ref{assumption:losses_iid_or_mixing}(i)}, by the Cauchy-Schwarz inequality and Hoeffding's lemma, for any $\lambda>0$ and $k\neq k^\star$,
\begin{align*}
\E\left[\exp\left(\lambda\sum_{s=1}^t Z_s^{(k)}\right)\right]
&\leq\left(\E\left[\exp\left(2\lambda\sum_{s=1}^t(\ell_s^{(k^\star)}-\mu_{k^\star})\right)\right]\right)^{1/2}\left(\E\left[\exp\left(-2\lambda\sum_{s=1}^t(\ell_s^{(k)}-\mu_k)\right)\right]\right)^{1/2}\\
&\leq\exp\left(\frac{t\lambda^2}{2}\right),
\end{align*}
where the last inequality follows from the temporal independence and boundedness of each expert's losses. Consequently, applying the exponential Markov inequality with $\lambda=\Delta_k/2$ yields
\begin{align*}
\P\left(\sum_{s=1}^t Z_s^{(k)}>\frac{\Delta_k t}{2}\right)&\leq\exp\left(-\frac{\lambda\Delta_k t}{2}+\frac{t\lambda^2}{2}\right)=\exp\left(-\frac{\Delta_k^2 t}{8}\right),
\end{align*}
which proves \eqref{hedge_ass_1}.

\medskip

\noindent  {\bf Case (ii)}. Under {\color{brown} Assumption~\ref{assumption:losses_iid_or_mixing}(iv)}, the process $\{Z_t^{(k)}\}_t$ is $\varphi$-mixing (see, for example, \cite{Dedecker2007a}) with coefficients $\varphi_n^{\chi,k}$, since $Z_t^{(k)}=-\chi_t^{(k)}+\Delta_k$ is an invertible affine transformation of the excess-loss process. Applying Lemma~\ref{lemma:rio_hoeffding_uniform_mixing} with $\widetilde{\theta}_t^{(k)}=1+4\sum_{n=1}^{t-1}\varphi_n^{\chi,k}$ yields
\begin{align*}
\P\left(\sum_{s=1}^{t}Z_s^{(k)}>\frac{\Delta_k t}{2}\right)\leq\sqrt{e}\exp\left(-\frac{\Delta_k^2 t}{32\widetilde{\theta}_t^{(k)}}\right),
\end{align*}
which yields \eqref{hedge_ass_2}.
\end{proof}

\begin{remark}%
\label{remark:hedge_concentr_uniform}
    Notice that given some bound 
    \begin{align*}
        \P\left(  L^{(k)}_t-L_{t}^{(k^\star)}<\frac{\Delta_k t }{2} \right)
        &\leq  a , %
        \end{align*}
    for all $k \not= k^\star$, one can use that
    \begin{align*}
        \left\{\min _{k \neq k^{\star}} L_t^{(k)}-L_t^{\left(k^\star\right)} <\frac{\Delta t}{2}\right\} \subseteq \bigcup_{k \neq k^{\star}}\left\{L_t^{(k)}-L_t^{\left(k^\star\right)}<\frac{\Delta_k t}{2}\right\}
    \end{align*}
    to obtain 
    \begin{align*}
        \P\left( \min _{k \neq k^{\star}} L_t^{(k)}-L_t^{\left(k^\star\right)} <\frac{\Delta t}{2} \right) \leq  (K-1) a .
    \end{align*}
\end{remark}

\begin{corollary}\label{corollary:simple_hedge_conc_bound_hoeffding}
    Let the conditions of Lemma~\ref{lemma:hedge_concentr_lemma} hold.
    Then the following bounds hold:
    \begin{description}
        \item[(i)] Under the conditions of case (i) of Lemma~\ref{lemma:hedge_concentr_lemma}
        \begin{align*}
            \P\left( L_{t}^{(k)}-L_{t}^{(k^\star)}<\frac{\Delta_k t }{2} \right)
            \leq  \exp \left(-\frac{ \Delta ^2 t}{8 }\right) .
        \end{align*}
        \item[(ii)] Under the conditions of case (ii) of Lemma~\ref{lemma:hedge_concentr_lemma} and, additionally, Assumption~\ref{ass:uniform summability}(ii), define $\overline{\widetilde{\theta}}_{\max}:=\max_{k\neq k^\star}\sup_{t\geq1}\widetilde{\theta}_t^{(k)}=1+4\max_{k\neq k^\star}\sum_{n=1}^{\infty}\varphi_n^{\chi,k}<\infty$. Then
        \begin{align*}
        \P\left(L_t^{(k)}-L_t^{(k^\star)}<\frac{\Delta_k t}{2}\right)\leq\sqrt{e}\exp\left(-\frac{\Delta^2 t}{32\overline{\widetilde{\theta}}_{\max}}\right).
        \end{align*}
    \end{description}
\end{corollary}

\begin{proof}[Proof of Corollary~\ref{corollary:simple_hedge_conc_bound_hoeffding}]
    To obtain case (i), we simply note that $\Delta_k \geq \Delta$ for all {\color{green} $k\neq k^\star$} on the right-hand side of \eqref{hedge_ass_1}.
    Case (ii) follows analogously from \eqref{hedge_ass_2}, using $\Delta_k\geq\Delta$ and $\widetilde{\theta}_t^{(k)}\leq\overline{\widetilde{\theta}}_{\max}$ for all $k\neq k^\star$ and $t\geq1$. The finiteness of $\overline{\widetilde{\theta}}_{\max}$ follows from Assumption~\ref{ass:uniform summability}(ii), since $\varphi_n^{\chi,k}\leq\varphi_n$.
\end{proof}

\subsubsection{Bernstein-type}

\begin{lemma}[Bernstein-type]
    \label{lem:hedge-bernstein}
    Let Assumptions~\ref{assumption:losses_0-1}-\ref{assumption:losses_iid_or_mixing} hold. Let $\eta_t = 2\sqrt{\log K/t}$ for $t \in [T]$ be the learning rate for decreasing Hedge. Assume that there exists a unique best expert $k^\star\in[K]$ with $\Delta>0$.
All bounds below are stated for $k\neq k^\star$. Define  $\widetilde{v}_k := \Var(\ell^{(k)}_t-\ell^{(k^\star)}_t)$, $k\in [K]$, $k\neq k^\star$.  Then
    \begin{description}
        \item[(i)] If {\color{brown} Assumption~\ref{assumption:losses_iid_or_mixing}(iii)} holds, then
        \begin{align}\label{hedge_ass_bern_case_1}
            \P\left( L_{t}^{(k)}-L_{t}^{(k^\star)}<\frac{\Delta_k t }{2}\right)
            &\leq
              \exp \left(-\frac{t \Delta_k^2}{8 (\widetilde{v}_k + \frac{1}{3}\Delta_k)}\right) .
        \end{align}
        \item[(ii)] If {\color{brown} Assumption~\ref{assumption:losses_iid_or_mixing}(iv)} holds, then
        \begin{align}\label{hedge_ass_bern_case_2_1}
            \P\left( L_{t}^{(k)}-L_{t}^{(k^\star)} < \frac{\Delta_k t }{2}\right)
            & \leq 
            \exp \left(-\frac{t \Delta_k^2}{{\color{green} 4} (\widetilde{\rho}^{(k)}_t)^2({\color{green} 32} \widetilde{v}_k + \Delta_k)}\right)
            , 
        \end{align}
        where $\widetilde{\rho}_t^{(k)}:=1+\sum_{n=1}^{t}\sqrt{\varphi_n^{\chi,k}}$, with $\varphi_n^{\chi,k}$ denoting the $\varphi$-mixing coefficients of the excess-loss process $\chi_t^{(k)}=\ell_t^{(k)}-\ell_t^{(k^\star)}$.
        \item[(iii)] If {\color{brown} Assumption~\ref{assumption:losses_iid_or_mixing}(iv)} holds and for each $k\neq k^\star$, $\{\ell_t^{(k)} - \ell_t^{(k^\star)}\}_t$ is  $\varphi$-mixing with coefficients satisfying 
        $\varphi_n^{\chi,k} \leq c_k \exp{(-b_k n^{\gamma_k})}$, $k\neq k^\star$,  with $b_k>0$, $c_k\ge0$, and $\gamma_k>0$. Define
        $
        \gamma_{\min}:=\min_{k\neq k^\star}\gamma_k$,
        $b_{\min}:=\min_{k\neq k^\star} b_k$,
        $c_{\max}:=\max_{k\neq k^\star} c_k
        $,
        and
        \begin{equation}\label{t0_hedge}
            t_0 
            :=
            \max \big\{\min \left\{m \geq 3 \:\vert\: m^2 \geq 3232\, c_{\max} 
            \textnormal{ and } 
            m (\log m)^{-2 / \gamma_{\min}} \geq 4\right\}, e^{3 / b_{\min}} \big\}.
        \end{equation}
        Then for all $t\geq t_0$
        \begin{align}\label{eq:hedge_bernstein_part-iii-bound}
            \P\left( L_{t}^{(k)}-L_{t}^{(k^\star)}<\frac{\Delta_k t }{2}\right)
            &\leq  2 \exp \left(-\frac{t \Delta_k^2}{32(\log t)^{2 / \gamma_k} \left( \widetilde{v}_k+\frac{1}{3}\Delta_k\right)}\right) .
        \end{align}
    \end{description}
\end{lemma}

\begin{proof}[Proof of Lemma~\ref{lem:hedge-bernstein}] 
    Analogously to the proof of Lemma~\ref{lemma:hedge_concentr_lemma}, we separately analyze \eqref{eq:pz} for cases (i) and (ii).
    
    \noindent{\bf Case (i)}. Under {\color{brown} Assumption~\ref{assumption:losses_iid_or_mixing}(iii)}, $\{Z_{t}^{(k)}\}_t$ are i.i.d.~for each $k\in [K]$. We apply the one-sided Bernstein inequality (see Lemma~\ref{lemma:iid_bernstein}) and obtain
    \begin{align*}%
        \P\left( \sum_{s=1}^{t}Z^{(k)}_s > \frac{\Delta_k t}{2} \right) \leq \exp \left(-\frac{ \Delta_k^2 t}{8 (\widetilde{v}_k + \frac{1}{3}\Delta_k)}\right),
    \end{align*}
    where we used that $M=2$, since $Z_t^{(k)}\in [-1+\Delta_k, 1+\Delta_k]\subset[-1,2]$. This expression is the inequality in \eqref{hedge_ass_bern_case_1}, as required.
    
    \medskip
    
    \noindent{\bf Case (ii)}. 
    Under {\color{brown}Assumption~\ref{assumption:losses_iid_or_mixing}(iv)}, , the processes $\{\chi_t^{(k)}\}_t$, where $\chi_t^{(k)}:=\ell_t^{(k)}-\ell_t^{(k^\star)}$, are strictly stationary and $\varphi$-mixing for each $k\neq k^\star$. Since $\chi_t^{(k)}\in[-1,1]$, we apply Lemma~\ref{lemma:samson_phi_mixing} with $U_t=-\chi_t^{(k)}$ and $M=1$. Noting that $\E[U_t]=-\Delta_k$, $\Var(U_t)=\widetilde{v}_k$, and $U_t-\E[U_t]=Z_t^{(k)}$, we obtain
\begin{align*}%
\P\left(\sum_{s=1}^{t}Z_s^{(k)}>\frac{\Delta_k t}{2}\right)&\leq\exp\left(-\frac{t\Delta_k^2}{128\widetilde{\rho}_t^{(k)}\widetilde{v}_k+4(\widetilde{\rho}_t^{(k)})^2\Delta_k}\right)\leq\exp\left(-\frac{t\Delta_k^2}{4(\widetilde{\rho}_t^{(k)})^2(32\widetilde{v}_k+\Delta_k)}\right),
\end{align*}
where $\widetilde{\rho}_t^{(k)}$ is defined as in the statement of the lemma. Since $\widetilde{\rho}_t^{(k)}\geq1$, inequality~\eqref{hedge_ass_bern_case_2_1} follows.
    
    \medskip
    
    \noindent{\bf Case (iii)}. The proof is analogous to the proof of case {(i)} but now using a one-sided Bernstein inequality for sums of geometrically $\varphi$-mixing random variables $\{Z_t^{(k)}\}_t$, see Theorem~\ref{theorem:huang_C_mixing}~\citep{hangStatisticalLearningKernelbased2015,hangBernsteintypeInequalityMixing2017}. Taking $\varepsilon=\frac{1}{2} \Delta_k t$ yields, for all $t\geq t_0$ with $t_0$ as in \eqref{t0_hedge}, obtained as the largest among all the experts and taking into account that $M=2$, the following bound:
    \begin{align*}
        \P\left( \sum_{s=1}^{t}Z^{(k)}_s > \frac{\Delta_k t}{2} \right) 
        & \leq
        2\exp \left(
        -\frac{\frac{1}{4}t^2 \Delta_k^2}{8(\log t)^{2 / \gamma_k}\left(t\widetilde{v}_k+\frac{ \Delta_k t}{3}\right)}\right)
        =
        2 \exp \left( -\frac{t \Delta_k^2}{32(\log t)^{2 / \gamma_k}\left(\widetilde{v}_k+\frac{ \Delta_k  }{3}\right)} \right) ,
    \end{align*}
    which is \eqref{eq:hedge_bernstein_part-iii-bound}. 
\end{proof}

\begin{corollary}\label{corollary:simple_hedge_conc_bound_bernstein}
    Let the conditions of Lemma~\ref{lem:hedge-bernstein} hold.
    Then the following bounds hold:
    \begin{description}
        \item[(i)] Under the conditions of case (i) of Lemma~\ref{lem:hedge-bernstein}
        \begin{align*}
            \P\left( L_{t}^{(k)}-L_{t}^{(k^\star)}<\frac{\Delta_k t }{2}\right)
            &\leq  \exp \left(-\frac{t \Delta^2}{8 (\widetilde{v}_{\max} + \frac{1}{3}\Delta) }\right) .
        \end{align*}
        \item[(ii)] Under the conditions of case (ii) of Lemma~\ref{lem:hedge-bernstein}, let, additionally, Assumption~\ref{ass:uniform summability roots}(ii) holds. Define $\overline{\widetilde{\rho}}_{\max}:=\max_{k\neq k^\star}\sup_{t\geq1}\widetilde{\rho}_t^{(k)}=1+\max_{k\neq k^\star}\sum_{n=1}^{\infty}\sqrt{\varphi_n^{\chi,k}}<\infty$. Then
        \begin{align*}%
            \P\left( L_{t}^{(k)}-L_{t}^{(k^\star)}<\frac{\Delta_k t }{2}\right)
            &\leq \exp \left(-\frac{t \Delta^2}{{\color{green} 4} \overline{\widetilde{\rho}}_{\max}^2 ({\color{green} 32} \widetilde{v}_{\max} + \Delta) }\right) . 
        \end{align*}
        \item[(iii)] Under the conditions of case (iii) of Lemma~\ref{lem:hedge-bernstein},
        \begin{align}\label{hedge_ass_bern_case_3}
            \P\left( L_{t}^{(k)}-L_{t}^{(k^\star)}<\frac{\Delta_k t }{2}\right)
            & \leq 2 \exp  \left( -\frac{t (\log t)^{-2 / \gamma_{\min}} \Delta^2}{ 32(\widetilde{v}_{\max} + \frac{1}{3} \Delta)} \right) 
        \end{align}
        for all $t\geq t_0$.
    \end{description}
\end{corollary}

\begin{proof}[Proof of Corollary~\ref{corollary:simple_hedge_conc_bound_bernstein}]
    Case (i) is obtained by using that $\Delta_k \geq \Delta $ for all $k\neq k^\star$, exactly as in the proof of Corollary~\ref{corollary:simple_ftl_conc_bound_bernstein}.
    For case (ii), one notes that on the right-hand side of \eqref{hedge_ass_bern_case_2_1} $\Delta_k\geq \Delta $ for all $k\neq k^\star$ and $\overline{\widetilde{\rho}}_{\max}^2 \geq (\widetilde{\rho}_t^{(k)})^2 $ for all $t\in [T]$ and $k\neq k^\star$. Notice now that, by Assumption~\ref{ass:uniform summability roots}(ii), $\overline{\widetilde{\rho}}_{\max}^2 \leq (1+\widetilde{C}_{2,\varphi})^2<\infty $.
    Lastly, to obtain the bound of case (iii), we additionally observe that 
    \begin{align*}
        (\log t)^{2 / \gamma_k} \left( \widetilde{v}_k + \frac{ 1}{3}\Delta_k\right)
        & 
        {\color{green} \leq (\log t)^{2/\gamma_{\min}} \frac{\Delta_k^2}{\Delta^2}
\left( \widetilde{v}_{\max}
+ \frac{1}{3}\Delta\right)},
    \end{align*}
    for all $k\neq k^\star$, {\color{green} using in the last inequality that $\gamma_k\geq \gamma_{\rm min}$, $t\geq 3$,  $x \mapsto x^2/(\widetilde{v}_{\max}+x/3)$ is increasing for $x>0$.}
\end{proof}

\subsubsection{Combined Bound}

The following results contain the necessary concentration results together with the bound on the expected leader regret as a function of sub-optimality gap $\Delta$. We present the results for the i.i.d.~setting as well as under the $\varphi$-mixing assumption on losses.

\begin{proposition}\label{proposition:hedge_main_concentration}
	Let Assumptions~\ref{assumption:losses_0-1}-\ref{assumption:losses_iid_or_mixing} hold. Let $\eta_t = 2\sqrt{\log K/t}$ for $t \in [T]$ be the learning rate for decreasing Hedge. Let  $k^\star\in [K]$ be the expert with the smallest expected loss among $K$ experts. Define $\widetilde{v}_k := \Var(\ell^{(k)}_t-\ell^{(k^\star)}_t)$, $k\in [K]$, $k\neq k^\star$.

For the dependent cases and each $k\neq k^\star$, define $\widetilde{\theta}_t^{(k)}:=1+4\sum_{n=1}^{t-1}\varphi_n^{\chi,k}$ and $\widetilde{\rho}_t^{(k)}:=1+\sum_{n=1}^{t}\sqrt{\varphi_n^{\chi,k}}$, where $\varphi_n^{\chi,k}:=\varphi_n(\{\chi_t^{(k)}\}_{t\in\mathbb Z})$ denotes the $\varphi$-mixing coefficient at lag $n$ of the excess-loss process $\chi_t^{(k)}:=\ell_t^{(k)}-\ell_t^{(k^\star)}$.

 Then, for every $k\neq k^\star$ and all $t\geq t_0$,
    \begin{equation*}
       \P\left( L_{t}^{(k)}-L_{t}^{(k^\star)}<\frac{\Delta_k t }{2}\right)
        \leq
        \min\left\{ H(t), B(t) \right\}
    \end{equation*}
    with the following cases:
    \begin{description}
        \item[(i)] If the dependence conditions of Lemma~\ref{lem:hedge-bernstein}(i) hold, then
        \begin{align*}
            H(t)=
              \exp \left(-\frac{\Delta_k^2 t }{8 }\right),\quad
            B(t)= \exp \left(-\frac{ \Delta_k^2 t}{8 (\widetilde{v}_k + \frac{1}{3}\Delta_k)}\right),\enspace {\it with} \enspace t_0=1.
        \end{align*}
         \item[(ii)] If the dependence conditions of Lemma~\ref{lem:hedge-bernstein}(ii) hold, then
         \begin{align*}
            H(t)=
            \sqrt{e}\exp \left(-\frac{\Delta_k^2 t }{{\color{green} 
            32\widetilde{\theta}_t^{(k)}}}\right),\quad B(t) =  
            \exp \left(-\frac{ \Delta_k^2 t }{{\color{green} 4} (\widetilde{\rho}^{(k)}_t)^2({\color{green} 32} \widetilde{v}_k + \Delta_k)}\right),\enspace {\it with} \enspace t_0=1.
        \end{align*}
        \item[(iii)] If the dependence conditions of Lemma~\ref{lem:hedge-bernstein}(iii) hold, then
        \begin{align*}          
H(t)=\sqrt{e}\exp\left(-\frac{t\Delta_k^2}{32\widetilde{\theta}_t^{(k)}}\right)
,\quad B(t) = 2 \exp \left(-\frac{ \Delta_k^2 t}{32(\log t)^{2 / \gamma_k} \left( \widetilde{v}_k+\frac{1}{3}\Delta_k\right)}\right) , 
        \end{align*}
        where $t_0$ is given explicitly in \eqref{t0_hedge} in Lemma~\ref{lem:hedge-bernstein},  and $\gamma_k > 0$ for all $k\neq k^\star$.
    \end{description}
\end{proposition}

\begin{proof}[Proof of Proposition~\ref{proposition:hedge_main_concentration}]
The proof is identical to the proof of Proposition~\ref{proposition:main_concentration}, where Lemma~\ref{lemma:hedge_concentr_lemma} and Lemma~\ref{lem:hedge-bernstein} are used instead of Lemma~\ref{lemma:ftl_concentr_lemma} and Lemma~\ref{lem:ftl-bernstein}, respectively.
\end{proof}

\subsection{FTL Regret Bounds}
\label{appendix:proofs_ftl_regret}

{\color{green}
The proof of Proposition~\ref{proposition:main_concentration} is a direct combination of the concentration results of Appendix~\ref{appendix:proofs_concentration_ftl}.
With this proposition at hand, we can extend the time-splitting approach used by \cite{mourtadaOptimalityHedgeAlgorithm2019} to studying the regret of FTL under dependence. 
}

\begin{proposition}[Hoeffding-type regret bounds]\label{proposition:ftl_hoeffding_regret_bound}
    Under the same setting of Lemma~\ref{lemma:ftl_concentr_lemma}, suppose that the FTL leader is almost surely unique at every round {\color{green} $t > 1$}. 
    \begin{description}
        \item[(i)] If Assumption~\ref{assumption:losses_iid_or_mixing}(i) holds, then 
        \begin{equation*}
            \E[ \overline{R}_{\textsc{ftl},T} ] 
            \leq
            2 + \frac{2 \log K + 4}{\Delta^2} .
        \end{equation*}
        \item[(ii)] Let Assumption~\ref{assumption:losses_iid_or_mixing}(ii) and Assumption~\ref{ass:uniform summability}(i) hold, and define $\overline{\theta}_{\max} := \max_{k \in [K]} \sup_{t \geq 1} \theta^{(k)}_t$. Then
        \begin{equation*}
            \E[ \overline{R}_{\textsc{ftl},T} ] 
            \leq
            3 + \frac{8 \overline{\theta}_{\max} (\log K + 4)}{\Delta^2} .
        \end{equation*}
    \end{description}
    The bounds are uniform in $T$.
\end{proposition}

Our proof proceeds by exploiting the time splitting arguments also used by \cite{mourtadaOptimalityHedgeAlgorithm2019} for studying decreasing Hedge. Informally, the key idea is that one may split $[T]$ into ``warm-up'' and ``effective'' periods. By choosing the splitting time, $t_0$, carefully, we can obtain bounds that depend only logarithmically on the number of experts.

\begin{proof}[Proof of Proposition~\ref{proposition:ftl_hoeffding_regret_bound}]
     Recall that for $t\in [T]$ we define {\color{green} $c_t=\1{ \textsc{ftl}_{t} \not= \textsc{ftl}_{t+1} }$} (the leader set changes at $t$). We also define $C_t$ as the total number of times the leader set changes up to time $t\in [T]$, that is $C_t = \sum_{\tau=1}^t c_\tau$, using the fact that almost surely there are no ties in determining the FTL expert during any round.
    Introduce now the ``warm-up time'' $t_0 \geq 2$, so that 
    \begin{equation}\label{CT}
        C_T =\sum_{t=1}^T c_t \leq  t_0 + \sum_{t=t_0+1}^T c_t,
    \end{equation}
    where for periods $1 \leq t \leq t_0$ we use $c_t\leq 1$.
    We now write
    \begin{align}
        c_t 
        &= {\color{green} \1{\textsc{ftl}_{t} \neq \textsc{ftl}_{t+1}}%
        \leq 
        \1{ \textsc{ftl}_{t} \neq \{1\} } + \1{ \textsc{ftl}_{t+1} \neq \{1\} }} .\label{ct}
    \end{align}

    \medskip

    \noindent \textbf{Case (i).}
    Using \eqref{ct} we write that
    \begin{align*}
        \E[c_t]
        &\leq \P( \textsc{ftl}_{t}  \neq \{1\} ) + \P( \textsc{ftl}_{t+1}  \neq \{1\} )
        \leq K \left( e^{-(t-1) \Delta^2/2} + e^{-t \Delta^2/2} \right) ,
    \end{align*}
    where in the last inequality we used the bound \eqref{eq:simple_ftl_conc_bound} in Corollary~\ref{corollary:simple_ftl_conc_bound}.
    
    Next,  we note that for $t \geq t_0 + 1$
    \begin{align*}%
        K \left( e^{-(t-1) \Delta^2/2} + e^{-t \Delta^2/2} \right)
        & =   \left( K e^{-t_0 \Delta^2/2}   \right) e^{- (t -t_0  - 1)\Delta^2/2} +   \left( K e^{-t_0 \Delta^2/2}   \right) e^{- (t -t_0)\Delta^2/2} \nonumber\\
        & \leq e^{- (t -t_0  - 1)\Delta^2/2} + e^{- (t -t_0)\Delta^2/2} .
    \end{align*}
   Finally, we choose $t_0 = \lceil \frac{2 \log K}{\Delta^2} \rceil$ using \eqref{eq:simple_ftl_conc_bound} in Corollary~\ref{corollary:simple_ftl_conc_bound}, and, taking expectation on both sides of \eqref{CT}, write
    \begin{equation*}
        \E[ C_T ] 
        \leq 
        1 + \frac{2 \log K}{\Delta^2} + \left( 1 + 2 \sum_{s=1}^\infty e^{- s \Delta^2/2} \right)
        \leq 
        2 + \frac{2 \log K}{\Delta^2} + \frac{4}{\Delta^2},
    \end{equation*}
    where in the last inequality Lemma~\ref{lemma:exp_series_limit} is applied. Finally, by Lemma~\ref{lemma:rooij_ftl_lemma} and noticing  that  $S_T \leq  1$, we obtain the following bound that does not depend on the terminal time $T$:
    \begin{equation*}
            \E[ \overline{R}_{\textsc{ftl},T} ] 
            \leq
            \E[ C_T ]\leq
            2 + \frac{2 \log K + 4}{\Delta^2},
        \end{equation*}
    as required.

    \medskip

    \noindent \textbf{Case (ii).} 
    We start by noticing that by Assumption~\ref{ass:uniform summability}(i), $\overline{\theta}_{\max} =\max_{k \in [K]} \sup_{t \geq 1}\{1+4 \sum_{n=1}^{t-1} \varphi_n^{(k)}\}\leq 1+4 C_{1,\varphi}<\infty$. Next,
    similarly to case (i), using \eqref{ct} we write that
    \begin{align*}
        \E[c_t]
        &\leq \P( \textsc{ftl}_{t}  \neq \{1\} ) + \P( \textsc{ftl}_{t+1}  \neq \{1\} )
        \leq K \sqrt{e} \left[ \exp\left(- \frac{(t-1) \Delta^2}{8 \overline{\theta}_{\max}} \right) + \exp\left(- \frac{t \Delta^2}{8 \overline{\theta}_{\max}} \right) \right],
    \end{align*}
    where in the last inequality we used the bound \eqref{eq:simple_ftl_mixing_bound} in Corollary~\ref{corollary:simple_ftl_conc_bound}.
    The same bound allows to set $t_0 = \lceil \frac{8 \overline{\theta}_{\max} \log K}{\Delta^2} \rceil$. 
    Hence, \eqref{CT} together with Lemma~\ref{lemma:exp_series_limit} yields
    \begin{equation*}
        \E[ C_T ] 
        \leq
        1 + \frac{8 \overline{\theta}_{\max} \log K}{\Delta^2} + \sqrt{e}\left(1 + 2 \sum_{s={\color{green} 1}}^\infty e^{- s \Delta^2/(8 \overline{\theta}_{\max})} \right)
        \leq 
        3 + \frac{8 \overline{\theta}_{\max} \log K}{\Delta^2} + \frac{32 \overline{\theta}_{\max}}{\Delta^2} .
    \end{equation*}
    Finally, utilizing this bound in Lemma~\ref{lemma:rooij_ftl_lemma} with  $S_T \leq  1$ concludes the proof.
\end{proof}
{\color{green}
Proposition~\ref{proposition:ftl_hoeffding_regret_bound}(i) {\color{green} can be strengthened under the additional assumption that the loss vectors are i.i.d., as in Assumption~\ref{assumption:losses_iid_or_mixing}(iii). 
Under i.i.d.\ losses, \cite{degennePerchet2016} established an $\mathcal{O}(\log K/\Delta)$ expected pseudo-regret bound for FTL in the full-information setting, allowing heterogeneous sub-optimality gaps. The following proposition establishes an explicit bound for expected regret relative to the best expert in hindsight under joint temporal independence, without requiring independence across experts or excluding ties among cumulative expert losses.
}
\begin{proposition}[FTL under jointly i.i.d.\ losses]\label{proposition:ftl_iid_sharp}
Let $K\geq 2$ and suppose that Assumptions~\ref{assumption:losses_0-1} and \ref{assumption:losses_iid_or_mixing}(iii) hold. Assume that there exists a unique best expert $k^\star\in[K]$ with $\Delta>0$. Then, for every $T\geq 1$, the expected regret of FTL satisfies
\begin{equation}
\label{eq:ftl_iid_bound}
    \E\big[\overline{R}_{\textsc{ftl},T}\big] \leq \frac{4(1+\log(K-1))+\log K}{\Delta}.
\end{equation}
\end{proposition}

\begin{proof}[Proof of Proposition~\ref{proposition:ftl_iid_sharp}]
Let $R^\star_{\textsc{ftl},T}:=\overline{L}_T-L_T^{(k^\star)}$ denote the cumulative pseudo-regret. Under Assumption~\ref{assumption:losses_iid_or_mixing}(iii), the loss vector $\boldsymbol{\ell}_t$ is independent of the FTL weights $\bm{\omega}_{\textsc{ftl},t}$, which depend only on the preceding loss vectors. Consequently,
\begin{equation}\label{eq:ftl_iid_pseudoregret}
\E[R^\star_{\textsc{ftl},T}]=\sum_{t=1}^T\sum_{k\neq k^\star}\Delta_k\E[\omega_{\textsc{ftl},t}^{(k)}].
\end{equation}
By the FTL weighting rule, for every $k\neq k^\star$ and $t\geq1$,
\begin{equation}\label{eq:exp_concentration_ftl}
\E[\omega_{\textsc{ftl},t}^{(k)}]\leq\P\left(L_{t-1}^{(k)}\leq L_{t-1}^{(k^\star)}\right)\leq\exp\left(-\frac{(t-1)\Delta_k^2}{2}\right),
\end{equation}
where the last inequality follows from Lemma~\ref{lemma:iid_hoeffding}, applied to the i.i.d. differences $\{\ell_s^{(k)}-\ell_s^{(k^\star)}\}_s\subset[-1,1]$, with expectation $\Delta_k>0$. For $t=1$, the inequality holds trivially. 

Since $0<\Delta\leq\Delta_k\leq1$, we can write
\begin{equation*}
\sum_{k\neq k^\star}\Delta_k\E[\omega_{\textsc{ftl},t}^{(k)}]=\int_0^1\sum_{\substack{k\neq k^\star\\\Delta_k>u}}\E[\omega_{\textsc{ftl},t}^{(k)}]\,du.
\end{equation*}
Since $\bm{\omega}_{\textsc{ftl},t}\in\Delta^{K-1}$ almost surely, we have $\omega_{\textsc{ftl},t}^{(k)}\geq0$ for every $k\in[K]$ and $\sum_{k=1}^K\omega_{\textsc{ftl},t}^{(k)}=1$. Therefore, for any $u\in[0,1]$,
\begin{equation*}
\sum_{\substack{k\neq k^\star\\\Delta_k>u}}\E[\omega_{\textsc{ftl},t}^{(k)}]\leq1.
\end{equation*}
On the other hand, since $\Delta_k\geq\max\{u,\Delta\}$ for every $k$ in the sum, \eqref{eq:exp_concentration_ftl} yields
\begin{equation*}
\sum_{\substack{k\neq k^\star\\\Delta_k>u}}\E[\omega_{\textsc{ftl},t}^{(k)}]\leq(K-1)\exp\left(-\frac{(t-1)\max\{u,\Delta\}^2}{2}\right).
\end{equation*}
Combining these inequalities, we obtain
\begin{equation}\label{eq:ftl_iid_weight_bound}
\sum_{\substack{k\neq k^\star\\\Delta_k>u}}\E[\omega_{\textsc{ftl},t}^{(k)}]\leq\min\left\{1,(K-1)\exp\left(-\frac{(t-1)\max\{u,\Delta\}^2}{2}\right)\right\}.
\end{equation}
Moreover, for any $M\geq1$ and $\alpha>0$, comparison with the integral of the nonincreasing function $x\mapsto\min\{1,Me^{-\alpha x}\}$ gives
\begin{equation}\label{eq:ftl_iid_series}
\sum_{s=0}^{\infty}\min\{1,Me^{-\alpha s}\}\leq1+\int_0^\infty\min\{1,Me^{-\alpha x}\}\,dx=1+\frac{1+\log M}{\alpha}.
\end{equation}
Combining \eqref{eq:ftl_iid_pseudoregret}, \eqref{eq:ftl_iid_weight_bound}, and \eqref{eq:ftl_iid_series},  we obtain
\begin{align*}
\E[R^\star_{\textsc{ftl},T}]&\leq\int_0^1\sum_{s=0}^{\infty}\min\left\{1,(K-1)\exp\left(-\frac{s\max\{u,\Delta\}^2}{2}\right)\right\}\,du\\
&\leq1+2(1+\log(K-1))\int_0^1\frac{du}{\max\{u,\Delta\}^2}=1+2(1+\log(K-1))\left(\frac{2}{\Delta}-1\right)\\
&\leq\frac{4(1+\log(K-1))}{\Delta},
\end{align*}
where the last inequality uses $0<\Delta\leq1$ and $K\geq2$.

Finally, using the regret decomposition
\begin{equation*}
\overline{R}_{\textsc{ftl},T}=R^\star_{\textsc{ftl},T}+D_T,
\end{equation*}
and noting that Lemma~\ref{lemma:mixing_DT}(i) gives $\E[D_T]\leq\log K/\Delta$ under Assumption~\ref{assumption:losses_iid_or_mixing}(iii), we obtain \eqref{eq:ftl_iid_bound},
as required.
\end{proof}
}

\begin{proposition}[Bernstein-type regret bounds]\label{proposition:ftl_bernstein_regret_bound} 
    Let the conditions of Lemma~\ref{lem:ftl-bernstein} hold.
    Assume that there are almost surely no ties in the FTL weights for all {\color{green} $t > 1$}, and define $v_{\max}:=\max_{k\in[K]} v_k$.
    \begin{description}
        \item[(i)] Under the conditions of case (i) of Lemma~\ref{lem:ftl-bernstein} it holds
        \begin{equation*}
            \E\big[\overline{R}_{\textsc{ftl},T}\big]
            \leq\
            2 + \left(\log(2K)+2\right) \frac{8v_{\max}+\tfrac{4}{3}\Delta}{\Delta^2} .
        \end{equation*}
        \item[(ii)] Under the conditions of case (ii) of Lemma~\ref{lem:ftl-bernstein}, and, additionally, under Assumption~\ref{ass:uniform summability roots}(i),  define  $\overline{\rho}_{\max}:=\max_{k\in[K]}\sup_{t\geq1}\rho_t^{(k)}$. Then
        \begin{equation*}
            \E\big[\overline{R}_{\textsc{ftl},T}\big]
            \leq\
            2 + \left(\log({\color{green} 2}K)+2 \right) {\color{green} \frac{4\overline{\rho}_{\max}^2(32v_{\max}+\Delta)}{\Delta^2}}.
        \end{equation*}
        \item[(iii)] Under the conditions of case (iii) of Lemma~\ref{lem:ftl-bernstein}, it holds 
        \begin{align*}
            \E\big[\overline{R}_{\textsc{ftl},T}\big]
            & \leq \widetilde{t}_0 + 1 + 2\sum_{t=\widetilde{t}_0}^T\exp\left( -\frac{t \widetilde{c}_{\Delta}}{ 2(\log t)^{\alpha} } \right),
        \end{align*}
        where $\alpha = 2/\gamma_{\rm min}$,
        \begin{align*}
            \widetilde{c}_{\Delta} & = 
            \frac{\Delta^2}{32(v_{\max} + \Delta/6)} ,
        \end{align*}
        and $\widetilde{t}_0$ is defined as 
        \begin{equation*}
           \widetilde{t}_0
            =
            \left\lceil  
                \inf \left\{ t  \geq 
                    \max\left\{ {e^\alpha}, {t_0 } \right\}  
                \,\bigg\vert\, \frac{t}{(\log t)^\alpha} \geq \frac{2 \log(4K)}{\widetilde{c}_{\Delta}} \right\}
            \right\rceil, 
        \end{equation*}   
        where $t_0$ is given in \eqref{t0}.
    \end{description}
    The bounds are uniform in $T$.
\end{proposition}

\begin{proof}[Proof of Proposition~\ref{proposition:ftl_bernstein_regret_bound}]
    We proceed analogously to the proof strategy used for Proposition~\ref{proposition:ftl_hoeffding_regret_bound}, but applying instead the  Bernstein-type results from Lemma~\ref{lem:ftl-bernstein}.
    
    \medskip

    \noindent \textbf{Case (i).} 
    By Corollary~\ref{corollary:simple_ftl_conc_bound_bernstein}(i)
    \begin{align}\label{eq:Pt}
        \P \left(\textsc{ftl}_{t+1}\neq\{k^\star\} \right)
         \le 2 (K-1) e^{- t c_\Delta}
         \le 2K e^{- t c_\Delta}
        \qquad\text{for all }t \geq 1,
    \end{align}
    with $
    c_\Delta = {\Delta^2}/{(8 v_{\max} + \tfrac{4}{3} \Delta)}
    $. 
    Recall that $C_T = \sum_{t=1}^T \1{ \textsc{ftl}_{t} \neq \textsc{ftl}_{t+1} }$, and let $t_0$ be a warm-up time so that, using \eqref{eq:Pt}, here too one obtains 
    \begin{align*}
        \mathbb{E}[C_T]
        & \leq t_0 + \sum_{t=t_0+1}^T \big( \P(\textsc{ftl}_{t}\neq\{k^\star\})+\P(\textsc{ftl}_{t+1}\neq\{k^\star\}) \big) \\
        & \leq t_0 + 2K \sum_{t=t_0+1}^T \left( e^{- (t-1) c_\Delta} + e^{- t c_\Delta} \right) \\
        & \leq t_0 + 2K e^{-t_0 c_\Delta} \sum_{t=t_0+1}^T \left( e^{- (t-t_0-1) c_\Delta} + e^{- (t-t_0) c_\Delta} \right) .
    \end{align*}
    Now, let the warm-up time be
    $t_0 = \left\lceil \frac{\log(2K)}{c_\Delta}\right\rceil$, 
    so that
    \begin{equation*}
        \mathbb{E}[C_T]
         \le t_0 + 1 + 2\sum_{s=1}^{\infty} e^{-s c_\Delta }
         = t_0 + 1 + \frac{2 e^{-c_\Delta}}{1-e^{-c_\Delta}}
         \le t_0 + 1 + \frac{2}{c_\Delta}.
    \end{equation*}
    Finally, $t_0\le \frac{\log(2K)}{c_\Delta}+1$ gives
    \begin{equation*}
        \mathbb{E}[C_T]
         \le \frac{\log(2K)}{c_\Delta} + 2 + \frac{2}{c_\Delta}
         = 2 + \left(\log(2K)+2 \right) \frac{1}{c_\Delta}
         = 2 + \left(\log(2K)+2 \right) \frac{8 v_{\max} +\tfrac{4}{3} \Delta}{\Delta^2}.
    \end{equation*}
    Since $\overline{R}_{\textsc{ftl},T}\le C_T$ under bounded losses ($S_T = 1$), the stated bound follows.

    \medskip
    
    \noindent \textbf{Case (ii).} 
    We start by noticing that by Assumption~\ref{ass:uniform summability roots}(i), $\overline{\rho}_{\max}^2 =(1+\max_{k \in [K]} \sum_{n=1}^{{\color{green} \infty}} \sqrt{\varphi_n^{(k)}})^2\leq (1+C_{2,\varphi})^2<\infty$. Next,
    by Corollary~\ref{corollary:simple_ftl_conc_bound_bernstein}(ii), for $t\geq 1$ it holds 
    \begin{equation*}
         \P \left(\textsc{ftl}_{{\color{green} t+1}}\neq\{k^\star\} \right)
         \leq 
         {\color{green} 2}(K-1) e^{
        -{t \, \widetilde{c}_{\Delta}}
        }\leq 
        {\color{green} 2}K e^{
        -{t \, \widetilde{c}_{\Delta}}
        },
    \end{equation*}
    where
    \begin{equation*}
       {\color{green} 
        \widetilde{c}_{\Delta} = \frac{\Delta^2}{4\overline{\rho}_{\max}^2(32v_{\max}+\Delta)}.}
    \end{equation*}
   Analogously to case (i), let the warm-up time be
    $t_0 = \left\lceil \frac{\log({\color{green} 2}K)}{\widetilde{c}_{\Delta}} \right\rceil$, then 
    \begin{align*}
        \mathbb{E}[C_T]
        & \leq t_0 + \sum_{t=t_0+1}^T {\color{green} \left(\P(\textsc{ftl}_{t}\neq\{k^\star\})+\P(\textsc{ftl}_{t+1}\neq\{k^\star\})\right)} \\
        & \leq t_0 + {\color{green} 2}K \sum_{t=t_0+1}^T \left( e^{- (t-1) \widetilde{c}_{\Delta}} + e^{- t \widetilde{c}_{\Delta}} \right) \\
        & \leq t_0 + {\color{green} 2}K e^{-t_0 \widetilde{c}_{\Delta}} \sum_{t=t_0+1}^T \left( e^{- (t-t_0-1) \widetilde{c}_{\Delta}} + e^{- (t-t_0) \widetilde{c}_{\Delta}} \right) \\
        & \leq \left( \frac{\log({\color{green} 2}K)}{\widetilde{c}_{\Delta}} + 1 \right) + \left( 1 + 2 \sum_{s=1}^{\infty} e^{-s \widetilde{c}_\Delta } \right) \\
        & \le \frac{\log({\color{green} 2}K)}{\widetilde{c}_{\Delta}} + 2 + \frac{2}{\widetilde{c}_{\Delta}} = {\color{green} 2 + \left(\log({\color{green} 2}K)+2\right)\frac{4\overline{\rho}_{\max}^2(32v_{\max}+\Delta)}{\Delta^2}},
    \end{align*}
    which concludes the proof.
    
    \medskip

    \noindent \textbf{Case (iii).} 
    By \eqref{eq:bernstein_part-ii-bound} in case (iii) of Lemma~\ref{lem:ftl-bernstein}, for $t \geq t_0$, with $t_0$ in \eqref{t0} (and also from the simplified bound in \eqref{eq:simple_ftl_mixing_bound_bernstein_iii}), it holds 
    \begin{equation*}
         \P \left(\textsc{ftl}_{t+1}\neq\{k^\star\} \right)
         \leq 
         4K \exp\left(
        -\frac{t \, \widetilde{c}_{\Delta}}{ (\log t)^{\alpha} }
        \right) ,
    \end{equation*}
    where
    \begin{equation*}
        \widetilde{c}_{\Delta} 
        = 
        \frac{\Delta^2}{32(v_{\max} + \Delta/6)}
    \end{equation*}
    and $\alpha = 2 / \gamma_{\min}$ for notational convenience.
    We use Lemma~\ref{lemma:logtail_t0} 
    with $b = 2 \log(4 K) / \widetilde{c}_{\Delta}$ to ensure
that the set defining $\widetilde{t}_0$ in the proposition is nonempty.
Since the function $t\mapsto t/(\log t)^\alpha$ is continuous
and nondecreasing for $t\geq e^\alpha$, the definition of $\widetilde{t}_0$ implies that, for all $t\geq\widetilde{t}_0$,
\begin{equation*}
    \frac{t}{(\log t)^\alpha}
    \geq
    \frac{2\log(4K)}{\widetilde{c}_{\Delta}}.
\end{equation*}
Hence,
    \begin{align*}
        4K\exp\left( -\frac{t \, \widetilde{c}_{\Delta}}{ (\log t)^{\alpha} } \right)
        & =\exp\left( \log(4K)-\frac{t \, \widetilde{c}_{\Delta}}{ (\log t)^{\alpha} } \right) \\
        & \leq \exp\left( \frac{t \, \widetilde{c}_{\Delta}}{2 (\log t)^{\alpha} } -\frac{t \, \widetilde{c}_{\Delta}}{ (\log t)^{\alpha} } \right)\leq \exp\left( -\frac{t \, \widetilde{c}_{\Delta}}{ 2(\log t)^{\alpha} } \right) .
    \end{align*}
    Proceeding as in cases (i) and (ii), we obtain
    \begin{align*}
        \mathbb{E}[C_T]
        & \leq \widetilde{t}_0 + 1 + 4 K \sum_{t=\widetilde{t}_0+1}^T\left\{\exp \left(-\frac{(t-1) \widetilde{c}_{\Delta}}{(\log (t-1))^\alpha}\right)+\exp \left(-\frac{t \widetilde{c}_{\Delta}}{(\log t)^\alpha}\right)\right\}\\
        & \leq \widetilde{t}_0 + 1 + 2\sum_{t=\widetilde{t}_0}^T\exp\left( -\frac{t \widetilde{c}_{\Delta}}{ 2(\log t)^{\alpha} } \right).
    \end{align*}
    Note that the series in the last line converges. Indeed, for any $C > 0$ there exists $C' > 0$ such that $\exp(- s C / (\log s)^\alpha) \lesssim \exp(- s^{C'})$. Hence, it holds that $\sum_{t=\widetilde{t}_0}^T\exp\left( -\frac{t \widetilde{c}_{\Delta}}{ 2(\log t)^{\alpha} } \right)\leq \int_{0}^\infty \exp(- s^{C'}) \textnormal{d}s = {\Gamma}(1/{C'}) / {C'} < \infty$, with ${\Gamma}$ the Gamma function.
    Finally, as $\E\left[\overline R_{\textsc{ftl},T}\right]\le \E[C_T]$, one obtains a uniform in $T$ bound, as required.
\end{proof}

\subsubsection{Combined Bounds for FTL}
\label{appendix:proofs_ftl_combined_bounds}

We can now obtain Theorem~\ref{theorem:ftl_regret_mixing} by incorporating the bounds previously derived.

\begin{proof}[Proof of Theorem~\ref{theorem:ftl_regret_mixing}]
    Part (i) follows immediately by combining Proposition~\ref{proposition:ftl_hoeffding_regret_bound}(i) and Proposition~\ref{proposition:ftl_bernstein_regret_bound}(i).
    Similarly, part (ii) is obtained by combining Proposition~\ref{proposition:ftl_hoeffding_regret_bound}(ii) and Proposition~\ref{proposition:ftl_bernstein_regret_bound}(ii), with the relevant constants and $t_0$ defined in the latter.
\end{proof}

\subsection{Hedge Regret Bounds}
\label{appendix:proofs_hedge_regret}

{\color{green}
\begin{lemma}
\label{lemma:mixing_DT}
Let $K\geq 2$ and suppose that Assumption~\ref{assumption:losses_0-1} holds. Assume that there exists a unique best expert $k^\star$ with sub-optimality gap $\Delta>0$ and define $D_T:=L_T^{(k^\star)}-\min_{k\in[K]}L_T^{(k)}$.
\begin{description}
    \item[(i)] If Assumption~\ref{assumption:losses_iid_or_mixing}(i) holds, then
\begin{equation*}%
        \E[D_T]\leq\frac{\log K}{\Delta}.
    \end{equation*}
    \item[(ii)] If Assumptions~\ref{assumption:losses_iid_or_mixing}(iv) and \ref{ass:uniform summability}(ii) hold, then
    \begin{equation*}%
        \E[D_T]\leq\frac{4\overline{\widetilde{\theta}}_{\max}}{\Delta}\left(\log(K-1)+\frac{3}{2}\right),
    \end{equation*}
    {\color{green} where $\overline{\widetilde{\theta}}_{\max}:=\max_{k\neq k^\star}\sup_{t\geq1}\widetilde{\theta}_t^{(k)}=1+4\max_{k\neq k^\star}\sum_{n=1}^{\infty}\varphi_n^{\chi,k}<\infty$, with $\varphi_n^{\chi,k}$ denoting the $\varphi$-mixing coefficients of the excess-loss process $\chi_t^{(k)}=\ell_t^{(k)}-\ell_t^{(k^\star)}$.}
\end{description}
Both bounds are uniform in $T$.
\end{lemma}
\begin{proof}
For every $k\neq k^\star$, let $\Delta_k:=\mu_k-\mu_{k^\star}\geq\Delta$ and define $Z_t^{(k)}:=\ell_t^{(k^\star)}-\ell_t^{(k)}+\Delta_k$. By construction, $\E[Z_t^{(k)}]=0$ and $|Z_t^{(k)}|\leq2$. Moreover,
\begin{equation*}
    D_T=\max\left\{0, \ \max_{k\neq k^\star}\left(\sum_{t=1}^T Z_t^{(k)}-T\Delta_k\right)\right\}.
\end{equation*}
\noindent\textbf{Case (i).} Under Assumption~\ref{assumption:losses_iid_or_mixing}(i), by the Cauchy-Schwarz inequality and Hoeffding's lemma (\citealp{Boucheron2013}, Lemma 2.2), for every $k\neq k^\star$,
\begin{align*}
    \E\left[e^{\Delta(L_T^{(k^\star)}-L_T^{(k)})}\right]
    &=e^{-T\Delta\Delta_k} \, \E\left[e^{\Delta\sum_{t=1}^T Z_t^{(k)}}\right]\\
    &\leq e^{-T\Delta\Delta_k}\left(\E\left[e^{2\Delta(L_T^{(k^\star)}-T\mu_{k^\star})}\right]\E\left[e^{-2\Delta(L_T^{(k)}-T\mu_k)}\right]\right)^{1/2}\\
    &\leq e^{-T\Delta\Delta_k+T\Delta^2/2}\leq1.
\end{align*}
Consequently, by the log-sum-exp inequality and Jensen's inequality,
    \begin{align}
    \label{eq:bound_EDt}
    \E[D_T]&\leq \frac{1}{\Delta} \log\left(\sum_{k=1}^K \E\left[ e^{\Delta(L_T^{(k^\star)}-L_T^{(k)})} \right]\right)\leq \frac{\log K}{\Delta}.
    \end{align}
    \noindent\textbf{Case (ii).} Under Assumption~\ref{assumption:losses_iid_or_mixing}(iv), the centered differences $\{Z_t^{(k)}\}_t$ are $\varphi$-mixing with coefficients bounded by those of the loss-vector process. Since $|Z_t^{(k)}|\leq2$, Lemma~\ref{lemma:rio_hoeffding_uniform_mixing} implies that, for every $u>0$,
\begin{equation*}
\P\left(\sum_{t=1}^T Z_t^{(k)}\geq u\right)\leq\sqrt{e}\exp\left(-\frac{u^2}{8\overline{\widetilde{\theta}}_{\max} T}\right).
\end{equation*}
Therefore, for every $x\geq0$, the union bound yields
\begin{align*}
\P(D_T>x)
&\leq\sum_{k\neq k^\star}\P\left(\sum_{t=1}^T Z_t^{(k)}>T\Delta_k+x\right)\leq\sqrt{e}(K-1)\exp\left(-\frac{(T\Delta+x)^2}{8\overline{\widetilde{\theta}}_{\max} T}\right)\\
&\leq\sqrt{e}(K-1)\exp\left(-\frac{\Delta x}{4\overline{\widetilde{\theta}}_{\max}}\right),
\end{align*}
where the last inequality follows from $(T\Delta+x)^2\geq2T\Delta x$. Since $D_T\geq0$, integration of the tail probability gives
\begin{align*}
\E[D_T]
&=\int_0^\infty\P(D_T>x)\,dx\leq\int_0^\infty\min\left\{1,\sqrt{e}(K-1)\exp\left(-\frac{\Delta x}{4\overline{\widetilde{\theta}}_{\max}}\right)\right\}\,dx\\
&=\frac{4\overline{\widetilde{\theta}}_{\max}}{\Delta}\left(\log(K-1)+\frac{3}{2}\right),
\end{align*}
which concludes the proof.
\end{proof}}

\begin{proposition}[Hoeffding-type regret bounds]\label{proposition:hedge_hoeffding_regret_bound}
    Let $K \geq 3$ and Assumptions~\ref{assumption:losses_0-1}-\ref{assumption:losses_iid_or_mixing} hold. {\color{green} Assume that there exists a unique best expert $k^\star\in[K]$ with $\Delta>0$.} Let $\eta_t = 2\sqrt{\log K/t}$ for $t \in [T]$ be the learning rate for decreasing Hedge.
    \begin{description}
        
\item[(i)] If Assumption~\ref{assumption:losses_iid_or_mixing}(i) holds, then
\begin{equation*}
    \E[\overline{R}_{\textsc{hdg},T}]\leq\frac{4\Delta\log K+25}{\Delta^2}.
\end{equation*}
\item[(ii)] If Assumption~\ref{assumption:losses_iid_or_mixing}(iii) holds, then
\begin{equation*}
    \E[\overline{R}_{\textsc{hdg},T}]\leq\frac{5\log K+25}{\Delta}.
\end{equation*}
        \item[(iii)] If {\color{brown} Assumptions~\ref{assumption:losses_iid_or_mixing}(iv) and \ref{ass:uniform summability}(ii) } hold,
        \begin{equation*}
            \E[\overline{R}_{\textsc{hdg},T}] 
            \leq 
            2 + \sqrt{\log K}+\frac{\sqrt{32 \overline{\widetilde{\theta}}_{\max}} \log K}{\Delta} +\frac{64 \overline{\widetilde{\theta}}_{\max}+16}{\Delta^2}+\frac{4 \overline{\widetilde{\theta}}_{\max}}{\Delta}\left(\log (K-1)+\frac{3}{2}\right),
        \end{equation*}
        {\color{green} where $\overline{\widetilde{\theta}}_{\max}:=\max_{k\neq k^\star}\sup_{t\geq1}\widetilde{\theta}_t^{(k)}=1+4\max_{k\neq k^\star}\sum_{n=1}^{\infty}\varphi_n^{\chi,k}<\infty$, with $\varphi_n^{\chi,k}$ denoting the $\varphi$-mixing coefficients of the excess-loss process $\chi_t^{(k)}=\ell_t^{(k)}-\ell_t^{(k^\star)}$.}
    \end{description}
\end{proposition}

\begin{proof}[Proof of Proposition~\ref{proposition:hedge_hoeffding_regret_bound}]

   {\color{green} We start by defining $D_t:=L_t^{(k^\star)}-\min_{k\in[K]}L_t^{(k)}\geq 0$. }As in the proof of Proposition~\ref{proposition:ftl_hoeffding_regret_bound}, for  $T>t_0$ we split the regret as follows    %
    {\color{green} 
    \begin{equation}\label{regret_bound}
        \overline{R}_{\textsc{hdg},T} 
        =\overline{R}_{\textsc{hdg},t_0}
        +\sum_{t=t_0+1}^{T}r_t^\star
        +D_T-D_{t_0}\leq 
        \overline{R}_{\textsc{hdg},t_0} + \sum_{t=t_0+1}^T r^\star_t +D_T,
    \end{equation}}
    where $\overline{R}_{\textsc{hdg},t_0}$ is the {\color{green} cumulative} regret up to $t_0$.
    By Proposition~\ref{proposition:regret_hedge} in \cite{chernovPredictionExpertAdvice2010}, decreasing Hedge satisfies 
     $\overline{R}_{\textsc{hdg},t_0} \leq \sqrt{t_0 \log K}$.
    Then, for any $t \in \{t_0+1, \ldots, T\}$, the expected instantaneous Hedge regret {\color{green} relative to the best expert in expectation satisfies 
    \begin{equation*}
        \E[ r_t^\star ]
        =
        \E\left[ \sum_{k\neq k^\star} \omega^{(k)}_{\textsc{hdg},t} (\ell_t^{(k)}-\ell_t^{(k^\star)})\right]
        \leq
        \sum_{k\neq k^\star} \E\left[ \omega^{(k)}_{\textsc{hdg},t} \right] ,
    \end{equation*}
    where the last inequality follows from boundedness of losses.}

    We now proceed to bound $\E[ \omega^{(k)}_{\textsc{hdg},t}]$.
    We mimic the steps in the proof of Theorem~2 in \cite{mourtadaOptimalityHedgeAlgorithm2019}. More explicitly, we first notice that 
    if $L_{t-1}^{(k)}-L_{t-1}^{(k^\star)}\geq\frac{\Delta_k(t-1)}{2}$ then with $\eta_t = 2\sqrt{\log K/t}$ one obtains
    \begin{align*}
        \omega^{(k)}_{\textsc{hdg},t} 
        =
        \frac{ \exp( -\eta_t (L^{(k)}_{t-1}-L^{(k^\star)}_{t-1}) ) }
        { 1+\sum_{j\neq k^\star} \exp( -\eta_t (L^{(j)}_{t-1}-L^{(k^\star)}_{t-1}) ) }
        &\leq 
        \exp\left(
        -2\sqrt{\frac{\log K}{t}}
        \frac{\Delta_k (t-1)}{2}\right)\nonumber\\
        &\leq 
        \exp\left(-\Delta_k \sqrt{\frac{(t-1)\log K}{2}}\right).
    \end{align*}
    Since $\omega^{(k)}_{\textsc{hdg},t}\leq 1$, one can write, for $t\geq 2$,
    \begin{align*}
        \omega^{(k)}_{\textsc{hdg},t} 
        {\color{green} \leq} \1{ L_{t-1}^{(k)}-L_{t-1}^{(k^\star)}<{\Delta_k(t-1)}/{2} }+
        \exp\left(-\Delta_k \sqrt{\frac{(t-1)\log K}{2}}\right),
    \end{align*}
    and, taking the expectation, one has 
    \begin{align}
    \label{eq:omega_bound}
        \E [\omega^{(k)}_{\textsc{hdg},t} ]
        {\color{green} \leq} \P\left( L_{t-1}^{(k)}-L_{t-1}^{(k^\star)}<{\Delta_k(t-1)}/{2} \right)+
        \exp\left(-\Delta_k \sqrt{\frac{(t-1)\log K}{2}}\right).
    \end{align}
    We now establish cases (i)–(iii), using the preceding weight bound for cases (i) and (iii).

    \medskip
    \medskip

    \noindent \textbf{Case (i).}
    Set $t_0 = \lceil \frac{8 \log K}{\Delta^2} \rceil$, which is chosen using \eqref{hedge_ass_1} together with Remark~\ref{remark:hedge_concentr_uniform}. Using case (i) of Corollary~\ref{corollary:simple_hedge_conc_bound_hoeffding} (see also the proof of Theorem 2 in \cite{mourtadaOptimalityHedgeAlgorithm2019}), we immediately note that, for $t \geq t_0 + 1$,
    \begin{align*}
        \sum_{k\neq k^\star} \E\left[ \omega^{(k)}_{\textsc{hdg},t} \right]
       & \leq
        (K-1)e ^{- (t-1) \Delta^2 / 8 } +  (K-1)e^{- \Delta \sqrt{(t-1) \log K / 2}\, }  \\
        & \leq ( (K-1) e^{-t_0 \Delta^2/8}) e^{- (t-t_0-1) \Delta^2/8} + ( (K-1) e^{- \Delta \sqrt{(t - 1) \log K/8}}) e^{- \Delta \sqrt{(t - 1) \log K/8}} \\
        & \leq e^{- (t-t_0-1) \Delta^2/8} + e^{- \Delta \sqrt{(t - 1) \log K/8}} ,
    \end{align*}
    where in the last inequality we used the fact that the choice of $t_0 \geq {8 \log (K)}/{\Delta^2}$ implies that $ (K-1) e^{-t_0 \Delta^2 / 8} \leq 1$. Then, for $t\geq t_0+1$ it holds that $ (K-1) e^{- \Delta \sqrt{(t - 1) \log K/8}}\leq 1$.
    Hence,
    \begin{equation}\label{eq:deg_hedge_iid_series}
        \sum_{t=t_0+1}^{\infty} \E[ r_t^\star ]
        \leq
        \sum_{s=0}^\infty e^{-s\Delta^2/8} + \sum_{s=1}^\infty e^{-(\Delta/\sqrt{8})\sqrt{{\color{green} s}}} \leq \dfrac{25}{\Delta^2},
    \end{equation}
    where we applied Lemma~\ref{lemma:exp_series_limit} since $\log K > 1$ (as we assumed $K \geq 3$).
    
    We now consider again the bound  $\overline{R}_{\textsc{hdg},t_0} \leq \sqrt{t_0 \log K}$. {\color{green} Since $K\geq 3$ and $\Delta\leq 1$, our choice of $t_0$ implies  $t_0\leq 9\log K/\Delta^2$, and hence $\overline{R}_{\textsc{hdg},t_0} \leq 3\log K/\Delta$. Therefore, using \eqref{regret_bound}, \eqref{eq:deg_hedge_iid_series}, and the bound in case (i) of Lemma~\ref{lemma:mixing_DT}, we obtain}
    \begin{equation*}
        \E\left[\overline{R}_{\textsc{hdg},T}\right]
        \leq {\color{green} \frac{3\log K}{\Delta} + \frac{\log K}{\Delta}} + \frac{25}{\Delta^2} 
        =
        \frac{4 \Delta \log K + 25}{\Delta^2} ,
    \end{equation*}
    which is the required bound. {\color{green} For $T\leq t_0$, the result follows directly from the worst-case regret bound.}

    \medskip

\noindent\textbf{Case (ii).} Under Assumption~\ref{assumption:losses_iid_or_mixing}(iii), Theorem~2 of \cite{mourtadaOptimalityHedgeAlgorithm2019} yields
\begin{equation*}
\E\left[\overline{L}_T-L_T^{(k^\star)}\right]\leq\frac{4\log K+25}{\Delta}.
\end{equation*}
Since Assumption~\ref{assumption:losses_iid_or_mixing}(iii) implies Assumption~\ref{assumption:losses_iid_or_mixing}(i), Lemma~\ref{lemma:mixing_DT}(i) gives $\E[D_T]\leq\log K/\Delta$. Consequently,
\begin{equation*}
\E[\overline{R}_{\textsc{hdg},T}]=\E\left[\overline{L}_T-L_T^{(k^\star)}\right]+\E[D_T]\leq\frac{5\log K+25}{\Delta},
\end{equation*}
which proves the result.

\medskip

    \noindent \textbf{Case (iii).} 
    Set $t_0 = \lceil \frac{ 32 \overline{\widetilde{\theta}}_{\max} \log K}{\Delta^2} \rceil$, which is chosen using case (ii) of Corollary~\ref{corollary:simple_hedge_conc_bound_hoeffding}. For $t \geq t_0 + 1$ it holds that
    {\color{green} 
    \begin{align*}
        \sum_{k\neq k^\star}\E\left[\omega^{(k)}_{\textsc{hdg},t}\right]
        &\leq \sqrt{e}\left((K-1)e^{-t_0\Delta^2/(32\overline{\widetilde{\theta}}_{\max})}\right)e^{-(t-t_0-1)\Delta^2/(32\overline{\widetilde{\theta}}_{\max})}\\[-5pt]
        &\quad+\left((K-1)e^{-\Delta\sqrt{(t-1)\log K/8}}\right)e^{-\Delta\sqrt{(t-1)\log K/8}}\\[5pt]
        &\leq 2e^{-(t-t_0-1)\Delta^2/(32\overline{\widetilde{\theta}}_{\max})}+e^{-\Delta\sqrt{(t-1)\log K/8}}.
    \end{align*}}
    {\color{green} The last inequality follows from $\sqrt{e}\leq2$ and the choice of $t_0$, which ensures that $(K-1)e^{-t_0\Delta^2/(32\overline{\widetilde{\theta}}_{\max})}\leq1$ and $(K-1)e^{-\Delta\sqrt{(t-1)\log K/8}}\leq1$ for $t\geq t_0+1$.}

    {\color{green} Therefore, since $K\geq3$, Lemma~\ref{lemma:exp_series_limit} yields}
    \begin{align}\label{proof-eq:hedge_depend_series_sum}
    \sum_{t=t_0+1}^T\E[r_t^\star] &\leq 2\sum_{s=0}^{\infty}e^{-s\Delta^2/(32\overline{\widetilde{\theta}}_{\max})}+\sum_{s=1}^{\infty}e^{-\Delta\sqrt{s/8}}\leq 2\left(1+\frac{32\overline{\widetilde{\theta}}_{\max}}{\Delta^2}\right)+\frac{16}{\Delta^2}.
\end{align}
    Once more, using that $\overline{R}_{\textsc{hdg},t_0} \leq \sqrt{t_0 \log K}$ and our choice of $t_0$, we obtain
    \begin{equation*}
        \overline{R}_{\textsc{hdg},t_0} \leq \sqrt{\log K} + \frac{\sqrt{32\overline{\widetilde{\theta}}_{\max}} \log K}{\Delta} .
    \end{equation*}
    Finally, Lemma~\ref{lemma:mixing_DT}(ii) gives $\E[D_T]\leq 4\overline{\widetilde{\theta}}_{\max}(\log(K-1)+3/2)/\Delta$. Using this with \eqref{regret_bound} together with \eqref{proof-eq:hedge_depend_series_sum} yields
    {\color{green} \begin{align*}
    \E\left[\overline{R}_{\textsc{hdg},T}\right]
    &\leq \sqrt{\log K}+\frac{\sqrt{32\overline{\widetilde{\theta}}_{\max}}\log K}{\Delta}+2+\frac{64\overline{\widetilde{\theta}}_{\max}+16}{\Delta^2}+\frac{4\overline{\widetilde{\theta}}_{\max}}{\Delta}\left(\log(K-1)+\frac{3}{2}\right),
\end{align*}
which is the required bound. For $T\leq t_0$, the result follows directly from the worst-case regret bound.}
\end{proof}

\begin{proposition}[Bernstein-type regret bounds]\label{proposition:hedge_bernstein_regret_bound} 
    Let $K\geq 3$. Assume that the conditions of Lemma~\ref{lem:hedge-bernstein} hold.
    Define $\widetilde{v}_{\max}:=\max_{k\neq k^\star} \widetilde{v}_k$.
    \begin{description}
        \item[(i)] Under the conditions of case (i) of Lemma~\ref{lem:hedge-bernstein}, it holds
        \begin{equation*}
            \E\big[\overline{R}_{\textsc{hdg},T}\big]
            \leq
                1 + \sqrt{\log K} 
                + \frac{ {\color{green} \left(4\sqrt{\frac{2}{3}}+1\right)} \Delta \log K + 8 (\widetilde{v}_{\max} + \frac{1}{3} \Delta) + 16 }{\Delta^2}.
        \end{equation*}
        \item[(ii)] Under the conditions of case (ii) of Lemma~\ref{lem:hedge-bernstein}, and additionally, under {\color{brown} Assumption~\ref{ass:uniform summability roots}(ii)}, it holds 
        \begin{align*}      
            \E\big[\overline{R}_{\textsc{hdg},T}\big] & \leq 1+\sqrt{\log K}+\frac{2\sqrt{33}\,\overline{\widetilde{\rho}}_{\max}\log K+4\overline{\widetilde{\theta}}_{\max}\left(\log(K-1)+\frac32\right)}{\Delta} \\
            &\quad\qquad +\frac{4(\overline{\widetilde{\rho}}_{\max})^2(32\widetilde{v}_{\max}+\Delta)+16}{\Delta^2}, 
        \end{align*}
        where $$\overline{\widetilde{\rho}}_{\max}:=\max_{k\neq k^\star}\sup_{t\geq1}\widetilde{\rho}_t^{(k)}=1+\max_{k\neq k^\star}\sum_{n=1}^{\infty}\sqrt{\varphi_n^{\chi,k}}$$ and $$\overline{\widetilde{\theta}}_{\max}:=\max_{k\neq k^\star}  \sup_{t\geq1}\widetilde{\theta}_t^{(k)}=1+4\max_{k\neq k^\star}\sum_{n=1}^{\infty}\varphi_n^{\chi,k}. $$
        \item[(iii)] Under the conditions of case (iii) of Lemma~\ref{lem:hedge-bernstein}, and additionally under Assumption~\ref{ass:uniform summability}(ii), it holds 
        \begin{align*}
            \E[\overline{R}_{\textsc{hdg},T}] 
            & \leq 
            \sqrt{\widetilde{t}^*_0 \log K} + \frac{16}{\Delta^2} + \sum_{s=\widetilde{t}^*_0}^\infty \exp\left( - \frac{ s \widetilde{c}_\Delta}{2 (\log s)^\alpha} \right) {\color{green} +\frac{4\overline{\widetilde{\theta}}_{\max}}{\Delta}\left(\log(K-1)+\frac32\right)},
        \end{align*}
        where $\alpha = 2/\gamma_{\rm min}$, $
            \widetilde{c}_\Delta := \frac{ \Delta^2 }{ 32(\widetilde{v}_{\max} + {\color{green} \tfrac{1}{3}} \Delta) }$, 
        and $\widetilde{t}^*_0$ is defined as 
        \begin{equation*}
           \widetilde{t}^*_0
            =
            \left\lceil  
                \inf \left\{ t  \geq 
                    \max\left\{ 3, {e^\alpha}, {t_0 } \right\}  
                \,\bigg\vert\, \frac{t}{(\log t)^\alpha} \geq \frac{ 2 \log(2K) }{ \widetilde{c}^*_\Delta } \right\}
            \right\rceil 
        \end{equation*} 
        where $t_0$ is given in \eqref{t0_hedge}, $\widetilde{c}^*_\Delta := \frac{ \Delta^2 }{ 32\, \nu^* }$, and $\nu^* := \max(1, \widetilde{v}_{\max} + {\color{green} \frac{1}{3}} \Delta)$.
    \end{description}
    The bounds are uniform in $T$.
\end{proposition}

\begin{proof}[Proof of Proposition~\ref{proposition:hedge_bernstein_regret_bound}]
    We proceed precisely as in the proof of Proposition~\ref{proposition:hedge_hoeffding_regret_bound} using the Bernstein-type results in Corollary~\ref{corollary:simple_hedge_conc_bound_bernstein}. The proof reduces to analyzing \eqref{eq:omega_bound} under the conditions of cases {\color{green} (i)--(iii)} as follows.
    
    \medskip

    \noindent \textbf{Case (i).} 
    Define $c_{\Delta} := \Delta^2 / (8 (\widetilde{v}_{\max} + \frac{1}{3}\Delta))$, $\nu^* := \max(1, \widetilde{v}_{\max} + \frac{1}{3}\Delta)$ and set $t_0 = \lceil \frac{8 \nu^* \log K}{ \Delta^2 } \rceil$. 
    Analogously to the proof of case (i) of Proposition~\ref{proposition:hedge_hoeffding_regret_bound}, we apply case (i) of Corollary~\ref{corollary:simple_hedge_conc_bound_bernstein} and, for $t \geq t_0 + 1$, find
    \begin{align*}
        \sum_{k\neq k^\star} \E\left[ \omega^{(k)}_{\textsc{hdg},t} \right]
        & \leq
        (K-1)e ^{- (t-1) c_{\Delta} } +  (K-1)e^{- \Delta \sqrt{(t-1) \log K / 2}\, }  \\
        & \leq \big( (K-1) e^{-t_0 c_{\Delta}} \big) e^{- (t-t_0-1) c_{\Delta}} + \big( (K-1) e^{- \Delta \sqrt{(t - 1) \log K/8}} \big) e^{- \Delta \sqrt{(t - 1) \log K/8}} \\
        & \leq e^{- (t-t_0-1) c_{\Delta}} + e^{- \Delta \sqrt{(t - 1) \log K/8}} .
    \end{align*}
    To obtain the last inequality we used, first, the fact that the choice of $t_0 \geq {8 \nu^* \log K}/{ \Delta^2 } \geq { \log (K)} / c_{\Delta}$ implies that $ (K-1) e^{-t_0 c_{\Delta}} \leq 1$ assuming $K\geq 3$. Second, for $t\geq t_0+1$ it holds that 
    \begin{align*}
        (K-1) \exp\left( - \Delta \sqrt{ (t-1)\frac{\log K}{8} } \right)
        &\leq (K-1) \exp\left( - \Delta \sqrt{ t_0\frac{\log K}{8} } \right)\\&
        \leq
        (K-1) \exp\left( - \sqrt{ \nu^* (\log K)^2 } \right)
        \leq 1 .
    \end{align*}
    Hence,
    \begin{equation*}
        \sum_{t=t_0+1}^{\infty} \E[ {\color{green} r_t^\star} ]
        \leq
        \sum_{s=0}^\infty e^{-s c_{\Delta}} + \sum_{s=1}^\infty e^{-(\Delta/\sqrt{8})\sqrt{s}} 
        \leq 
        1 + \frac{1}{ c_{\Delta}} + \frac{16}{\Delta^2},
    \end{equation*}
    where we applied Lemma~\ref{lemma:exp_series_limit} since $\log K > 1$ for $K\geq 3$.
    
    We now consider again the bound $\overline{R}_{\textsc{hdg},t_0} \leq \sqrt{t_0 \log K}$. 
    Let us observe that, since losses $\ell_t^{(k)} \in [0, 1]$ for all $k \in [K]$, {\color{green} $\widetilde{v}_{\max} \leq 1$}, and hence $\nu^* \leq 4/3$.
    Using our choice of $t_0$, we obtain
    $\overline{R}_{\textsc{hdg},t_0} 
        \leq \sqrt{({ \frac{32}{3} \log K}/{\Delta^2} + 1)\log K}
        \leq \sqrt{\log K} + 4 \sqrt{\frac{2}{3}} \log K / \Delta $.
    Therefore, using \eqref{regret_bound}, {\color{green} together with the bound $\E[D_T]\leq\log K/\Delta$ from \eqref{eq:bound_EDt}, we obtain} 
    \begin{equation*}
        \E\left[\overline{R}_{\textsc{hdg},T} \right]
        \leq 
        \sqrt{\log K} + \frac{ 4 \sqrt{\frac{2}{3}} \log K }{\Delta} + 1+\dfrac{1}{ {c_{\Delta}}} + \dfrac{16}{\Delta^2} {\color{green} +\frac{\log K}{\Delta}}.
    \end{equation*}
    {\color{green} Substituting  the definition of $c_\Delta$ and collecting the terms gives the required bound. For $T\leq t_0$, the result follows directly from the worst-case regret bound.}

    \medskip

    \medskip

    \noindent \textbf{Case (ii).}
    Similarly to case (i) above, we use case (ii) of Corollary~\ref{corollary:simple_hedge_conc_bound_bernstein} and introduce ${c}_{\Delta} := \Delta^2 / ({\color{green} 4} \overline{\widetilde{\rho}}_{\max}^2 ({\color{green} 32} \widetilde{v}_{\max} + \Delta))$ and ${\nu}^* := \max({\color{green} 2}, {\color{green} 32} \widetilde{v}_{\max} + \Delta)$. We set $t_0 := \left\lceil \frac{{\color{green} 4} \overline{\widetilde{\rho}}_{\max}^2 {\nu}^* \log K }{ \Delta^2 } \right\rceil$.
    For $t \geq t_0 + 1$, we thus obtain the bound
    \begin{align*}
        \sum_{k\neq k^\star} \E\left[ \omega^{(k)}_{\textsc{hdg},t} \right]
        & \leq \big( (K-1) e^{-t_0 {c}_{\Delta}} \big) e^{- (t-t_0-1) {c}_{\Delta}} + \big( (K-1) e^{- \Delta \sqrt{(t - 1) \log K/8}} \big) e^{- \Delta \sqrt{(t - 1) \log K/8}} \\
        & \leq e^{- (t-t_0-1) {c}_{\Delta}} + e^{- \Delta \sqrt{(t - 1) \log K/8}} .
    \end{align*}
    The first summand follows from $t_0 \geq {\color{green} 4} \overline{\widetilde{\rho}}_{\max}^2 {\nu}^* \log K / \Delta^2 \geq \log K / {c}_{\Delta} $, while the second summand is due to noticing that
    \begin{equation*}
        (K-1) \exp\left( - \Delta \sqrt{ \frac{{\color{green} 4} \overline{\widetilde{\rho}}_{\max}^2 {\nu}^* \log K }{ \Delta^2 } \cdot \frac{\log K}{8} } \right)
        \leq   
(K-1)\exp\left(-\overline{\widetilde{\rho}}_{\max}\sqrt{\nu^*/2}\,\log K\right)
        \leq 1 ,
    \end{equation*}
    where we use that $\overline{\widetilde{\rho}}_{\max} \geq 1$ by definition and $K\geq 3$.

    We now follow similar derivations are before, this time noting that ${\nu}^* \leq {\color{green} 33}$ leads to  
    $\overline{R}_{\textsc{hdg},t_0} 
        \leq \sqrt{({ {\color{green} 132} \overline{\widetilde{\rho}}_{\max}^2 \log K}/{\Delta^2} + 1)\log K}$.  Lemma~\ref{lemma:mixing_DT}(ii) applies because Assumption~\ref{ass:uniform summability roots}(ii) implies Assumption~\ref{ass:uniform summability}(ii). Hence, $\E[D_T]\leq\frac{4\overline{\widetilde{\theta}}_{\max}}{\Delta}\left(\log(K-1)+\frac32\right)$. Combining this inequality with \eqref{regret_bound} and the preceding bounds, plugging the value of ${c}_{\Delta}$ and collecting the terms  yields
    \begin{equation*}
        \E\left[\overline{R}_{\textsc{hdg},T}\right]\leq1+\sqrt{\log K}+\frac{2\sqrt{33}\,\overline{\widetilde{\rho}}_{\max}\log K+4\overline{\widetilde{\theta}}_{\max}\left(\log(K-1)+\frac32\right)}{\Delta}+\frac{4\overline{\widetilde{\rho}}_{\max}^2(32\widetilde{v}_{\max}+\Delta)+16}{\Delta^2}.
    \end{equation*}
    as required. 
    For $T\leq t_0$, the result follows directly from the worst-case regret bound $\overline{R}_{\textsc{hdg},T}\leq\sqrt{T\log K}$.

    \medskip

    \noindent \textbf{Case (iii).} 
    We now combine the proof strategy from previous parts with that of Proposition~\ref{proposition:ftl_bernstein_regret_bound}, case (iii). To apply the bound in \eqref{hedge_ass_bern_case_3}, Corollary~\ref{corollary:simple_hedge_conc_bound_bernstein}, for convenience we introduce $\alpha = 2 / \gamma_{\min}$.
    Let again $\nu^* := \max(1, \widetilde{v}_{\max} + {\color{green} \frac{1}{3}} \Delta)$, and define
    \begin{equation*}
        \widetilde{c}_\Delta := \frac{ \Delta^2 }{ 32 (\widetilde{v}_{\max} + {\color{green} \frac{1}{3}} \Delta) } ,
        \qquad
        \widetilde{c}^*_\Delta := \frac{ \Delta^2 }{ 32\, \nu^* } ,
    \end{equation*}
    as well as 
    \begin{equation*}
       \widetilde{t}^*_0
        =
        \left\lceil  
            \inf \left\{ t  \geq 
                \max\left\{ 3, {e^\alpha}, {t_0 } \right\}  
            \,\bigg\vert\, \frac{t}{(\log t)^\alpha} \geq \frac{ 2 \log(2K) }{ \widetilde{c}^*_\Delta } \right\}
        \right\rceil ,
    \end{equation*}
    which is well-defined thanks to Lemma~\ref{lemma:logtail_t0}.
    For $t \geq \widetilde{t}^*_0 + 1$, we obtain
    \begin{align*}
        \sum_{k\neq k^\star} \E\left[ \omega^{(k)}_{\textsc{hdg},t} \right]
        & \leq 2 (K-1) \exp\left( - \frac{(t-1) \widetilde{c}_\Delta}{(\log(t-1))^\alpha} \right) +  (K-1)e^{- \Delta \sqrt{(t-1) \log K / 2} } \\
        & = I  + II ,
    \end{align*}
    where we can study terms $I$ and $II$ separately. 
    With regards to the former, note that
    \begin{align*}
        I 
        & \leq \exp\left( \log(2 K) - \frac{(t-1) \widetilde{c}_\Delta}{(\log(t-1))^\alpha} \right) \\
        & \leq \exp\left( \frac{ (t-1) \widetilde{c}^*_\Delta}{2 (\log(t-1))^\alpha} - \frac{ (t-1) \widetilde{c}_\Delta}{(\log(t-1))^\alpha} \right)
            \leq \exp\left( - \frac{ (t-1) \widetilde{c}_\Delta}{2 (\log(t-1))^\alpha} \right) ,
    \end{align*}
    where we have used that $\widetilde{c}^*_\Delta \leq \widetilde{c}_\Delta$ by definition of $\nu^*$.
    Turning to $II$, much like in previous parts, we can split the exponential factor to have 
    \begin{equation*}
        II = \big( (K-1) e^{- \Delta \sqrt{(t - 1) \log K/8}} \big) e^{- \Delta \sqrt{(t - 1) \log K/8}} .
    \end{equation*}
    We now use that, since $\widetilde{t}^*_0 \geq 3$, it holds 
    \begin{equation*}
        \widetilde{t}^*_0
        \geq 
        (\log \widetilde{t}^*_0)^\alpha \log(2 K) \frac{64 \nu^* }{\Delta^2} 
        \geq
        \log(2 K) \frac{64 \nu^* }{\Delta^2} ,
    \end{equation*}
    and hence
    \begin{align*}
        (K-1) e^{- \Delta \sqrt{(t - 1) \log K/8}}
        & \leq 
        (K-1) \exp\left( - \Delta \sqrt{ \log(2 K) \frac{64 \nu^* }{\Delta^2} \cdot \frac{\log K}{8} } \right)
        \leq 1 
    \end{align*}
    when $t \geq \widetilde{t}^*_0 + 1$, since $\log(2 K) \geq \log K$.
    Putting the two bounds together, one finds, for $t$ sufficiently large, 
    \begin{equation*}%
        \sum_{k\neq k^\star} \E\left[ \omega^{(k)}_{\textsc{hdg},t} \right]
        \leq
        \exp\left( - \frac{ (t-1) \widetilde{c}_\Delta}{2 (\log(t-1))^\alpha} \right) + \exp\left( - \Delta \sqrt{(t - 1) \frac{\log K}{8} } \right) .
    \end{equation*}

    Once more, we use the result that $\overline{R}_{\textsc{hdg}, {\color{green} \widetilde{t}^*_0}} \leq \sqrt{\widetilde{t}^*_0 \log K}$ to control the warm-up period and Lemma~\ref{lemma:exp_series_limit}. We obtain an overall regret bound
    \begin{equation*}
        \sqrt{\widetilde{t}^*_0 \log K} + \frac{16}{\Delta^2} + \sum_{s=\widetilde{t}^*_0}^\infty \exp\left( - \frac{ s \widetilde{c}_\Delta}{2 (\log s )^\alpha} \right) .
    \end{equation*}
    {\color{green} Finally, Assumptions~\ref{assumption:losses_iid_or_mixing}(iv) and \ref{ass:uniform summability}(ii) allow us to apply Lemma~\ref{lemma:mixing_DT}(ii). Incorporating the resulting bound on $\E[D_T]$ into \eqref{regret_bound} completes the proof. 
    For $T\leq\widetilde{t}^*_0$, the result follows directly from the worst-case regret bound $\overline{R}_{\textsc{hdg},T}\leq\sqrt{T\log K}$.
    }

\end{proof}

\subsubsection{Combined Bounds for Hedge}

We can now obtain Theorem~\ref{theorem:hedge_regret_mixing} by incorporating the bounds previously derived.

\begin{proof}[Proof of Theorem~\ref{theorem:hedge_regret_mixing}]
Part (i) follows by combining Proposition~\ref{proposition:hedge_hoeffding_regret_bound}(i)--(ii) and Proposition~\ref{proposition:hedge_bernstein_regret_bound}(i), noting that Assumption~\ref{assumption:losses_iid_or_mixing}(iii) implies Assumption~\ref{assumption:losses_iid_or_mixing}(i). Similarly, part (ii) follows by combining Proposition~\ref{proposition:hedge_hoeffding_regret_bound}(iii) and Proposition~\ref{proposition:hedge_bernstein_regret_bound}(ii).
\end{proof}

\newpage

\section{Doubling-trick Hedge}
\label{appendix:Doubling-trick Hedge}
For $K\geq 2$, constant Hedge can be modified to work in setups where the time horizon $T$ is not known in advance, or when the learning process is intended to continue, by using the so-called \emph{doubling trick} \citep{Cesa-Bianchi2006,shalev_online}. 
The idea is to partition time into phases of exponentially increasing length, and to restart Hedge at the beginning of each phase with a learning rate tuned for that phase length.

{\color{green} Specifically, let phase $r$ cover rounds $t\in[2^{r-1},2^r-1]$, and set
\begin{equation*}
\eta_r=\sqrt{\frac{8\log K}{2^{r-1}}}.
\end{equation*}
At the beginning of each phase, the forecaster resets the cumulative losses to zero and applies Hedge with the constant learning rate $\eta_r$. By Assumption~\ref{assumption:losses_0-1} and \eqref{eq:hedge_regret_bound}, the regret incurred during phase $r$, including a possibly incomplete final phase, is bounded by $\sqrt{2^{r-1}\log K/2}$. Since the number of phases up to time $T$ is $R=\lfloor\log_2 T\rfloor+1$, summing the phasewise bounds yields
\begin{equation*}
\overline{R}_T\leq\sqrt{\frac{\log K}{2}}\sum_{r=1}^{R}2^{(r-1)/2}=\sqrt{\frac{\log K}{2}}\frac{2^{R/2}-1}{\sqrt{2}-1}\leq\frac{\sqrt{T\log K}}{\sqrt{2}-1}.
\end{equation*}
Therefore, the doubling-trick Hedge algorithm does not require advance knowledge of the time horizon and retains the worst-case regret rate $\mathcal{O}(\sqrt{T\log K})$.}

\newpage

\section{Direct Combination Methods}
\label{appendix:Direct Combination Methods}
We consider three direct forecast combination methods: the \textit{cross-sectional Median} (Median), \textit{Simple Averaging} (SA), and the \textit{Rolling Mean Squared Error} (RollMSE). These methods are straightforward to implement  and serve as benchmarks in our empirical experiments in Section~\ref{section:application}.

{\color{green}
\paragraph{Median.} A simple alternative to linear forecast averaging is to combine expert predictions using their cross-sectional median, defined as 
\begin{equation*}
    \widehat{Y}_{t+h}^{\text {Median }}:=\operatorname{median}\left\{\widehat{Y}_{t+h}^{(k)}: k \in[K]\right\}.
\end{equation*}
For even $K$, we define the median as the average of the two central ordered forecasts. Unlike methods that weight experts according to past forecasting performance, the median does not use historical losses and is robust to extreme individual predictions.}

\paragraph{Simple Averaging (SA).}
One can assign equal time-invariant weights to all $K$ experts
\begin{equation*}
    \omega^{(k)}_{\mathrm{SA},t}= \frac{1}{K} 
    \quad\text{for all }
    k \in [K], \enspace 
    t \in [T].
\end{equation*}
SA is straightforward and can reduce variance when many experts are accurate. However, it is inflexible, as it does not adapt to variability in time in relative performance across experts and implicitly relies on a subset of experts performing well at all times. Moreover, it also requires that no expert faces catastrophically high losses over the forecasting window, as this can increase the forecaster's losses unless additional tools are used (for example, trimming or winsorization).

\paragraph{Rolling MSE (RollMSE).} The rolling MSE scheme is specifically tailored to the setting of squared losses for experts. At decision time $t \in [T]$ and forecasting horizon $h\in [H]$, only losses from decisions
$\tau \le t-h$ have realized (their outcomes are revealed at $\tau+h \le t$). Fix a window length  $r\ge1$ and define the effective window length
\begin{equation*}
	r(t,h):=\min\{ r, (t-h)_+\},\quad (x)_+:=\max\{x,0\}.
\end{equation*}
If $r(t,h)=0$ ($t\leq h$), use an initialization rule (for example, equal weights).
Otherwise, for each expert $k\in[K]$ set
\begin{equation*}
	\mathrm{MSE}^{(k)}_{t,h,r} 
    =\frac{1}{r(t,h)}
    \sum_{\tau=t-h-r(t,h)+1}^{ t-h}
    \ell^{(k)}_{\tau,h}
    =\frac{1}{r(t,h)}
    \sum_{\tau=t-h-r(t,h)+1}^{ t-h}
    \left(Y_{\tau+h}-\widehat Y^{(k)}_{\tau+h}  \right)^2,
\end{equation*}
and construct the rolling MSE weights according to
\begin{equation*}
	\omega^{(k)}_{\mathrm{RollMSE},t,h}
    = 
    \frac{  \left({\mathrm{MSE}}^{(k)}_{t,h,r}+\varepsilon   \right)^{-1}} {\sum_{j=1}^K   \left({\mathrm{MSE}}^{(j)}_{t,h,r}+\varepsilon   \right)^{-1}},
\end{equation*}
with $\varepsilon>0$ a small constant.

This approach has proven effective in practice and does not require estimating forecast error correlations~\citep{timmermann2006forecast}.
The key practical issue is the selection of the window size $r$. Selecting a window size too large can produce weights that react too slowly to sudden changes in loss values across experts, while a small $r$ can overreact to recent performance. Moreover, the weighting of MSEs within the window is imposed to be uniform. For the extreme case $r = 1$ ($t > h$ and $r(t,h) = 1$), the weights depend only on the most recent realized $h$-ahead error, that is
\begin{align*}
	\omega^{(k)}_{\mathrm{RollMSE},t,h}
    & = \dfrac{\left(\ell^{(k)}_{t-h,h}+\varepsilon \right)^{-1}}
    {\sum_{j=1}^K \left(\ell^{(j)}_{t-h,h}+\varepsilon \right)^{-1}} = \dfrac{\left(\left(Y_{(t-h)+h}-\widehat Y^{(k)}_{(t-h)+h}\right)^2 +\varepsilon \right)^{-1}}
    {\sum_{j=1}^K \left(\left(Y_{(t-h)+h}-\widehat Y^{(j)}_{(t-h)+h}\right)^2 +\varepsilon \right)^{-1}}\\
    &= \dfrac{\left(\left(Y_t-\widehat Y^{(k)}_{t}\right)^2 +\varepsilon \right)^{-1}}
    {\sum_{j=1}^K \left(\left(Y_t-\widehat Y^{(j)}_{t}\right)^2 +\varepsilon \right)^{-1}}.
\end{align*}

\newpage

\color{red}

\section{Gap Dependence of FTL and Decreasing Hedge Bounds}
\label{appendix:gap_dependence_bounds}

The expected regret bounds for Follow-the-Leader and decreasing Hedge discussed in Section~\ref{section:expert_ensembles} exhibit different dependence on the sub-optimality gap $\Delta$. While our Hoeffding- and Bernstein-type bounds generally contain terms proportional to $1/\Delta^2$, sharper bounds of order $\mathcal{O}(\log K/\Delta)$ are available under the stronger assumption of jointly i.i.d.\ loss vectors, as established in Proposition~\ref{proposition:ftl_iid_sharp} for FTL and Proposition~\ref{proposition:hedge_hoeffding_regret_bound}(ii) for decreasing Hedge (see also \citealp{degennePerchet2016,mourtadaOptimalityHedgeAlgorithm2019}). Nevertheless, all these gap-dependent bounds become increasingly conservative as $\Delta\to0$. For decreasing Hedge, the gap-independent worst-case bound of order $\mathcal{O}(\sqrt{T\log K})$ provides an alternative guarantee for a fixed forecasting horizon $T$. Although FTL does not admit a comparable worst-case guarantee in general, under jointly i.i.d.\ loss vectors it also admits gap-independent expected-regret bounds of order $\mathcal{O}(\sqrt{T\log K})$, even when losses are dependent across experts (see \citealp{kotlowski2018}).
This can be relevant in practice, as looking, for example, at Figures~\ref{fig:plot_ecdf_msfe_medium_seed} and \ref{fig:plot_ecdf_msfe_medium_seed_leak}, the empirical CDFs of relative MSFE suggest that many  MFESN experts have similar observed forecasting performance, although this does not directly determine the population sub-optimality gap $\Delta$. Of course, using fewer experts than our empirical choice of $K=1000$ may result in more defined separation between individual models, but one would expect this comes at the cost of worse online aggregation performance. 
In the following, we examine this concern from two different points of view. 
First, we discuss a potential extension of our framework that allows one to consider ``grouped'' experts based on loss quantiles. 
Second, we discuss simulation and empirical evidence that suggests these concerns are not linked to large losses in performance from a practical perspective. 

\subsection{Expert Grouping by Quantile}

Let us consider a quantile-based extension of $\Delta$, where we replace the minimum gap between expected losses with a quantile of the gap distribution across experts. 
Formally, let us define $\Delta_q$ as the $q$-th quantile of $\left\{ \Delta(k) := \mu_k-\mu_{k^{\star}} \right\}_{k \neq k^{\star}}$, $q \in (0,1)$. By construction, at most $\lceil q(K-1) \rceil$ experts have respective gaps that are smaller than $\Delta_q$, and the regret contribution from these experts can be bounded separately.
For simplicity, let us work with i.i.d. loss vectors, as in Assumption~\ref{assumption:losses_iid_or_mixing}(iii), and again, without loss of generality, assume that $\mu^\star \equiv \mu_1 < \mu_2 \leq \ldots \leq \mu_K$. Introduce $k_q := 1+\lceil q(K-1) \rceil$, and define the $q$-quantile gap as $\Delta_q:=\mu_{k_q}-\mu_{1}$
that is, the empirical $q$-quantile of the gaps among the $K-1$ sub-optimal experts. By construction $\Delta_q \geq \Delta$ for all $q > 1/(K-1)$, with equality possible only at the presence of ties.
 
The expected instantaneous pseudo-regret can be decomposed by splitting the expert pool into two groups: those with indices $k<k_q$, containing $k_q -1$ experts, including the optimal expert, and those with indices $k\geq k_q$, containing $K-k_q+1$ experts. Under Assumption  E below, experts in the first group have gaps strictly smaller than $\Delta_q$, while those in the second group have gaps at least $\Delta_q$. 
We can write
\begin{equation*}
    \E[ \overline{\ell}_t - \ell_t^\star ] 
    = \sum_{k \neq k^\star} \mathbb{E}[\omega^{(k)}_t]\,\Delta_k 
    = \underbrace{\sum_{k =2}^{k_q-1}\mathbb{E}[\omega^{(k)}_t]\,\Delta_k}_{\displaystyle I} + \underbrace{\sum_{k \geq k_q} \mathbb{E}[\omega^{(k)}_t]\,\Delta_k}_{\displaystyle II}.
\end{equation*}

We make the following assumption on the separation of the two groups.

\medskip

\noindent
\textbf{Assumption E} (Group separation)\textbf{.} 
It holds $\mu_{k_q} > \mu_{k_q-1} \geq 0$.

\medskip

As Assumption E allows for $\mu_{k_q-1} \geq \ldots \geq \mu_1 \geq 0$, we consider the group of experts collected in term $I$ as the ``near-tied optimal'' experts. On the other hand, term $II$ covers the ``well-separated sub-optimal'' experts.
We can then apply very similar arguments to the one used, for example, in Appendix~\ref{appendix:proofs_concentration_ftl} for FTL weights, to show weight concentration away from high-expected-loss experts. Introducing 
\begin{equation*}
    \delta_q := \frac{\mu_{k_q} - \mu_{k_q-1}}{2},
\end{equation*}
similar derivations to those used in Lemma~\ref{lemma:ftl_concentr_lemma} yield, for all $k\geq k_q$ and $t\geq2$,
\begin{equation*}
    \E[\omega^{(k)}_{\textsc{ftl},t}]\leq \P\left( L_{t-1}^{(k)} \leq L_{t-1}^{(k_q-1)} \right)
    \leq \exp \left(-{(t-1)\left(\mu_k-\mu_{k_q-1}\right)^2}/{2}\right)\leq \exp\left(-2(t-1)\delta_q^2\right),
\end{equation*}
where the second inequality is obtained by applying Hoeffding's inequality in Lemma~\ref{lemma:iid_hoeffding} to i.i.d. loss differences $\ell_s^{(k_q-1)}-\ell_s^{(k)}\in[-1,1]$, similarly to the arguments used in the proof of Lemma~\ref{lemma:ftl_concentr_lemma}, and noting that $\mu_k-\mu_{k_q-1}\geq 2\delta_q$. Since $\Delta_k\leq 1$, we obtain
\begin{equation*}
    \sum_{k \geq k_q} \E[\omega^{(k)}_{\textsc{ftl},t}] \Delta_k
    \leq (K-k_q+1) \exp\left(-2(t-1)\delta_q^2\right),
\end{equation*}
and hence the group of sub-optimal experts, collectively, has weights that concentrate towards zero.
Analogous arguments for the near-tied optimal experts in term $I$ deliver that FTL weights concentrate within this group.

Thus, by extending the proof strategy employed in Proposition~\ref{proposition:ftl_hoeffding_regret_bound}, we expect that analogous bounds could be obtained for this grouped setting. The main difference is that an appropriate concept of \textit{group regret} would need to be defined, measuring the regret associated with selecting experts from the set  $\{ k_q, \ldots, K \}$ rather than from  $\{ 1, \ldots, k_q-1 \}$. Since this extension is beyond the scope of this work, we do not pursue this framework here, although such a generalization could be particularly relevant in settings where groups of experts have similar expected losses, as may occur  in ensembles of nested prediction models.

\subsection{Small Empirical Gaps and Performance}

From a practical point of view, in both Monte Carlo simulations and applications, we observe that, even when experts are not clearly separated, online combination methods perform well.

In both the medium-MD (Section~\ref{section:application}) and small-MD (Appendix~\ref{appendix:small_md_results}) applications, we observe that MSFE ECDFs are smooth in the lower tail, where the best-performing models lie. Nonetheless, in both settings, FTL and Hedge-type online combination schemes report meaningful, consistent improvements over both the benchmarks (SA, RollMSE, and Median) and the baseline models from \cite{ballarin2022reservoir}. These improvements are also robust to Monte Carlo resampling of the ensembles, as we show in Appendix~\ref{appendix:robustness_resampling}.
Therefore, even though our theoretical analysis from Section~\ref{section:expert_ensembles} follows the well-developed literature on online combination and is fundamentally tied to the sub-optimality gap, empirical evidence suggests that the schemes we study can perform well even when differences in observed forecasting performance across experts are small.

Furthermore, in Appendix~\ref{appendix:simulation_ar1_losses} we numerically evaluate our regret guarantees against empirical regret through direct simulation of the loss process; this allows us to compute the bounds derived in Section~\ref{section:expert_ensembles} explicitly. 
Our results suggest that, while theoretical regret bounds capture the pattern of increasing regret as $\Delta$ shrinks, they are also very conservative with respect to the impact of both sub-optimality gap and time dependence on performance. 
While Figures~\ref{fig:ar1-c035-gap}-\ref{fig:ar1-c035-q} do indicate that both FTL and decreasing Hedge are negatively affected when $\Delta \to  0$, the effect is much milder than predicted by Theorems~\ref{theorem:ftl_regret_mixing} and \ref{theorem:hedge_regret_mixing}. Moreover, in the AR(1) simulation designs considered, increasing loss persistence has little effect on the empirical regret of either online weighting method. These findings underscore the effectiveness of Follow-the-Leader and Hedge-type forecast combination, but also point to an important direction for future research---namely, developing substantially sharper bounds within the established regret framework we leverage \citep{Cesa-Bianchi2006,rooijFollowLeaderIf2014,mourtadaOptimalityHedgeAlgorithm2019}, or devising novel approaches to better explain the performance of these methods in stochastic settings where small sub-optimality gaps render  the corresponding theoretical bounds overly conservative.

\color{black}

\newpage

\section{Algorithms}
\label{appendix:algorithms}

In this section, we provide pseudocode for the algorithm implementation of the ensemble combination schemes discussed in Section~\ref{section:expert_ensembles} and then used in Section~\ref{section:application}.
{\color{green} 
Since our empirical application concerns one-step-ahead forecasting ($h=1$), we omit the forecasting-horizon subscript and retain the time-indexing convention of
Section~\ref{section:preliminaries}. Thus, at round $t$, the experts predict $Y_{t+1}$ using weights $\bm{\omega}_t$, computed from cumulative losses
$\bm{L}_{t-1}$. After observing $Y_{t+1}$, the cumulative losses are updated to $\bm{L}_t$.
}
An explicit subscript for each combination method is also added (see the notation of Section~\ref{section:expert_ensembles}).

\bigskip

\begin{algorithm}
\setstretch{1.25}
\caption{-- Follow-the-Leader (FTL)}\label{algortihm:ftl}
    \begin{algorithmic}
    \Require Ensemble of experts indexed by $k \in [K]$
    \State Initialize cumulative losses as $L^{(k)}_{{\color{green} 0}} = 0$ for all $k \in [K]$ 
    \State Initialize weights as $ \omega_{\textsc{ftl},{{\color{green} 1}}}^{(k)} = 1/K $ for all $k \in [K]$ \Comment{Equal weighting init}
    \For{$t = 1, 2, \ldots, T$}
        \LComment{Prediction step}
        \State Experts make predictions $\{ \widehat{Y}^{(k)}_{t+1} \}_{k=1}^K$
        \State Ensemble prediction is 
        $ \widehat{Y}_{t+1} := \sum_{k=1}^K \omega_{\textsc{ftl},{{\color{green} t}}}^{(k)} \widehat{Y}_{t+1}^{(k)} $
        \State Outcome $Y_{t+1}$ is realized %
        \LComment{Weights update step}
        \State Each  $k\in [K]$ updates its cumulative loss {{\color{green} $L^{(k)}_{t} = L^{(k)}_{t-1} + \ell(Y_{t+1}, \widehat{Y}_{t+1}^{(k)})$ }}
        \State Compute set of experts with minimal cumulative loss, $\textsc{ftl}_{t+1} = \arg\min_{k \in [K]} L^{(k)}_{{{\color{green} t}}}$
        \State Set $\omega_{\textsc{ftl},t+{{\color{green} 1}}}^{(k)} = \1{ k \in \textsc{ftl}_{t+1} } / \abs{\textsc{ftl}_{{{\color{green} t+1}}}}$  for all $k \in [K]$
    \EndFor
    \State \Output Ensemble predictions $\{ \widehat{Y}_{t+1} \}_{t=1}^T$, ensemble weights $\{\bm{\omega}_{\textsc{ftl},{{{\color{green} t}}}}\}_{t=1}^T$
    \end{algorithmic}
\end{algorithm}

\newpage

\begin{algorithm}[H]
\setstretch{1.25}
\caption{-- Constant Hedge}\label{algortihm:const_hedge}
    \begin{algorithmic}
    \Require An ensemble of experts indexed by $k \in [K]$, constant learning rate $\eta > 0$
    \State Initialize weights as $ \omega_{\textsc{hdg},{{\color{green} 1}}}^{(k)} = 1/K $ for all $k \in [K]$ \Comment{Equal weighting init}
    \For{$t = 1, 2, \ldots, T$}
        \LComment{Prediction step}
        \State Experts make predictions $\{ \widehat{Y}^{(k)}_{t+1} \}_{k=1}^K$
        \State Ensemble prediction is 
        $ \widehat{Y}_{t+1} := \sum_{k=1}^K \omega_{\textsc{hdg},{{\color{green} t}}}^{(k)} \widehat{Y}_{t+1}^{(k)} $
        \State Outcome $Y_{t+1}$ is realized 
        \LComment{Weights update step}
        \State Each $k$th expert, $k\in [K]$, faces instantaneous loss $\ell^{(k)}_{{{\color{green} t}}} = \ell(Y_{t+1}, \widehat{Y}_{t+1}^{(k)})$
        \State Compute $v^{(k)}_{t} := \omega^{(k)}_{\textsc{hdg},{{\color{green} t}}} \exp( -\eta \ell^{(k)}_{{{\color{green} t}}} )$   for all $k \in [K]$ \Comment{Pre-normalization weights}
        \State Set $\omega_{\textsc{hdg},{{\color{green} t+1}}}^{(k)} = v^{(k)}_{{{\color{green} t}}} / \sum_{j=1}^K v^{(j)}_{{{\color{green} t}}} $   for all $k \in [K]$
    \EndFor
    \State \Output Ensemble predictions $\{ \widehat{Y}_{t+1} \}_{t=1}^T$, ensemble weights $\{\bm{\omega}_{\textsc{hdg},{{\color{green} t}}}\}_{t=1}^T$
    \end{algorithmic}
    \vspace{5pt}
    \setstretch{1.}
    \textbf{Note:} Online weight updates using instantaneous losses $\ell^{(k)}_{{{\color{green} t}}}$ are more numerically stable compared to calculating Hedge weights based on cumulative losses.
\end{algorithm}

\newpage

\begin{algorithm}[H]
\setstretch{1.25}
\caption{-- Decreasing Hedge}\label{algortihm:decreasing_hedge}
    \begin{algorithmic}
    \Require An ensemble of experts indexed by $k \in [K]$, learning rate scale $c_0 > 0$ \\(for example $c_0 = 2$, based on worst-case bounds in  \citealp{mourtadaOptimalityHedgeAlgorithm2019})
    \State {\color{green} Initialize cumulative losses as $L^{(k)}_{{\color{green} 0}} = 0$ for all $k \in [K]$ }
    \State Initialize weights as $ \omega_{\textsc{dh},{{\color{green} 1}}}^{(k)} = 1/K $ for all $k \in [K]$ \Comment{Equal weighting init}
    \For{$t = 1, 2, \ldots, T$}
        \LComment{Prediction step}
        \State Experts make predictions $\{ \widehat{Y}^{(k)}_{t+1} \}_{k=1}^K$
        \State Ensemble prediction is 
        $ \widehat{Y}_{t+1} := \sum_{k=1}^K \omega_{\textsc{dh},{{\color{green} t}}}^{(k)} \widehat{Y}_{t+1}^{(k)} $
        \State Outcome $Y_{t+1}$ is realized 
        \LComment{Weights update step}
        \State Each expert $k\in [K]$ faces instantaneous loss $\ell^{(k)}_{{{\color{green} t}}} = \ell(Y_{t+1}, \widehat{Y}_{t+1}^{(k)})$
        \State Each expert $k\in [K]$ updates its cumulative loss $L^{(k)}_{{{\color{green} t}}} = L^{(k)}_{{{\color{green} t-1}}} + \ell(Y_{t+1}, \widehat{Y}_{t+1}^{(k)})$
        \State Find ${\color{green} m_t} = \min_{k \in [K]} L^{(k)}_{{{\color{green} t}}}$
        \State Set $\eta_{t+1} = c_0 \sqrt{\log K / (t+1)}$ \Comment{Decreasing learning rate}
        \State Compute $v^{(k)}_{{{\color{green} t}}} := \exp( -\eta_{{{\color{green} t+1}}} (L^{(k)}_{t} - {\color{green} m_t}) )$   for all $k \in [K]$ \Comment{Pre-normalization weights}
        \State Set $\omega_{\textsc{dh},t+1}^{(k)} = v^{(k)}_{{{\color{green} t}}} / \sum_{j=1}^K v^{(j)}_{{{\color{green} t}}} $    for all $k \in [K]$
    \EndFor
    \State \Output Ensemble predictions $\{ \widehat{Y}_{t+1} \}_{t=1}^T$, ensemble weights $\{\bm{\omega}_{\textsc{dh},{{\color{green} t}}}\}_{t=1}^T$
    \end{algorithmic}
    \vspace{5pt}
    \setstretch{1.}
    \textbf{Note:} Hedge weights are calculated by centering the cumulative losses with  ${\color{green} m_t=\min_{k\in[K]}L_t^{(k)}}$ to avoid numerical issues.
\end{algorithm}

\newpage

\begin{algorithm}[H]
\setstretch{1.25}
\caption{-- AdaHedge}\label{algorithm:adahedge}
    \begin{algorithmic}
    \Require An ensemble of experts indexed by $k \in [K]$
    \State Initialize cumulative losses as $L^{(k)}_0 = 0$ for all $k \in [K]$ 
    \State Initialize cumulative mixability gap $\nabla_0 = 0$ and  cumulative mix loss $\overline{M}_0 = 0$
    \For{$t = 1, 2, \ldots, T$}
        \LComment{Weight computation}
        \State {\color{green} Set $m_{t-1} := \min_{k\in[K]} L_{t-1}^{(k)}$}
        \If{$\nabla_{t-1}=0$}
            \State {\color{green} Set $\eta_t=\infty$ and
            $v_t^{(k)}=\1{L_{t-1}^{(k)}=m_{t-1}}$
            for all $k\in[K]$}
        \Else
            \State {\color{green} Set $\eta_t=\log K/\nabla_{t-1}$ and
            $v_t^{(k)}=
            \exp\left(-\eta_t(L_{t-1}^{(k)}-m_{t-1})\right)$
            for all $k\in[K]$}
        \EndIf
        \State Compute $\overline{v}_t=\sum_{j=1}^K v_t^{(j)}$
        \State Set $\omega_{\textsc{ah},t}^{(k)}
        = v_t^{(k)}/\overline{v}_t$
        for all $k\in[K]$ \Comment{Current-round weights}

        \LComment{Prediction step}
        \State Experts make predictions $\{ \widehat{Y}^{(k)}_{t+1} \}_{k=1}^K$
        \State Ensemble prediction is 
        $ \widehat{Y}_{t+1} := \sum_{k=1}^K \omega_{\textsc{ah},{\color{green} t}}^{(k)} \widehat{Y}_{t+1}^{(k)} $
        \State Outcome $Y_{t+1}$ is realized 
        \LComment{Mixability gap update step}
        \State Each expert $k\in [K]$ faces instantaneous loss $\ell^{(k)}_{{\color{green} t}} = \ell(Y_{t+1}, \widehat{Y}_{t+1}^{(k)})$
        \State Each expert $k\in [K]$ updates its cumulative loss $L^{(k)}_{{\color{green} t}} = L^{(k)}_{t-1} + \ell^{(k)}_{t}$
        \State {\color{green} Compute $m_t := \min_{k\in[K]} L_t^{(k)}$}
        
        \State Compute forecaster's loss
        $\overline{\ell}_t =
        \bm{\omega}_{\textsc{ah},t}^{\top}\bm{\ell}_t$

        \LComment{Mix-loss computation}
        \If{$\eta_t=\infty$}
            \State Set
            $\overline{m}_t=m_t-m_{t-1}$
        \Else
            \State Compute
            $\overline{v}_t^{\mathrm{new}}
            =\sum_{k=1}^K
            \exp\left(-\eta_t(L_t^{(k)}-m_t)\right)$

            \State Set
            $\overline{m}_t
            =m_t-m_{t-1}
            -\eta_t^{-1}
            \log\left(
                {\overline{v}_t^{\mathrm{new}}}/{\overline{v}_t}
            \right)$
            \Comment{Numerically stable mix loss}
        \EndIf

        \State Update cumulative mix loss
        $\overline{M}_t
        =\overline{M}_{t-1}+\overline{m}_t$

        \State Update cumulative mixability gap
        $\nabla_t=\nabla_{t-1}
        +\max\{0,\overline{\ell}_t-\overline{m}_t\}$
    \EndFor

    \State \Output Ensemble predictions
    $\{\widehat{Y}_{t+1}\}_{t=1}^T$,
    ensemble weights
    $\{\bm{\omega}_{\textsc{ah},t}\}_{t=1}^T$
    \end{algorithmic}
    \vspace{5pt}
    \setstretch{1.}

\end{algorithm}

\clearpage
\newpage

\section{Summary of Data}\label{appendix:data}

The dataset is collected from various sources, including the Federal Reserve Bank of St.~Louis Monthly (FRED-MD) and Quarterly (FRED-QD) Macroeconomic Databases (see \citealp{McCracken2016, McCracken2021} for more details), as well as daily series sourced from Refinitiv Datastream and Wharton WRDS. 
The selection of predictors in the medium-MD dataset includes indicators from ten macroeconomic and financial categories (see \citealp{Stock1996, Stock2002}). 
{\color{brown}
For the small-MD dataset, a smaller subset of 9 predictors is used. 
}%
For both datasets, the preprocessing of macroeconomic indicators follows the standard guidelines outlined in \cite{McCracken2016, McCracken2021}. 
The estimation and evaluation sample spans a period from January 1st, 1990, to December 31st, 2019.
{\color{brown}
We use the same releases of FRED-MD and QD originally employed in \cite{ballarin2022reservoir} (June~2021).
}%
The initial estimation sample covers the period from 1990Q1 to 2007Q4, preceding the Great Financial Crisis. {\color{brown}For cross-validation only, we extend the sample
to 1975Q1-2007Q4.} The testing and online learning sample spans 2008Q1-2019Q4, encompassing a total of 48 forecasting rounds.
{\color{brown}
Table~\ref{table:dataset} contains a complete list of the data used in our empirical experiments, grouped by frequency, as well as the pre-processing transformation applied to each series.
}%

\begin{table}[ht!]
\begingroup
\centering
\footnotesize
\begin{tabularx}{\linewidth}{lcllX}
    \textbf{Start Date} & \textbf{T} & \textbf{Code} & \textbf{Name} & \textbf{Description} \\
    \midrule
    \multicolumn{5}{l}{\footnotesize%
        Quarterly~ 
        \raisedrule[0.2em]{0.2pt}} \\[2pt]
    31/03/1959 & 5 & GDPC1 & Y & Real Gross Domestic Produce\\[5pt]
    \multicolumn{5}{l}{\footnotesize%
        Monthly~ 
        \raisedrule[0.2em]{0.2pt}} \\[2pt]
    30/01/1959 & 5 & INDPRO & XM1 & Industrial Production Index\\
    30/01/1959 & 5 & PAYEMS & XM4 & Payroll All Employees: Total nonfarm\\
    30/01/1959 & 4 & HOUST & XM5 & Housing Starts: Total New Privately Owned\\
    30/01/1959 & 5 & RETAILx & XM7 & Retail and Food Services Sales\\
    31/01/1973 & 5 & TWEXMMTH & XM11 & Nominal effective exchange rate US\\
    30/01/1959 & 2 & FEDFUNDS & XM12 & Effective Federal Funds Rate\\
    30/01/1959 & 1 & BAAFFM & XM14 & Moody’s Baa Corporate Bond Minus FEDFUNDS\\
    30/01/1959 & 1 & COMPAPFFx & XM15 & 3-Month Commercial Paper Minus FEDFUNDS\\
    30/01/1959 & 2 & CUMFNS & XM2 & Capacity Utilization: Manufacturing\\
    30/01/1959 & 2 & UNRATE & XM3 & Civilian Unemployment Rate\\
    30/01/1959 & 5 & DPCERA3M086SBEA & XM6 & Real personal consumption expenditures\\
    30/01/1959 & 5 & AMDMNOx & XM8 & New Orders for Durable Goods\\
    31/01/1978 & 2 & UMCSENTx & XM9 & Consumer Sentiment Index\\
    30/01/1959 & 6 & WPSFD49207 & XM10 & PPI: Finished Goods\\
    30/01/1959 & 1 & AAAFFM & XM13 & Moody’s Aaa Corporate Bond Minus FEDFUNDS\\
    30/01/1959 & 1 & TB3SMFFM & XM16 & 3-Month Treasury C Minus FEDFUNDS\\
    30/01/1959 & 1 & T10YFFM & XM17 & 10-Year Treasury C Minus FEDFUNDS\\
    30/01/1959 & 2 & GS1 & XM18 & 1-Year Treasury Rate\\
    30/01/1959 & 2 & GS10 & XM19 & 10-Year Treasury Rate\\
    30/01/1959 & 1 & GS10-TB3MS & XM20 & 10-Year Treasury Rate - 3-Month Treasury Bill\\[5pt]
    \multicolumn{5}{l}{\footnotesize%
        Daily~ 
        \raisedrule[0.2em]{0.2pt}} \\[2pt]
    30/01/1959 & 8 & DJINDUS & XD3 & DJ Industrial price index\\
    31/12/1963 & 8 & S\&PCOMP & XD1 & S\&P500 price index\\
    01/05/1982 & 1 & ISPCS00-S\&PCOMP* & XD2 & S\&P500 basis spread\\
    11/09/1989 & 8 & SP5EIND & XD4 & S\&P Industrial price index\\
    31/12/1969 & 8 & GSCITOT & XD5 & Spot commodity price index\\
    10/01/1983 & 8 & CRUDOIL & XD6 & Spot price oil\\
    02/01/1979 & 8 & GOLDHAR & XD7 & Spot price gold\\
    30/03/1982 & 8 & WHEATSF & XD8 & Spot price wheat\\
    01/11/1983 & 8 & COCOAIC,COCINU** & XD9 & Spot price cocoa\\
    30/03/1983 & 1 & NCLC.03-NCLC.01 & XD10 & Futures price oil term structure\\
    30/10/1978 & 1 & NGCC.03-NGCC.01 & XD11 & Futures price gold term structure\\
    02/01/1975 & 1 & CWFC.03-CWFC.01 & XD12 & Futures price wheat term structure\\
    02/01/1973 & 1 & NCCC.03-NCCC.01 & XD13 & Futures price cocoa term structure\\
    \bottomrule
\end{tabularx}
\endgroup

\vspace{-0.5em}
\footnotesize
\singlespacing
    Notes: 
    ``Start Date'' is the date for which the series is first available (before data transformations). Following \cite{McCracken2016} and \cite{McCracken2021}, the transformation codes in column ``T'' indicate with D for difference and log for natural logarithm, 1: none, 2: D, 3: DD, 4: Log, 5: Dlog, 6: DDlog, 7: percentage change, 8: GARCH volatility. ``Code'' is the code in the FRED-QD and FRED-MD datasets for quarterly and monthly data, and the Datastream mnemonic for the remaining frequencies. Missing values due to public holidays are interpolated by averaging over the previous five observations. *: Available until 20/09/2021. **: Average before 29/12/2017, COCINUS mean adjusted thereafter. 
\vspace{1em}
\caption{
Input and output variables, frequencies, and transformations (adapted from \cite{ballarin2022reservoir}). 
}
\label{table:dataset}
\end{table}

\clearpage
\newpage

\section{Summary of Models}
\label{appendix:models_summary}
Tables~\ref{tab:model_list}, \ref{tab:ensemble_class_list}, and \ref{tab:gdp-concise-design} summarize the forecasting models, ensemble constructions, and implementation parameters of the combination methods used in our empirical analysis, respectively.

\begin{table}[H]
\renewcommand{\baselinestretch}{1.2}
\centering
\small %
\begin{tabularx}{\textwidth}{p{2.7cm} p{7cm} p{3.8cm}} %
    \textbf{Model Name} & \textbf{Description} & \textbf{Specification} \\ 
    \toprule
    Mean &  \begin{tabular}[t]{@{}l@{}}Unconditional mean of outcome series\\ over the estimation sample.\end{tabular} & -- \\
    \Xhline{0.1pt}
    AR(1) &  \begin{tabular}[t]{@{}l@{}}Autoregressive model of the output \\series estimated using OLS.\end{tabular} & -- \\
    \Xhline{0.1pt}
    DFM A & \begin{tabular}[t]{@{}l@{}}Stock aggregation, \\ VAR(1) factor process.\end{tabular} & \begin{tabular}[t]{@{}l@{}}Factors: 10\end{tabular} \\
    \Xhline{0.1pt}
    DFM B & \begin{tabular}[t]{@{}l@{}}Almon aggregation, \\VAR(1) factor process.\end{tabular} & Same as DFM A \\
    \Xhline{0.1pt}
    S-MFESN~A & \begin{tabular}[t]{@{}l@{}}S-MFESN model:\\ Sparse-normal $\widetilde{A}$,\\ sparse-uniform $\widetilde{C}$, $\widetilde{{\zeta}}=0$.\\ Isotropic ridge regression fit.\end{tabular} & \begin{tabular}[t]{@{}l@{}}Reservoir dimension: 30 \\ Sparsity: 33.3\% \\ $\rho = 0.5$, $\gamma = 1$, $\alpha = 0.1$\end{tabular} \\
    \Xhline{0.1pt}
    S-MFESN~B & \begin{tabular}[t]{@{}l@{}}Same as S-MFESN~A except for\\ larger reservoir dimension and lower\\ sparsity ratio.\end{tabular} & \begin{tabular}[t]{@{}l@{}}Reservoir dimension: 120 \\ Sparsity: 8.3\% \\ $\rho = 0.5$, $\gamma = 1$, $\alpha = 0.1$\end{tabular} \\
    \Xhline{0.1pt}
    M-MFESN~A & \begin{tabular}[t]{@{}l@{}}M-MFESN model:\\ Monthly and daily freq. reservoirs. \\ Sparse-normal $\widetilde{A}_{1}$, $\widetilde{A}_{2}$,\\ sparse-uniform $\widetilde{C}_{1}$, $\widetilde{C}_{2}$,  $\widetilde{{\zeta}}_{1}=\widetilde{{\zeta}}_{2}=0$.\\ Isotropic ridge regression fit.\end{tabular} & \begin{tabular}[t]{@{}l@{}}Reservoir dimension: \\Month$=100$, Day$=20$ \\ Sparsity: Month$=10\%$, Day$=50\%$ \\ Month: $\rho = 0.5$, $\gamma = 1.5$, $\alpha = 0$\\ Day: $\rho = 0.5$, $\gamma = 0.5$, $\alpha = 0.1$ \end{tabular} \\
    \Xhline{0.1pt}
    M-MFESN~B &  \begin{tabular}[t]{@{}l@{}}Same as M-MFESN~A with different \\ values of hyperparameters.\end{tabular} & \begin{tabular}[t]{@{}l@{}}Reservoir dimension:\\ Month$=100$, Day$=20$ \\ Sparsity: Month$=10\%$, Day$=50\%$ \\ Month: $\rho = 0.08$, $\gamma = 0.25$, $\alpha = 0.3$\\ Day: $\rho = 0.01$, $\gamma = 0.01$, $\alpha = 0.99$ \end{tabular} \\
    \bottomrule
\end{tabularx}
\caption{%
    Models and ensembles applied in empirical forecasting exercises (for non-ensemble models see Table~4.1 in \cite{ballarin2022reservoir}). MFESN hyperparameters are defined for normalized state equations \eqref{eq:esn_hyper_normalized}-\eqref{eq:esn_state_new}.
}
\label{tab:model_list}
\end{table}

\clearpage

{

\begin{table}[p!]
\color{green}
\renewcommand{\baselinestretch}{1.2}
\centering
\small %
\begin{tabularx}{\textwidth}{p{3.2cm} p{6.6cm} p{3.4cm}} %
    \textbf{Ensemble Name} & \textbf{Description} & \textbf{Specification} \\ 
    \toprule
    
    EN-MFESN-RP & \begin{tabular}[t]{@{}l@{}} Applied to all types of MFESN:\\ 1000 distinct models generated with\\ independently randomly drawn\\ reservoir state coefficients.\end{tabular} & \begin{tabular}[t]{@{}l@{}}Same specification as the corres- \\ ponding baseline MFESN model.\end{tabular} \\
    \Xhline{0.1pt}
    EN-MFESN-$\alpha$RP & \begin{tabular}[t]{@{}l@{}} Applied to all types of MFESN:\\ 1000 distinct models with 5 different\\ leak rates, 200 draws of reservoir state\\ coefficients per $\alpha$ value.\end{tabular} & \begin{tabular}[t]{@{}l@{}}Same specification as the corres- \\ ponding baseline MFESN model, \\ except for leak rate: \\ $\alpha \in \big\{0.1, 0.3, 0.5, 0.7, 0.9 \big\}$. \end{tabular} \\
    \Xhline{0.1pt}
    EN-MFESN-$\lambda$RP & \begin{tabular}[t]{@{}l@{}} Applied to all types of MFESN:\\ 1000 distinct models with 5 different\\ ridge regularization factors, 200 draws\\ of reservoir state coefficients per $\lambda$\\ value. No cross-validation used.\end{tabular} & \begin{tabular}[t]{@{}l@{}}Same specification as the corres- \\ ponding baseline MFESN model, \\ except for ridge penalty strength: \\ $\lambda \in  \big\{ 0.01, 0.1, 1, 10, 100\}$. \end{tabular} \\
    \Xhline{0.1pt}
    EN-MFESN-$\alpha\lambda$RP & \begin{tabular}[t]{@{}l@{}} Applied to all types of MFESN:\\ 1000 distinct models with 5 different\\ leak rates and 5 different ridge regula-\\ rization factor, 40 draws of reservoir\\ state coefficients per $(\alpha, \lambda)$ value pair.\\ No cross-validation used. \end{tabular} & \begin{tabular}[t]{@{}l@{}}Same specification as the corres- \\ ponding baseline MFESN model, \\ except for leak rate and penalty: \\ $\alpha \in \big\{0.1, 0.3, 0.5, 0.7, 0.9 \big\}$, \\ $\lambda \in  \big\{ 0.01, 0.1, 1, 10, 100\}$. \end{tabular} \\
    
    \bottomrule
\end{tabularx}
\caption{%
    Ensembles applied in empirical forecasting exercises. 
}
\label{tab:ensemble_class_list}
\end{table}

}

\clearpage

{

\providecommand{\raisedrule}[2][0em]{%
    \leavevmode\leaders\hbox{\rule[#1]{1pt}{#2}}\hfill\kern0pt
}

\begin{table}[p!]
\color{brown}
\centering
\small
\setlength{\tabcolsep}{3pt}
\setlength{\extrarowheight}{1pt}
\renewcommand{\arraystretch}{1}
\linespread{1.2}\selectfont

\begin{tabularx}{\linewidth}{
    @{}p{.24\linewidth}
    p{.12\linewidth}
    >{\raggedright\arraybackslash}X@{}
}
    \textbf{Item} & \textbf{Symbol} & \textbf{Specification}\\
    \toprule

    \multicolumn{3}{@{}l@{}}{
        \footnotesize\textit{Fixed variables}
        ~\raisedrule[0.2em]{0.2pt}
    }\\[3pt]
    
    Experts per pool & $K$ & $1000$\\
    
    Test sample size & $T_{\mathrm{test}}$ &
    48 quarters, 2008Q1--2019Q4.\\

    \multicolumn{3}{@{}l@{}}{
        \footnotesize\textit{Forecast combination}
        ~\raisedrule[0.2em]{0.2pt}
    }\\[1pt]
    
    Median \& SA & -- &
    No tuning parameters needed.\\
    
    RollMSE & $r,\varepsilon$ &
    $r=4$, $\varepsilon=10^{-9}$;
    initialized with equal weights until four losses are available; see Appendix~\ref{appendix:Direct Combination Methods}.\\
    
    FTL & -- &
    No tuning parameters needed, see Algorithm~\ref{algortihm:ftl}. Numerical ties are determined using absolute and relative tolerances of $10^{-9}$ and $10^{-5}$, respectively. Equal weights among experts tied for the smallest cumulative loss.\\
    
    Hedge & $\eta$ &
    $10^4$; see Algorithm~\ref{algortihm:const_hedge}.\\
    
    DecHedge & $c_0$ &
    $10^4$ with
    {\color{green} $\eta_{t+1}=c_0\sqrt{\log K/(t+1)}$}; see Algorithm~\ref{algortihm:decreasing_hedge}.\\
    
    AdaHedge & $\eta_t$ &
    Adaptive rate set via mixability gap; see 
    Algorithm~\ref{algorithm:adahedge}.\\
    
    \bottomrule
\end{tabularx}

\caption{Empirical design and forecast combination parameters.
Model hyperparameters, ridge penalties, and fixed combination parameters
are not re-tuned at each forecast origin. 
}
\label{tab:gdp-concise-design}
\end{table}

}

\clearpage

\subsection{Echo State Networks in Brief}
\label{appendix:esn_in_brief}

Echo State Networks (ESNs) are a particular family of recurrent neural networks with randomly sampled connectivity weights. 
ESN models are, in general, nonlinear state-space systems that, in the forecasting setting, are defined by the following equations:
\begin{align}
  {X}_{t} &= \alpha X_{t-1} + (1-\alpha)\sigma(A X_{t-1}+ C Z_{t} + {\zeta}), \label{eq:esn_state} \\ 
    Y_{t+1} &= b + W^{\top} X_{t} + \epsilon_{t+1}, \nonumber %
\end{align}
where $A \in \Real^{D \times D}$ is the reservoir matrix, $C \in \Real^{D \times d}$ is the input matrix, ${\zeta} \in \Real^D$ is the input shift, $\alpha \in [0, 1)$ is the leak rate, and ${\theta} := (b, W\transp)\transp \in \Real^{D+1}$ with $W \in \mathbb{R}^D$ denotes the readout coefficients. The map $\sigma : \Real \rightarrow \Real$ is an activation function applied elementwise and we assume that $\sigma$ is the hyperbolic tangent in the rest of the paper.
We refer to $A$, $C$, and ${\zeta}$ as state parameters. These parameters are \textit{randomly drawn} and \textit{kept fixed} throughout, while $\theta$ is the only set of parameters that need to be estimated.

The randomly drawn parameters $A$, $C$, and ${\zeta}$ in \eqref{eq:esn_state} are constructed by first sampling $\widetilde{A}$, $\widetilde{C}$, and $\widetilde{{\zeta}}$ from appropriately chosen distributions, and then normalizing as follows:
\begin{align}\label{eq:esn_hyper_normalized}
    \overline{A} = {\widetilde{A}}/{\rho(\widetilde{A})}, \quad \overline{C} = {\widetilde{C}}/{\|\widetilde{C}\|_2}, \quad \overline{{\zeta}} = {\widetilde{{\zeta}}}/{\|\widetilde{{\zeta}}\|_2},
\end{align}
where $\rho(\widetilde{A})$ denotes the spectral radius of $\widetilde{A}$. We then define $A = \rho \overline{A}$, $C = \gamma \overline{C}$, and ${\zeta} = \varsigma \overline{{\zeta}}$, which allows to write the state equation \eqref{eq:esn_state} as 
\begin{equation}\label{eq:esn_state_new}
    X_{t} = \alpha X_{t-1} + (1-\alpha)\sigma( \rho \overline{A}X_{t-1} + \gamma \overline{C}Z_{t} + \varsigma \overline{{\zeta}} ) .
\end{equation}
The tuple ${\varphi}:= (\alpha, \rho, \gamma, \varsigma)$ is referred to as the hyperparameters of the ESN. Specifically, $\alpha \in [0, 1)$ is the so-called leak rate, $\rho \in \mathbb{R}^+$ denotes the spectral radius of the reservoir matrix $A$, $\gamma \in \mathbb{R}^+$ is the input scaling, and $\varsigma \in \mathbb{R}^+$ corresponds to the shift scaling. These hyperparameters play a crucial role in defining the properties of the state map in \eqref{eq:esn_state_new}. Although there is no clear consensus on how to best select $\varphi$, in this work we construct ensembles of ESN models that not only differ in the draws of the state parameters $A$, $C$, and $\zeta$, but also in the state hyperparameters. This allows us to explore a wide range of different reservoir architectures, and thus sidestep the problem of model tuning.\footnote{\cite{ballarin2022reservoir} propose a general data-driven algorithm to tune $\varphi$. However, the resulting hyperparameter optimization problem is computationally complex and highly non-convex.}

As is common in the ESN literature, we consider the ridge regression estimator for $\theta$. More precisely, let
$X := (X_1, \ldots, X_{T-1})^\top \in \mathbb{R}^{(T-1)\times D}$ 
and $Y := (Y_2, \ldots, Y_T)^\top \in \mathbb{R}^{T-1}$.
Define
$
M= \big(\mathbb{I}_{T-1} - \frac{1}{T-1} \mathbf{1}_{T-1} \mathbf{1}_{T-1}^\top\big)$, $ \widetilde{Y} = M Y$, and $
\widetilde{X} = M X. 
$
Then the ridge regression estimator for $\theta$ is given by
\begin{align}
    \widehat{W}_\lambda &= \arg\min_{W \in \mathbb{R}^D} \{\left \| \widetilde{Y}  -\widetilde{X} W\right \|_2^2 + \lambda \left \| W \right \|_2^2 \}
    = \big( \widetilde{X}^{\top}\widetilde{X} + \lambda \mathbb{I}_D \big)^{-1} \widetilde{X}^{\top} \widetilde{Y},\label{w_est}\\
    \widehat{b}_{\lambda} &= \frac{1}{T-1} {\bf{1}}_{T-1} ^\top \left(Y  -  X \widehat{W}_\lambda\right),\label{b_est}
\end{align}
where $\lambda \in \mathbb{R}^+$ is the regularization hyperparameter. In practice, when the regularization parameter is tuned, we rely on time-series-adapted cross-validation (CV) (see, e.g., \citealp{bergmeir2012use,bergmeir2018,Hyndman2013}). In ensemble specifications with fixed regularization parameters, we instead consider a prespecified grid of values.

\clearpage
\newpage

\color{brown}\renewcommand{\normalcolor}{\color{brown}}

\section{Additional Empirical Results}
\label{appendix:additional_empirics}

\subsection{Hyperparameter Sensitivity}
\label{section:hyperpars_sensitivity}

We examine the effectiveness of different model groups at a more granular level than the aggregate analysis in Section~\ref{section:forecast_perf}. For the $\alpha$RP ensembles, the total budget of $K=1000$ models is partitioned into subsets of $200$ MFESN specifications, each corresponding to a distinct leak rate. This allows us to explore whether there is substantial variability in forecasting performance across these subsets.

Figure~\ref{fig:plot_ecdf_by_leak_msfe_medium_seed_leak} displays the overlaid ECDFs of EN-MFESN-$\alpha$RP ensemble models grouped by leak rate. The plots reveal pronounced variability across both the leak rates $\alpha$ and model types. 
For S-MFESNs (types A and B), a leak rate of $\alpha = 0.7$ consistently dominates, whereas both low and high values ($\alpha = 0.1$ and $\alpha = 0.9$) perform considerably worse than the ensemble as a whole.
The behavior of M-MFESN ensembles is more nuanced: for both A- and B-type models, $\alpha=0.9$ produces the worst results by a wide margin, while low leak rates ($\alpha \in \{0.1, 0.3\}$) degrade the high-performing experts and shrink the upper tail, compressing the distribution around the median. In contrast, $\alpha = 0.5$ and $\alpha = 0.7$ yield the best overall results. The former uniformly dominates the total $\alpha$RP ECDF, while the latter produces a pronounced rise up to approximately the 3$^\text{rd}$ decile, indicating that a substantial fraction of models populate the lower MSFE tail of the error distribution with this leak rate value.

These findings suggest revisiting the leak-rate choices in \cite{ballarin2022reservoir}. Across MFESN classes, the most effective leak rates differ substantially from the original specifications. For example, the baseline S-MFESNs are constructed with $\alpha = 0.1$ (Table~\ref{tab:model_list}), while Figure~\ref{fig:plot_ecdf_by_leak_msfe_medium_seed_leak}(a)-(b) suggest that the values around $0.7$ are a uniformly more robust choice. 
However, this insight comes at a price. Attempting to study multiple hyperparameters with a comparable level of precision, such as distinct monthly and daily frequency reservoir leak rates in M-MFESNs, would require the ensemble budget to grow exponentially. 
Model combination therefore offers a parsimonious and adaptive alternative, relying on a fixed-budget set of models rather than the construction of a single optimal specification. This adaptivity is particularly valuable for MFESNs, where randomness induces persistent model heterogeneity that cannot be corrected ex-post, a feature shared by many modern neural network architectures.

\begin{figure}[p!]
    \centering
    \begin{subfigure}[b]{0.45\textwidth}
        \includegraphics[width=\textwidth]{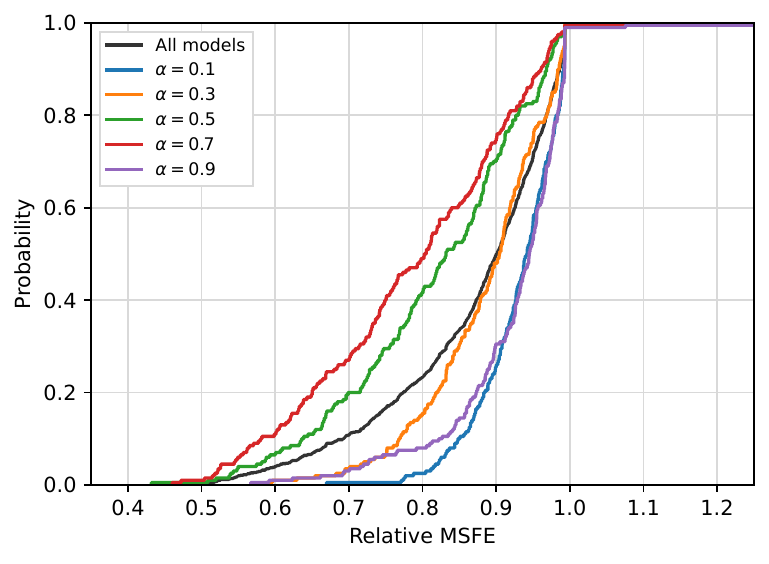}
        \caption{S-MFESN~A}
    \end{subfigure}
    \hfill
    \begin{subfigure}[b]{0.45\textwidth}
        \includegraphics[width=\textwidth]{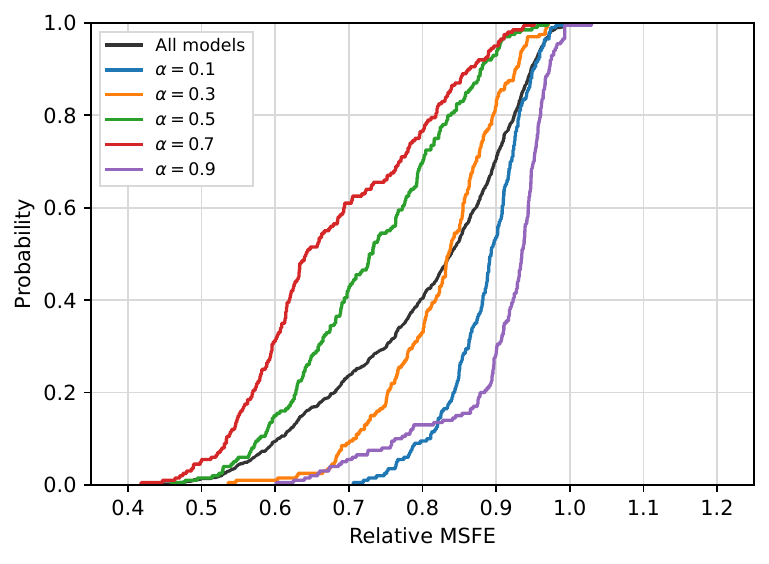}
        \caption{S-MFESN~B}
    \end{subfigure}
    
    \vspace{0.5cm} %
    
    \begin{subfigure}[b]{0.45\textwidth}
        \includegraphics[width=\textwidth]{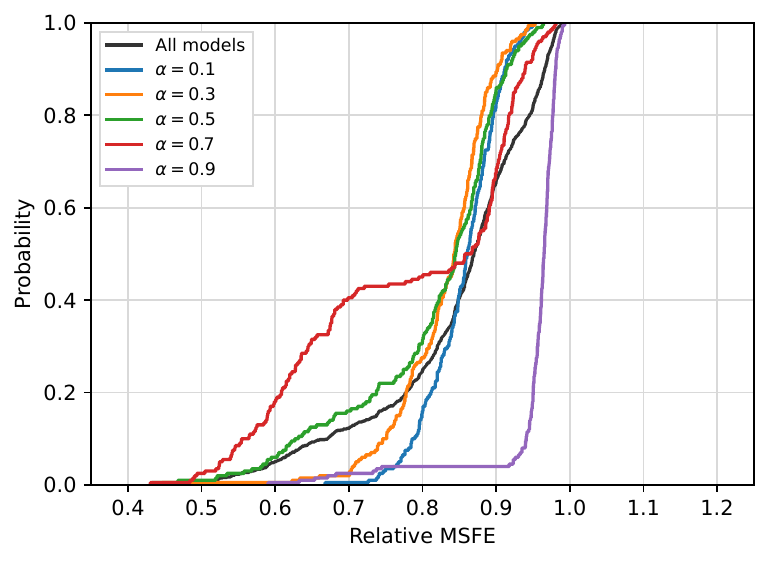}
        \caption{ {\color{green}M-MFESN~A}}
    \end{subfigure}
    \hfill
    \begin{subfigure}[b]{0.45\textwidth}
        \includegraphics[width=\textwidth]{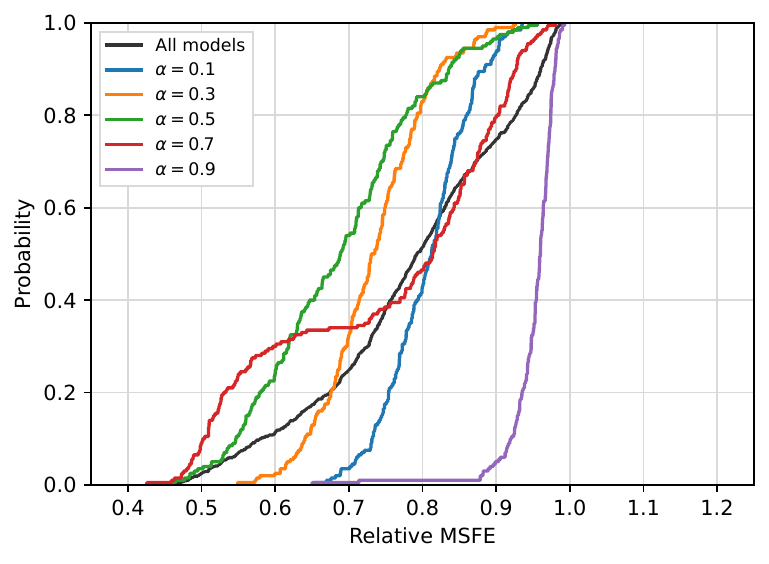}
        \caption{M-MFESN~B}
    \end{subfigure}
    \caption{EN-MFESN-$\alpha$RP ensemble: MSFE relative to {\color{brown}the in-sample mean}. Plots provide the MSFE ECDF of the entire ensemble (black curve) and models grouped by leak rate $\alpha$ (colored).}
    \label{fig:plot_ecdf_by_leak_msfe_medium_seed_leak}
\end{figure}

\clearpage
\newpage

\subsection{Additional Results for the Medium-MD Dataset}
\label{appendix:medium_md_additional}

This appendix provides supplementary results for the main forecast combination application of Section~\ref{section:application}. 
First, we provide a discussion of results obtained with EN-MFESN-$\alpha$RP and EN-MFESN-$\lambda$RP ensembles.
Second, we address our use of the squared loss---which is unbounded---as the combination loss in our main empirical design. Our analysis from Section~\ref{section:application} is repeated, but we perform online expert combination using a thresholded (``clipped'') variant of MSE.

\subsubsection{MSFE Plots for \texorpdfstring{$\alpha$RP}{alphaRP} and
\texorpdfstring{$\lambda$RP}{lambdaRP} Ensembles}

We provide the companion plots to Figures~\ref{fig:plot_ecdf_msfe_medium_seed}-\ref{figure:medium_alpha_lambda_common_grid} for MFESN ensembles applied to the medium-MD dataset.

The results for EN-MFESN-$\alpha$RP ensembles, displayed in Figure~\ref{fig:plot_ecdf_msfe_medium_seed_leak}, broadly align with those obtained with EN-MFESN-RP ensembles, while exhibiting further performance gains due to the increased model richness induced by leak-rate variation. 
Once again, one can see that individual MFESN draws depend significantly on properties of one incidental sampling instance: for example, M-MFESN~B has relative MSFE below the median, whereas S-MFESN~A remains noticeably above it. Importantly, the left tail of the ECDF changes substantially in some cases (see, e.g., Figure~\ref{fig:plot_ecdf_msfe_medium_seed_leak}(c)), indicating that high-performing models (previously unavailable in EN-MFESN-RP) may emerge when leak rate $\alpha$ varies. A further distinction is that no single aggregation method dominates across all EN-MFESN-$\alpha$RP ensembles: FTL and AdaHedge alternate in delivering superior results depending on the MFESN type, as illustrated by the contrast between Panels (a)  and (c) and Panels (b) and (d) in Figure~\ref{fig:plot_ecdf_msfe_medium_seed_leak}. This suggests that introducing leak-rate variability fundamentally changes ensemble composition and, in turn, the relative performance of combination schemes.

In contrast, the error distribution of experts in the EN-MFESN-$\lambda$RP changes significantly. In particular, we can notice that the ECDFs of the relative MSFEs are much less smooth than those obtained in EN-MFESN-RP and EN-MFESN-$\alpha$RP setups. We attribute this to the fact that, by choosing $\lambda$ within a coarse discrete grid, we may force experts to have poor generalization properties.
Indeed, comparing Figure~\ref{fig:plot_ecdf_msfe_medium_seed_leak} with Figure~\ref{figure:medium_lambda_common_grid}, we see that in the latter a substantial portion of experts is concentrated above unit relative MSFE: indeed, 200 models are fitted with ridge penalization of $\lambda = 100$, producing forecasts that are close to zero.
Moreover, it appears that for some model classes the EN-MFESN-$\lambda$RP design can actually hurt forecast combination performance. For example, Panels (a) and (c) in Figure~\ref{figure:medium_lambda_common_grid} indicate that FTL performs worse than benchmark combination schemes (simple average and rolling MSE), while it dominates either in both EN-MFESN-RP and EN-MFESN-$\alpha$RP experiments. 
Panel (b) also shows that, compared to the $\alpha$RP setting, all combination schemes perform worse when applied to S-MFESN~B ensembles than when only leak rates are diversified.
Therefore, even though in some specific cases expert diversification based on $\lambda$ alone may be more effective than using different leak rates (cf. results for M-MFESN~B models in Figure~\ref{figure:medium_lambda_common_grid}(d) versus Figure~\ref{fig:plot_ecdf_msfe_medium_seed_leak}(d)), such gains appear to be circumscribed to specific model specifications. 

\clearpage
\newpage

\begin{figure}[p!]
    \centering
    \begin{subfigure}[b]{0.45\textwidth}
        \includegraphics[width=\textwidth]{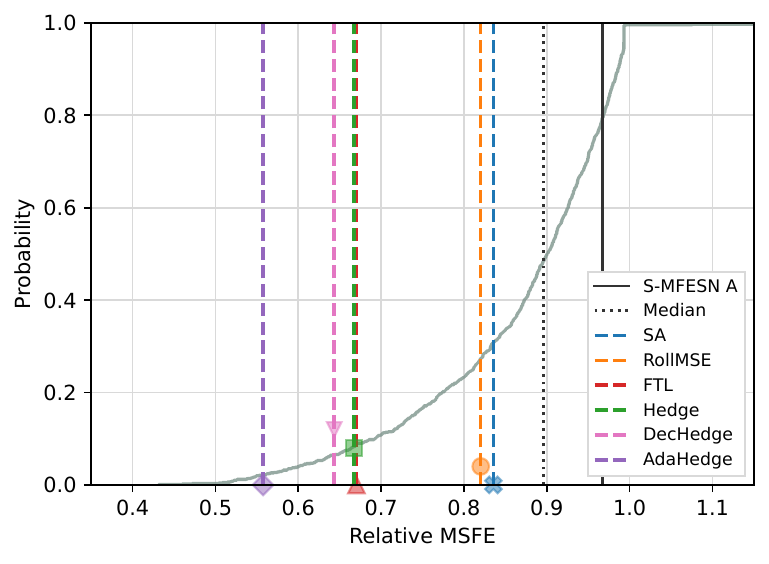}
        \caption{S-MFESN~A}
    \end{subfigure}
    \hfill
    \begin{subfigure}[b]{0.45\textwidth}
        \includegraphics[width=\textwidth]{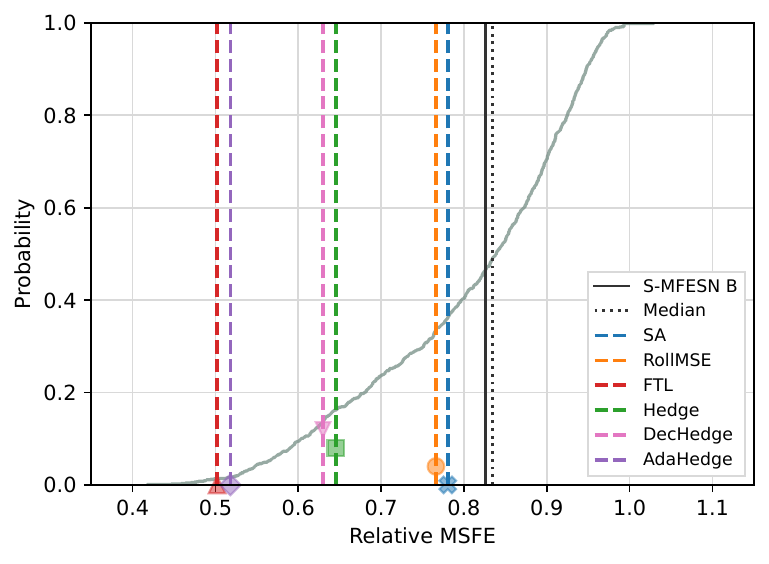}
        \caption{S-MFESN~B}
    \end{subfigure}
    
    \vspace{0.5cm} %
    
    \begin{subfigure}[b]{0.45\textwidth}
        \includegraphics[width=\textwidth]{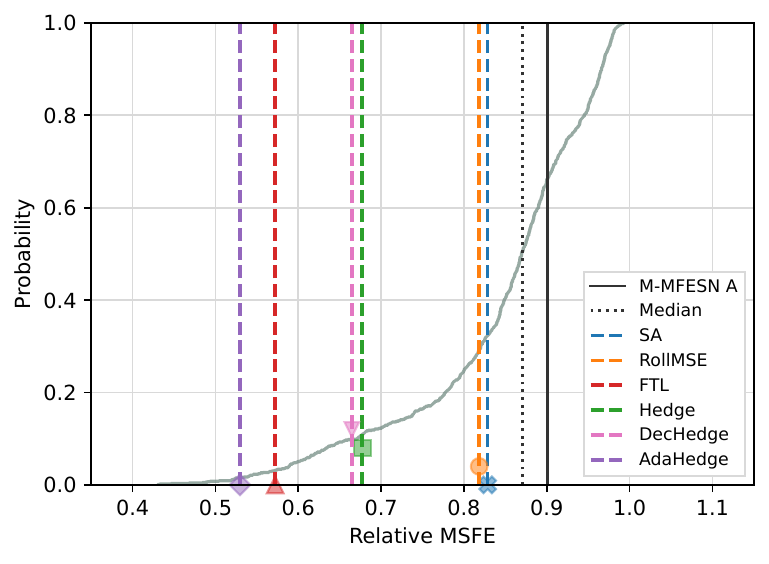}
        
        \caption{{\color{green}M-MFESN~A}}
    \end{subfigure}
    \hfill
    \begin{subfigure}[b]{0.45\textwidth}
        \includegraphics[width=\textwidth]{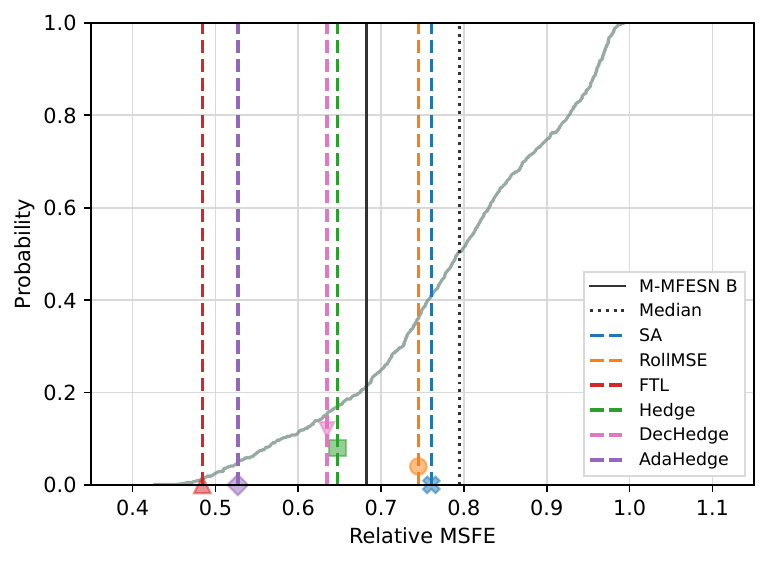}
        \caption{M-MFESN~B}
    \end{subfigure}
    \caption{EN-MFESN-$\alpha$RP ensembles of S-MFESN~A~(a), S-MFESN~B~(b), M-MFESN~A~(c), and M-MFESN~B~(d) models. Plots display MSFE relative to {\color{brown}the in-sample mean as a} baseline for ensemble aggregation methods (colored dashed lines), the performance of Median (black dotted line), specific draws of S-MFESN~A/B and M-MFESN~A/B from \cite{ballarin2022reservoir} (black solid line), and the ECDF of the relative MSFE across individual ensemble models (gray curve). Colored markers are included to distinguish overlapping lines.}
    \label{fig:plot_ecdf_msfe_medium_seed_leak}
\end{figure}

\clearpage

\begin{figure}[p!]
    \centering
    \begin{subfigure}[b]{0.45\textwidth}
        \includegraphics[width=\textwidth]{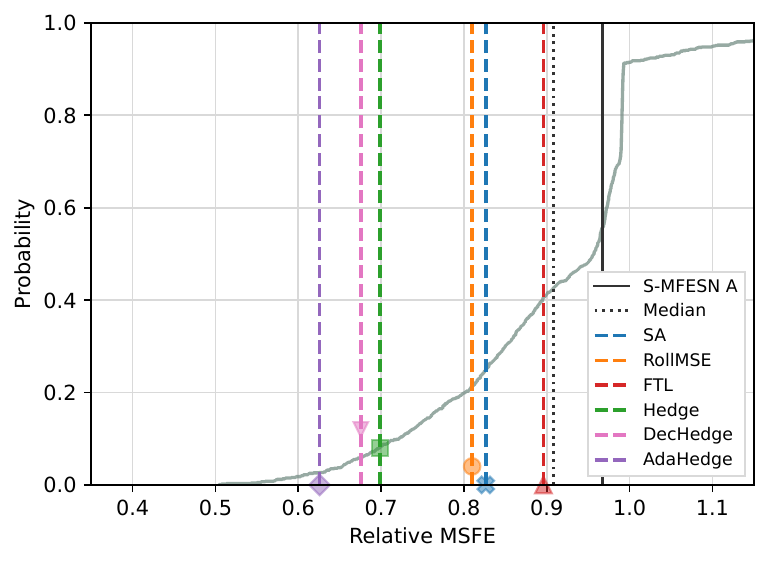}
        \caption{S-MFESN~A}
    \end{subfigure}
    \hfill
    \begin{subfigure}[b]{0.45\textwidth}
        \includegraphics[width=\textwidth]{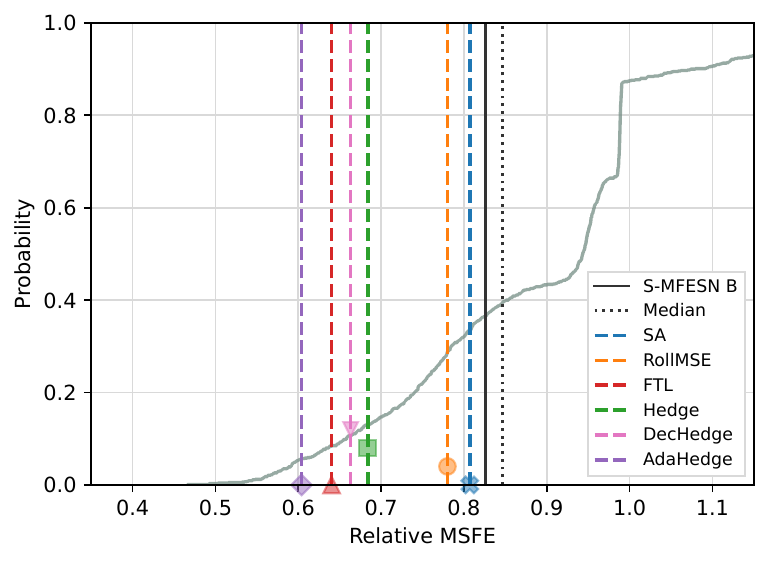}
        \caption{S-MFESN~B}
    \end{subfigure}

    \vspace{0.5cm}

    \begin{subfigure}[b]{0.45\textwidth}
        \includegraphics[width=\textwidth]{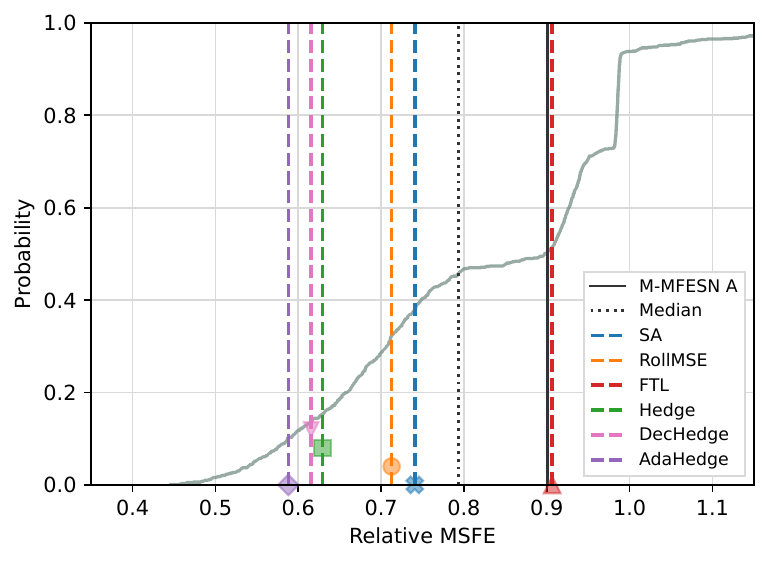}
        \caption{M-MFESN~A}
    \end{subfigure}
    \hfill
    \begin{subfigure}[b]{0.45\textwidth}
        \includegraphics[width=\textwidth]{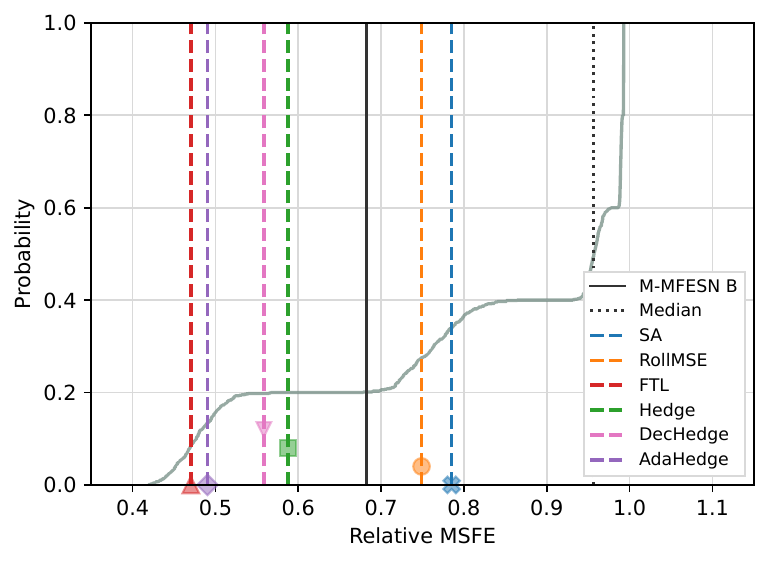}
        \caption{M-MFESN~B}
    \end{subfigure}

    \caption{EN-MFESN-$\lambda$RP ensembles of
    S-MFESN~A~(a), S-MFESN~B~(b), M-MFESN~A~(c), and
    M-MFESN~B~(d) models. Plots display MSFE relative to {\color{brown} the in-sample mean as a} baseline for ensemble aggregation methods (colored dashed lines), the performance of Median (black dotted line), specific draws of S-MFESN~A/B and M-MFESN~A/B from \cite{ballarin2022reservoir} (black solid line), and the ECDF of the relative MSFE across individual ensemble models (gray curve). Colored markers are included to distinguish overlapping lines.}
    \label{figure:medium_lambda_common_grid}
\end{figure}

\clearpage
\newpage

\subsubsection{Combination with Quantile-clipped Losses}
\label{appendix:clipped_se_loss}

Since our theoretical framework requires that all expert loss processes belong almost surely to the unit interval (Assumption~\ref{assumption:losses_0-1}), we check whether using a bounded loss instead of MSE substantially affects our empirical results.
In practice, we expect practitioners to combine models using the same mean-squared error we have employed in Section~\ref{section:application}, which is commonly used as both training loss and forecasting evaluation measure. Nonetheless, it may be the case that even more effective weights can be constructed by using non-convex losses, since, in fact, we make no assumptions on whether the loss process itself is defined in terms of a convex function; we refer to \cite{fuRobustRegressionGeneral2023} for a recent discussion of various bounded losses in the setting of robust regression.

For the sake of simplicity, we focus on a \textit{clipped} version of the squared loss, and the clipping threshold is chosen by the following bootstrapping procedure. 
Let $\mathcal{Y}_{\text{train}} := \{Y_{t}\}_{t=1}^{T_\text{train}}$ be the subset of GDP growth observations used for training. We sample with replacement $T^*$ pairs of observations from $\mathcal{Y}_{\text{train}}$, $\{ ( Y_{1,\tau}, Y_{2,\tau} ) \}_{\tau=1}^{T^*}$, where $Y_{1,\tau}$ and $Y_{2,\tau}$ are drawn independently from each other. Then, we construct the squared differences $\text{se}^2_{\tau} := ( Y_{1,\tau} - Y_{2,\tau} )^2$, and compute their $p^{\text{th}}$ quantile, $\widehat{c}_p^{\,2}$, with $p \in [0, 1]$. We use the resulting quantile as a data-dependent clipping threshold for squared forecast errors. The quantile level refers to the distribution of squared differences between resampled training observations and does not necessarily correspond to the same quantile of the out-of-sample forecast-error distribution. %
We construct the \textit{$p$-quantile clipped combination loss} as
\begin{equation*}
    \ell^{(k,p)}_{t,h}
    :=
    \min\left\{
        \frac{(Y_{t+h}-\widehat{Y}^{(k)}_{t+h})^2}
             {\widehat{c}_p^{\,2}},\,1
    \right\} .
\end{equation*}

We repeat the exercise in Table~\ref{table:msfe_medium_results_all} using clipped squared-error losses for online combination with FTL, Hedge, DecHedge and AdaHedge. We set the clipping threshold by choosing either $p = 0.95$ or $p = 0.99$ and $T^* = 20000$ pairs resamplings. For constant and decreasing Hedge, we set $\eta_p=\eta\widehat{c}_p^{\,2}$ and $c_{0,p}=c_0\widehat{c}_p^{\,2}$, respectively, where $\eta=c_0=10^4$ are the learning-rate parameters used with unclipped combination losses; otherwise, the experimental designs are identical to the ones used for Table~\ref{table:msfe_medium_results_all}. To ensure comparability between results, in this setting, too, we employ (unclipped) MSFE to measure prediction performance.

Our findings are collected in Tables~\ref{table:medium_md_clipped_combination_95_scaled_48q} and \ref{table:medium_md_clipped_combination_99_scaled_48q}, where clipping is performed at the 95\% and 99\% quantiles, respectively. For ease of comparison, in parentheses we also report the percentage variation from the matching relative MSFE in our primary empirical exercise, cf. Table~\ref{table:msfe_medium_results_all}.
Concentrating on Table~\ref{table:medium_md_clipped_combination_99_scaled_48q}---which we believe provides a more robust clipping threshold, given that our testing period includes the 2008 financial crisis---we can see that we can split combination schemes into two groups based on how much clipping affects MSFE: constant and decreasing Hedge, which tend to be more heavily influenced, and FTL and AdaHedge, which show less variation in performance. 
We observe that the MSFE variations for FTL and AdaHedge do not exceed $6.23\%$; Follow-the-Leader, in particular, shows no change in MSFE at the reported precision in 11 out of 16 MFESN model classes and settings with respect to results in Section~\ref{section:application}.
On the other hand, non-adaptive Hedge schemes are more consistently affected by clipping: constant Hedge shows relative MSFE increases between $2.13\%$ and $7.39\%$, while decreasing Hedge exhibits increases between $1.32\%$ and $8.60\%$.
These results help confirm that Follow-the-Leader and AdaHedge combination schemes tend not only to be more performant and effective at online forecasting combination in our macroeconomic experiments, but they also appear to be less sensitive to structural perturbations of the loss. This confirmation is particularly helpful for practitioners, because, as we mentioned above, from an applied point of view it is not necessarily clear which combination loss is a priori optimal for a given task. It would be especially interesting to formally study the choice of the loss itself as that of tuning parameter.

\clearpage

\begin{table}[p!]
\centering
\small
\setlength\tabcolsep{0pt}
\setlength\extrarowheight{0pt}
\renewcommand{\arraystretch}{0.98}
\selectfont
\captionsetup{font=small,skip=8pt}

\makebox[\textwidth][c]{
\begin{tabularx}{\textwidth}{@{} >{\centering\arraybackslash}p{2.6cm} @{\extracolsep{\fill}}*{4}{>{\centering\arraybackslash}X} @{}}
    & \multicolumn{4}{c}{Ensemble} \\
    \cline{2-5}
    \addlinespace[2pt]
    Model & FTL & Hedge & DecHedge & AdaHedge \\
    \toprule
    \multicolumn{5}{l}{\footnotesize EN-MFESN-RP (random parameters resampling)~\raisedrule[0.2em]{0.2pt}} \\[2pt]
    S-MFESN~A & 0.653 & 0.892 & 0.893 & 0.710 \\
     & {\footnotesize\textit{(-1.58\%)}} & {\footnotesize\textit{(+6.31\%)}} & {\footnotesize\textit{(+5.53\%)}} & {\footnotesize\textit{(-0.16\%)}} \\[1.5pt]
    S-MFESN~B & 0.580 & 0.861 & 0.862 & 0.633 \\
     & {\footnotesize\textit{(-6.33\%)}} & {\footnotesize\textit{(+4.40\%)}} & {\footnotesize\textit{(+5.05\%)}} & {\footnotesize\textit{(-5.11\%)}} \\[1.5pt]
    M-MFESN~A & 0.660 & 0.859 & 0.860 & 0.678 \\
     & {\footnotesize\textit{(0.00\%)}} & {\footnotesize\textit{(+2.70\%)}} & {\footnotesize\textit{(+2.10\%)}} & {\footnotesize\textit{(+1.39\%)}} \\[1.5pt]
    M-MFESN~B & 0.579 & 0.708 & 0.706 & 0.592 \\
     & {\footnotesize\textit{(-1.47\%)}} & {\footnotesize\textit{(+2.87\%)}} & {\footnotesize\textit{(+2.94\%)}} & {\footnotesize\textit{(+0.26\%)}} \\[1.5pt]
    \addlinespace[3pt]
    \multicolumn{5}{l}{\footnotesize EN-MFESN-$\alpha$RP (random parameters resampling \& varying leak rates)~\raisedrule[0.2em]{0.2pt}} \\[2pt]
    S-MFESN~A & 0.559 & 0.725 & 0.716 & 0.545 \\
     & {\footnotesize\textit{(-16.53\%)}} & {\footnotesize\textit{(+8.64\%)}} & {\footnotesize\textit{(+11.25\%)}} & {\footnotesize\textit{(-2.25\%)}} \\[1.5pt]
    S-MFESN~B & 0.514 & 0.691 & 0.680 & 0.529 \\
     & {\footnotesize\textit{(+2.35\%)}} & {\footnotesize\textit{(+7.09\%)}} & {\footnotesize\textit{(+8.02\%)}} & {\footnotesize\textit{(+2.10\%)}} \\[1.5pt]
    M-MFESN~A & 0.535 & 0.738 & 0.734 & 0.536 \\
     & {\footnotesize\textit{(-6.44\%)}} & {\footnotesize\textit{(+9.03\%)}} & {\footnotesize\textit{(+10.48\%)}} & {\footnotesize\textit{(+1.23\%)}} \\[1.5pt]
    M-MFESN~B & 0.452 & 0.679 & 0.668 & 0.510 \\
     & {\footnotesize\textit{(-6.61\%)}} & {\footnotesize\textit{(+4.94\%)}} & {\footnotesize\textit{(+5.25\%)}} & {\footnotesize\textit{(-3.19\%)}} \\[1.5pt]
    \addlinespace[3pt]
    \multicolumn{5}{l}{\footnotesize EN-MFESN-$\lambda$RP (random parameters resampling \& varying ridge penalties)~\raisedrule[0.2em]{0.2pt}} \\[2pt]
    S-MFESN~A & 0.660 & 0.731 & 0.720 & 0.616 \\
     & {\footnotesize\textit{(-26.33\%)}} & {\footnotesize\textit{(+4.60\%)}} & {\footnotesize\textit{(+6.52\%)}} & {\footnotesize\textit{(-1.61\%)}} \\[1.5pt]
    S-MFESN~B & 0.677 & 0.720 & 0.704 & 0.620 \\
     & {\footnotesize\textit{(+5.79\%)}} & {\footnotesize\textit{(+5.29\%)}} & {\footnotesize\textit{(+6.14\%)}} & {\footnotesize\textit{(+2.62\%)}} \\[1.5pt]
    M-MFESN~A & 1.112 & 0.664 & 0.649 & 0.643 \\
     & {\footnotesize\textit{(+22.65\%)}} & {\footnotesize\textit{(+5.52\%)}} & {\footnotesize\textit{(+5.47\%)}} & {\footnotesize\textit{(+9.32\%)}} \\[1.5pt]
    M-MFESN~B & 0.453 & 0.613 & 0.580 & 0.492 \\
     & {\footnotesize\textit{(-3.76\%)}} & {\footnotesize\textit{(+4.37\%)}} & {\footnotesize\textit{(+3.79\%)}} & {\footnotesize\textit{(+0.33\%)}} \\[1.5pt]
    \addlinespace[3pt]
    \multicolumn{5}{l}{\footnotesize EN-MFESN-$\alpha\lambda$RP (random parameters resampling \& varying leak rates and ridge penalties)~\raisedrule[0.2em]{0.2pt}} \\[2pt]
    S-MFESN~A & 0.534 & 0.690 & 0.673 & 0.578 \\
     & {\footnotesize\textit{(-11.77\%)}} & {\footnotesize\textit{(+6.85\%)}} & {\footnotesize\textit{(+9.24\%)}} & {\footnotesize\textit{(+1.79\%)}} \\[1.5pt]
    S-MFESN~B & 0.587 & 0.632 & 0.611 & 0.521 \\
     & {\footnotesize\textit{(+1.38\%)}} & {\footnotesize\textit{(+4.94\%)}} & {\footnotesize\textit{(+5.66\%)}} & {\footnotesize\textit{(-1.27\%)}} \\[1.5pt]
    M-MFESN~A & 0.806 & 0.655 & 0.640 & 0.590 \\
     & {\footnotesize\textit{(+17.47\%)}} & {\footnotesize\textit{(+3.82\%)}} & {\footnotesize\textit{(+4.39\%)}} & {\footnotesize\textit{(+1.81\%)}} \\[1.5pt]
    M-MFESN~B & 0.493 & 0.674 & 0.646 & 0.503 \\
     & {\footnotesize\textit{(+6.13\%)}} & {\footnotesize\textit{(+7.06\%)}} & {\footnotesize\textit{(+7.08\%)}} & {\footnotesize\textit{(+4.03\%)}} \\[1.5pt]
    \bottomrule
\end{tabularx}
}
\caption{
    Relative MSFE of quarterly U.S.\ GDP growth predictions with respect to the in-sample mean, with clipped squared loss used for online combination. The clipping threshold uses the bootstrapped 95th-percentile squared differences ($p = 0.95$); unclipped squared-error loss is used for evaluation. Italic percentages in parentheses report MSFE changes relative to the same method under no-clipping; see Table~\ref{table:msfe_medium_results_all}.
}
\label{table:medium_md_clipped_combination_95_scaled_48q}
\end{table}

\begin{table}[p!]
\centering
\small
\setlength\tabcolsep{0pt}
\setlength\extrarowheight{0pt}
\renewcommand{\arraystretch}{0.98}
\selectfont
\captionsetup{font=small,skip=8pt}

\makebox[\textwidth][c]{
\begin{tabularx}{1\textwidth}{@{} >{\centering\arraybackslash}p{2.6cm} @{\extracolsep{\fill}}*{4}{>{\centering\arraybackslash}X} @{}}
    & \multicolumn{4}{c}{Ensemble} \\
    \cline{2-5}
    \addlinespace[2pt]
    Model & FTL & Hedge & DecHedge & AdaHedge \\
    \toprule
    \multicolumn{5}{l}{\footnotesize EN-MFESN-RP (random parameters resampling)~\raisedrule[0.2em]{0.2pt}} \\[2pt]
    S-MFESN~A & 0.653 & 0.880 & 0.886 & 0.708 \\
     & {\footnotesize\textit{(-1.58\%)}} & {\footnotesize\textit{(+4.89\%)}} & {\footnotesize\textit{(+4.66\%)}} & {\footnotesize\textit{(-0.45\%)}} \\[1.5pt]
    S-MFESN~B & 0.594 & 0.855 & 0.857 & 0.641 \\
     & {\footnotesize\textit{(-4.05\%)}} & {\footnotesize\textit{(+3.67\%)}} & {\footnotesize\textit{(+4.52\%)}} & {\footnotesize\textit{(-3.92\%)}} \\[1.5pt]
    M-MFESN~A & 0.660 & 0.855 & 0.857 & 0.678 \\
     & {\footnotesize\textit{(0.00\%)}} & {\footnotesize\textit{(+2.17\%)}} & {\footnotesize\textit{(+1.77\%)}} & {\footnotesize\textit{(+1.39\%)}} \\[1.5pt]
    M-MFESN~B & 0.579 & 0.707 & 0.706 & 0.592 \\
     & {\footnotesize\textit{(-1.47\%)}} & {\footnotesize\textit{(+2.78\%)}} & {\footnotesize\textit{(+2.85\%)}} & {\footnotesize\textit{(+0.26\%)}} \\[1.5pt]
    \addlinespace[3pt]
    \multicolumn{5}{l}{\footnotesize EN-MFESN-$\alpha$RP (random parameters resampling \& varying leak rates)~\raisedrule[0.2em]{0.2pt}} \\[2pt]
    S-MFESN~A & 0.670 & 0.706 & 0.689 & 0.562 \\
     & {\footnotesize\textit{(0.00\%)}} & {\footnotesize\textit{(+5.86\%)}} & {\footnotesize\textit{(+7.10\%)}} & {\footnotesize\textit{(+0.86\%)}} \\[1.5pt]
    S-MFESN~B & 0.502 & 0.680 & 0.667 & 0.518 \\
     & {\footnotesize\textit{(0.00\%)}} & {\footnotesize\textit{(+5.34\%)}} & {\footnotesize\textit{(+5.91\%)}} & {\footnotesize\textit{(+0.06\%)}} \\[1.5pt]
    M-MFESN~A & 0.572 & 0.727 & 0.722 & 0.532 \\
     & {\footnotesize\textit{(0.00\%)}} & {\footnotesize\textit{(+7.39\%)}} & {\footnotesize\textit{(+8.60\%)}} & {\footnotesize\textit{(+0.39\%)}} \\[1.5pt]
    M-MFESN~B & 0.484 & 0.672 & 0.658 & 0.527 \\
     & {\footnotesize\textit{(0.00\%)}} & {\footnotesize\textit{(+3.87\%)}} & {\footnotesize\textit{(+3.70\%)}} & {\footnotesize\textit{(+0.06\%)}} \\[1.5pt]
    \addlinespace[3pt]
    \multicolumn{5}{l}{\footnotesize EN-MFESN-$\lambda$RP (random parameters resampling \& varying ridge penalties)~\raisedrule[0.2em]{0.2pt}} \\[2pt]
    S-MFESN~A & 0.840 & 0.723 & 0.704 & 0.632 \\
     & {\footnotesize\textit{(-6.23\%)}} & {\footnotesize\textit{(+3.46\%)}} & {\footnotesize\textit{(+4.18\%)}} & {\footnotesize\textit{(+1.00\%)}} \\[1.5pt]
    S-MFESN~B & 0.640 & 0.713 & 0.691 & 0.610 \\
     & {\footnotesize\textit{(0.00\%)}} & {\footnotesize\textit{(+4.19\%)}} & {\footnotesize\textit{(+4.22\%)}} & {\footnotesize\textit{(+0.98\%)}} \\[1.5pt]
    M-MFESN~A & 0.907 & 0.654 & 0.638 & 0.602 \\
     & {\footnotesize\textit{(0.00\%)}} & {\footnotesize\textit{(+3.97\%)}} & {\footnotesize\textit{(+3.71\%)}} & {\footnotesize\textit{(+2.43\%)}} \\[1.5pt]
    M-MFESN~B & 0.470 & 0.600 & 0.566 & 0.490 \\
     & {\footnotesize\textit{(0.00\%)}} & {\footnotesize\textit{(+2.13\%)}} & {\footnotesize\textit{(+1.32\%)}} & {\footnotesize\textit{(+0.00\%)}} \\[1.5pt]
    \addlinespace[3pt]
    \multicolumn{5}{l}{\footnotesize EN-MFESN-$\alpha\lambda$RP (random parameters resampling \& varying leak rates and ridge penalties)~\raisedrule[0.2em]{0.2pt}} \\[2pt]
    S-MFESN~A & 0.597 & 0.677 & 0.648 & 0.576 \\
     & {\footnotesize\textit{(-1.36\%)}} & {\footnotesize\textit{(+4.90\%)}} & {\footnotesize\textit{(+5.11\%)}} & {\footnotesize\textit{(+1.29\%)}} \\[1.5pt]
    S-MFESN~B & 0.579 & 0.620 & 0.590 & 0.529 \\
     & {\footnotesize\textit{(0.00\%)}} & {\footnotesize\textit{(+2.98\%)}} & {\footnotesize\textit{(+2.07\%)}} & {\footnotesize\textit{(+0.26\%)}} \\[1.5pt]
    M-MFESN~A & 0.686 & 0.655 & 0.637 & 0.578 \\
     & {\footnotesize\textit{(0.00\%)}} & {\footnotesize\textit{(+3.80\%)}} & {\footnotesize\textit{(+3.86\%)}} & {\footnotesize\textit{(-0.22\%)}} \\[1.5pt]
    M-MFESN~B & 0.465 & 0.659 & 0.627 & 0.484 \\
     & {\footnotesize\textit{(0.00\%)}} & {\footnotesize\textit{(+4.67\%)}} & {\footnotesize\textit{(+4.00\%)}} & {\footnotesize\textit{(-0.03\%)}} \\[1.5pt]
    \bottomrule
\end{tabularx}
}
\caption{
    Relative MSFE of quarterly U.S.\ GDP growth predictions with respect to the in-sample mean, with clipped squared loss used for online combination. The clipping threshold uses the bootstrapped 99th-percentile squared differences ($p = 0.99$); unclipped squared-error loss is used for evaluation. Italic percentages in parentheses report MSFE changes relative to the same method under no-clipping; see Table~\ref{table:msfe_medium_results_all}.
}
\label{table:medium_md_clipped_combination_99_scaled_48q}
\end{table}

\clearpage
\newpage

\subsection{Forecasting Results for the Small-MD Dataset}
\label{appendix:small_md_results}

This appendix reports the results for the same experimental setup described in Section~\ref{section:application} applied to the small-MD dataset, which features only $9$ regressors (against the $33$ used in the medium-MD application). For all MFESN types and ensemble designs (RP, $\alpha$RP, $\lambda$RP, and $\alpha\lambda$RP), the construction of experts and the evaluation of the ensemble performance mimic the one employed in the medium-MD setting.

Table~\ref{table:msfe_small_results_all} reports the Small-MD
counterpart of Table~\ref{table:msfe_medium_results_all}.
In this reduced data setting, we can see that FTL is overwhelmingly dominant in terms of combination performance, with AdaHedge ranking second. When compared to the medium-MD setting, all baseline MFESN models, as well as benchmark DFMs, have better forecasting performance; given the somewhat short length of the evaluation sample, this is not unexpected. 
Nonetheless, the picture for online combination is not as clear-cut. In the EN-MFESN-RP setting, A-type MFESN ensembles yield lower-performing combinations (comparing the lowest relative MSFE across methods) than their counterparts in the medium-MD experiments in Table~\ref{table:msfe_medium_results_all}---a similar pattern also emerges in the EN-MFESN-$\alpha$RP setup with S- versus M-MFESN models, and specifically with M-MFESN~B experts in both EN-MFESN-$\lambda$RP and EN-MFESN-$\alpha\lambda$RP setups. Figures~\ref{fig:plot_ecdf_msfe_small_seed}-\ref{figure:small_alpha_lambda_common_grid} complement these results by displaying the distributions of individual experts' relative MSFEs and the performance of the combination methods across the four ensemble designs.

It is not, then, unequivocally clear whether online combination can be seen as a supplement or, potentially, a replacement for the selection of the \textit{information set} used for prediction. We partially explore this question in the simulations of Appendix~\ref{appendix:simulation_random_subset}, where we construct simple random-subset linear experts. Nevertheless, the linear nature of both the DGP and the experts does not provide broad insight into the behavior of FTL and Hedge-type schemes in the much richer nonlinear multi-frequency ESN setting we study in our applications. Still, linking online combination with model selection appears to us as both natural and desirable, to the extent that some of these algorithms (FTL, AdaHedge) are also free of tuning parameters. We consider this a particularly interesting avenue for future research.

\clearpage

\begin{table}[p!]
\centering
\small
\setlength\tabcolsep{0pt}
\setlength\extrarowheight{1pt}
\linespread{1.01}
\selectfont
\begin{tabularx}{\textwidth}{@{} >{\centering\arraybackslash}p{2.5cm} @{\extracolsep{\fill}}*{8}{>{\centering\arraybackslash}X} @{}}

    & & \multicolumn{7}{c}{Ensemble} \\
    \cline{3-9}
    \addlinespace[2pt]
    Model & Baseline & Median & SA & RollMSE & FTL & Hedge & DecHedge & AdaHedge \\
    \toprule
    Mean & 1.000 & -- & -- & -- & -- & -- & -- & -- \\
    AR(1) & 0.758 & -- & -- & -- & -- & -- & -- & -- \\
    DFM A & 0.799 & -- & -- & -- & -- & -- & -- & -- \\
    DFM B & 0.885 & -- & -- & -- & -- & -- & -- & -- \\[2pt]
    \multicolumn{9}{l}{\footnotesize%
        EN-MFESN-RP (random parameters resampling)~
        \raisedrule[0.2em]{0.2pt}} \\[2pt]
    S-MFESN~A & 0.721 & 0.820 & 0.813 & 0.811 & \ul{ 0.695 } & 0.774 & 0.773 & \bf 0.690 \\[-4pt]
    & {\footnotesize\textit{-27.85\%}}
    & {\footnotesize\textit{-18.04\%}}
    & {\footnotesize\textit{-18.68\%}}
    & {\footnotesize\textit{-18.92\%}}
    & {\footnotesize\ul{\textit{-30.47\%}}}
    & {\footnotesize\textit{-22.62\%}}
    & {\footnotesize\textit{-22.72\%}}
    & {\footnotesize\bfseries\textit{-31.00\%}} \\
    S-MFESN~B & 0.758 & 0.766 & 0.764 & 0.762 & \bf 0.602 & 0.733 & 0.732 & \ul{ 0.616 } \\[-4pt]
    & {\footnotesize\textit{-24.20\%}}
    & {\footnotesize\textit{-23.36\%}}
    & {\footnotesize\textit{-23.62\%}}
    & {\footnotesize\textit{-23.79\%}}
    & {\footnotesize\bfseries\textit{-39.84\%}}
    & {\footnotesize\textit{-26.68\%}}
    & {\footnotesize\textit{-26.79\%}}
    & {\footnotesize\ul{\textit{-38.39\%}}} \\
    M-MFESN~A & 0.802 & 0.799 & 0.800 & 0.799 & \bf 0.688 & 0.784 & 0.783 & \ul{ 0.711 } \\[-4pt]
    & {\footnotesize\textit{-19.80\%}}
    & {\footnotesize\textit{-20.08\%}}
    & {\footnotesize\textit{-19.96\%}}
    & {\footnotesize\textit{-20.07\%}}
    & {\footnotesize\bfseries\textit{-31.23\%}}
    & {\footnotesize\textit{-21.63\%}}
    & {\footnotesize\textit{-21.70\%}}
    & {\footnotesize\ul{\textit{-28.91\%}}} \\
    M-MFESN~B & 0.590 & 0.651 & 0.655 & 0.654 & \bf 0.505 & 0.631 & 0.629 & \ul{ 0.528 } \\[-4pt]
     & {\footnotesize\textit{-40.96\%}}
    & {\footnotesize\textit{-34.93\%}}
    & {\footnotesize\textit{-34.49\%}}
    & {\footnotesize\textit{-34.63\%}}
    & {\footnotesize\bfseries\textit{-49.50\%}}
    & {\footnotesize\textit{-36.87\%}}
    & {\footnotesize\textit{-37.15\%}}
    & {\footnotesize\ul{\textit{-47.22\%}}} \\
    \multicolumn{9}{l}{\footnotesize%
        EN-MFESN-$\alpha$RP (random parameters resampling \& varying leak rates)~
        \raisedrule[0.2em]{0.2pt}} \\[2pt]
    S-MFESN~A & 0.721 & 0.829 & 0.811 & 0.804 & \bf 0.506 & 0.704 & 0.700 & \ul{ 0.577 } \\[-4pt]
    & {\footnotesize\textit{-27.85\%}}
    & {\footnotesize\textit{-17.13\%}}
    & {\footnotesize\textit{-18.92\%}}
    & {\footnotesize\textit{-19.55\%}}
    & {\footnotesize\bfseries\textit{-49.38\%}}
    & {\footnotesize\textit{-29.65\%}}
    & {\footnotesize\textit{-29.96\%}}
    & {\footnotesize\ul{\textit{-42.29\%}}} \\
    S-MFESN~B & 0.758 & 0.768 & 0.769 & 0.763 & \bf 0.443 & 0.681 & 0.672 & \ul{ 0.544 } \\[-4pt]
    & {\footnotesize\textit{-24.20\%}}
    & {\footnotesize\textit{-23.15\%}}
    & {\footnotesize\textit{-23.11\%}}
    & {\footnotesize\textit{-23.73\%}}
    & {\footnotesize\bfseries\textit{-55.72\%}}
    & {\footnotesize\textit{-31.88\%}}
    & {\footnotesize\textit{-32.81\%}}
    & {\footnotesize\ul{\textit{-45.65\%}}} \\
    M-MFESN~A & 0.802 & 0.846 & 0.850 & 0.847 & \bf 0.716 & 0.811 & 0.811 & \ul{ 0.722 } \\[-4pt]
    & {\footnotesize\textit{-19.80\%}}
    & {\footnotesize\textit{-15.38\%}}
    & {\footnotesize\textit{-14.99\%}}
    & {\footnotesize\textit{-15.30\%}}
    & {\footnotesize\bfseries\textit{-28.36\%}}
    & {\footnotesize\textit{-18.90\%}}
    & {\footnotesize\textit{-18.88\%}}
    & {\footnotesize\ul{\textit{-27.81\%}}} \\
    M-MFESN~B & 0.590 & 0.734 & 0.747 & 0.740 & \bf 0.540 & 0.680 & 0.675 & \ul{ 0.561 } \\[-4pt]
    & {\footnotesize\textit{-40.96\%}}
    & {\footnotesize\textit{-26.65\%}}
    & {\footnotesize\textit{-25.33\%}}
    & {\footnotesize\textit{-26.04\%}}
    & {\footnotesize\bfseries\textit{-46.05\%}}
    & {\footnotesize\textit{-32.03\%}}
    & {\footnotesize\textit{-32.47\%}}
    & {\footnotesize\ul{\textit{-43.86\%}}} \\
    
    \multicolumn{9}{l}{\footnotesize%
        EN-MFESN-$\lambda$RP
        (random parameters resampling \& varying ridge penalties)~
        \raisedrule[0.2em]{0.2pt}} \\[2pt]
    S-MFESN~A
    & 0.721 & 0.858 & 0.786 & 0.772
    & {\bfseries 0.535} & 0.669 & 0.646 & \ul{0.570} \\[-4pt]
    & {\footnotesize\textit{-27.85\%}}
    & {\footnotesize\textit{-14.19\%}}
    & {\footnotesize\textit{-21.39\%}}
    & {\footnotesize\textit{-22.83\%}}
    & {\footnotesize\bfseries\textit{-46.47\%}}
    & {\footnotesize\textit{-33.08\%}}
    & {\footnotesize\textit{-35.35\%}}
    & {\footnotesize\ul{\textit{-43.02\%}}} \\
    S-MFESN~B
    & 0.758 & 0.794 & 0.726 & 0.707
    & {\bfseries 0.501} & 0.622 & 0.599 & \ul{0.534} \\[-4pt]
    & {\footnotesize\textit{-24.20\%}}
    & {\footnotesize\textit{-20.63\%}}
    & {\footnotesize\textit{-27.41\%}}
    & {\footnotesize\textit{-29.30\%}}
    & {\footnotesize\bfseries\textit{-49.86\%}}
    & {\footnotesize\textit{-37.83\%}}
    & {\footnotesize\textit{-40.11\%}}
    & {\footnotesize\ul{\textit{-46.59\%}}} \\
    M-MFESN~A
    & 0.802 & 0.759 & 0.740 & 0.727
    & 0.632 & 0.653 & \ul{0.625} & {\bfseries 0.565} \\[-4pt]
    & {\footnotesize\textit{-19.80\%}}
    & {\footnotesize\textit{-24.07\%}}
    & {\footnotesize\textit{-25.96\%}}
    & {\footnotesize\textit{-27.35\%}}
    & {\footnotesize\textit{-36.83\%}}
    & {\footnotesize\textit{-34.66\%}}
    & {\footnotesize\ul{\textit{-37.48\%}}}
    & {\footnotesize\bfseries\textit{-43.46\%}} \\
    M-MFESN~B
    & 0.590 & 0.957 & 0.803 & 0.771
    & {\bfseries 0.499} & 0.628 & 0.603 & \ul{0.511} \\[-4pt]
    & {\footnotesize\textit{-40.96\%}}
    & {\footnotesize\textit{-4.26\%}}
    & {\footnotesize\textit{-19.69\%}}
    & {\footnotesize\textit{-22.85\%}}
    & {\footnotesize\bfseries\textit{-50.07\%}}
    & {\footnotesize\textit{-37.20\%}}
    & {\footnotesize\textit{-39.70\%}}
    & {\footnotesize\ul{\textit{-48.91\%}}} \\
    \multicolumn{9}{l}{\footnotesize%
        EN-MFESN-$\alpha\lambda$RP
        (random parameters resampling \& varying leak rates and ridge penalties)~
        \raisedrule[0.2em]{0.2pt}} \\[2pt]
    S-MFESN~A
    & 0.721 & 0.879 & 0.776 & 0.757
    & \ul{0.509} & 0.633 & 0.610 & {\bfseries 0.507} \\[-4pt]
    & {\footnotesize\textit{-27.85\%}}
    & {\footnotesize\textit{-12.12\%}}
    & {\footnotesize\textit{-22.37\%}}
    & {\footnotesize\textit{-24.26\%}}
    & {\footnotesize\ul{\textit{-49.06\%}}}
    & {\footnotesize\textit{-36.68\%}}
    & {\footnotesize\textit{-38.99\%}}
    & {\footnotesize\bfseries\textit{-49.27\%}} \\
    S-MFESN~B
    & 0.758 & 0.814 & 0.722 & 0.700
    & \ul{0.540} & 0.592 & 0.569 & {\bfseries 0.498} \\[-4pt]
    & {\footnotesize\textit{-24.20\%}}
    & {\footnotesize\textit{-18.60\%}}
    & {\footnotesize\textit{-27.80\%}}
    & {\footnotesize\textit{-30.02\%}}
    & {\footnotesize\ul{\textit{-46.02\%}}}
    & {\footnotesize\textit{-40.84\%}}
    & {\footnotesize\textit{-43.05\%}}
    & {\footnotesize\bfseries\textit{-50.21\%}} \\
    M-MFESN~A
    & 0.802 & 0.807 & 0.761 & 0.746
    & {\bfseries 0.492} & 0.648 & 0.618 & \ul{0.529} \\[-4pt]
    & {\footnotesize\textit{-19.80\%}}
    & {\footnotesize\textit{-19.25\%}}
    & {\footnotesize\textit{-23.90\%}}
    & {\footnotesize\textit{-25.42\%}}
    & {\footnotesize\bfseries\textit{-50.76\%}}
    & {\footnotesize\textit{-35.22\%}}
    & {\footnotesize\textit{-38.17\%}}
    & {\footnotesize\ul{\textit{-47.10\%}}} \\
    M-MFESN~B
    & 0.590 & 0.967 & 0.841 & 0.815
    & \ul{0.541} & 0.665 & 0.641 & {\bfseries 0.526} \\[-4pt]
    & {\footnotesize\textit{-40.96\%}}
    & {\footnotesize\textit{-3.29\%}}
    & {\footnotesize\textit{-15.90\%}}
    & {\footnotesize\textit{-18.50\%}}
    & {\footnotesize\ul{\textit{-45.89\%}}}
    & {\footnotesize\textit{-33.53\%}}
    & {\footnotesize\textit{-35.86\%}}
    & {\footnotesize\bfseries\textit{-47.37\%}} \\
    \bottomrule
\end{tabularx}
\color{black}
\caption{%
Relative MSFE of quarterly U.S. GDP growth predictions with respect to the in-sample mean for the Small-MD dataset. Baselines are benchmarks and models in \cite{ballarin2022reservoir}.
Ensemble size is $K=1000$ for each MFESN specification.
For EN-MFESN-$\lambda$RP and EN-MFESN-$\alpha\lambda$RP, the ridge penalties are $\lambda\in\{0.01,0.1,1,10,100\}$.
Performance changes in percentage with respect to the in-sample mean are shown in italic. Best performing combinations are displayed in bold; second-best performing combinations are underlined.
}
\label{table:msfe_small_results_all}
\end{table}

\begin{figure}[p!]
    \centering
    \begin{subfigure}[b]{0.45\textwidth}
        \includegraphics[width=\textwidth]{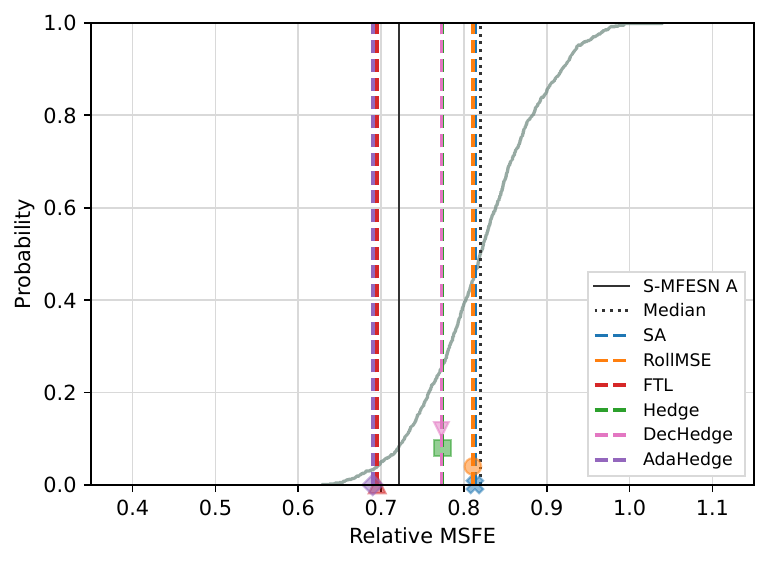}
        \caption{S-MFESN~A}
    \end{subfigure}
    \hfill
    \begin{subfigure}[b]{0.45\textwidth}
        \includegraphics[width=\textwidth]{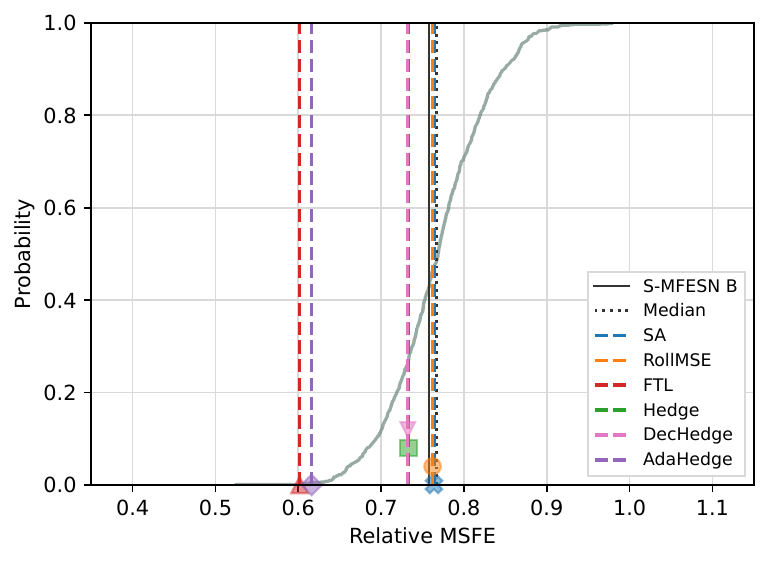}
        \caption{S-MFESN~B}
    \end{subfigure}

    \vspace{0.5cm}

    \begin{subfigure}[b]{0.45\textwidth}
        \includegraphics[width=\textwidth]{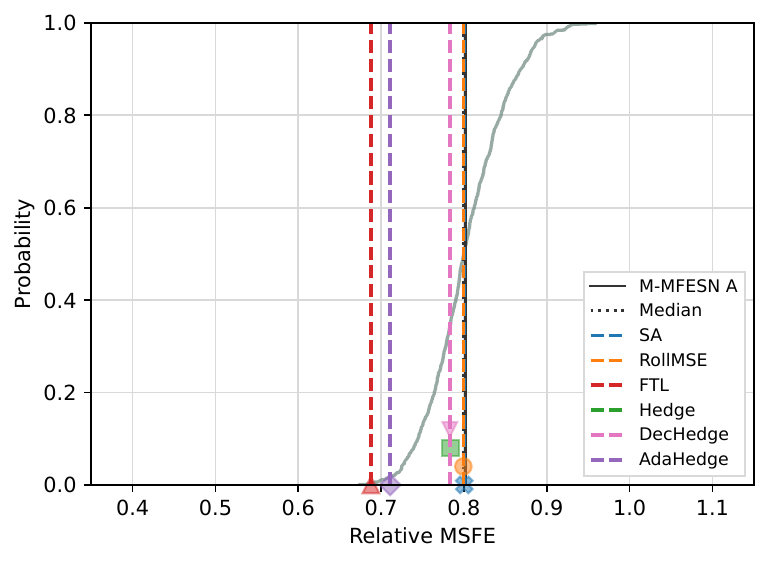}
        \caption{M-MFESN~A}
    \end{subfigure}
    \hfill
    \begin{subfigure}[b]{0.45\textwidth}
        \includegraphics[width=\textwidth]{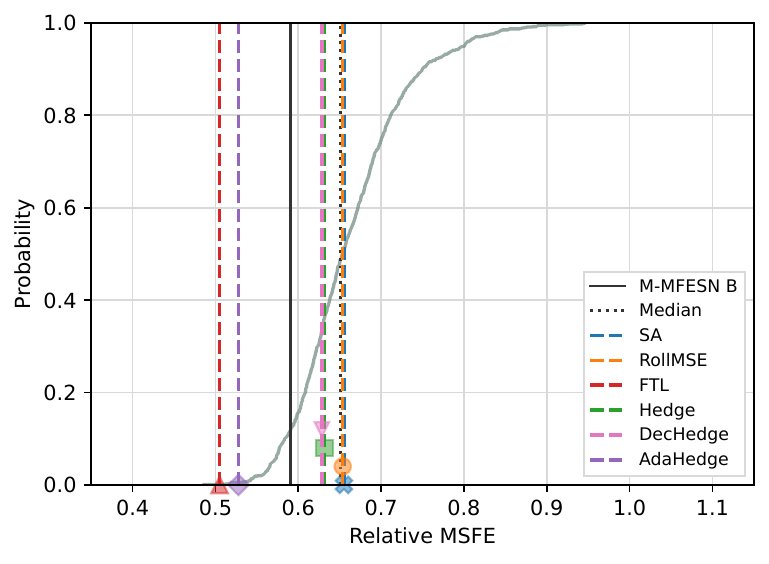}
        \caption{M-MFESN~B}
    \end{subfigure}
    \caption{EN-MFESN-RP ensembles of S-MFESN~A~(a), S-MFESN~B~(b), M-MFESN~A~(c), and M-MFESN~B~(d) models for the Small-MD dataset. Plots display MSFE relative to {\color{brown} the in-sample mean as a} baseline for ensemble aggregation methods (colored dashed lines), the Median forecast-combination performance (black dotted line), specific draws of S-MFESN~A/B and M-MFESN~A/B from \cite{ballarin2022reservoir} (black solid line), and the ECDF of the relative MSFE across individual ensemble models (gray curve). Colored markers are included to distinguish overlapping lines.}
    \label{fig:plot_ecdf_msfe_small_seed}
\end{figure}

\begin{figure}[p!]
    \centering
    \begin{subfigure}[b]{0.45\textwidth}
        \includegraphics[width=\textwidth]{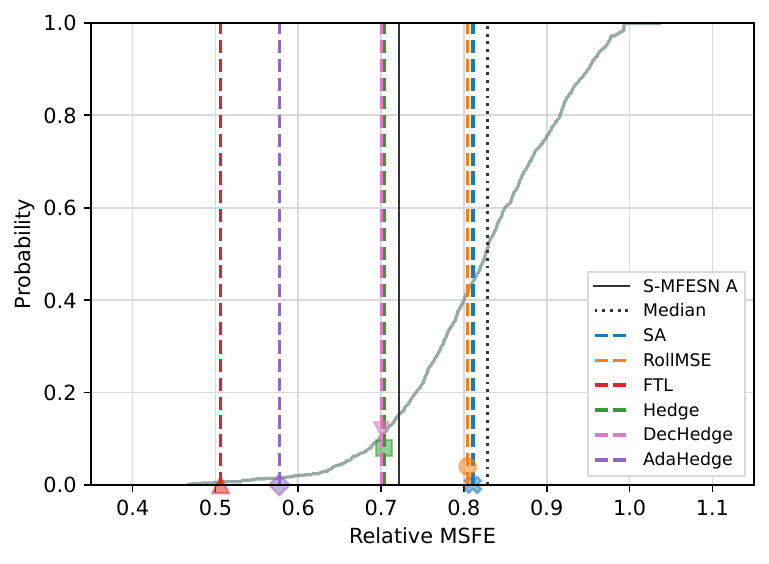}
        \caption{S-MFESN~A}
    \end{subfigure}
    \hfill
    \begin{subfigure}[b]{0.45\textwidth}
        \includegraphics[width=\textwidth]{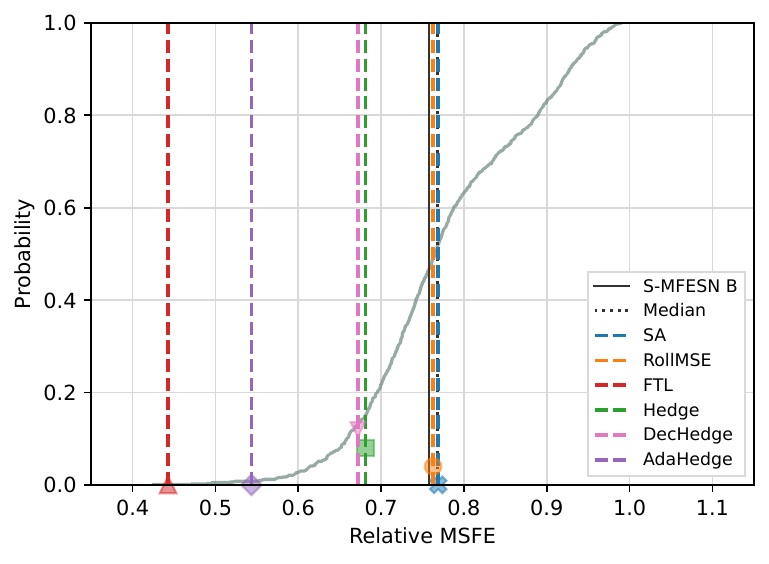}
        \caption{S-MFESN~B}
    \end{subfigure}

    \vspace{0.5cm}

    \begin{subfigure}[b]{0.45\textwidth}
        \includegraphics[width=\textwidth]{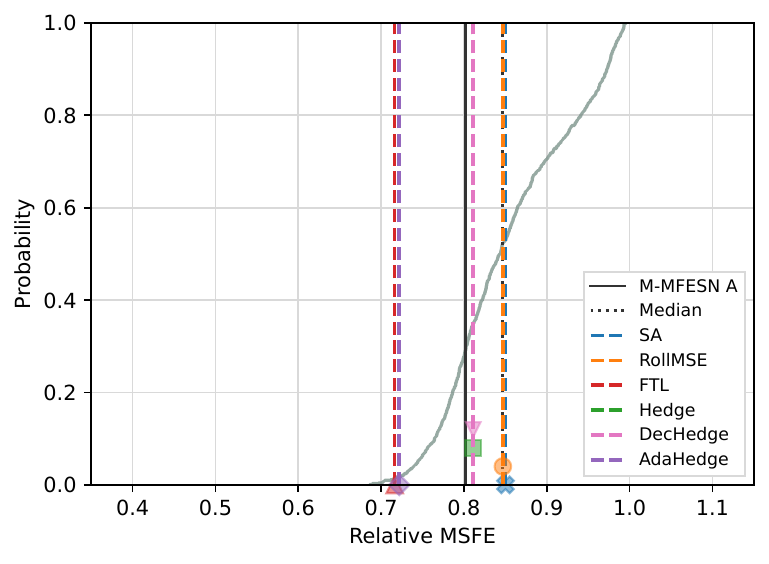}
        \caption{M-MFESN~A}
    \end{subfigure}
    \hfill
    \begin{subfigure}[b]{0.45\textwidth}
        \includegraphics[width=\textwidth]{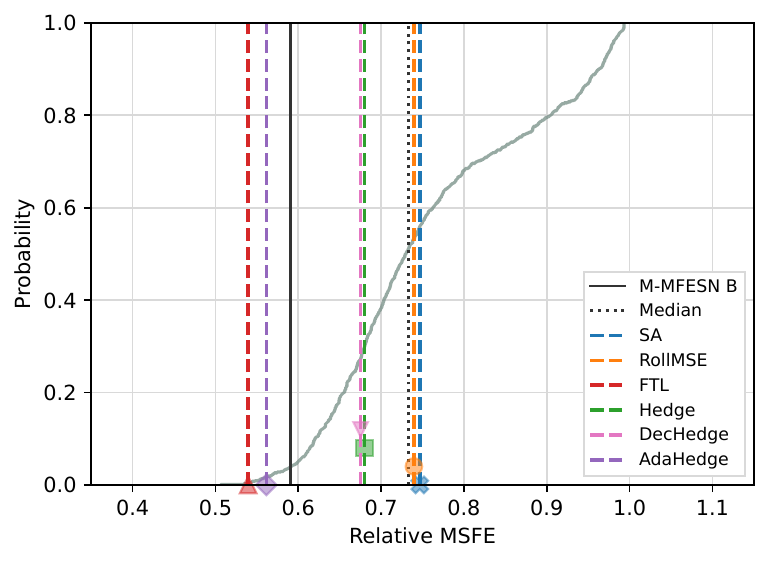}
        \caption{M-MFESN~B}
    \end{subfigure}
    \caption{EN-MFESN-$\alpha$RP ensembles of S-MFESN~A~(a), S-MFESN~B~(b), M-MFESN~A~(c), and M-MFESN~B~(d) models for the Small-MD dataset. Plots display MSFE relative to {\color{brown} the in-sample mean as a} baseline for ensemble aggregation methods (colored dashed lines), the Median forecast-combination performance (black dotted line), specific draws of S-MFESN~A/B and M-MFESN~A/B from \cite{ballarin2022reservoir} (black solid line), and the ECDF of the relative MSFE across individual ensemble models (gray curve). Colored markers are included to distinguish overlapping lines.}
    \label{fig:plot_ecdf_msfe_small_seed_leak}
\end{figure}

\begin{figure}[p!]
    \centering
    \begin{subfigure}[b]{0.45\textwidth}
        \includegraphics[width=\textwidth]{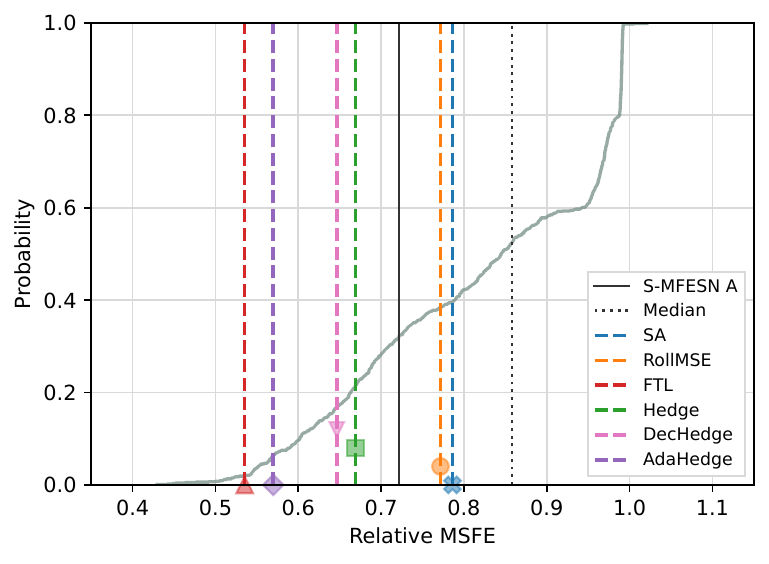}
        \caption{S-MFESN~A}
    \end{subfigure}
    \hfill
    \begin{subfigure}[b]{0.45\textwidth}
        \includegraphics[width=\textwidth]{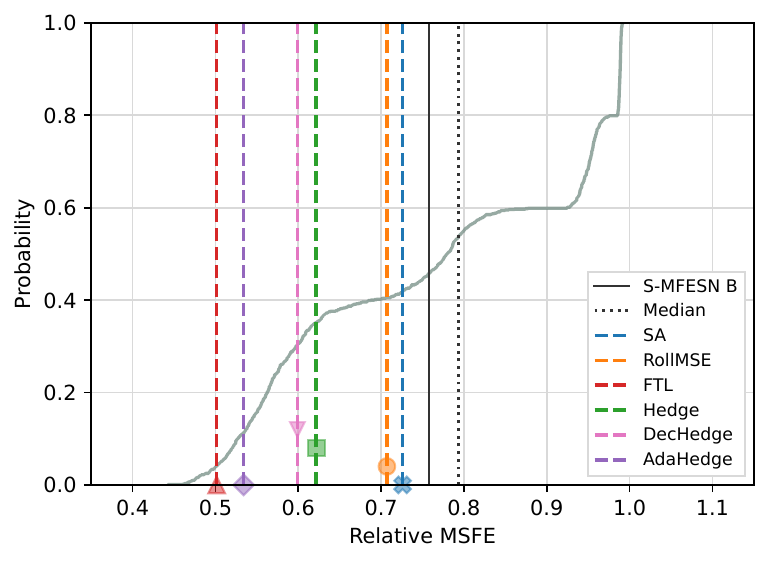}
        \caption{S-MFESN~B}
    \end{subfigure}

    \vspace{0.5cm}

    \begin{subfigure}[b]{0.45\textwidth}
        \includegraphics[width=\textwidth]{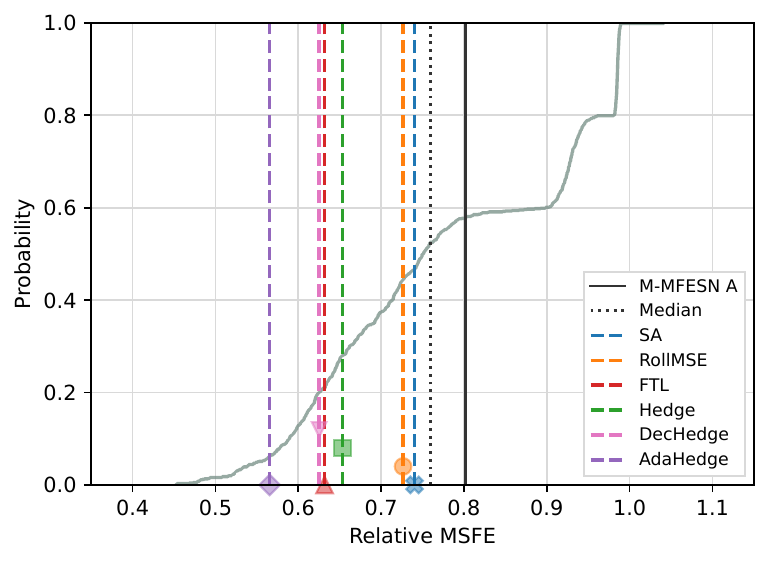}
        \caption{M-MFESN~A}
    \end{subfigure}
    \hfill
    \begin{subfigure}[b]{0.45\textwidth}
        \includegraphics[width=\textwidth]{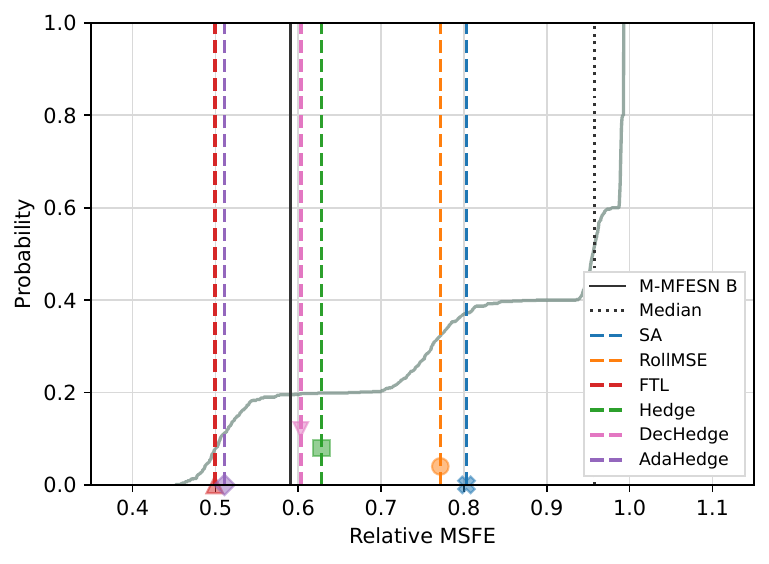}
        \caption{M-MFESN~B}
    \end{subfigure}

    \caption{EN-MFESN-$\lambda$RP ensembles of
    S-MFESN~A~(a), S-MFESN~B~(b), M-MFESN~A~(c), and
    M-MFESN~B~(d) models for the Small-MD dataset. Plots display MSFE relative to {\color{brown} the in-sample mean as a} baseline for ensemble aggregation methods (colored dashed lines), the Median forecast-combination performance (black dotted line), specific draws of S-MFESN~A/B and M-MFESN~A/B from \cite{ballarin2022reservoir} (black solid line), and the ECDF of the relative MSFE across individual ensemble models (gray curve). Colored markers are included to distinguish overlapping lines.}
    \label{figure:small_lambda_common_grid}
\end{figure}

\begin{figure}[p!]
    \centering
    \begin{subfigure}[b]{0.45\textwidth}
        \includegraphics[width=\textwidth]{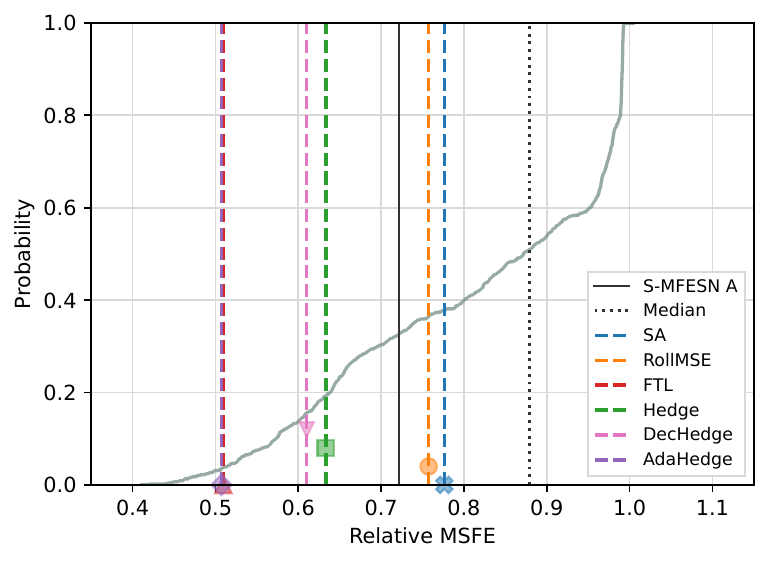}
        \caption{S-MFESN~A}
    \end{subfigure}
    \hfill
    \begin{subfigure}[b]{0.45\textwidth}
        \includegraphics[width=\textwidth]{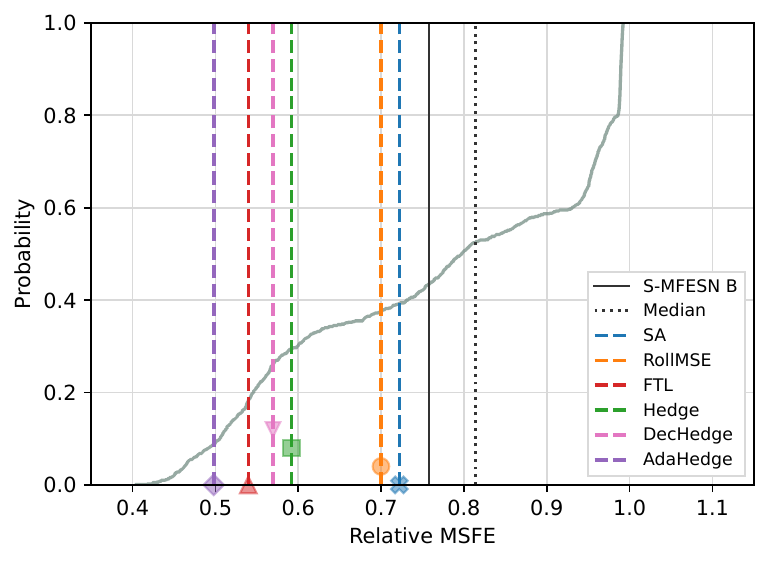}
        \caption{S-MFESN~B}
    \end{subfigure}

    \vspace{0.5cm}

    \begin{subfigure}[b]{0.45\textwidth}
        \includegraphics[width=\textwidth]{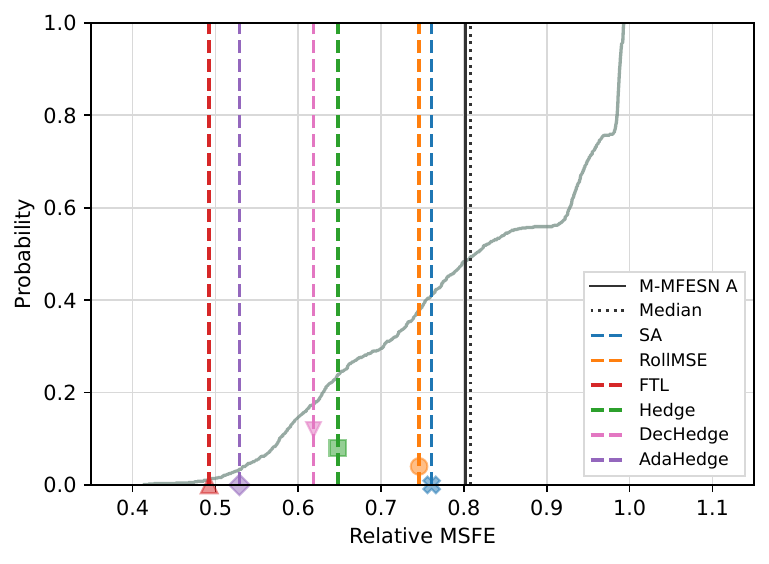}
        \caption{M-MFESN~A}
    \end{subfigure}
    \hfill
    \begin{subfigure}[b]{0.45\textwidth}
        \includegraphics[width=\textwidth]{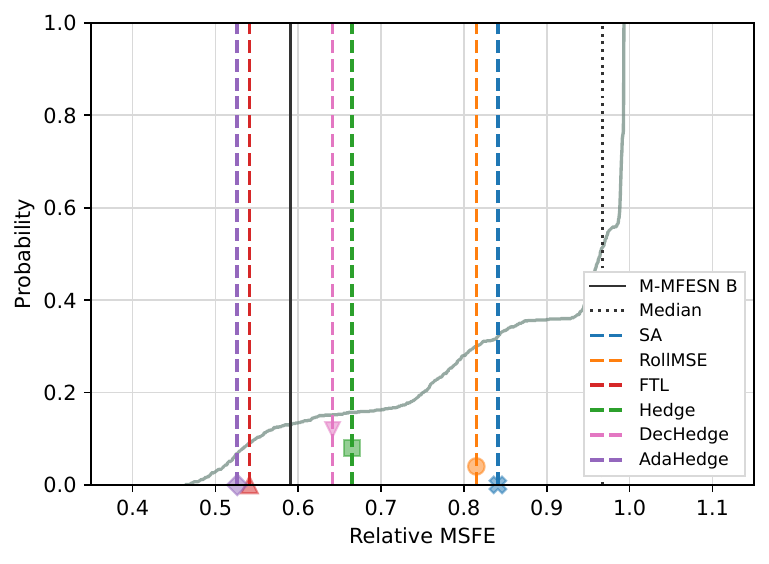}
        \caption{M-MFESN~B}
    \end{subfigure}

    \caption{EN-MFESN-$\alpha\lambda$RP ensembles of
    S-MFESN~A~(a), S-MFESN~B~(b), M-MFESN~A~(c), and
    M-MFESN~B~(d) models for the Small-MD dataset. Plots display MSFE relative to {\color{brown} the in-sample mean as a} baseline for ensemble aggregation methods (colored dashed lines), the Median forecast-combination performance (black dotted line), specific draws of S-MFESN~A/B and M-MFESN~A/B from \cite{ballarin2022reservoir} (black solid line), and the ECDF of the relative MSFE across individual ensemble models (gray curve). Colored markers are included to distinguish overlapping lines.}
    \label{figure:small_alpha_lambda_common_grid}
\end{figure}

\clearpage
\newpage

\subsection{Ensemble Robustness to Expert Resampling}
\label{appendix:robustness_resampling}

We extend the empirical analysis in
Section~\ref{section:application} by repeating the construction
of each MFESN ensemble $B=100$ times and comparing its forecast
combinations with a CV-selected expert from the same pool. We evaluate one-step-ahead forecasting performance for U.S.\ quarterly GDP growth using the medium-MD data sample. The preprocessing, model specifications, and forecast combination are performed identically to the empirical exercises in Section~\ref{section:application}.

We summarize the design in Table~\ref{tab:gdp-expanded-concise-design} and identify each specification by a tuple $\mathfrak{D}=(a,e)$, where $a$ denotes the MFESN model and $e$ the ensemble construction, giving 16 combinations. For each design tuple $\mathfrak{D}$, we repeat the expert construction $B=100$ times, with $K=1000$ experts in each pool. Note that the observations are fixed across replications, and only the expert draws change.  
As in Section~\ref{section:application}, ridge penalties are
tuned separately for each expert in RP and $\alpha$RP. For $\lambda$RP and $\alpha\lambda$RP ensembles, we instead use the grid $\lambda\in \{0.01, 0.1, 1, 10, \allowbreak 100\}$, which is common across all MFESN model specifications and replications. Each reservoir draw is evaluated at all five values of $\lambda$, without per-expert ridge-penalty tuning.

Within each pool, the \textit{CV-selected expert} is the candidate with the lowest validation loss among the $K=1000$ experts. This selection uses the same pre-evaluation CV splits and
criterion as ridge-penalty tuning for RP and $\alpha$RP, following the time-series CV scheme in \cite{ballarin2022reservoir}. The validation design is reported in Table~\ref{tab:gdp-expanded-concise-design}. Thus, CV selects the single-expert benchmark in all four
constructions. The CV-selected expert and all fitted coefficients remain
fixed throughout the evaluation period; no evaluation observations
enter ridge-penalty tuning or expert selection. In fact, for both $\lambda$RP and $\alpha\lambda$RP settings, the CV criterion can readily be evaluated at each candidate's fixed penalty, without further optimization over $\lambda$. This procedure may be seen as a computationally-efficient alternative to per-expert ridge penalty tuning, since it exploits the fact that the experts in the ensemble have to be instantiated regardless of how $\lambda$ is tuned.
In these experiments, too, we use squared-error loss for online weight updates and forecast evaluation. As in Section~\ref{section:application}, experts are
estimated using 1990Q1--2007Q4 and evaluated over the
48 quarters, 2008Q1--2019Q4.

Table~\ref{table:gdp_expanded_ratio_means} shows the relative MSFEs of benchmark and online combination methods with respect to the CV-selected expert. The MSFE of the CV-selected expert with respect to the in-sample mean is reported in parentheses; the table then reports the ratio between each method's mean MSFE (computed across the $B=100$ replications) and the corresponding mean MSFE of the CV-selected expert. We use this relative performance metric to evaluate how well, on average, our online combination schemes compare to a ``direct'' expert selection criterion that does not adapt over the forecasting run and that can be readily computed using only the training and validation splits of the data. 
The results in Table~\ref{table:gdp_expanded_ratio_means} indicate that both FTL and AdaHedge methods robustly outperform the CV-selected expert in all scenarios, while constant and decreasing Hedge appear to also yield MSFE reductions in $\lambda$RP and $\alpha\lambda$RP experiments.

Table~\ref{table:gdp_cv_rank_win_rates} complements these average results by reporting the out-of-sample rank of the CV-selected expert within each pool and the percentage of replications in which FTL and AdaHedge outperform this expert. The CV-selected expert's out-of-sample rank varies substantially across ensemble specifications. FTL achieves a lower MSFE than the CV-selected expert in 61--95\% of replications across the different settings, while the corresponding percentages for AdaHedge range from 50\% to 100\%. These results provide additional evidence that the average improvements reported in Table~\ref{table:gdp_expanded_ratio_means} are not driven solely by a small number of particularly favorable expert-pool realizations.

We further extend the empirical analysis by comparing the
EN-MFESN-$\alpha\lambda$RP ensembles with an
\textit{equal-budget single model} (EB). Unlike the CV-selected expert
in Table~\ref{table:gdp_expanded_ratio_means}, this model is
selected by tuning $\alpha$ and $\lambda$ for the first reservoir
draw of each pool, under an approximately matched computational
budget. For each model specification, we retain the $B=100$ replications,
with $K=1000$ experts in each pool. As in
Table~\ref{tab:gdp-expanded-concise-design}, each ensemble is
constructed from $40$ reservoir draws evaluated at all
combinations of the five leak rates and five ridge penalties.
For the single model, we use a finer $41\times41$ grid, with
equally spaced values of $\alpha\in[0.1,0.9]$ and
$\log_{10}\lambda\in[-2,2]$. All values in the original leak-rate
and ridge-penalty grids are retained. The results in Table~\ref{table:gdp_same_budget_ratio_means} indicate that FTL, Hedge, DecHedge and AdaHedge all yield lower mean MSFE than the equal-budget single model in all specifications.

\clearpage

\providecommand{\raisedrule}[2][0em]{%
    \leavevmode\leaders\hbox{\rule[#1]{1pt}{#2}}\hfill\kern0pt
}

\begin{table}[p!]
\centering
\small
\setlength{\tabcolsep}{3pt}
\setlength{\extrarowheight}{0pt}
\renewcommand{\arraystretch}{1}
\linespread{1.0}\selectfont
\captionsetup{font={small,stretch=1},skip=6pt}

\begin{tabularx}{\linewidth}{
    @{}p{.24\linewidth}
    p{.12\linewidth}
    >{\raggedright\arraybackslash}X@{}
}
    \textbf{Item} & \textbf{Symbol} & \textbf{Specification}\\
    \toprule
    \midrule
    
    \multicolumn{3}{@{}l@{}}{
        \footnotesize\textit{Design tuple} $\mathfrak{D}$
        ~\raisedrule[0.2em]{0.2pt}
    }\\[3pt]
    
    Model specification & $a$ &
    S-MFESN~A, S-MFESN~B, M-MFESN~A, M-MFESN~B.\\
    
    Ensemble construction & $e$ &
    RP: $1000$ draws, candidate-specific tuned $\lambda$;\\
    
    & & $\alpha$RP: $200$ draws $\times$ $5$ leak rates,
    candidate-specific tuned $\lambda$;\\
    
    & & $\lambda$RP: $200$ draws $\times$ $5$ fixed ridge penalties;\\
    
    & & $\alpha\lambda$RP: $40$ draws $\times$ $5$ leak rates
    $\times$ $5$ fixed ridge penalties.\\ & & For all constructions, the same reservoir and input
matrices are used within each draw across the specified
leak rates and ridge penalties; different draws are
independent.\\[2pt]
    
    \multicolumn{3}{@{}l@{}}{
        \footnotesize\textit{Fixed variables}
        ~\raisedrule[0.2em]{0.2pt}
    }\\[3pt]
    
    Replications\newline
    (per tuple $\mathfrak{D}$) & $B$ & $100$\\
    
    Experts per pool & $K$ & $1000$\\
    
    \begin{tabular}[t]{@{}l@{}} Out-of-sample \\ evaluation period \end{tabular} & $T_{\mathrm{test}}$ &
    48 quarters, 2008Q1--2019Q4.\\
    
    Leak-rate grid & -- &
    $\{0.1,0.3,0.5,0.7,0.9\}$ for $\alpha$RP and $\alpha\lambda$RP.\\
    
    Fixed ridge-penalty grid & -- & $\{0.01,0.1,1,10,100\}$ for $\lambda$RP and $\alpha\lambda$RP.\\
    
    Expert selection & $k_{\mathrm{CV}}$ & Candidate minimizing pre-evaluation validation MSFE across the ten folds within each pool of $K=1000$ experts.\\
    Selected-expert update & -- & Selected expert remains fixed throughout the test period.\\
    \begin{tabular}[t]{@{}l@{}} Validation and \\ estimation design \end{tabular} & -- & \textit{Validation:} Following \cite{ballarin2022reservoir},
    ten-fold expanding-window CV on 1975Q1--2007Q4,
    with five validation quarters per split.
    The criterion is mean squared validation error averaged
    across the ten folds, used for candidate-specific
    ridge-penalty tuning in RP and $\alpha$RP and for selecting
    $k_{\mathrm{CV}}$ in all four constructions.
    \newline
    \textit{Estimation:} With penalties tuned or assigned,
    each expert's intercept and readout weights are estimated
    on 1990Q1--2007Q4, as in
    Section~\ref{section:application}, and kept fixed
    throughout out-of-sample forecasting.\\

    \midrule
    \multicolumn{3}{@{}l@{}}{
        \footnotesize\textit{Forecast combination}
        ~\raisedrule[0.2em]{0.2pt}
    }\\[3pt]
    
    Median \& SA & -- &
    No tuning parameters needed.\\
    
    RollMSE & $r,\varepsilon$ &
    $r=4$, $\varepsilon=10^{-9}$;
    initialized with equal weights until four losses are available; see Appendix~\ref{appendix:Direct Combination Methods}.\\
    
    FTL & -- &
    No tuning parameters needed, see Algorithm~\ref{algortihm:ftl}. Numerical ties are determined using absolute and relative tolerances of $10^{-9}$ and $10^{-5}$, respectively. Equal weights among experts tied for the smallest cumulative loss.\\
    
    Hedge & $\eta$ &
    $10^4$; see Algorithm~\ref{algortihm:const_hedge}.\\
    
    DecHedge & $c_0$ &
    $10^4$ with
    {\color{green} $\eta_{t+1}=c_0\sqrt{\log K/(t+1)}$}; see Algorithm~\ref{algortihm:decreasing_hedge}\\
    
    AdaHedge & $\eta_t$ &
    Adaptive rate set via mixability gap; see 
    Algorithm~\ref{algorithm:adahedge}.\\

\bottomrule
\end{tabularx}

\caption{Additional empirical design and forecast combination settings. The combination parameters match those in the empirical exercise of Section~\ref{section:application} (see Tables~\ref{tab:model_list}, \ref{tab:ensemble_class_list}, and \ref{tab:gdp-concise-design}). Model hyperparameters, ridge penalties, and fixed combination parameters are not re-tuned at each forecast origin. }
\label{tab:gdp-expanded-concise-design}
\end{table}

\clearpage

\providecommand{\raisedrule}[2][0em]{\leavevmode\leaders\hbox{\rule[#1]{1pt}{#2}}\hfill\kern0pt}

\begin{table}[t!]
\centering
\small
\setlength{\tabcolsep}{0pt}
\setlength{\extrarowheight}{1pt}
\linespread{1.2}\selectfont

\begin{tabular*}{\textwidth}{
    @{\extracolsep{\fill}}
    >{\centering\arraybackslash}p{2.5cm}
    *{3}{>{\centering\arraybackslash}m{.09\textwidth}}
    >{\centering\arraybackslash}m{.095\textwidth}
    *{2}{>{\centering\arraybackslash}m{.09\textwidth}}
    >{\centering\arraybackslash}m{.105\textwidth}
    >{\centering\arraybackslash}m{.11\textwidth}
    @{}
}
& & \multicolumn{7}{c}{Ensemble} \\
\cline{3-9}
\addlinespace[2pt]
Model & CV & Median & SA & {RollMSE} & FTL & Hedge & {DecHedge} & {AdaHedge} \\
\toprule

\multicolumn{9}{l}{\footnotesize
EN-MFESN-RP (random parameters resampling)~
\raisedrule[0.2em]{0.2pt}} \\[2pt]
S-MFESN~A & (0.771) & 1.218 & 1.192 & 1.188 & \textbf{0.900} & 1.085 & 1.078 & \underline{0.916} \\
S-MFESN~B & (0.734) & 1.204 & 1.192 & 1.189 & \textbf{0.919} & 1.134 & 1.137 & \underline{0.951} \\
M-MFESN~A & (0.732) & 1.200 & 1.192 & 1.190 & \textbf{0.949} & 1.154 & 1.158 & \underline{0.978} \\
M-MFESN~B & (0.601) & 1.223 & 1.216 & 1.211 & \textbf{0.878} & 1.141 & 1.138 & \underline{0.927} \\

\addlinespace[4pt]
\multicolumn{9}{l}{\footnotesize
EN-MFESN-$\alpha$RP (random parameters resampling \& varying leak rates)~
\raisedrule[0.2em]{0.2pt}} \\[2pt]
S-MFESN~A & (0.692) & 1.305 & 1.230 & 1.209 & \underline{0.860} & 1.003 & 0.965 & \textbf{0.823} \\
S-MFESN~B & (0.638) & 1.294 & 1.226 & 1.205 & \textbf{0.817} & 1.029 & 1.010 & \underline{0.837} \\
M-MFESN~A & (0.678) & 1.286 & 1.221 & 1.205 & \textbf{0.807} & 1.004 & 0.977 & \underline{0.814} \\
M-MFESN~B & (0.547) & 1.461 & 1.387 & 1.355 & \textbf{0.910} & 1.177 & 1.152 & \underline{0.952} \\

\addlinespace[4pt]
\multicolumn{9}{l}{\footnotesize
EN-MFESN-$\lambda$RP (random parameters resampling \& varying ridge penalties)~
\raisedrule[0.2em]{0.2pt}} \\[2pt]
S-MFESN~A & (0.815) & 1.115 & 1.015 & 0.991 & 0.895 & 0.851 & \underline{0.808} & \textbf{0.765} \\
S-MFESN~B & (0.756) & 1.116 & 1.074 & 1.035 & \underline{0.856} & 0.911 & 0.891 & \textbf{0.800} \\
M-MFESN~A & (0.696) & 1.126 & 1.059 & 1.023 & 0.985 & 0.909 & \underline{0.884} & \textbf{0.844} \\
M-MFESN~B & (0.688) & 1.393 & 1.147 & 1.094 & \textbf{0.686} & 0.855 & 0.813 & \underline{0.704} \\

\addlinespace[4pt]
\multicolumn{9}{l}{\footnotesize
EN-MFESN-$\alpha\lambda$RP (random parameters resampling \& varying leak rates and ridge penalties)~
\raisedrule[0.2em]{0.2pt}} \\[2pt]
S-MFESN~A & (0.675) & 1.319 & 1.187 & 1.150 & 0.901 & 0.945 & \underline{0.894} & \textbf{0.815} \\
S-MFESN~B & (0.672) & 1.207 & 1.137 & 1.094 & \underline{0.832} & 0.916 & 0.884 & \textbf{0.791} \\
M-MFESN~A & (0.706) & 1.137 & 1.046 & 1.024 & 0.876 & 0.881 & \underline{0.852} & \textbf{0.806} \\
M-MFESN~B & (0.549) & 1.763 & 1.509 & 1.452 & \textbf{0.894} & 1.140 & 1.086 & \underline{0.909} \\

\addlinespace[2pt]
\bottomrule
\end{tabular*}

\caption{
Relative MSFE of quarterly U.S.\ GDP growth predictions.
The CV column reports the relative MSFE of the CV-selected
expert with respect to the in-sample mean, with values shown
in parentheses. The remaining columns report relative MSFE
of the forecast combinations with respect to the CV-selected
expert. All relative MSFEs are calculated as ratios of mean
MSFEs across replications. For each design tuple $\mathfrak{D}=(a,e)$ in
Table~\ref{tab:gdp-expanded-concise-design}, means are
calculated separately across $B=100$ replications, each
with a pool of $K=1000$ experts. For $\lambda$RP and $\alpha\lambda$RP, the fixed ridge
penalties are $\lambda\in \Lambda$,
common to all model specifications and replications.
The CV-selected expert minimizes pre-evaluation validation
loss within its pool. Best performing combinations are displayed in bold. Second-best performing combinations are underlined.
}
\label{table:gdp_expanded_ratio_means}
\end{table}

\clearpage

\begin{table}[t!]
\centering
\small
\setlength{\tabcolsep}{0pt}
\setlength{\extrarowheight}{1pt}
\linespread{1.2}\selectfont

\begin{tabular*}{\textwidth}{@{\extracolsep{\fill}}
    >{\centering\arraybackslash}m{2.5cm}
    >{\centering\arraybackslash}m{1.8cm}
    >{\centering\arraybackslash}m{4.3cm}
    *{2}{>{\centering\arraybackslash}m{2.2cm}}@{}}
\shortstack{Model\\specification} & Ensemble &
\shortstack{CV OOS rank\\median [IQR]} &
\shortstack{FTL\\wins(\%)} &
\shortstack{AdaHedge\\wins (\%)}\\
\toprule

\multirow{4}{*}{S-MFESN~A}
& RP                   & $22\,[9,\,65]$      & 74 & 78 \\
& $\alpha$RP           & $68\,[36,\,136]$    & 78 & 87 \\
& $\lambda$RP          & $237\,[179,\,297]$  & 76 & 98 \\
& $\alpha\lambda$RP    & $100\,[49,\,223]$   & 67 & 83 \\

\addlinespace[4pt]
\multirow{4}{*}{S-MFESN~B}
& RP                   & $13\,[6,\,40]$      & 83 & 72 \\
& $\alpha$RP           & $119\,[83,\,184]$   & 95 & 95 \\
& $\lambda$RP          & $228\,[175,\,310]$  & 89 & 100 \\
& $\alpha\lambda$RP    & $235\,[101,\,306]$  & 77 & 87 \\

\addlinespace[4pt]
\multirow{4}{*}{M-MFESN~A}
& RP                   & $6\,[2,\,14]$       & 65 & 50 \\
& $\alpha$RP           & $122\,[89,\,144]$   & 95 & 94 \\
& $\lambda$RP          & $261\,[202,\,352]$  & 61 & 100 \\
& $\alpha\lambda$RP    & $311\,[260,\,354]$  & 86 & 97 \\

\addlinespace[4pt]
\multirow{4}{*}{M-MFESN~B}
& RP                   & $29\,[12,\,69]$     & 95 & 87 \\
& $\alpha$RP           & $67\,[42,\,105]$    & 81 & 69 \\
& $\lambda$RP          & $222\,[203,\,262]$  & 95 & 91 \\
& $\alpha\lambda$RP    & $118\,[87,\,140]$   & 80 & 77 \\

\addlinespace[2pt]
\bottomrule
\end{tabular*}

\caption{
Out-of-sample ranks of the CV-selected expert and MSFE win
rates for quarterly U.S.\ GDP growth predictions. The experimental design, expert pools and CV-selected experts
are identical to those used in
Table~\ref{table:gdp_expanded_ratio_means}, with $B=100$
replications and $K=1000$ experts per pool.
Within each pool, experts are ranked by out-of-sample MSFE,
with rank $1$ denoting the lowest MSFE.
The CV OOS rank column reports the median and the 25th and
75th percentiles in brackets of the CV-selected expert's rank
across $B=100$ replications, rounded to the nearest integer.
The remaining columns report the percentages of replications
in which FTL and AdaHedge yield strictly lower MSFE than
the CV-selected expert from the same pool.
}
\label{table:gdp_cv_rank_win_rates}
\end{table}

\clearpage

\providecommand{\raisedrule}[2][0em]{\leavevmode\leaders\hbox{\rule[#1]{1pt}{#2}}\hfill\kern0pt}

\begin{table}[t!]
\centering
\small
\setlength{\tabcolsep}{0pt}
\setlength{\extrarowheight}{1pt}
\linespread{1.2}\selectfont

\begin{tabular*}{\textwidth}{
    @{\extracolsep{\fill}}
    >{\centering\arraybackslash}p{2.5cm}
    *{6}{>{\centering\arraybackslash}m{.09\textwidth}}
    >{\centering\arraybackslash}m{.105\textwidth}
    >{\centering\arraybackslash}m{.11\textwidth}
    @{}
}
    & & \multicolumn{7}{c}{Ensemble} \\
    \cline{3-9}
    \addlinespace[2pt]
    Model & EB & Median & SA & {RollMSE}
    & FTL & Hedge & {DecHedge} & {AdaHedge} \\
    \toprule
    
    \multicolumn{9}{l}{\footnotesize
    EN-MFESN-$\alpha\lambda$RP (random parameters resampling \& varying leak rates and ridge penalties)~
    \raisedrule[0.2em]{0.2pt}} \\[2pt]
    S-MFESN~A & (0.841) & 1.058 & 0.951 & 0.922 & 0.723 & 0.757 & \underline{0.717} & \textbf{0.653} \\
    S-MFESN~B & (0.719) & 1.127 & 1.062 & 1.022 & \underline{0.776} & 0.855 & 0.825 & \textbf{0.739} \\
    M-MFESN~A & (0.830) & 0.967 & 0.890 & 0.871 & 0.745 & 0.749 & \underline{0.725} & \textbf{0.686} \\
    M-MFESN~B & (0.695) & 1.392 & 1.191 & 1.147 & \textbf{0.706} & 0.900 & 0.857 & \underline{0.717} \\
    
    \addlinespace[2pt]
    \bottomrule
\end{tabular*}

\caption{
Relative MSFE of quarterly U.S.\ GDP growth predictions. The EB column reports the relative MSFE of the equal-budget single model with respect to the in-sample mean, with values shown in parentheses. The remaining columns report relative MSFE of the forecast combinations with respect to the equal-budget single model. All relative MSFEs are calculated as ratios of mean MSFEs across $B=100$ replications, each with a pool of $K=1000$ experts. Best performing combinations are displayed in bold. Second-best performing combinations are underlined.
}
\label{table:gdp_same_budget_ratio_means}
\end{table}

\clearpage
\newpage

\color{black}\renewcommand{\normalcolor}{\color{black}}

\subsection{Combination Weight Dynamics}
\label{appendix:Combination Weights Dynamics}

An important question regarding our ensembles is whether the weight vectors associated with different combination strategies evolve in a stable manner over time, or instead show significant variation throughout the online learning run.
In other words, we investigate whether the combination weights stabilize over the forecasting period. Here, we focus on type-A S-MFESN and M-MFESN models within the $\alpha$RP ensembles, since Figure~\ref{fig:plot_ecdf_msfe_medium_seed_leak} shows that these specifications have the most curvature in their MSFE ECDFs. Figure~\ref{fig:plot_weights_type_B} provides similar weight plots for the B-type S- and M-MFESN $\alpha$RP ensembles.

Our findings indicate that weight dynamics can differ significantly across different classes of MFESN ensembles.
In the upper row of Figure~\ref{fig:plot_weights_type_A}, both AdaHedge and FTL weights for the S-MFESN~A ensembles change substantially over the forecasting exercise. For AdaHedge (Panel~(a) in Figure~\ref{fig:plot_weights_type_A}), the weight vector $\bm{\omega}_{\textsc{ah},t}$ remains very close to its equal-weight initialization until the second quarter of 2010. After that, a clearer ranking starts to emerge slowly, but fails to stabilize. A similar pattern is visible for the Follow-the-Leader (Panel~(b) of Figure~\ref{fig:plot_weights_type_A}), whose weights $\bm{\omega}_{\textsc{ftl},t}$ fluctuate heavily until 2012Q2, after which the leader remains unchanged.
By contrast, the dynamics of M-MFESN~A ensembles' weights are more stable. Panels~(c) and~(d) of Figure~\ref{fig:plot_weights_type_A} demonstrate that AdaHedge and FTL weights settle quickly and change slowly over time. For AdaHedge, the weights $\bm{\omega}_{\textsc{ah},t}$ stabilize after the end of the Great Financial Crisis downturn. Over 80\% of the cumulative weight concentrates on just 5 experts, and the top-ranked expert receives over 30\% of the weight as early as 2009Q1. Furthermore, the weights  $\bm{\omega}_{\textsc{ftl},t}$ of the FTL demonstrate quick convergence to a stable leader within less than one year from the start of the prediction exercise, indicating that the same M-MFESN~A expert maintains the lowest cumulative observed loss throughout the remaining forecasting period.

It is important to emphasize that these marked differences in weight dynamics do not imply large differences in overall forecasting performance. Table~\ref{table:msfe_medium_results_all} in the main text shows that these differences in weight dynamics do not translate directly into differences in forecasting performance. In particular, AdaHedge achieves similar relative MSFE for S-MFESN~A and M-MFESN~A ensembles, despite their markedly different weight dynamics.
From an interpretability perspective, more stable weights provide clearer insight into the mechanism of the ensemble: Figure~\ref{fig:plot_weights_type_A}(d) strongly suggests an online form of model selection carried out iteratively during the forecasting exercise. 
An open question is hence whether there exist inherent trade-offs between ensemble ``interpretability'' (e.g., provably simple weight dynamics for model instances from given classes of nonlinear models) and predictive performance. We leave the analysis of this question to future work.

\begin{figure}[p!]
    \centering
    \begin{subfigure}[b]{0.49\textwidth}
        \includegraphics[width=\textwidth]{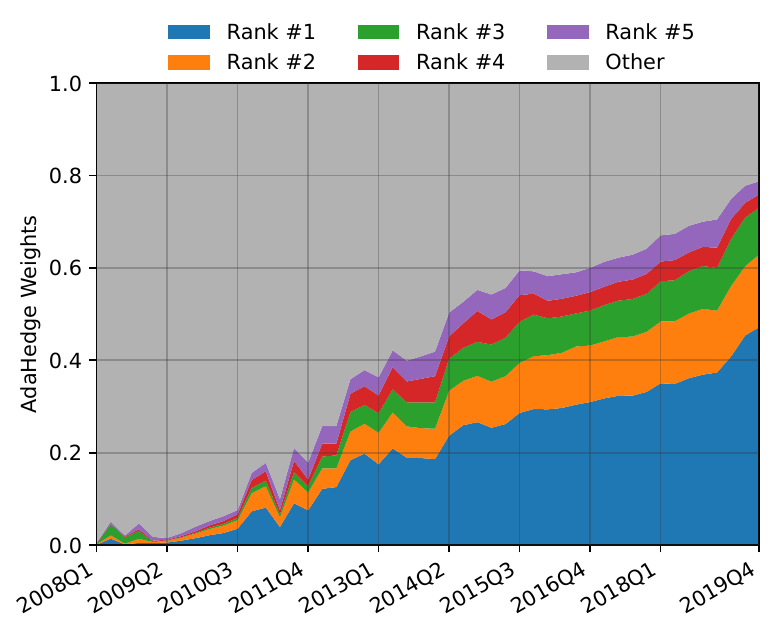}
        \caption{S-MFESN~A -- AdaHedge}
    \end{subfigure}
    \hfill
    \begin{subfigure}[b]{0.49\textwidth}
        \includegraphics[width=\textwidth]{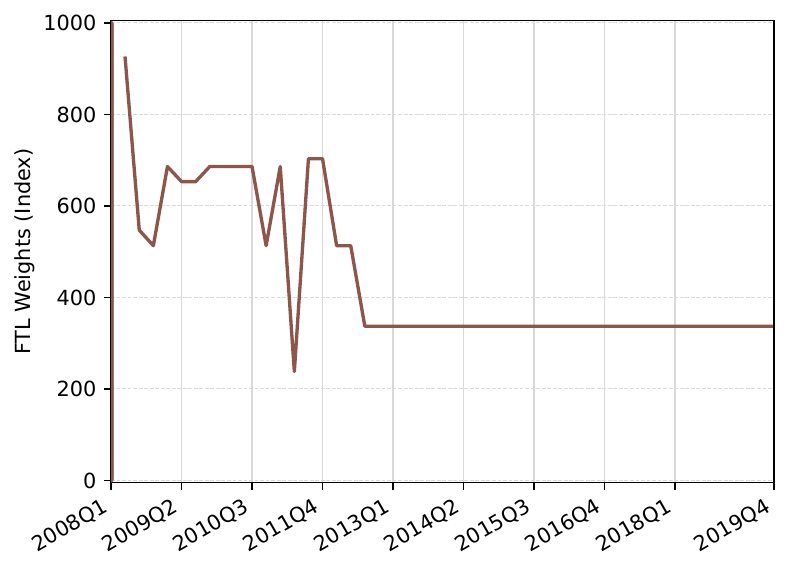}
        \caption{S-MFESN~A -- Follow-the-Leader}
    \end{subfigure}
    
    \vspace{0.5cm} %

    \begin{subfigure}[b]{0.49\textwidth}
        \includegraphics[width=\textwidth]{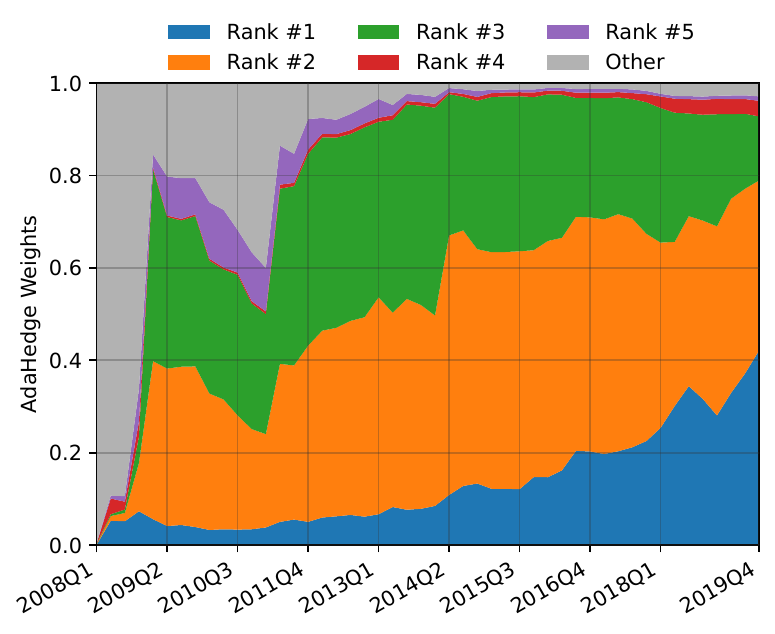}
        \caption{{\color{green}M-MFESN~A -- AdaHedge}}
    \end{subfigure}
    \hfill
    \begin{subfigure}[b]{0.49\textwidth}
        \includegraphics[width=\textwidth]{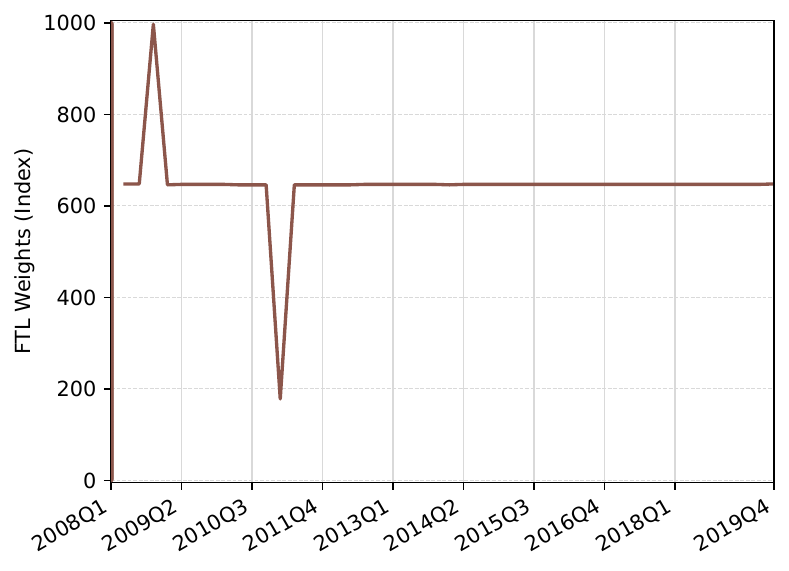}
        \caption{{\color{green}M-MFESN~A -- Follow-the-Leader}}
    \end{subfigure}
    \caption{EN-MFESN-$\alpha$RP ensemble: AdaHedge weights (left column) and FTL leader indices (right column) over the forecasting interval. For FTL, the index $\textsc{ftl}_t$ of the leader model is shown. }
    \label{fig:plot_weights_type_A}
\end{figure}

\begin{figure}[p!]
    \centering
    \begin{subfigure}[b]{0.49\textwidth}
        \includegraphics[width=\textwidth]{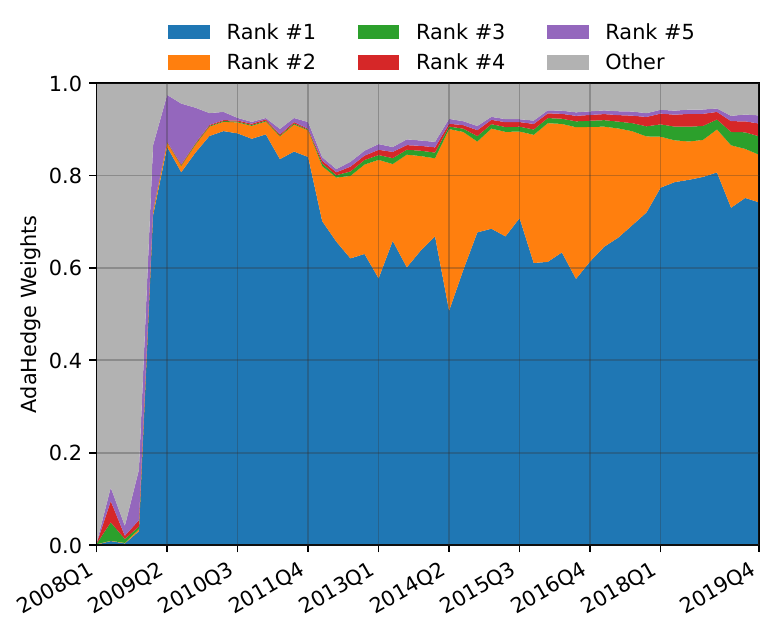}
        \caption{S-MFESN~B -- AdaHedge}
    \end{subfigure}
    \hfill
    \begin{subfigure}[b]{0.49\textwidth}
        \includegraphics[width=\textwidth]{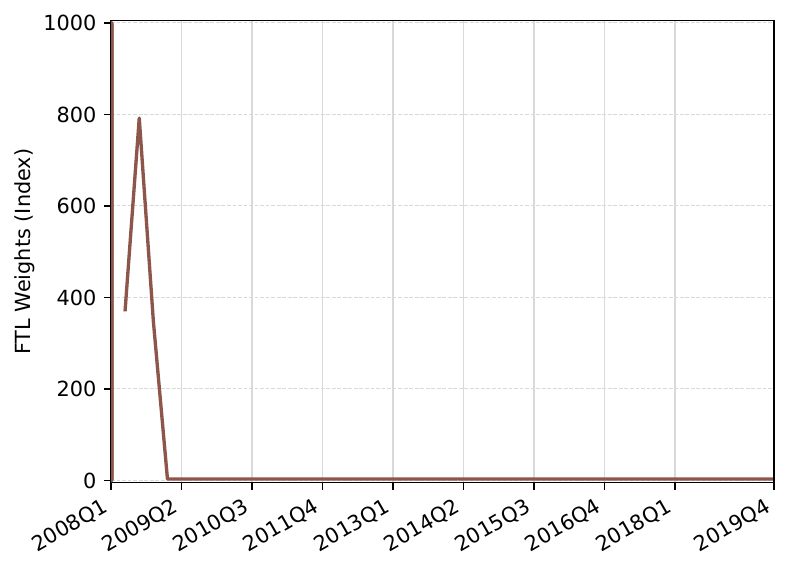}
        \caption{S-MFESN~B -- Follow-the-Leader}
    \end{subfigure}
    
    \vspace{0.5cm} %

    \begin{subfigure}[b]{0.49\textwidth}
        \includegraphics[width=\textwidth]{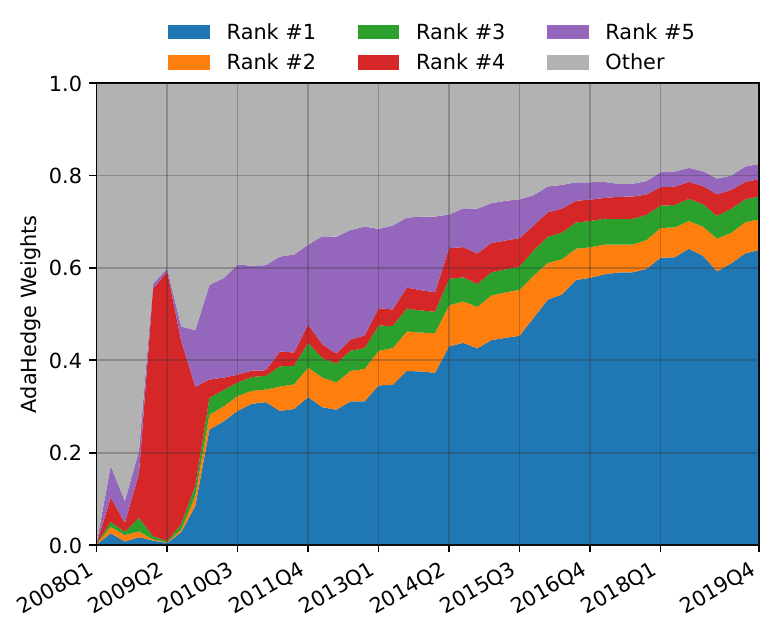}
        \caption{M-MFESN~B -- AdaHedge}
    \end{subfigure}
    \hfill
    \begin{subfigure}[b]{0.49\textwidth}
        \includegraphics[width=\textwidth]{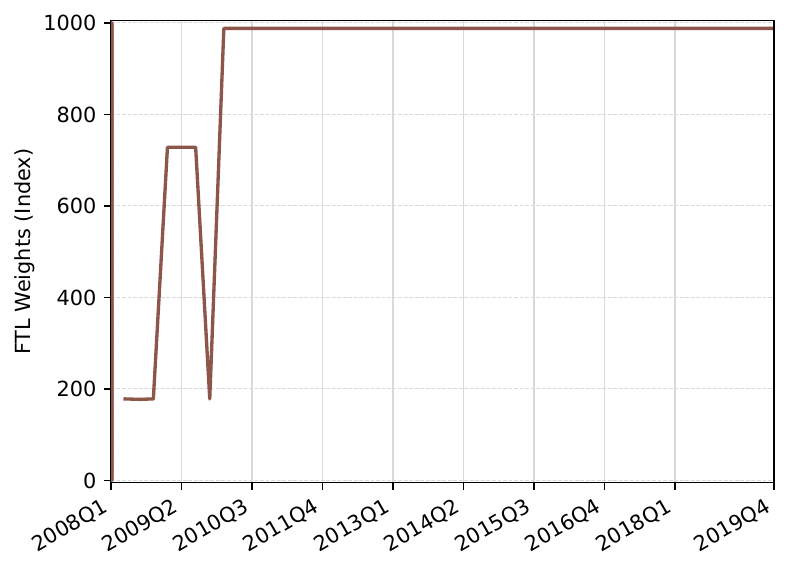}
        \caption{M-MFESN~B -- Follow-the-Leader}
    \end{subfigure}
    \caption{EN-MFESN-$\alpha$RP ensemble: AdaHedge weights (left column) and FTL leader indices (right column) over the forecasting interval. For FTL, the index $\textsc{ftl}_t$ of the leader model is shown. }
    \label{fig:plot_weights_type_B}
\end{figure}

\clearpage
\newpage

\color{green}\renewcommand{\normalcolor}{\color{green}}

\section{Simulations}
\label{appendix:simulations}

In this appendix, we provide simulation evidence to support our theoretical derivations and empirical findings under controlled DGPs. 
First, we construct an ensemble forecasting scenario that is linear and Gaussian, where experts are constructed via random sampling of a subset of the regressors (the size of which we vary). We consider different dependence and expert settings, ensemble sizes and sample sizes, and evaluate performance over 500 Monte Carlo replications.
Then, we study explicitly how robust Follow-the-Leader and decreasing Hedge online aggregation methods are with respect to the sub-optimality gap by direct simulation of the loss process. This simulation design allows us to explicitly compute the regret bounds obtained in Section~\ref{section:expert_ensembles}, and we compare them to the mean empirical regrets we obtain via Monte Carlo simulation. We consider different scenarios with $\Delta$ ranging from $0.01$ to $0.30$ and also with different persistence.

\subsection{Combination with Random-subset Experts}
\label{appendix:simulation_random_subset}

We study forecast combination using random-subset linear regression models. Each expert is defined in terms of a randomly sampled subset of indices, $s_k \subseteq \{1, \ldots, J\}$, from a total $J \in \N$ variables that are observed and can be used for prediction. This randomization approach for forecasting has been studied by \cite{bootForecastingUsingRandom2019}, among others; the randomization approach more generally is now well-studied in econometrics, see e.g. \cite{casarinRandomProjectionMethods2021,dinhRandomSubspaceLocal2024}.  
In our Monte Carlo simulations, we examine ensembles that either include or exclude an \textit{oracle expert}, which uses all
30 predictors and population coefficients, and, therefore, is both optimal in the MSE sense and requires no parameter estimation.

\paragraph{Data-generating process.} We consider a linear data-generating process with $J=30$ predictors: 
\begin{align*}
    X_{t,j} &= \rho_jX_{t-1,j}+\xi_{t,j},
    \quad \xi_{t,j}\sim N\!\left(0,\sigma_{\xi,j}^2\right),
    \quad j=1,\ldots,J, \\
    Y_t &= \sum_{j=1}^{J} \beta_j X_{t,j}+\varepsilon_t,
    \quad \varepsilon_t\sim N(0,1),
\end{align*}
where $\sigma_{\xi,j}^2=v_X(1-\rho_j^{2})$,  $v_X={4}/({1-0.9^2})$, and $\beta_j={1}/{J}$ for $j=1,\ldots,J$. All innovations are mutually independent across time and predictors. 
Each random expert $k$ is defined by a subset $s_k \subseteq \{1, \ldots, J\}$, and yields the forecast
\begin{equation*}
    \widehat{Y}^{(k)}_t = \widehat{\beta}_0^{(k)} +  \sum_{j \in s_k} \widehat{\beta}^{(k)}_j X_{t,j} ,
\end{equation*}
where coefficients $\widehat{\beta}_0^{(k)}$ and $\{ \widehat{\beta}^{(k)}_j \}$ are estimated by OLS using the training sample $T_{\text{train}}$. The subsets $s_k$ are drawn without replacement within each expert, with their cardinalities determined by the subset regime in Table~\ref{tab:mc-design}.

We consider two persistence specifications. In the base DGP, $\rho_j\in\{0.1,0.5,0.9\}$, with ten predictors at each persistence value. In the robustness DGP, $\rho_j\in\{0.05,0.25,0.45\}$, again with ten predictors at each level. By construction, $\operatorname{Var}(X_{t,j})=v_X$ in both designs.\footnote{The predictor processes are initialized at zero; 200 pre-sample observations are used and discarded subsequently  to reduce dependence on the initial condition.}

Forecasts of $Y_t$ use the predictors $X_t$ available at the time of prediction.

\begin{table}[!t]
    \centering
    \small
    \setlength{\tabcolsep}{3pt}
    \setlength{\extrarowheight}{1pt}
    \linespread{1.1}\selectfont

    \begin{tabularx}{\linewidth}{@{}p{.33\linewidth}p{.15\linewidth}>{\raggedright\arraybackslash}X@{}}
    \toprule
    Experiment setting & Symbol & Specification \\
    \midrule

    \multicolumn{3}{@{}l@{}}{\footnotesize\textit{Design tuple} $\mathfrak{D}$~\raisedrule[0.2em]{0.2pt}} \\[1pt]

    Oracle inclusion
        & \textsc{oracle}
        & Y: included; N: excluded.  \\[3pt]

    DGP specification
        & \textsc{dgp}
        & Base: $\rho_j \in \{0.1,0.5,0.9\}$, 10 predictors per persistence value; \newline Robust: $\rho_j \in \{0.05,0.25,0.45\}$, 10 predictors per persistence value. \vspace{0.2cm} \\

    Subset regime
        & \textsc{sr}
        & Small: 5 predictors per expert; \newline Mixed: $p/5$ experts at each subset size 5, 10, 15, 20, and 25; \newline Large: 29 predictors per expert. \vspace{0.2cm}\\

    \makecell[l]{Number of random experts}
        & $p$
        & 10, 100, 200 \\[3pt]

    Training sample length
        & $T_{\text{train}}$
        & 100, 300, 500   \\[7pt]

    \multicolumn{3}{@{}l@{}}{\footnotesize\textit{Fixed quantities}~\raisedrule[0.2em]{0.2pt}} \vspace{0.2cm}\\

Number of predictors
        & $J$
        & 30 \\[3pt]

    Test sample length
        & $T_{\text{test}}$
        & 1000 \\

    \makecell[l]{MC replications}
        & $B$
        & 500 per design tuple $\mathfrak{D}$\\

    \bottomrule
    \end{tabularx}
    
    \caption{Monte Carlo design for simulations with random-subset experts.}
    \label{tab:mc-design}
\end{table}

\paragraph{Simulation design.}
The simulation design varies along five dimensions summarized in Table~\ref{tab:mc-design}: oracle inclusion, DGP persistence, subset regime, number of random experts, and training sample length. We denote a design by $\mathfrak{D} := (\textsc{oracle}, \textsc{dgp}, \textsc{sr}, p, T_{\text{train}})$ yielding $108$ design types in total.

For each design $\mathfrak{D}$ and Monte Carlo replication, we simulate   $T = T_{\text{train}} + T_{\text{test}}$ observations and construct $p$ random experts. Each expert is estimated by OLS on the first $T_{\text{train}}$ observations and remains fixed over the subsequent test sample of length $ T_{\text{test}}$. In the Small and Large regimes, each expert uses, respectively, $5$ and $29$ randomly selected predictors. In the Mixed regime, $p/5$ experts are constructed at each subset size $5,10,15,20$, and $25$.  When $\textsc{oracle}=\text{Y}$, we additionally include the correctly specified model containing all $J=30$ predictors, so that $K = p + \1{ \textsc{oracle} = \text{Y} }$. 

The aggregation methods are evaluated over the subsequent $T_{\mathrm{test}}=1000$ observations. At each evaluation time $t$, weights depend only on expert losses observed before $t$. The squared loss is used both for online weight updates and forecast evaluation, and performance is measured by out-of-sample MSFE; these choices match the combination and evaluation losses used in Section~\ref{section:application}. Each design $\mathfrak{D}$ is replicated  $B = 500$ times.

\paragraph{Results.}

Tables~\ref{tab:mc-raw-results-oracle} and \ref{tab:mc-raw-results-no_oracle} report mean relative MSFE with respect to the (ex-post) realized best expert, with and without the oracle, respectively.
When the oracle is included, the online expert methods substantially outperform the benchmarks---i.e., the direct combination schemes and the median---in the small and mixed subset regimes, as well as at short sample sizes, as shown in Table~\ref{tab:mc-raw-results-oracle}. This difference is especially pronounced in the Small regime (Panel~B), where the benchmarks' relative MSFE remains around $1.43$--$1.54$, while FTL and Hedge-type schemes stay close to one. The advantage is smaller in the Mixed regime (Panel~C) and, in the Large regime (Panel~D), decreases as $T_{\mathrm{train}}$ increases. With respect to average performance across subset regimes (Panel~A), all online methods remain close to the realized best expert, while the benchmarks yield uniformly higher MSFE.  
Excluding the oracle changes this behavior, as shown in Table~\ref{tab:mc-raw-results-no_oracle}, and the differences between benchmarks and online expert combination schemes become much smaller. In the Small regime (Panel~B), the benchmarks perform better, with rolling MSE and simple average winning out. In the Mixed regime (Panel C), SA and RollMSE perform best for the shortest training sample, whereas AdaHedge and decreasing Hedge achieve the lowest relative MSFEs for the longer training samples. Lastly, in the Large regime (Panel~D), the differences between methods are generally small except for  $T_{\text{train}} = 100$, and decreasing Hedge achieves the lowest MSE by a minor margin. 
Overall, comparison of Tables~\ref{tab:mc-raw-results-oracle} and \ref{tab:mc-raw-results-no_oracle} shows that, in the linear Gaussian setting, online aggregation exhibits a more pronounced performance advantage over benchmark combination methods when the oracle is included, consistent with the presence of a distinctly strong expert in the ensemble.

\clearpage
\begin{table}[t!]
\centering
\footnotesize
\setlength{\tabcolsep}{0pt}
\setlength{\extrarowheight}{1pt}
\renewcommand{\arraystretch}{1}
\linespread{1}\selectfont

\begin{tabular*}{\linewidth}{@{\extracolsep{\fill}}>{\centering\arraybackslash}m{.28\linewidth}*{6}{>{\centering\arraybackslash}m{.12\linewidth}}@{}}
\toprule
& \multicolumn{3}{c}{Base DGP}
& \multicolumn{3}{c}{Robust DGP} \\
\cmidrule(lr){2-4}\cmidrule(l){5-7}
Method / $T_{\mathrm{train}}$
& $100$ & $300$ & $500$ & $100$ & $300$ & $500$ \\
\midrule

\multicolumn{7}{@{}l@{}}{\footnotesize
Panel~A: \textit{Average across subset regimes}
~\raisedrule[0.2em]{0.2pt}} \\
\addlinespace[0pt]
Median & 1.553 & 1.285 & 1.254 & 1.419 & 1.271 & 1.248 \\
\addlinespace[2pt]
SA & 1.510 & 1.271 & 1.243 & 1.390 & 1.258 & 1.238 \\
RollMSE & 1.443 & 1.248 & 1.222 & 1.372 & 1.245 & 1.225 \\
\addlinespace[4pt]
FTL & 1.011 & 1.010 & 1.010 & 1.010 & 1.010 & 1.010 \\
Hedge & 1.009 & 1.009 & \ul{1.009}
& 1.009 & \ul{1.009} & \ul{1.009} \\
DecHedge & \textbf{1.009} & \ul{1.009} & 1.009
& \textbf{1.009} & 1.009 & 1.009 \\
AdaHedge & \ul{1.009} & \textbf{1.009} & \textbf{1.009} & \ul{1.009} & \textbf{1.009} & \textbf{1.009} \\

\addlinespace[2pt]
\multicolumn{7}{@{}l@{}}{\footnotesize
Panel~B: \textit{Small subset regime}
~\raisedrule[0.2em]{0.2pt}} \\
\addlinespace[0pt]
Median & 1.544 & 1.502 & 1.493 & 1.513 & 1.492 & 1.487 \\
\addlinespace[2pt]
SA & 1.516 & 1.478 & 1.470 & 1.487 & 1.469 & 1.465 \\
RollMSE & 1.458 & 1.436 & 1.432 & 1.461 & 1.446 & 1.442 \\
\addlinespace[4pt]
FTL & 1.011 & \ul{1.009} & \ul{1.008}
& 1.010 & \ul{1.009} & \ul{1.008} \\
Hedge & \ul{1.010} & 1.009 & 1.009
& 1.010 & 1.010 & 1.009 \\
DecHedge & \textbf{1.009} & \textbf{1.008} & \textbf{1.008}
& \textbf{1.008} & \textbf{1.008} & \textbf{1.008} \\
AdaHedge & 1.010 & 1.009 & 1.009 & \ul{1.010} & 1.009 & 1.009 \\

\addlinespace[2pt]
\multicolumn{7}{@{}l@{}}{\footnotesize
Panel~C: \textit{Mixed subset regime}
~\raisedrule[0.2em]{0.2pt}} \\
\addlinespace[0pt]
Median & 1.334 & 1.216 & 1.196 & 1.283 & 1.206 & 1.191 \\
\addlinespace[2pt]
SA & 1.324 & 1.211 & 1.192 & 1.275 & 1.202 & 1.188 \\
RollMSE & 1.283 & 1.187 & 1.169 & 1.267 & 1.188 & 1.172 \\
\addlinespace[4pt]
FTL & 1.015 & 1.014 & 1.015 & 1.014 & 1.014 & 1.015 \\
Hedge & \textbf{1.009} & \textbf{1.009} & \textbf{1.009}
& \textbf{1.009} & \textbf{1.009} & \textbf{1.009} \\
DecHedge & 1.010 & 1.010 & 1.010
& 1.010 & 1.010 & 1.010 \\
AdaHedge & \ul{1.010} & \ul{1.009} & \ul{1.009} & \ul{1.010} & \ul{1.009} & \ul{1.009} \\

\addlinespace[2pt]
\multicolumn{7}{@{}l@{}}{\footnotesize
Panel~D: \textit{Large subset regime}
~\raisedrule[0.2em]{0.2pt}} \\
\addlinespace[0pt]
Median & 1.782 & 1.138 & 1.074 & 1.460 & 1.114 & 1.066 \\
\addlinespace[2pt]
SA & 1.689 & 1.124 & 1.067 & 1.410 & 1.104 & 1.060 \\
RollMSE & 1.588 & 1.119 & 1.065 & 1.388 & 1.102 & 1.060 \\
\addlinespace[4pt]
FTL & \textbf{1.006} & \textbf{1.007} & \textbf{1.008}
& \textbf{1.006} & \textbf{1.007} & \textbf{1.008} \\
Hedge & 1.009 & 1.009 & 1.008
& 1.009 & 1.009 & 1.008 \\
DecHedge & \ul{1.007} & 1.009 & 1.011
& \ul{1.007} & 1.010 & 1.011 \\
AdaHedge & 1.008 & \ul{1.008} & \ul{1.008} & 1.007 & \ul{1.007} & \ul{1.008} \\
\bottomrule
\end{tabular*}

\caption{Relative out-of-sample MSFE with the oracle included among the experts ($\textsc{oracle}=\text{Y}$). Each entry is the mean MSFE of the indicated combination method divided by the mean MSFE of the ex-post realized best individual expert over the corresponding designs and Monte Carlo replications. The expert pool contains $K=p+1$ experts. Panel~A: averages over the three subset regimes and $p\in\{10,100,200\}$; Panels B--D: averages across $p$ within subset regime. The DGP specification and $T_{\mathrm{train}}$ are fixed within each column with $B=500$ per design. Constant Hedge uses $\eta=0.5$; decreasing Hedge uses $c_0=2$. Values below one indicate that the method achieves lower mean MSFE than the ex-post realized best individual expert. Within each Panel and column, best performing
combinations are displayed in bold; second-best performing combinations are underlined.}
\label{tab:mc-raw-results-oracle}
\end{table}

\clearpage

\begin{table}[t!]
\centering
\footnotesize
\setlength{\tabcolsep}{0pt}
\setlength{\extrarowheight}{1pt}
\renewcommand{\arraystretch}{1}
\linespread{1}\selectfont

\begin{tabular*}{\linewidth}{@{\extracolsep{\fill}}>{\centering\arraybackslash}m{.28\linewidth}*{6}{>{\centering\arraybackslash}m{.12\linewidth}}@{}}
\toprule
& \multicolumn{3}{c}{Base DGP}
& \multicolumn{3}{c}{Robust DGP} \\
\cmidrule(lr){2-4}\cmidrule(l){5-7}
Method / $T_{\mathrm{train}}$
& $100$ & $300$ & $500$ & $100$ & $300$ & $500$ \\
\midrule

\multicolumn{7}{@{}l@{}}{\footnotesize
Panel~A: \textit{Average across subset regimes}
~\raisedrule[0.2em]{0.2pt}} \\
\addlinespace[0pt]
Median & 1.023 & 1.005 & 1.009 & 0.991 & 1.002 & 1.006 \\
\addlinespace[2pt]
SA & 1.005 & 0.999 & 1.005 & \ul{0.980} & 0.997 & 1.003 \\
RollMSE & \textbf{0.981} & \textbf{0.989} & 0.995
& \textbf{0.979} & \ul{0.994} & 0.999 \\
\addlinespace[4pt]
FTL & 1.028 & 1.028 & 1.028 & 1.031 & 1.030 & 1.031 \\
Hedge & 1.012 & 1.011 & 1.010 & 1.012 & 1.011 & 1.011 \\
DecHedge & 0.998 & 0.995 & \ul{0.995}
& 0.996 & 0.994 & \ul{0.994} \\
AdaHedge & \ul{0.982} & \ul{0.991} & \textbf{0.992}
& 0.982 & \textbf{0.990} & \textbf{0.991} \\

\addlinespace[2pt]
\multicolumn{7}{@{}l@{}}{\footnotesize
Panel~B: \textit{Small subset regime}
~\raisedrule[0.2em]{0.2pt}} \\
\addlinespace[0pt]
Median & 0.982 & 0.979 & 0.979 & 0.969 & 0.972 & 0.973 \\
\addlinespace[2pt]
SA & \ul{0.974} & \ul{0.974} & \ul{0.976}
& \ul{0.962} & \ul{0.968} & \ul{0.970} \\
RollMSE & \textbf{0.953} & \textbf{0.963} & \textbf{0.966}
& \textbf{0.961} & \textbf{0.967} & \textbf{0.969} \\
\addlinespace[4pt]
FTL & 1.033 & 1.031 & 1.031 & 1.036 & 1.035 & 1.036 \\
Hedge & 1.014 & 1.014 & 1.014 & 1.015 & 1.015 & 1.015 \\
DecHedge & 0.997 & 0.996 & 0.996 & 0.996 & 0.995 & 0.995 \\
AdaHedge & 0.983 & 0.988 & 0.990
& 0.981 & 0.986 & 0.987 \\

\addlinespace[2pt]
\multicolumn{7}{@{}l@{}}{\footnotesize
Panel~C: \textit{Mixed subset regime}
~\raisedrule[0.2em]{0.2pt}} \\
\addlinespace[0pt]
Median & 0.930 & 1.029 & 1.053 & 0.942 & 1.035 & 1.055 \\
\addlinespace[2pt]
SA & \ul{0.921} & 1.023 & 1.048 & \textbf{0.934} & 1.030 & 1.051 \\
RollMSE & \textbf{0.901} & 1.007 & 1.029
& \ul{0.936} & 1.021 & 1.039 \\
\addlinespace[4pt]
FTL & 1.039 & 1.033 & 1.034 & 1.040 & 1.035 & 1.035 \\
Hedge & 1.013 & 1.010 & 1.009 & 1.013 & 1.010 & 1.009 \\
DecHedge & 0.993 & \ul{0.991} & \ul{0.989}
& 0.992 & \ul{0.990} & \ul{0.989} \\
AdaHedge & 0.959 & \textbf{0.982} & \textbf{0.984}
& 0.962 & \textbf{0.982} & \textbf{0.983} \\

\addlinespace[2pt]
\multicolumn{7}{@{}l@{}}{\footnotesize
Panel~D: \textit{Large subset regime}
~\raisedrule[0.2em]{0.2pt}} \\
\addlinespace[0pt]
Median & 1.151 & 1.016 & 1.005 & 1.066 & 1.008 & 1.002 \\
\addlinespace[2pt]
SA & 1.116 & 1.008 & 1.001 & 1.046 & 1.002 & 0.998 \\
RollMSE & 1.084 & 1.005 & \ul{0.999} & 1.043 & \ul{1.002} & \ul{0.998} \\
\addlinespace[4pt]
FTL & 1.013 & 1.018 & 1.018 & 1.016 & 1.019 & 1.019 \\
Hedge & 1.008 & 1.007 & 1.006 & 1.008 & 1.006 & 1.006 \\
DecHedge & \textbf{1.002} & \textbf{0.998} & \textbf{0.998}
& \textbf{1.001} & \textbf{0.998} & \textbf{0.998} \\
AdaHedge & \ul{1.003} & \ul{1.003} & 1.003
& \ul{1.002} & 1.003 & 1.003 \\
\bottomrule
\end{tabular*}

\caption{Relative out-of-sample MSFE with the oracle excluded from the experts ($\textsc{oracle}=\text{N}$). Each entry is the mean MSFE of the indicated combination method divided by the mean MSFE of the ex-post realized best individual expert over the corresponding designs and Monte Carlo replications. The expert pool contains $K=p$ experts. Panel~A: averages over the three subset regimes and $p\in\{10,100,200\}$; Panels B--D: averages across $p$ within subset regime. The DGP specification and $T_{\mathrm{train}}$ are fixed within each column with $B=500$ per design. Constant Hedge uses $\eta=0.5$; decreasing Hedge uses $c_0=2$. Values below one indicate that the method achieves lower mean MSFE than the ex-post realized best individual expert. Within each Panel and column, best performing
combinations are displayed in bold; second-best performing combinations are underlined.}
\label{tab:mc-raw-results-no_oracle}
\end{table}

\clearpage
\newpage

\subsection{Combination with Simulated AR(1) Losses}
\label{appendix:simulation_ar1_losses}

The regret bounds in Section~\ref{section:expert_ensembles} depend on two key quantities in the time series setting, namely, the sub-optimality gap $\Delta$ and measures of temporal dependence in the expert losses. These quantities are generally unknown in applications, and estimating dependence coefficients is itself a challenging statistical problem (see, for example, \cite{khaleghiInferringMixingProperties2023}). Therefore, the numerical values of the theoretical regret bounds are difficult to assess in practice.

In this appendix, we perform Monte Carlo simulations in which we control the population loss gaps and the temporal dependence of the expert losses. We generate expert losses as bounded AR(1) processes and evaluate Follow-the-Leader (FTL) and decreasing Hedge across different levels of persistence and sub-optimality gaps, while keeping each expert's loss variance fixed. To distinguish the second-best expert from the remaining experts, we introduce an additional positive gap between their expected losses. We then compare the mean cumulative regret of both methods with their respective  Hoeffding- and Bernstein-type upper bounds in Theorems~\ref{theorem:ftl_regret_mixing} and~\ref{theorem:hedge_regret_mixing}, respectively.

\subsubsection{Data-generating Process}
\label{subsec:data_generating_process}

For each expert $k \in [K]$, losses follow the AR(1) process
\begin{equation}
    \ell_t^{(k)}-\mu_k
    =q\bigl(\ell_{t-1}^{(k)}-\mu_k\bigr)
    +a\sqrt{1-q^2}\,U_t^{(k)},\qquad
    U_t^{(k)}\overset{\mathrm{i.i.d.}}{\sim}
    \operatorname{Unif}[-\sqrt{3},\sqrt{3}],
\label{eq:ar1_losses_dgp}
\end{equation}
where $0\leq q<1$ controls temporal persistence, $a>0$ is the common stationary standard deviation, and $\mu_k$ is the expected loss of expert $k$. The innovations are assumed independent across experts, time, and Monte Carlo replications. 
The initial states are obtained by truncating the stationary representation to 65 innovations, resulting in a negligible approximation error.\footnote{The stationary representation is $\ell_t^{(k)}=\mu_k+a\sqrt{1-q^2}\sum_{j=0}^{\infty}q^jU_{t-j}^{(k)}$. Truncating the sum at $j=64$ yields an approximation error bounded by $aA_q q^{65}\leq 1.04\times10^{-23}$ over the design grid, where $A_q=\sqrt{3}\sqrt{(1+q)/(1-q)}$.}

Under stationarity, we have
\begin{equation*}
    \mathbb E[\ell_t^{(k)}]=\mu_k,\qquad
    \operatorname{Var}(\ell_t^{(k)})=a^2,\enspace \text{and} \enspace 
    \operatorname{Corr}(\ell_t^{(k)},\ell_{t+n}^{(k)})=q^n, \qquad n\geq0.
\label{eq:ar1-c035-moments}
\end{equation*}
We specify the population mean losses as
\begin{equation*}
    \mu_1=\mu^\star,\qquad
    \mu_2=\mu^\star+\Delta_1,\qquad
    \mu_k=\mu^\star+\Delta_2\quad(k=3,\ldots,K),\qquad
    \Delta_2=\Delta_1+\delta,
\end{equation*}
where $\Delta_1, \delta>0$. Expert $k=1$ has the lowest expected loss, expert $k=2$ is the population second-best expert, and the remaining $K-2$ experts have an additional expected-loss gap of $\delta$ relative to the second-best expert and of $\Delta_2$ relative to the best expert. The minimum sub-optimality gap is $\Delta=\min_{k\ne1}(\mu_k-\mu_1)=\Delta_1$.

In the following, we further ensure that our parametrization of the loss process ensures that it lies in the interval $[0, 1]$, so that we do not need to apply any clipping or normalization.

\subsubsection{Dependence and \texorpdfstring{$\varphi$}{phi}-mixing}

Parameter $q$ controls temporal persistence, and setting $q=0$ implies that losses are independent over time. For $0 \leq q < 1/2$, we derive explicit geometric upper bounds on the $\varphi$-mixing coefficients $\{\varphi_n^{(k)}\}_{n\geq1}$ of the individual loss processes $\{\ell_t^{(k)}\}_t$.\footnote{The requirement $q < 1/2$ comes from using a one-step Doeblin minorization argument; a more refined multi-step minorization, which we do not pursue here, would allow us to extend the calculations to $q \in [1/2, 1)$.}
These bounds imply summability and square-root summability of the corresponding mixing coefficients. 

We also derive upper bounds on the mixing coefficients $\{\varphi_n^{\chi,k}\}_{n\geq1}$ of the excess-loss processes
\begin{equation*}
    \chi_t^{(k)}
    = \ell_{t}^{(k)} - \ell_{t}^{(1)}
    =(\mu_k-\mu_1) + a \bigl( Z_{t}^{(k)} - Z_{t}^{(1)} \bigr), \quad k=2,\ldots, K,
\end{equation*}
where $Z_t^{(k)}$ is defined in \eqref{eq:ar1-c035-z}. These coefficients enter the Hedge regret guarantees.

\paragraph{Mixing upper bounds.}

We introduce the centered and rescaled loss process 
\begin{equation}
    Z_{t}^{(k)} := \frac{ \ell^{(k)}_t - \mu_k }{ a } ,
\label{eq:ar1-c035-z}
\end{equation}
which, following \eqref{eq:ar1_losses_dgp}, fulfills 
\begin{equation*}
    Z_{t}^{(k)} = q Z_{t-1}^{(k)} + \sqrt{1-q^2}\, U_{t}^{(k)} ,
    \qquad 
    U_{t}^{(k)} \stackrel{\mathrm{i.i.d.}}{\sim} \mathrm{Unif}[-\sqrt3,\sqrt3].
\end{equation*}
Note that, under the stationarity assumption, 
$ \ell_t^{(k)}-\mu_k = a \sqrt{1-q^2}\sum_{j=0}^{\infty}q^j U_{t-j}^{(k)} $,
and it is simple to show that
\begin{equation}
    |\ell_t^{(k)}-\mu_k| \leq a \sqrt{3} \sqrt{\frac{1+q}{1-q}}.
\label{eq:ar1_support_ell-mu}
\end{equation}
Hence,
\begin{equation}
    |Z_{t}^{(k)}| \le A_q,
    \quad\text{where}\quad
    A_q := \sqrt{3} \sqrt{\frac{1+q}{1-q}}.
\label{eq:ar1_support_z}
\end{equation}
Since the transformation  $\ell_t^{(k)}=\mu_k+aZ_t^{(k)}$ is invertible, the processes  $\{\ell_t^{(k)}\}_t$ and $\{Z_t^{(k)}\}_t$ have identical $\phi$-mixing coefficients.

Set $I_q := [-A_q , \, A_q]$, so that, by \eqref{eq:ar1_support_z}, $Z_t^{(k)} \in I_q$ for all $t \in \Int$.
 We also have that, for any $z \in I_q$,
\begin{equation*}
    Z_t^{(k)} \mid (Z_{t-1}^{(k)} = z) \;\in\; I_q(z) ,
    \quad\text{where}\quad 
    I_q(z) \,:=\, \bigl[ \, q z - b_q , \; q z + b_q \, \bigr], \ b_q := \sqrt{3 (1 - q^{2})}.
\end{equation*}
Moreover, $Z_t^{(k)}$ is uniformly distributed over $I_q(z)$ conditional on $Z_{t-1}^{(k)} = z$, since $U_{t}^{(k)}$ is uniformly distributed.
The process $\{ Z_t^{(k)} \}_{t \in \Int}$ is a time-homogeneous Markov chain on the state space $( I_q, \mathcal{B}({I}_q) )$, with $\mathcal{B}({I}_q)$ a Borel $\sigma$-algebra on $I_q$ and the transition kernel $\P_q$ such that $\mathbb{P}_q(z,F)$ is a regular version of the conditional probability $\mathbb{P}(Z_t^{(k)}\in F\mid Z_{t-1}^{(k)}=z)$ for all $F\in\mathcal{B}({I}_q)$~\citep{doukhan1995mixing}. 
The transition kernel $\P_q$ is given by
\begin{equation*}
    \P_q( z, F )
    = \frac{ \lambda({ F \cap I_q(z) }) }{ \lambda({ I_q(z) }) }
    = \frac{ \lambda({ F \cap I_q(z) }) }{ 2 b_q } ,
\end{equation*}
for every $z\in I_q$ and $F\in\mathcal B(I_q)$, where $\lambda$ denotes the Lebesgue measure.

Introduce
\begin{equation*}
    C_q := \bigcap_{z \in I_q} I_q(z)
    = \bigl[ q A_q - b_q, \ -q A_q + b_q \bigr].
\end{equation*}
For $q < 1/2$, one has $q A_q < b_q$, so $C_q$ has positive width $\lambda({ C_q }) = 2 (b_q - q A_q) > 0$. Let then $\nu_q$ be the uniform probability measure over $C_q$. Since $C_q \subseteq I_q(z)$ for all $z \in I_q$, we have
\begin{equation*}
    \P_q( z, F ) 
    = \frac{ \lambda({ F \cap I_q(z) }) }{ 2 b_q }
    \geq \frac{ \lambda({ F \cap C_q }) }{ 2 b_q } 
    = \frac{ \lambda({ C_q }) }{ 2 b_q } \, \frac{ \lambda({ F \cap C_q }) }{ \lambda({ C_q }) } 
    = \epsilon_q \nu_q(F) ,
\end{equation*}
where 
\begin{equation*}
    \epsilon_q := \frac{\lambda( C_q )}{2 b_q}
    = 1 - \frac{q A_q}{b_q}
    = 1 - \frac{q}{1 - q}
    = \frac{1 - 2 q}{1 - q} > 0 
\end{equation*}
is the one-step Doeblin minorization constant. Hence, for $0<q<1/2$, the Markov transition kernel $\P_q$ satisfies the one-step Doeblin minorization condition
\begin{equation}
    \P_q( z, F ) \geq \epsilon_q \nu_q(F) ,
\label{eq:ar1-c035-minorization}
\end{equation}
for every $z\in I_q$ and $F\in\mathcal B(I_q)$.

We can now apply the Doeblin recurrence theorem (cf. Theorem~1 in \citealp{doukhan1995mixing}, Chapter~2.4) to obtain
\begin{equation}
    \varphi_n^{(k)} \leq 2 r_q^n ,
\label{eq:ar1-c035-mixing}
\end{equation}
where $r_q := 1-\epsilon_q = {q}/({1-q})$.

Fix now $k \neq 1$ and consider the pair
\begin{equation*}
    \overline{Z}_t^{(k)} := \bigl( Z_t^{(k)} , \, Z_t^{(1)} ,\bigr)
\end{equation*}
which has product transition kernel $\P_q \otimes \P_q$ since innovations $U_t^{(k)}$ and $U_t^{(1)}$ are independent, and $\P_q$ is common across all $k \in [K]$. From \eqref{eq:ar1-c035-minorization}, for all $z_k,z_1\in I_q$ and $F_k,F_1\in\mathcal{B}(I_q)$, it holds
\begin{equation*}
    (P_q \otimes P_q) \bigl( (z_k, z_1), \, F_k \times F_1 \bigr)
    \;\ge\; \epsilon_q^{2} \, (\nu_q \otimes \nu_q) (F_k \times F_1) ,
\label{eq:ar1-c035-product-minorization}
\end{equation*}
which, by application of the monotone class theorem (see e.g. \cite{kallenbergFoundationsModernProbability2021}, Chapter 1), extends to all measurable sets on the product space $( I_q \times I_q, \mathcal{B}({I}_q) \otimes \mathcal{B}({I}_q) )$. 
Therefore, the minorization condition applied to $\{ \overline{Z}_t^{(k)} \}_{t \in \Int}$ with constant $\epsilon_q^{2} > 0$, meaning that 
\begin{equation*}
    \varphi_n\big( \{ \overline{Z}_t^{(k)} \}_{t \in \Int} \big) \leq 2 (1 - \epsilon_q^{2})^n ,
\end{equation*}
where $\{ \varphi_n ( \{ \overline{Z}_t^{(k)} \}_{t \in \Int} ) \}_{n \in \N}$ are the $\varphi$-mixing coefficients of the pairs process.
Since $\chi_t^{(k)}$ is a measurable function of $\overline{Z}_t^{(k)}$, its $\varphi$-mixing coefficients, $\{ \overline{\varphi}_n^{(k)} \}_{n \in \N}$ are upper bounded by 
\begin{equation}
    \overline{\varphi}_n^{(k)} 
    \leq \varphi_n\big( \{ \overline{Z}_t^{(k)} \}_{t \in \Int} \big)
    \leq 2 \overline{r}_q^n ,
\label{eq:ar1-c035-mixing-xi}
\end{equation}
where we define $\overline{r}_q := 1 - \epsilon_q^{2}$.

\paragraph{Regret dependence coefficients.}
\begingroup
\emergencystretch=1em
We evaluate the dependence constants in Theorems~\ref{theorem:ftl_regret_mixing}(ii) and~\ref{theorem:hedge_regret_mixing}(ii) at $\Delta=\Delta_1$, $v_{\max}=a^2$, and $\widetilde v_{\max}=2a^2$. From \eqref{eq:ar1-c035-mixing}, we obtain
\par
\endgroup

\begin{equation}
\begin{aligned}
\overline{\theta}_{\textsc{ftl}}^H &\leq 1+8\sum_{n=1}^{\infty}r_q^n=1+\frac{8r_q}{1-r_q}=:\theta_{\textsc{ftl}}^H(q),\\
\overline{\rho}_{\textsc{ftl}}^B &\leq 1+\sqrt{2}\sum_{n=1}^{\infty}r_q^{n/2}=1+\sqrt{2}\frac{\sqrt{r_q}}{1-\sqrt{r_q}}=:\rho_{\textsc{ftl}}^B(q).
\end{aligned}
\label{eq:ar1-c035-bound-constants_ftl}
\end{equation}
Similarly, using \eqref{eq:ar1-c035-mixing-xi}, we obtain
\begin{equation}
\begin{aligned}
\overline{\theta}_{\textsc{hdg}}^H &\leq 1+8\sum_{n=1}^{\infty}\overline{r}_q^n=1+\frac{8\overline{r}_q}{1-\overline{r}_q}=:\theta_{\textsc{hdg}}^H(q),\\
\overline{\rho}_{\textsc{hdg}}^B &\leq 1+\sqrt{2}\sum_{n=1}^{\infty}\overline{r}_q^{n/2}=1+\sqrt{2}\frac{\sqrt{\overline{r}_q}}{1-\sqrt{\overline{r}_q}}=:\rho_{\textsc{hdg}}^B(q).
\end{aligned}
\label{eq:ar1-c035-bound-constants_hedge}
\end{equation}

The decreasing-Hedge rate uses $c_0=2$, as required by Theorem~\ref{theorem:hedge_regret_mixing}. It is also important to emphasize that the final regret bounds depend on the minimum gap $\Delta_1$, but not on the separation gap $\delta$.

\begin{table}[t!]
\centering
\small
\setlength{\tabcolsep}{3pt}
\setlength{\extrarowheight}{1pt}
\linespread{1.2}\selectfont
\begin{tabularx}{\linewidth}{@{}p{.30\linewidth}p{.13\linewidth}>{\raggedright\arraybackslash}X@{}}
    \toprule
    Variable & Symbol & Specification\\
    \midrule
    \multicolumn{3}{@{}l@{}}{\footnotesize\textit{Design tuple} $\mathfrak D$~\raisedrule[.2em]{.2pt}}\\[3pt]
    AR persistence & $q$ & $0,\ 0.1,\ 0.2,\ 0.3,\ 0.4,\ 0.45$\\
    Sub-optimality gap & $\Delta_1$ & $0.01,\ 0.025,\ 0.05,\ 0.1,\ 0.2,\ 0.3$ \\[4pt]
    \multicolumn{3}{@{}l@{}}{\footnotesize\textit{Fixed variables}~\raisedrule[.2em]{.2pt}}\\[1pt]
    Minimum expected loss & $\mu^\star$ & $0.3$ \\
    Runner-up separation gap & $\delta$ & $0.1$, with $\Delta_2=\Delta_1 + \delta$.\\
    Loss innovation scale & $a$ & $0.3 \sqrt{ {0.55}/({3 \cdot 1.45}) } \approx 0.10667$ \\
    Number of experts & $K$ & 100: one population-best expert, one runner-up, and 98 experts with higher expected loss.\\
    Evaluation length & $T$ & 1000\\
    MC replications per tuple & $B$ & 500\\
    Decreasing Hedge constant & $c_0$ & $2$; {\color{green} $\eta_{t+1}=c_0\sqrt{\log K/(t+1)}$}; see Algorithm~\ref{algortihm:decreasing_hedge}.\\
    \bottomrule
\end{tabularx}
\caption{Monte Carlo design for bounded AR(1) expert losses.}
\label{tab:ar1-c035-design}
\end{table}

\paragraph{Experimental design.}

To apply \eqref{eq:ar1-c035-bound-constants_ftl} and \eqref{eq:ar1-c035-bound-constants_hedge}, in our simulations we choose $q$ within the interval $[0, 0.5)$. 
Given $\mu^\star$, $\Delta_1$ and $\delta$, let $\delta_\mu := \min_{k \in [K]} \min( \mu_k, 1 - \mu_k )$ be the minimum gap between the experts' expected loss and the endpoints of the unit interval.
Based on \eqref{eq:ar1_support_ell-mu}, to ensure that $\ell^{(k)}_t \in [0, 1]$ for all $t$ and $k \in [K]$, we require 
\begin{equation*}
    a \sqrt{3} \sqrt{\frac{1+q}{1-q}} \leq \delta_\mu ,
\end{equation*}
which implies that $0 \leq \mu_k - \delta_\mu \leq \ell^{(k)}_t \leq \mu_k + \delta_\mu \leq 1$ by construction.
Therefore, for a given design, choosing the innovation scaling factor $a$ as
\begin{equation*}
    a(\mu^\star, \Delta_1, \delta, q) := \delta_\mu \, \sqrt{ \frac{1-q}{3(1+q)} }
\label{eq:ar1_design_a}
\end{equation*}
ensures the range of all loss processes is within the unit interval. 

In our experimental design, we vary $q$ and $\Delta_1$ over a finite grid, while setting $\mu^* = 0.3$ and $\delta = 0.1$; we identify each specification by a design tuple $\mathfrak{D} = (q,\Delta_1)$, with the grids and fixed quantities reported in Table~\ref{tab:ar1-c035-design}.
We fix $a$ to be uniform across all design tuples: since $\Delta_1\leq0.3$ throughout the grid, we have $\mu_k\in[0.3,0.7]$ for all experts and hence $\delta_\mu=0.3$ for every design tuple. Moreover,  $a(\mu^\star, \Delta_1, \delta, q)$ is an decreasing function of $q$, and hence the choice
\begin{equation*}
    a = 0.3 \, \sqrt{ \frac{0.55}{3 \cdot 1.45} } \approx 0.10667 
\end{equation*}
ensures $aA_q\leq0.3=\delta_\mu$ for all $q\leq0.45$ and that therefore $\ell_t^{(k)}\in[0,1]$ for every $t$, $k$, and design tuple $\mathfrak{D} = (q,\Delta_1)$.

For each of the $36$ design tuples, we generate $B=500$ independent replications of $K=100$ expert loss sequences of length $T=1000$. Innovation draws are common across design tuples and combination methods.

\subsubsection{Results}

We report the results of our simulations and the corresponding theoretical regret bounds for selected design tuples in Figures~\ref{fig:ar1-c035-gap} and~\ref{fig:ar1-c035-q}, which show the Monte Carlo mean cumulative regret (upper Panels) and the theoretical upper bounds on expected regret (lower Panels) as functions of the sub-optimality gap $\Delta_1$ and persistence parameter $q$, respectively. 
Figure~\ref{fig:ar1-c035-gap} shows that, although our bounds qualitatively capture the increase in regret as $\Delta_1 \to 0$, they are extremely conservative. This reflects, for example,  the constants in the concentration inequalities and the upper bounds on mixing coefficients used in deriving Theorems~\ref{theorem:ftl_regret_mixing} and \ref{theorem:hedge_regret_mixing}.
A more interesting finding is that temporal dependence appears to have little effect on empirical regret in our simulation design. As Figure~\ref{fig:ar1-c035-gap} suggests and Figure~\ref{fig:ar1-c035-q} further illustrates, for each fixed value of $\Delta_1$, the mean empirical regret of both FTL and decreasing Hedge varies little over the considered range of $q$. In contrast, the theoretical bounds increase sharply with $q$, showing deterioration as $q$ approaches $1/2$. Hence, our theoretical guarantees substantially overstate the sensitivity of empirical regret to temporal dependence in this setting.

\clearpage

\begin{figure}[p!]
    \centering
    \begin{subfigure}[b]{0.48\textwidth}
        \includegraphics[width=\textwidth]{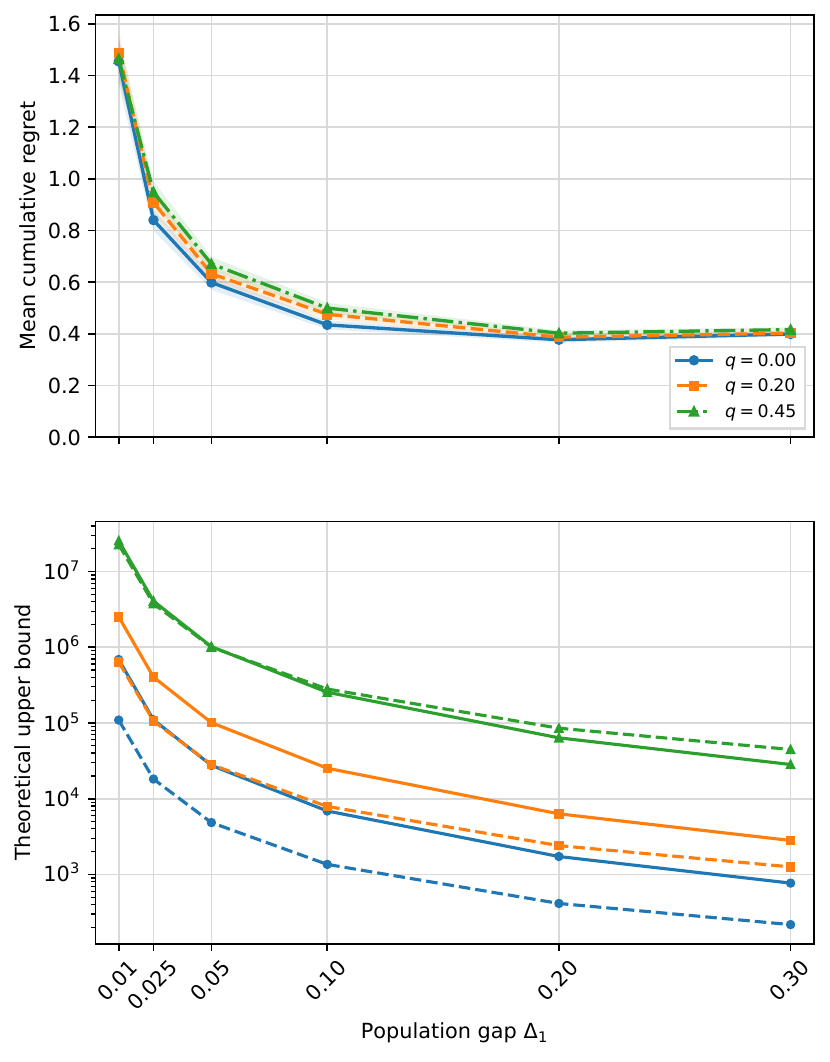}
        \caption{FTL}
    \end{subfigure}
    \hfill
    \begin{subfigure}[b]{0.48\textwidth}
        \includegraphics[width=\textwidth]{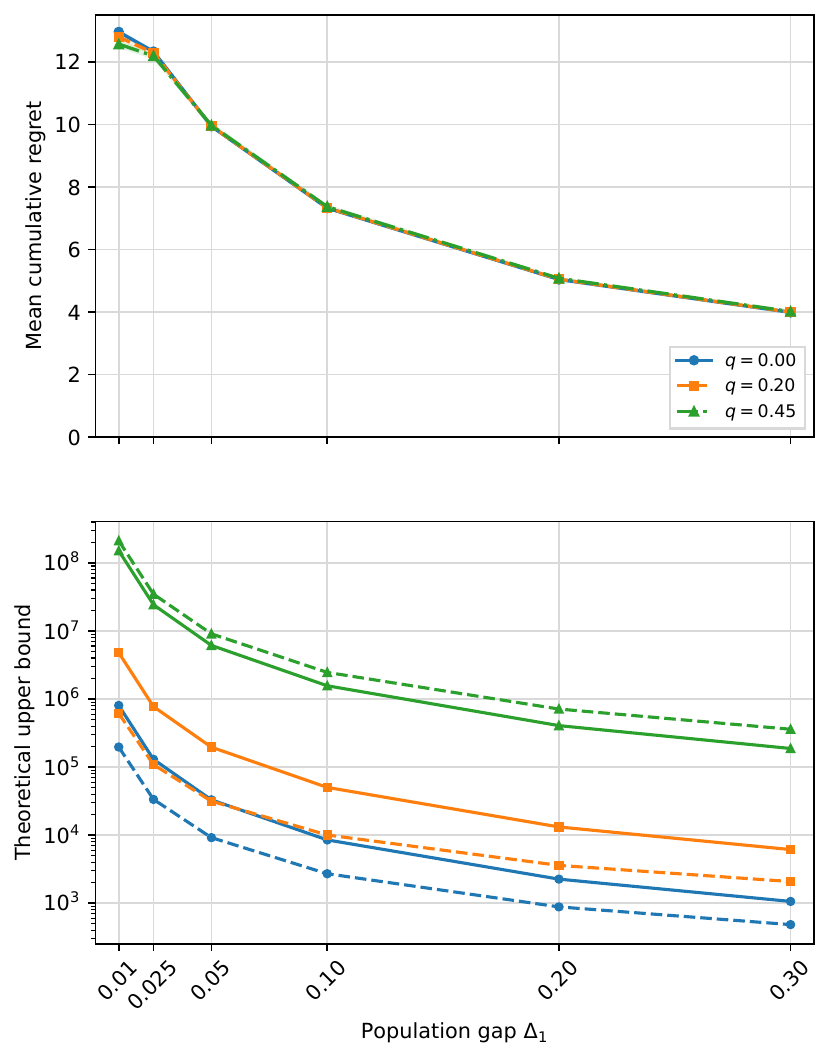}
        \caption{Decreasing Hedge}
    \end{subfigure}
    \caption{Cumulative regret and theoretical upper bounds as the population gap varies, for
$q\in\{0,0.20,0.45\}$.
The upper Panels show Monte Carlo mean regret; the lower Panels show
the mixing Hoeffding (solid) and Bernstein (dashed) bounds on a logarithmic
scale, using~\eqref{eq:ar1-c035-bound-constants_ftl}-\eqref{eq:ar1-c035-bound-constants_hedge}.
Regret is cumulative weighted-expert loss minus the smallest cumulative
loss among the $K$ experts in the same replication.
Each upper-Panel point averages regret over $B=500$ replications at a fixed
$(q,\Delta_1)$, with $K=100$ and $T=1000$.
Shading shows pointwise 95\% confidence intervals for the Monte Carlo mean.
The mixing bounds are used for all $q$, including $q=0$.}
    \label{fig:ar1-c035-gap}
\end{figure}

\clearpage

\begin{figure}[p!]
    \centering
    \begin{subfigure}[b]{0.48\textwidth}
        \includegraphics[width=\textwidth]{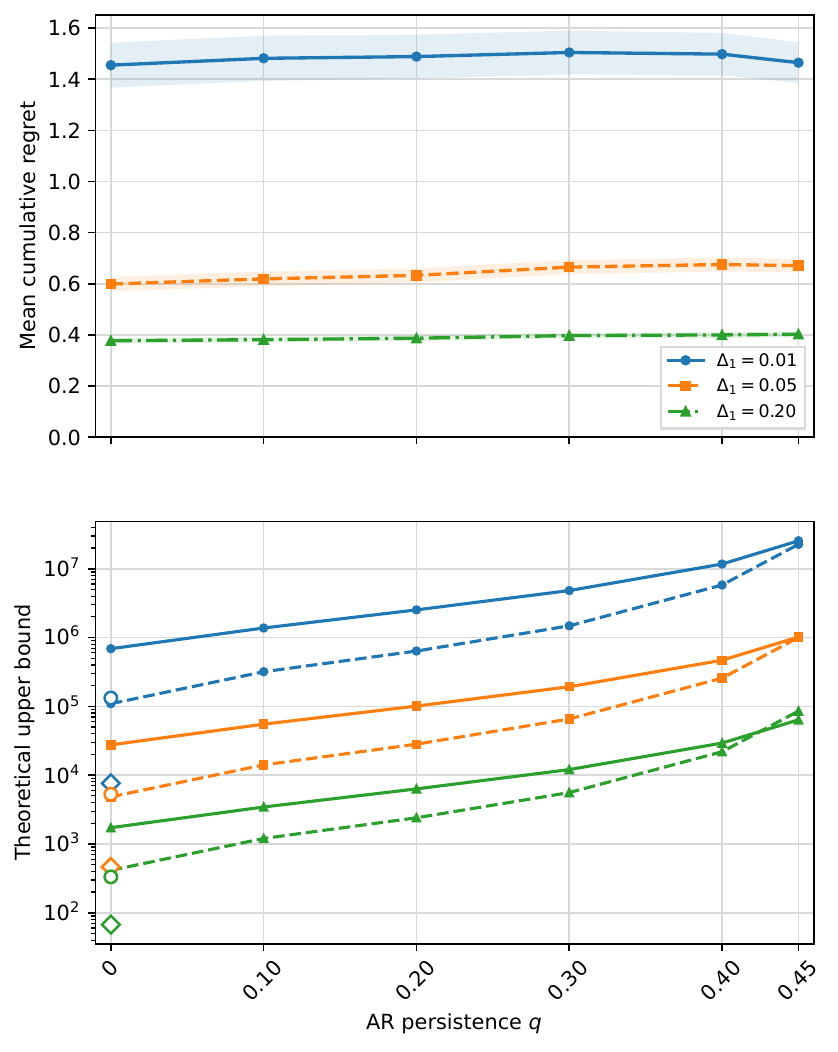}
        \caption{FTL}
    \end{subfigure}
    \hfill
    \begin{subfigure}[b]{0.48\textwidth}
        \includegraphics[width=\textwidth]{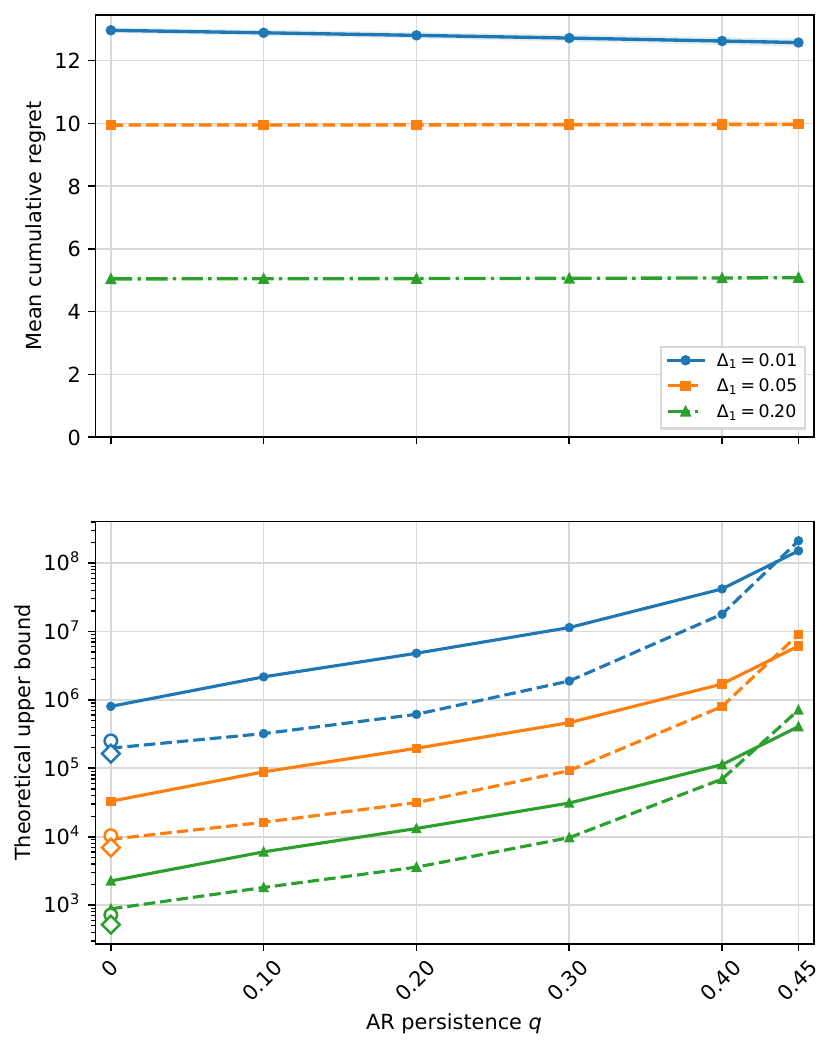}
        \caption{Decreasing Hedge}
    \end{subfigure}
    \caption{Cumulative regret and theoretical upper bounds as temporal persistence varies, for
$\Delta_1\in\{0.01,0.05,0.20\}$.
The upper Panels show Monte Carlo mean regret; the lower Panels show
the mixing Hoeffding (solid) and Bernstein (dashed) bounds on a logarithmic
scale, using~\eqref{eq:ar1-c035-bound-constants_ftl}-\eqref{eq:ar1-c035-bound-constants_hedge}.
Regret is cumulative weighted-expert loss minus the smallest cumulative
loss among the $K$ experts in the same replication.
Each upper-Panel point averages regret over $B=500$ replications at a fixed
$(q,\Delta_1)$, with $K=100$ and $T=1000$.
Shading shows pointwise 95\% confidence intervals for the Monte Carlo mean.
At $q=0$, hollow circles and diamonds additionally show the i.i.d.\ Hoeffding- and Bernstein-type bounds, respectively, from Theorem~\ref{theorem:ftl_regret_mixing}(i) for FTL and Propositions~\ref{proposition:hedge_hoeffding_regret_bound}(i) and~\ref{proposition:hedge_bernstein_regret_bound}(i) for decreasing Hedge.}
    \label{fig:ar1-c035-q}
\end{figure}

\color{black}

\end{document}